\pdfoutput=1
\documentclass[reprint,amsmath,amssymb,aps,prx,superscriptaddress,floatfix,longbibliography]{revtex4-2}

\usepackage{graphicx}
\usepackage{bm}
\usepackage{amsthm}
\usepackage{mathtools}
\usepackage{xcolor}
\usepackage{tikz}
\usetikzlibrary{arrows.meta,positioning,calc,shapes.geometric,decorations.pathreplacing}
\usepackage{xurl}
\usepackage{etoolbox}
\AtBeginEnvironment{thebibliography}{}
\usepackage[colorlinks=true,linkcolor=blue!60!black,citecolor=blue!60!black,urlcolor=blue!60!black,filecolor=blue!60!black,hypertexnames=false]{hyperref}
\newcommand{\papertitle}{Arithmetic of qudit dimensions governs erasure correction on the Fano plane and stabilizer-state entanglement}
\hypersetup{pdftitle={\papertitle},pdfauthor={Xidan Song and Ruifeng Cao}}

\newcommand{\Stab}{\mathrm{Stab}}
\newcommand{\StabZ}[2]{\mathrm{Stab}^{(#1)}_{#2}}
\newcommand{\Lc}{\mathcal{L}}
\newcommand{\LcZ}[2]{\mathcal{L}^{(#1)}_{#2}}
\newcommand{\Lab}{\mathcal{L}^{\mathrm{ab}}}

\newcommand{\QLR}{\mathrm{QLR}}
\newcommand{\GF}{\mathrm{GF}}
\newcommand{\Z}{\mathbb{Z}}
\newcommand{\R}{\mathbb{R}}
\newcommand{\Q}{\mathbb{Q}}

\newcommand{\pr}{\mathrm{pr}}
\newcommand{\Dual}{\mathsf{D}}
\newcommand{\cone}{\operatorname{cone}}

\newcommand{\rank}{\operatorname{rank}}

\newcommand{\hF}{h_{\mathrm{F}}}
\newcommand{\hNF}{h_{\mathrm{NF}}}
\newcommand{\rF}{r_{\mathrm{F}}}
\newcommand{\rNF}{r_{\mathrm{NF}}}

\newcommand{\ket}[1]{\lvert #1\rangle}
\newcommand{\bra}[1]{\langle #1\rvert}
\newcommand{\abs}[1]{\lvert #1\rvert}
\newcommand{\norm}[1]{\lVert #1\rVert}

\newcommand{\LA}{\textup{(LA)}}
\newcommand{\LB}{\textup{(LB)}}
\newcommand{\fnine}{f_{9}}

\newcommand{\eps}{\varepsilon}
\newcommand{\Ingl}{\mathrm{Ing}}
\newcommand{\ie}{i.e.}

\newcommand{\ind}[1]{\mathbf{1}_{#1}}
\newcommand{\ftwelve}{f_{12}}
\newcommand{\rk}{\operatorname{rk}}
\newcommand{\SNF}{S_{\mathrm{NF}}}
\newcommand{\SF}{S_{\mathrm{F}}}
\newcommand{\Rg}{\mathcal{R}}
\newcommand{\Kf}{\mathbb{K}}
\newcommand{\hcap}{h_{\cap}}
\newcommand{\reprelease}{R48}

\renewcommand{\S}{\ensuremath{\mathsection}}

\theoremstyle{plain}
\newtheorem{theorem}{Theorem}
\newtheorem{proposition}{Proposition}
\newtheorem{lemma}{Lemma}
\newtheorem{corollary}{Corollary}
\theoremstyle{definition}
\newtheorem{definition}{Definition}
\newtheorem{example}{Example}
\theoremstyle{remark}

\theoremstyle{plain}
\newtheorem{applemma}{Lemma}
\numberwithin{applemma}{section}
\newtheorem{appproposition}[applemma]{Proposition}

\begin{document}

\title{\papertitle}

\newcommand{\authornote}{These authors contributed equally to this work. Corresponding author: Ruifeng Cao, \href{mailto:ruifeng.cao@manchester.ac.uk}{\nolinkurl{ruifeng.cao@manchester.ac.uk}}. Xidan Song: \href{mailto:liuqingdamafen@whu.edu.cn}{\nolinkurl{liuqingdamafen@whu.edu.cn}}.}
\author{Xidan Song}
\thanks{\authornote}
\affiliation{School of Computer Science, Wuhan University, Wuhan 430072, China}

\author{Ruifeng Cao}
\thanks{\authornote}
\affiliation{School of Engineering, The University of Manchester, Manchester M13 9PL, United Kingdom}

\begin{abstract}
Qudit codes often inherit the parameters of qubit codes, but not necessarily their correctable erasures. We show that a code encoding one qudit of dimension $D$ into seven such qudits that corrects every erasure Steane's code corrects exists exactly when $D$ is a power of $2$; for the qutrit color code, exactly when $D$ is odd. This needs no stabilizer assumption: these codes are, up to local unitaries, stabilizer codes of qubits, respectively prime-power qudits. The two patterns differ in one line, which completes the Fano plane only where $1+1=0$. As ideal secret sharing schemes, these codes exist exactly where classical ones do: entanglement does not help, unlike for Euler's $36$ officers. The same arithmetic makes the entropy cones of qubit and odd-dimensional stabilizer states incomparable from six parties on. At five, the stabilizer cone is strictly smaller than the quantum linear rank cone, answering an open question.
\end{abstract}

\maketitle

\section{Introduction}
\label{sec:intro}

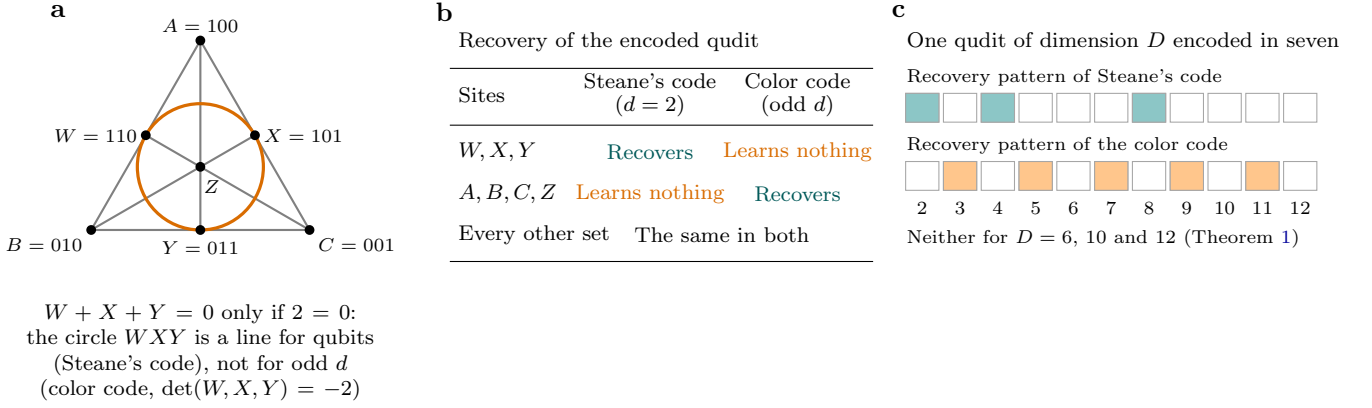
\begin{figure*}[t]
\centering
\begin{tikzpicture}[font=\footnotesize,
  pt/.style={circle, fill=black, inner sep=1.3pt},
  lab/.style={font=\scriptsize}]
\node[anchor=west, font=\normalsize\bfseries] at (-2.1,2.05) {a};
\begin{scope}[scale=1.15]
\coordinate (A) at (90:1.45);
\coordinate (B) at (210:1.45);
\coordinate (C) at (330:1.45);
\coordinate (W) at ($(A)!0.5!(B)$);
\coordinate (X) at ($(A)!0.5!(C)$);
\coordinate (Y) at ($(B)!0.5!(C)$);
\coordinate (Z) at (0,0);
\draw[thick, gray] (A)--(B)--(C)--cycle;
\draw[thick, gray] (A)--(Y); \draw[thick, gray] (B)--(X); \draw[thick, gray] (C)--(W);
\draw[very thick, orange!85!black] (Z) circle (0.725);
\foreach \p in {A,B,C,W,X,Y,Z} \node[pt] at (\p) {};
\node[lab, above] at (A) {$A=100$};
\node[lab, below left] at (B) {$B=010$};
\node[lab, below right] at (C) {$C=001$};
\node[lab, left] at (W) {$W=110$};
\node[lab, right] at (X) {$X=101$};
\node[lab, below=1pt] at (Y) {$Y=011$};
\node[lab, inner sep=0.5pt] at ($(Z)+(-60:0.26)$) {$Z$};
\end{scope}
\node[align=center, text width=4.7cm] at (0,-2.45) {$W+X+Y=0$ only if $2=0$:\\ the circle $WXY$ is a line for qubits\\ (Steane's code), not for odd $d$\\ (color code, $\det(W,X,Y)=-2$)};
\node[anchor=west, font=\normalsize\bfseries] at (3.0,2.05) {b};
\begin{scope}[xshift=3.3cm, yshift=1.2cm]
\node[anchor=west, align=left] at (0,0.45) {Recovery of the encoded qudit};
\draw (0,0.1) -- (5.6,0.1);
\node[anchor=west] at (0,-0.25) {Sites};
\node[align=center] at (2.65,-0.25) {Steane's code\\[-1pt] ($d=2$)};
\node[align=center] at (4.6,-0.25) {Color code\\[-1pt] (odd $d$)};
\draw (0,-0.65) -- (5.6,-0.65);
\node[anchor=west] at (0,-1.0) {$W,X,Y$};
\node[text=teal!70!black] at (2.65,-1.0) {Recovers};
\node[text=orange!85!black] at (4.6,-1.0) {Learns nothing};
\node[anchor=west] at (0,-1.55) {$A,B,C,Z$};
\node[text=orange!85!black] at (2.65,-1.55) {Learns nothing};
\node[text=teal!70!black] at (4.6,-1.55) {Recovers};
\node[anchor=west, align=left] at (0,-2.1) {Every other set};
\node[align=center] at (3.6,-2.1) {The same in both};
\draw (0,-2.45) -- (5.6,-2.45);
\end{scope}
\node[anchor=west, font=\normalsize\bfseries] at (9.03,2.05) {c};
\begin{scope}[xshift=9.33cm, yshift=1.2cm, cell/.style={draw=gray!70, minimum width=0.44cm, minimum height=0.38cm, inner sep=0pt}]
\node[anchor=west, align=left] at (-0.1,0.45) {One qudit of dimension $D$ encoded in seven};
\node[anchor=west, font=\scriptsize] at (-0.1,-0.02) {Recovery pattern of Steane's code};
\node[anchor=west, font=\scriptsize] at (-0.1,-0.92) {Recovery pattern of the color code};
\foreach \n in {2,...,12} {
  \node[cell] at ({(\n-2)*0.5+0.22},-0.42) {};
  \node[cell] at ({(\n-2)*0.5+0.22},-1.32) {};
  \node[font=\scriptsize] at ({(\n-2)*0.5+0.22},-1.75) {$\n$};
}
\foreach \n in {2,4,8} \node[cell, fill=teal!45] at ({(\n-2)*0.5+0.22},-0.42) {};
\foreach \n in {3,5,7,9,11} \node[cell, fill=orange!45] at ({(\n-2)*0.5+0.22},-1.32) {};
\node[anchor=west, font=\scriptsize] at (-0.1,-2.15) {Neither for $D=6$, $10$ and $12$ (Theorem~\ref{thm:ideal})};
\end{scope}
\end{tikzpicture}
\caption{Steane's code, its qudit counterpart and the arithmetic that separates them. \textbf{a}~The seven sites of Steane's code are the points of the Fano plane, the nonzero vectors of $\GF(2)^3$, with $Z=111$ at the center; its lines are the triples with zero sum, among them $WXY$ (orange circle). The seven-qudit color code of $\Z_d$ qudits uses the same sites. \textbf{b}~Which sets of sites recover the encoded qudit in Steane's code and in the color code of odd $d$ (teal: recovers; orange: learns nothing); the two differ only on $\{W,X,Y\}$ and $\{A,B,C,Z\}$. \textbf{c}~Dimensions $D\le12$ in which a code that encodes one qudit of dimension $D$ into seven such qudits, stabilizer code or not, can have the recovery pattern of Steane's code (teal: powers of~$2$) or that of the color code (orange: odd~$D$).}
\label{fig:steane}
\end{figure*}

Quantum codes for qudits, $d$-level quantum systems, are often obtained by carrying qubit codes over, as for the color codes of qudits~\cite{Watson2015}, the smallest of which keeps the length, rate and distance of Steane's code (Sec.~\ref{sec:qss}). These parameters do not determine which sets of erasures a code corrects: qubit stabilizer codes with the parameters of Steane's code cannot correct between $7$ and $23$ of the $35$ erasures of three sites, depending on the code (Sec.~\ref{sec:qss}). That finer structure, the recovery pattern, decides which losses of sites a code survives, and in quantum secret sharing~\cite{CleveGottesmanLo1999,Gottesman2000} it is prescribed: it fixes which sets of shares recover the secret and which learn nothing about it. Whether a recovery pattern survives a change of dimension is a question about all codes, not only stabilizer codes: nonadditive codes can encode more than every stabilizer code of the same length and distance~\cite{RHSS1997}. And entanglement can make possible what classical combinatorics forbids: no pair of orthogonal Latin squares of order~$6$ exists (Euler's problem of $36$ officers), yet an absolutely maximally entangled state of four qudits of dimension $6$ does; with one qudit as the reference, it is a code that encodes one qudit of dimension $6$ into three, any two of which recover it~\cite{Rather2022}, and it is not a stabilizer code up to local unitaries (Sec.~\ref{sec:qss}).

Here we show that for the recovery patterns of Steane's code and of its qudit counterpart, the seven-qudit color code, entanglement offers no such escape, and that the arithmetic of the dimension, rather than its size, decides whether they survive. If a reference system $R$ is maximally entangled with the encoded information of a code, a set $T$ of sites can recover it exactly when $I(R;T)=2S(R)$ and learns nothing about it exactly when $I(R;T)=0$~\cite{Imai2005}, and the erasure of a set of sites is correctable exactly when the set learns nothing. Steane's seven-qubit code corrects exactly the erasures of the sets of qubits that contain no line of the Fano plane, whose points are the seven nonzero vectors of $\GF(2)^3$, with $\GF(2)$ the field of two elements, and whose lines are the triples with zero sum (Fig.~\ref{fig:steane}a). The seven-qudit color code of odd local dimension recovers the encoded qudit from the same sets of sites, except that the three sites $110$, $101$ and $011$ learn nothing and the other four sites recover it (Fig.~\ref{fig:steane}b); both codes have distance $3$. A code that encodes one qudit of dimension $D$ into seven qudits of the same dimension and corrects every erasure that Steane's code corrects exists, stabilizer code or not, if and only if $D$ is a power of $2$, and one that corrects every erasure that the color code of odd dimension corrects if and only if $D$ is odd; such codes have exactly these recovery patterns and are, up to local unitaries, stabilizer codes, of qubits and of qudits of odd dimensions, respectively (Theorem~\ref{thm:ideal}; Fig.~\ref{fig:steane}c), and for $D=2$ every code with Steane's pattern is Steane's code up to local unitaries (Sec.~\ref{sec:qss}). Ideal classical secret sharing schemes for these access structures exist in exactly the same dimensions, by results of Mat\'u\v{s}~\cite{Matus1999} (Supplementary Note~\ref{app:qss-general}). Theorem~\ref{thm:ideal} is their quantum counterpart but does not follow from them, since quantum codes need not come from classical schemes: for the $36$ officers a quantum code exists although the corresponding classical object, a pair of orthogonal Latin squares, does not.

One line of the Fano plane decides. The sites $110$, $101$ and $011$ are collinear because $110+101=011$, which needs $1+1=0$; in the color code of odd dimension they are not (Fig.~\ref{fig:steane}a). The proof, which uses no computer, turns this into a statement about all codes in three stages (Sec.~\ref{sec:qss}). Put the reference at the origin of the cube $\{0,1\}^3$ and the seven sites at the other vertices, the points of the Fano plane; a line of the Fano plane together with the origin is then a plane of the affine geometry $\mathrm{AG}(3,2)$ of the cube, a set of four vertices with zero sum, and so are the four sites off the line. First, when all systems have dimension $D$, the recovery pattern alone fixes the entanglement entropies of the sites and the reference: a theorem of Mat\'u\v{s} on polymatroids~\cite{Matus2012}, applied to a lift of the entropies, and Lehman's theorem on connected matroids~\cite{Lehman1964} show that they are given by the rank function of $\mathrm{AG}(3,2)$ for Steane's pattern and of the real affine cube $R_8$ for that of the color code. Second, these entropies force a stabilizer state up to local unitaries (Steps~1--5 of the proof in Supplementary Note~\ref{app:qss-general}). For Steane's pattern, the four pairs of systems along each direction of the cube form a Greenberger--Horne--Zeilinger (GHZ) state, by the theorem of Hayden, Jozsa, Petz and Winter on equality in strong subadditivity~\cite{HJPW2004}, and comparing two directions within each plane of $\mathrm{AG}(3,2)$ yields operators from nice error bases~\cite{Knill1996} that fix the state; together they generate a stabilizer group of order $D^8$. Third, in that stabilizer state the line through $110$, $101$ and $011$ forces every phase-space label $x$, the vector that indexes a Weyl operator, to satisfy $2x=0$, whereas without that line every label has odd order (Step~6).

Three consequences follow. First, for quantum error correction, with one qudit of the same dimension per site, codes outside the stabilizer formalism gain nothing for these patterns, and the arithmetic of the dimension decides which pattern a code can have: on qutrits, which qudit processors provide~\cite{Ringbauer2022} and whose error correction in an oscillator has been demonstrated~\cite{Brock2025}, a code that encodes one qutrit into seven qutrits can have the color code's pattern but not Steane's, and in dimension $6$ neither pattern is possible. With several qudits per site the obstruction disappears asymptotically for general codes, since, when many qudits are encoded, copies of Steane's code embedded into $\Z_d$ qudits, which are not stabilizer codes of $\Z_d$ qudits, approach $7$ qudits per encoded qudit; but stabilizer codes still pay for the arithmetic: whenever $d$ has an odd prime factor, those of $\Z_d$ qudits with Steane's recovery pattern need at least $28/3$ qudits per encoded qudit, and conversely a qubit stabilizer code with the recovery pattern of the color code of odd dimension needs, at some site, at least $4/3$ times as many qubits as it encodes (Proposition~\ref{prop:qss}; Sec.~\ref{sec:qss}). Second, for quantum secret sharing, entanglement does not enlarge the set of dimensions in which ideal schemes for these access structures exist, unlike for the $36$ officers; which other recovery patterns are rigid (Sec.~\ref{sec:qss}), and for which access structures entanglement helps, are open (Discussion). Third, the same arithmetic constrains the entanglement of stabilizer states, as we now explain.

Stabilizer states, which underlie quantum error correction~\cite{Gottesman1997,CalderbankShor1996,Steane1996} and measurement-based computation~\cite{Raussendorf2001,Hein2004}, exist for qudits of every local dimension $d$~\cite{Gottesman1999qudit,Hostens2005}. For $n$ parties and a purifying system, the entanglement entropies of all subsystems form an entropy vector; those of stabilizer states generate the stabilizer entropy cone, an inner approximation of the quantum entropy cone (Sec.~\ref{sec:reduction}). Both halves of the paper rest on the lift of Mat\'u\v{s}: for stabilizer states it gives a dictionary between their entropies and the ranks of classical codes (Proposition~\ref{thm:reduction}), which for prime $d$ assembles known facts~\cite{GrossWalter2013,LMRW2013,Osthus2024}. Through it, the same line of the Fano plane separates the entanglement of qubits from that of odd dimensions: from six parties on, no stabilizer state of odd local dimension reproduces, up to scaling and even in the limit, the entanglement entropies of the logical zero state of Steane's code with six of its qubits as parties and the seventh as the purifier (Theorem~\ref{thm:steane}), and for every $d\ge3$ no qubit stabilizer state reproduces those of a logical basis state of the color code of $\Z_d$ qudits (Theorem~\ref{thm:reverse}). With enough parties, $\Z_{d'}$ qudits reproduce all entropy patterns of $\Z_d$ qudits, their entropy vectors up to scaling and in the limit, if and only if $d$ divides $d'$ (Theorem~\ref{thm:divisibility}): $\Z_6$ qudits reproduce the patterns of qubits and of qutrits alike, qutrits and $\Z_4$ qudits are incomparable, and Galois qudits of dimension $4$, built on the field $\GF(4)$, reproduce exactly the qubit patterns, whereas $\Z_4$ qudits reproduce strictly more. This addresses a question of Osthus on what characteristic-dependent linear rank inequalities, which hold over fields of some characteristics but not of others, imply for qudits of different local dimensions~\cite[Sec.~4.4]{Osthus2024}: they make some entropy patterns realizable only in some local dimensions. For two prime local dimensions $p<q$ and at least $2p+3$ parties, this already follows through the dictionary from inequalities of Pe\~na-Macias and Sarria~\cite{PenaSarria2021}; at six parties, where the characteristic-dependent inequalities that we checked are blind, two transfer arguments give it (Sec.~\ref{sec:six}), and composite local dimensions need the inequalities of Lemma~\ref{lem:JI} (Methods).

Up to four parties nothing depends on the local dimension (Proposition~\ref{thm:four}). At five parties the dictionary turns the purifier into one more variable, and an inequality from a single common information (Sec.~\ref{sec:reduction}), which defines a facet of the stabilizer cone in every local dimension, shows that the stabilizer cone is strictly smaller than the quantum linear rank cone $\QLR_5$ of Bao, Cheng, Hern\'andez-Cuenca and Su (Theorem~\ref{thm:five}). This answers their question whether the two cones coincide~\cite{BCHS2020,BCHS2020gap}, on which they had reported ``significant evidence for and none against'' equality~\cite[Sec.~5]{BCHS2020gap}; published six-variable inequalities of Dougherty, Freiling and Zeger~\cite{DFZ2009} also separate the cones once the purifier is a variable. Figure~\ref{fig:landscape} summarizes what is known; the main question left open is whether the five-party cone depends on the local dimension.

Unless stated otherwise, the local dimension is the alphabet $\Z_d$ of the generalized Pauli group, and each party may hold any number of qudits. Except for Theorem~\ref{thm:ideal}, our separations and lower bounds concern stabilizer states only: embedding each qubit of the Steane state into a qutrit gives a qutrit state with the same entropies that is not a stabilizer state, so they are not dimension witnesses for general states. Section~\ref{sec:methods-prior} of the Methods details which parts of this paper are new and which rest on known results. Every computer-assisted step rests on an exact certificate, an exact recomputation or an exhaustive enumeration, and the data and code are public (Methods).

\begin{figure*}[t]
\centering
\begin{tikzpicture}[font=\footnotesize]
\colorlet{same}{teal!22}
\colorlet{diff}{orange!40}
\colorlet{undec}{gray!14}
\def\xx#1{5.6+1.06*(#1-1)}
\def\yy#1{-0.86*(#1-1)}
\newcommand{\band}[5]{
  \fill[#4] ({\xx{#2}-0.5},{\yy{#1}-0.34}) rectangle ({\xx{#3}+0.5},{\yy{#1}+0.34});
  \node[align=center, inner sep=0pt] at ({(\xx{#2}+\xx{#3})/2},{\yy{#1}}) {#5};}
\newcommand{\rowlabel}[2]{\node[anchor=east, align=right] at (4.95,{\yy{#1}}) {#2};}
\node[anchor=west, font=\normalsize\bfseries] at (-0.15,0.62) {a};
\rowlabel{1}{$\Stab_n$ versus $\QLR_n$}
\rowlabel{2}{$\StabZ{2}{n}\subseteq\StabZ{d}{n}$ for odd $d$?}
\rowlabel{3}{$\StabZ{d}{n}\subseteq\StabZ{2}{n}$ for odd $d$?}
\rowlabel{4}{$\StabZ{3}{n}$ versus $\StabZ{5}{n}$}
\rowlabel{5}{$\StabZ{4}{n}\subseteq\StabZ{2}{n}$?}
\band{1}{1}{4}{same}{Equal}
\band{1}{5}{5}{diff}{$\subsetneq$\\ {\scriptsize Thm.~\ref{thm:five}}}
\band{1}{6}{11}{undec}{Open (Methods)}
\band{2}{1}{4}{same}{Yes (Prop.~\ref{thm:four})}
\band{2}{5}{5}{undec}{Open}
\band{2}{6}{11}{diff}{No: Steane state (Thm.~\ref{thm:steane})}
\band{3}{1}{4}{same}{Yes (Prop.~\ref{thm:four})}
\band{3}{5}{5}{undec}{Open}
\band{3}{6}{11}{diff}{No: non-Fano state (Thm.~\ref{thm:reverse})}
\band{4}{1}{4}{same}{Equal (Prop.~\ref{thm:four})}
\band{4}{5}{7}{undec}{Open}
\band{4}{8}{11}{diff}{Incomparable\\ (Thm.~\ref{thm:primes})}
\band{5}{1}{4}{same}{Yes (Prop.~\ref{thm:four})}
\band{5}{5}{5}{undec}{Open}
\band{5}{6}{11}{diff}{No (Thm.~\ref{thm:reverse})}
\foreach \n in {1,...,11} \node at ({\xx{\n}},{\yy{5}-0.65}) {$\n$};
\node at ({\xx{11}+0.85},{\yy{5}-0.65}) {$\cdots$};
\node at ({(\xx{1}+\xx{11})/2},{\yy{5}-1.1}) {Number of parties $n$};
\node[anchor=west, font=\normalsize\bfseries] at (-0.15,-5.15) {b};
\draw[thick] (0.3,-6.8) -- (0.95,-6.8);
\node[anchor=west] at (1.05,-6.8) {Inclusion for all $n$, strict from $n=6$ on};
\draw[thick, dashed] (0.3,-7.4) -- (0.95,-7.4);
\node[anchor=west] at (1.05,-7.4) {Inclusion for all $n$, strict for larger $n$};
\begin{scope}[xshift=1.2cm, dn/.style={circle, draw, thick, fill=white, minimum size=0.62cm, inner sep=0pt}]
\node[dn] (d3)  at (6.8,-8.3) {$3$};
\node[dn] (d2)  at (9.0,-8.3) {$2$};
\node[dn] (d5)  at (12.2,-8.3) {$5$};
\node[dn] (d7)  at (13.4,-8.3) {$7$};
\node[dn] (d11) at (14.6,-8.3) {$11$};
\node[dn] (d9)  at (6.2,-7.0) {$9$};
\node[dn] (d6)  at (7.9,-7.0) {$6$};
\node[dn] (d4)  at (9.6,-7.0) {$4$};
\node[dn] (d10) at (12.2,-7.0) {$10$};
\node[dn] (d12) at (8.6,-5.7) {$12$};
\node[dn] (d8)  at (10.4,-5.7) {$8$};
\draw[thick] (d2)--(d4) (d2)--(d6) (d2)--(d10) (d3)--(d6) (d5)--(d10) (d4)--(d8) (d4)--(d12);
\draw[thick, dashed] (d3)--(d9) (d6)--(d12);
\node[font=\scriptsize] at (6.8,-8.85) {Qutrit};
\node[font=\scriptsize] at (9.0,-8.85) {Qubit, $\GF(2^k)$};
\end{scope}
\end{tikzpicture}
\caption{Dependence of stabilizer entropy cones on the number of parties and the local dimension. \textbf{a}~Each row compares two cones for $n$ parties: $\StabZ{d}{n}$, the stabilizer entropy cone of $\Z_d$ qudits; $\Stab_n$, that of all local dimensions together; and $\QLR_n$, the quantum linear rank cone (Sec.~\ref{sec:reduction} and Methods). Teal: equal, or the inclusion holds; orange: different, or the inclusion fails; gray: open; Thm., Theorem; Prop., Proposition. Row~1 for $n\le4$: Refs.~\cite{LMRW2013,BCHS2020}. \textbf{b}~Inclusions among the stabilizer cones of $\Z_d$ qudits for $2\le d\le12$ (Theorem~\ref{thm:divisibility}): with enough parties, the cone of $\Z_d$ qudits lies strictly inside that of $\Z_{d'}$ qudits exactly when an upward path leads from $d$ to $d'$, that is, when $d$ is a proper divisor of $d'$; cones joined by no such path are incomparable. Dashed edges are strict from $n=8$ ($3$ and $9$) and from $n=14$ ($6$ and $12$) on (Supplementary Note~\ref{app:cubes-divisibility}). Node~$2$ also stands for Galois qudits of dimension $2^k$ (Methods).}
\label{fig:landscape}
\end{figure*}
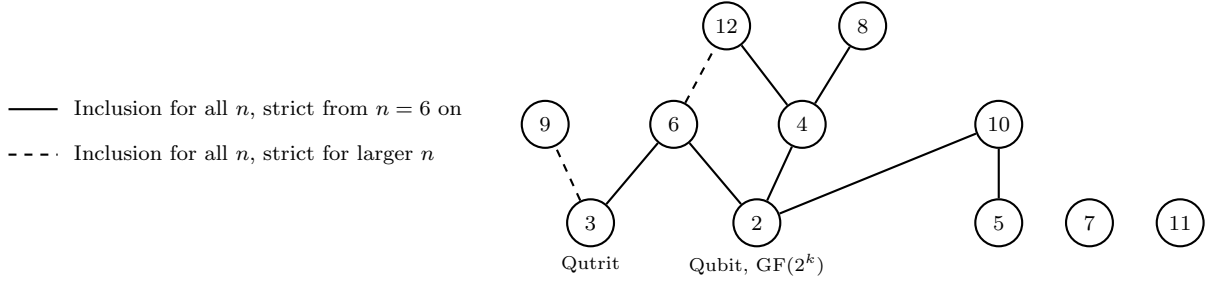

\section{Results}
\label{sec:results}
\subsection{Rigid recovery patterns and their cost in qudits}
\label{sec:qss}

Label seven sites by the nonzero points of $\GF(2)^3$: $A$, $B$ and $C$ are $100$, $010$ and $001$, then $W$, $X$ and $Y$ are $110$, $101$ and $011$, and $Z$ is $111$. The seven lines of the Fano plane are $ABW$, $ACX$, $BCY$, $AYZ$, $BXZ$, $CWZ$ and $WXY$, and every point lies on three of them (Fig.~\ref{fig:steane}a). Steane's $[[7,1,3]]$ code~\cite{Steane1996prl,Steane1996} has one qubit on each point; it is the smallest triangular color code~\cite{Bombin2006}, and the color code of $\Z_d$ qudits on the same seven sites~\cite{Watson2015} is a $[[7,1,3]]_d$ code for every $d\ge2$, Steane's code for $d=2$ (Supplementary Note~\ref{app:qss-color}). A \emph{stabilizer state} of qudits of local dimension $d$ is, up to a phase, the unique common $+1$ eigenvector of an abelian group of generalized Pauli operators of $\Z_d$ that contains no nontrivial scalar~\cite{Gottesman1999qudit,Hostens2005} (for qudits of several dimensions, of tensor products of their generalized Pauli operators); a \emph{stabilizer code} is the common $+1$ eigenspace of an abelian group of generalized Pauli operators without nontrivial scalars. The simplest stabilizer states are the \emph{CSS states}, after Calderbank, Shor and Steane~\cite{CalderbankShor1996,Steane1996}: for a subgroup (a classical code) $\mathcal{C}$ of $\Z_d^m$, where $m$ is the number of qudits and the coordinates are distributed among the systems, $\ket{\mathcal{C}}=\abs{\mathcal{C}}^{-1/2}\sum_{c\in\mathcal{C}}\ket{c}$.

For codes, the local dimension matters through their exact sets of correctable erasures. Let a reference system $R$ be maximally entangled with the encoded system, and let a system $E$ purify the joint state of $R$ and the sites ($E$ is trivial for a code and allows mixed encodings). A set $T$ of sites is \emph{authorized} if $I(R;T)=2S(R)$, so that it can recover the encoded state, and \emph{unauthorized} if $I(R;T)=0$, so that it learns nothing about it; for a code, the erasure of an unauthorized set is correctable~\cite{Imai2005}. In quantum secret sharing~\cite{CleveGottesmanLo1999,Gottesman2000}, the encoded system is a secret and the sites are shares held by players, and a \emph{scheme} for an access structure $\Gamma$, a family of sets of players that is closed under taking supersets and contains the set of all players but not the empty set, is a state in which the sets in $\Gamma$ are authorized and all other sets are unauthorized. Its information ratio is the largest ratio of the logarithm of the dimension of a share to that of the secret, and the scheme is \emph{ideal} if this ratio is $1$. A \emph{stabilizer scheme} of local dimension $d$ is one in which the joint state of $R$, the shares and $E$ is a stabilizer state of $\Z_d$ qudits or, for a prime power $d$, of Galois qudits of dimension $d$, whose generalized Pauli operators are labeled by the field $\GF(d)$ through its trace~\cite{Ketkar2006}; its secret \emph{consists of $m$ qudits} if $S(R)=m\log d$, as for a code that encodes $m$ qudits. Logarithms are to base $2$.

Put $R$ at the origin of the cube $\{0,1\}^3$ and the seven players at the other vertices, labeled as above. The qubit \emph{cube state}, the CSS state of the code of affine functions on the cube (the extended Hamming code), is $(\ket{0}_R\ket{0_L}+\ket{1}_R\ket{1_L})/\sqrt2$ for Steane's code. Its access structure $\Gamma_{\mathrm F}$ is the recovery pattern of Steane's code: a set of players is authorized exactly when it contains a line of the Fano plane. Over $\GF(p)$ with $p$ odd, the cube state is in the same way $\sum_j\ket{j}_R\ket{j_L}/\sqrt p$ for the seven-qudit color code (Supplementary Note~\ref{app:qss-color}); its access structure $\Gamma_{\mathrm{NF}}$ differs from $\Gamma_{\mathrm F}$ only in that $\{W,X,Y\}$ is unauthorized and $\{A,B,C,Z\}$ is authorized (Fig.~\ref{fig:steane}b). In matroid terms, the two are the ports at $R$ of the binary affine cube $\mathrm{AG}(3,2)$ and of the real affine cube $R_8$~\cite{Oxley2011}, that is, the families of sets $T$ of players such that $T\cup\{R\}$ contains a circuit (a minimal dependent set) through $R$; the access structures of ideal classical schemes are matroid ports~\cite{BrickellDavenport1991}. The \emph{planes} of these matroids, their sets of four points of rank $3$, are the fourteen sets of four vertices of the cube with zero sum for $\mathrm{AG}(3,2)$, and for $R_8$ the twelve of them other than $\{R,W,X,Y\}$ and $\{A,B,C,Z\}$, the sets of four coplanar vertices of the real cube. Both schemes are ideal instances of the construction of Sarvepalli and Raussendorf; the qubit one is, up to the labeling of the players, their example built from the extended Hamming code~\cite[Sec.~III~B]{SarvepalliRaussendorf2010}. The requirement concerns the exact set of correctable erasures, not the distance: the color code has the length, rate and distance of Steane's code in every local dimension. In Theorem~\ref{thm:ideal} below, $D$ is the dimension of the encoded system and of every site, with nothing else assumed about them; $d$ is the dimension of qudits of which a site may hold several, such as the $\Z_d$ qudits of stabilizer codes and states.

\begin{theorem}[Codes and ideal schemes]\label{thm:ideal}
Let $D\ge2$, and consider codes that encode one qudit of dimension $D$ into seven qudits of dimension $D$, one at each point of the Fano plane. A code that corrects the erasure of every set of sites that contains no line exists if and only if $D$ is a power of $2$, and one that corrects the erasure of every set of sites outside $\Gamma_{\mathrm{NF}}$ exists if and only if $D$ is odd; such codes have exactly the recovery pattern $\Gamma_{\mathrm F}$, respectively $\Gamma_{\mathrm{NF}}$. More generally, a quantum secret sharing scheme, stabilizer scheme or not, pure or mixed, whose secret and seven shares all have dimension $D$ exists for $\Gamma_{\mathrm F}$ if and only if $D$ is a power of $2$, and for $\Gamma_{\mathrm{NF}}$ if and only if $D$ is odd. In every such scheme, and so for every such code, the joint state of $R$ and the shares is pure and, up to a unitary on each system, a stabilizer state: for $\Gamma_{\mathrm F}$ of qubits, $\log_2D$ on each system, and for $\Gamma_{\mathrm{NF}}$ of qudits of odd prime-power dimensions $n_1,\dots,n_r$ on each system, with $n_1\cdots n_r=D$. Moreover, let $r_{\mathcal M}$ be the rank function of $\mathcal M=\mathrm{AG}(3,2)$, respectively $\mathcal M=R_8$. Every pure state of $R$ and the seven players, each of dimension $D$, whose entropies are $S(X)=(2r_{\mathcal M}(X)-\abs X)\log D$ for every set $X$ of these eight systems is, up to a unitary on each system, a stabilizer state of qubits, respectively of qudits of odd prime-power dimensions, and such states exist exactly when $D$ is a power of $2$, respectively odd.
\end{theorem}

So no code of qudits whose dimension is not a power of $2$, stabilizer code or not, encodes one qudit with the recovery pattern of Steane's code and one qudit per site, no code of qudits of even dimension does so with the recovery pattern of the color code of odd dimension, and in dimension $6$ neither is possible. A code that corrects every erasure that Steane's code corrects, or every erasure that the color code of odd local dimension corrects, has exactly the recovery pattern of that code, because a code is a pure state together with the reference, so a set of sites is unauthorized exactly when its complement is authorized, and both access structures are self-dual, a set of players being authorized exactly when its complement is not (for $\Gamma_{\mathrm F}$, a set of points of the Fano plane contains no line exactly when its complement contains one; Supplementary Note~\ref{app:qss-general}). Every code that encodes $k$ qubits with Steane's pattern and $k$ qubits per site is a qubit stabilizer code up to local unitaries. Since the entropies are continuous on the compact set of states, approximate schemes in the excluded dimensions violate the access conditions by an amount bounded away from zero, although we have no explicit bound (Supplementary Note~\ref{app:qss-general}).

We call a recovery pattern \emph{rigid} in dimension $D$ if every code that has it and encodes one qudit of dimension $D$ with one such qudit per site is, up to a unitary on each site, a stabilizer code of qudits of some dimensions, as in Theorem~\ref{thm:ideal}. By Theorem~\ref{thm:ideal}, the two patterns are rigid in every dimension, nontrivially so in the dimensions in which codes with them exist: powers of $2$ for Steane's pattern and odd dimensions for that of the color code. The threshold pattern in which any two of three sites recover is not rigid in dimension $6$: the code of the $36$ officers has it, and no stabilizer code does (Supplementary Note~\ref{app:qss-general}). For $D=2$ rigidity holds in a stronger form, uniqueness: every code that encodes one qubit into seven qubits and corrects every erasure that Steane's code corrects is Steane's code up to local unitaries (Supplementary Note~\ref{app:qss-general}), much as every code that encodes one qubit into five with distance $3$ is the five-qubit code~\cite{Rains1999}. Steane's code is not determined by its parameters, since other qubit stabilizer codes with the same length, rate and distance have up to $23$ uncorrectable erasures of three sites (Supplementary Note~\ref{app:qss-color}), but it is determined by its recovery pattern. For $D\ge3$ we do not know whether the codes of Theorem~\ref{thm:ideal} are unique up to local unitaries.

The dichotomy has classical precedents: networks built on the Fano and non-Fano matroids show that linear network coding depends on the characteristic~\cite{DFZ2005insufficiency}, and Mat\'u\v{s} proved that the Fano matroid, $\mathrm{AG}(3,2)$ contracted at any point, has representations by partitions, the combinatorial form of ideal classical schemes, only in degrees that are powers of $2$, and the non-Fano matroid, $R_8$ contracted at any point, only in odd degrees~\cite{Matus1999}. Theorem~\ref{thm:ideal} is a quantum counterpart of both halves, which shows in addition that these codes are stabilizer codes up to local unitaries. Its proof uses no computer (Fig.~\ref{fig:ideal}; Supplementary Note~\ref{app:qss-general}): applied to the lift $S(X)+\sum_{i\in X}S(i)$ of the entropies, in which the access conditions become linear equations (Methods), a theorem of Mat\'u\v{s} on polymatroids (monotone submodular set functions)~\cite{Matus2012} and Lehman's theorem that a connected matroid is determined by its circuits through one element~\cite{Lehman1964} give the entropies, $(2r_{\mathcal M}(X)-\abs X)\log D$ with $r_{\mathcal M}$ the rank function of $\mathcal M=\mathrm{AG}(3,2)$ or $R_8$. From these alone, the planes of $\mathcal M$ yield commuting stabilizers, whose phase-space labels, the vectors that index their Weyl operators, force every label $x$ to satisfy $2x=0$ for $\Gamma_{\mathrm F}$ and to have odd order for $\Gamma_{\mathrm{NF}}$, the difference being whether $\{R,W,X,Y\}$ is a plane. So rigidity holds already for entropies, as the last part of Theorem~\ref{thm:ideal} states.

\begin{figure}[t]
\centering
\begin{tikzpicture}[font=\scriptsize,
  pt/.style={circle, fill=black, inner sep=0.9pt},
  rp/.style={circle, fill=orange!85!black, inner sep=1.3pt},
  lab/.style={font=\tiny, inner sep=1pt}]
\newcommand{\cubeat}[2]{%
  \coordinate (R) at (#1,#2);           \coordinate (A) at ({#1+1.1},#2);
  \coordinate (B) at ({#1+0.48},{#2+0.36}); \coordinate (W) at ({#1+1.58},{#2+0.36});
  \coordinate (C) at (#1,{#2+1.05});     \coordinate (X) at ({#1+1.1},{#2+1.05});
  \coordinate (Y) at ({#1+0.48},{#2+1.41}); \coordinate (Z) at ({#1+1.58},{#2+1.41});}
\newcommand{\cubeedges}{\draw[gray!55] (R)--(A)--(W)--(B)--cycle (C)--(X)--(Z)--(Y)--cycle (R)--(C) (A)--(X) (B)--(Y) (W)--(Z);}
\newcommand{\cubepoints}{%
  \node[rp] at (R) {}; \foreach \p in {A,B,C,W,X,Y,Z} \node[pt] at (\p) {};
  \node[lab, below left] at (R) {$R$}; \node[lab, below right] at (A) {$A$}; \node[lab, above left] at (B) {$B$};
  \node[lab, right] at (W) {$W$}; \node[lab, left] at (C) {$C$}; \node[lab, below right] at (X) {$X$};
  \node[lab, above left] at (Y) {$Y$}; \node[lab, above right] at (Z) {$Z$};}
\node[anchor=north west, font=\normalsize\bfseries] at (0,2.25) {a};
\cubeat{1.3}{0.35}
\fill[teal!22] (R)--(A)--(W)--(B)--cycle;
\cubeedges
\draw[very thick, teal!70!black] (A)--(W)--(B);
\cubepoints
\node at (2.1,-0.1) {Pure, entropies of $\mathrm{AG}(3,2)$};
\node[anchor=north west, font=\normalsize\bfseries] at (4.25,2.25) {b};
\cubeat{5.55}{0.35}
\cubeedges
\draw[very thick, teal!70!black] (R)--(A) (B)--(W) (C)--(X) (Y)--(Z);
\cubepoints
\node at (6.35,-0.1) {GHZ along each direction};
\node[anchor=north west, font=\normalsize\bfseries] at (0,-0.45) {c};
\cubeat{1.3}{-2.35}
\fill[orange!22] (C)--(X)--(Z)--(Y)--cycle;
\cubeedges
\draw[very thick, orange!85!black] (C)--(X)--(Z)--(Y)--cycle;
\cubepoints
\node at (2.1,-2.8) {$D^2$ stabilizers on each plane};
\node[anchor=north west, font=\normalsize\bfseries] at (4.25,-0.45) {d};
\coordinate (fA) at ($(6.35,-1.5)+(90:1.1)$);
\coordinate (fB) at ($(6.35,-1.5)+(210:1.1)$);
\coordinate (fC) at ($(6.35,-1.5)+(330:1.1)$);
\coordinate (fW) at ($(fA)!0.5!(fB)$);
\coordinate (fX) at ($(fA)!0.5!(fC)$);
\coordinate (fY) at ($(fB)!0.5!(fC)$);
\coordinate (fZ) at (6.35,-1.5);
\draw[gray!70] (fA)--(fB)--(fC)--cycle (fA)--(fY) (fB)--(fX) (fC)--(fW);
\draw[very thick, orange!85!black] (fZ) circle (0.55);
\foreach \p in {fA,fB,fC,fW,fX,fY,fZ} \node[pt] at (\p) {};
\node[lab, above=1pt] at (fA) {$u$};
\node[lab, below left] at (fB) {$v$};
\node[lab, below right] at (fC) {$w$};
\node[lab, left] at (fW) {$u{+}v$};
\node[lab, right] at (fX) {$u{+}w$};
\node[lab, below] at (fY) {$v{+}w$};
\node[lab, below, fill=white, inner sep=0.3pt] at ($(fZ)+(0,-0.06)$) {$u{+}v{+}w$};
\node at (6.35,-2.8) {$2x=0$: all qudits are qubits};
\end{tikzpicture}
\caption{Idea of the proof of Theorem~\ref{thm:ideal}. \textbf{a}~The reference $R$ (orange dot) at the origin of the cube $\{0,1\}^3$, the seven players at the other vertices (Fig.~\ref{fig:steane}a); a line of the Fano plane, such as $ABW$ (thick), forms with $R$ a plane of $\mathrm{AG}(3,2)$ (shaded). The recovery pattern makes the state pure and fixes its entropies, $S(X)=(2r(X)-\abs X)\log D$, where $r$ is the rank function of $\mathrm{AG}(3,2)$ and $D$ the dimension of each system~\cite{Matus2012,Lehman1964}. \textbf{b}~Along each of the seven directions of the cube, the four pairs of systems (thick, for one direction) form a Greenberger--Horne--Zeilinger (GHZ) state (Supplementary Lemma~\ref{lem:K-ghz}, which rests on the structure of equality in strong subadditivity~\cite{HJPW2004}). \textbf{c}~For each plane, such as $\{C,X,Y,Z\}$ (orange), comparing its splittings into pairs along two directions gives $D^2$ commuting stabilizers built from nice error bases~\cite{Knill1996}; at each system their factors generate, up to scalars, one group, the same for every direction and plane, and the groups of different systems are isomorphic (Steps~2 and~3 of Supplementary Note~\ref{app:qss-general}); all planes together give $D^8$ stabilizers, so the state is a stabilizer state. \textbf{d}~Normal form of the phase-space labels of the stabilizers that act trivially on $R$, at the positions of Fig.~\ref{fig:steane}a, with $u$, $v$ and $w$ in the phase space $V$ of one system; since $\{R,W,X,Y\}$ is a plane, the labels at $W$, $X$ and $Y$ (circle) force $2x=0$ on $V$. For the pattern of the color code, $R_8$ replaces $\mathrm{AG}(3,2)$, only the directions of the edges and of the body diagonal are used, and every label has odd order (Supplementary Note~\ref{app:qss-general}).}
\label{fig:ideal}
\end{figure}
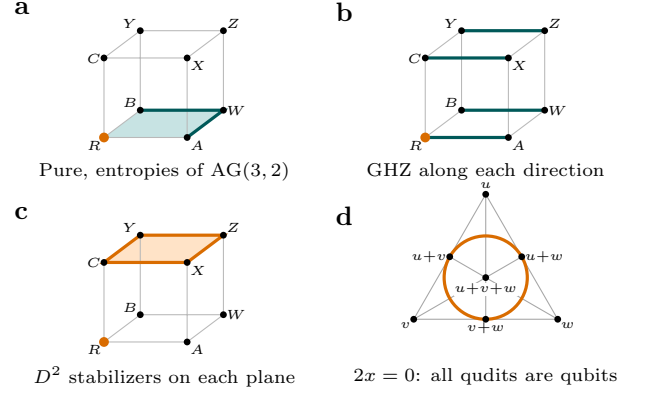

The two patterns differ on the erasure channel, in which each site is erased independently with probability $\varepsilon$. Steane's code fails exactly when the erased sites contain a line, and the color code of odd local dimension exactly when they form a set in $\Gamma_{\mathrm{NF}}$; the two differ only on $\{W,X,Y\}$, fatal for Steane's code alone, and $\{A,B,C,Z\}$, fatal for the color code alone, so the failure probability of the color code is lower by $\varepsilon^3(1-\varepsilon)^4-\varepsilon^4(1-\varepsilon)^3=\varepsilon^3(1-\varepsilon)^3(1-2\varepsilon)$. Every qubit stabilizer code that encodes one qubit into seven with distance at least $3$ fails for at least seven sets of three sites (exhaustive search), and hence more often than the color code for every $0<\varepsilon<1/2$, by at least the same margin (Supplementary Note~\ref{app:qss-color}). So, at equal erasure probability per site, the pattern of odd local dimensions is the better of the two, with failure probability $6\varepsilon^3$ against at least $7\varepsilon^3$ to leading order, although its sites are larger; Steane's pattern matters where it is prescribed, as an access structure in secret sharing. Whether a qubit code that encodes one qubit into seven, outside the stabilizer formalism, fails less often than the color code we do not know; by Theorem~\ref{thm:ideal}, none corrects all the erasures that the color code corrects. For even $d\ge4$, the color code of $\Z_d$ qudits recovers the encoded qudit from exactly the sets that lie in both patterns, while $\{W,X,Y\}$ and $\{A,B,C,Z\}$ learn part of it; so the erasure of a set of sites is uncorrectable exactly when the set lies in either pattern, and for every $0<\varepsilon<1$ the code fails more often than Steane's code and than the color code of odd local dimension (Supplementary Note~\ref{app:qss-color}).

For stabilizer codes, the arithmetic of the local dimension also bounds the sizes of the shares from below.

\begin{proposition}[Stabilizer schemes]\label{prop:qss}
For the access structures $\Gamma_{\mathrm F}$ and $\Gamma_{\mathrm{NF}}$:
\begin{enumerate}
\item[(i)] The qubit cube state is an ideal scheme for $\Gamma_{\mathrm F}$. Every stabilizer scheme for $\Gamma_{\mathrm F}$, pure or mixed, of odd local dimension ($\Z_d$ qudits with $d$ odd, or Galois qudits of odd characteristic) satisfies $\sum_{P\in\ell}S(P)\ge4\,S(R)$ for every line $\ell$ of the Fano plane, so its information ratio is at least $4/3$, and on average its shares are at least $4/3$ times as large as the secret. For $\Z_d$ qudits with $d$ not a power of $2$, in every stabilizer scheme for $\Gamma_{\mathrm F}$ whose secret consists of $m$ qudits the players of every line hold together at least $4m$ qudits, with the same conclusions.
\item[(ii)] For every odd prime $p$, the cube state over $\GF(p)$ is an ideal scheme for $\Gamma_{\mathrm{NF}}$. Every stabilizer scheme for $\Gamma_{\mathrm{NF}}$ of qubits or of Galois qudits of characteristic $2$, pure or mixed, satisfies $S(P)+S(P')+S(P'')\ge4\,S(R)$ for any three of the players $A$, $B$, $C$ and $Z$, so its information ratio is at least $4/3$. For $\Z_d$ qudits with $d\equiv2\pmod4$, in every stabilizer scheme for $\Gamma_{\mathrm{NF}}$ whose secret consists of $m$ qudits any three of these players hold together at least $4m$ qudits, so the information ratio is at least $4/3$.
\item[(iii)] For every odd $d$, every $d\equiv2\pmod4$ with $d\ne2$, and Galois qudits of odd characteristic, $\Gamma_{\mathrm F}$ has a stabilizer scheme with information ratio $7/5$, whose secret consists of five qudits. For every $d\equiv2\pmod4$, qubits included, and for Galois qudits of characteristic $2$, so has $\Gamma_{\mathrm{NF}}$, with a secret of ten qudits.
\item[(iv)] For $\Z_d$ qudits, $\Gamma_{\mathrm F}$ has an ideal stabilizer scheme, pure or mixed, if and only if $d=2$, and $\Gamma_{\mathrm{NF}}$ has one if and only if $d$ is odd.
\end{enumerate}
\end{proposition}

Summing (i) over the seven lines, each player lying on three of them, shows that a stabilizer code with the recovery pattern of Steane's code needs at least $28/3$ qudits per encoded qudit, against $7$ for qubits, whenever its local dimension is not a power of $2$; and by (ii), a qubit stabilizer code with the recovery pattern of the color code has a site that holds at least $4/3$ times as many qubits as it encodes. For prime local dimensions, stabilizer schemes are linear schemes (Supplementary Lemma~\ref{lem:qss-linear}), the bound of (i) follows from bounds of Jafari and Khazaei and of Pe\~na-Macias for linear schemes~\cite{JafariKhazaei2021,PenaMacias2023}, and that of (ii) from a bound of Jafari and Khazaei; for composite local dimensions, where the lift of a stabilizer state is the rank function of subgroups of a finite abelian group, one for each system (Sec.~\ref{sec:reduction}), we need an inequality for seven subgroups in the configuration of the Fano plane that holds in every finite abelian group of odd order (Lemma~\ref{lem:fano-odd}, Methods). Theorem~\ref{thm:ideal} gives part~(iv) for $\Gamma_{\mathrm{NF}}$ and shows that $d$ is a power of $2$ for $\Gamma_{\mathrm F}$; that $d=2$, which excludes $\Z_4$, $\Z_8$, \dots, follows from a normal form of the phase-space labels of the stabilizer group (Methods). For $D=4$, for instance, the codes of Theorem~\ref{thm:ideal} with Steane's pattern are, up to local unitaries, stabilizer codes of two qubits per site, such as two copies of Steane's code, and by~(iv) none is, even up to local unitaries, a stabilizer code of one $\Z_4$ qudit per site. We do not know the optimal ratios, which lie between $4/3$ and $7/5$ for $\Gamma_{\mathrm F}$ in every odd local dimension (Supplementary Note~\ref{app:qss-upper}).

Codes that are not stabilizer codes of $\Z_d$ qudits avoid the bound of (i) for large secrets: restricting the secret of copies of the qubit cube state to a subspace and embedding the qubits of each player into qudits of dimension $d$ gives schemes for $\Gamma_{\mathrm F}$ whose secret consists of $k$ qudits and whose shares consist of $n$ qudits, with $n/k\to1$; for $d=3$ and $k=4$ the scheme uses $35$ qutrits, against at least $38$ for every qutrit stabilizer scheme with a secret of four qutrits (Supplementary Note~\ref{app:qss-general}). So, for this recovery pattern and large secrets, the cost in (i) is a cost of stabilizer codes of $\Z_d$ qudits with $d$ not a power of $2$, not of the pattern: copies of Steane's code embedded into qudits of dimension $d$, which are not stabilizer codes of such qudits, avoid it. With shares of the size of the secret, Theorem~\ref{thm:ideal} excludes this.
\subsection{A dictionary between stabilizer-state entanglement and classical codes}
\label{sec:reduction}
\label{sec:prelim}

The remaining results concern the entanglement entropies of stabilizer states; the tool behind them, and behind the bounds of Proposition~\ref{prop:qss}, is a dictionary between these entropies and the ranks of classical codes. We consider $n$ parties $1,\dots,n$ and a purifying system $0$, and write $N=\{0,1,\dots,n\}$. For a pure state on $N$ the \emph{entropy vector} assigns to every $X\subseteq N$ the von Neumann entropy $S(X)$ of the reduced state on $X$, in bits. It is \emph{pure-symmetric}: $S(\emptyset)=0$ and $S(X)=S(N\setminus X)$. Stabilizer states and CSS states are defined as in Sec.~\ref{sec:qss}, and each system may hold any number of qudits. We write $\StabZ{d}{n}$ for the closed convex cone generated by the entropy vectors of stabilizer states of $\Z_d$ qudits on $N$; since entropy vectors add under tensor products, it is the closure of the positive multiples of stabilizer entropy vectors.

The closure of the set of entropy vectors of all quantum states on $N$ is a convex cone, the \emph{quantum entropy cone}, whose description beyond three parties is open~\cite{Pippenger2003}. The entropy vectors of stabilizer states, whose entropies are ranks of matrices~\cite{Fattal2004,Hein2004}, generate the \emph{stabilizer entropy cone}, an inner approximation of the quantum entropy cone that contains the hypergraph and holographic entropy cones (Discussion). Cones compare \emph{entropy patterns}, entropy vectors up to a positive factor, in the limit: a dependence on $d$ at the level of cones is one that no number of qudits per party, no number of copies and no limit removes. The inequalities previously stated for stabilizer entropy vectors hold in every local dimension: every balanced classical information inequality~\cite{GrossWalter2013}, and every \emph{linear rank inequality} that follows from common information, the first being the Ingleton inequality~\cite{LMRW2013}; linear rank inequalities are those that the dimensions of sums of subspaces satisfy over every field, and common information adds the intersection of two subspaces as a new variable. \emph{Characteristic-dependent} linear rank inequalities hold over fields of some characteristics but not of others. Bao, Cheng, Hern\'andez-Cuenca and Su introduced the \emph{quantum linear rank cone} $\QLR_n$, defined by the linear rank inequalities in $n$ variables (Sec.~\ref{sec:fourfive}); it contains the stabilizer cone for up to five parties and equals it for up to four~\cite{LMRW2013,BCHS2020}.

On the classical side, if $U_i$ ($i\in\mathcal{E}$) are subgroups of a finite abelian group, then $h(X)=\log\abs{\sum_{i\in X}U_i}$ is the \emph{rank function} of the subgroup arrangement; over $\GF(p)$ it is $\log p$ times the dimension of the span, the rank function of a linear code. We write $\LcZ{d}{m}$ for the closed convex cone generated by the rank functions of subgroup arrangements on $m$ elements in groups of exponent dividing $d$. For a function $h$ on the subsets of $N$, the \emph{dual} is $h^*(X)=\sum_{i\in X}h(i)+h(N\setminus X)-h(N)$, also written $\Dual h$, and $h$ is \emph{self-dual} if $h^*=h$; the self-dual elements of a cone form its \emph{self-dual slice}. The central tool of this paper is the \emph{lift} of a pure-symmetric vector, which Mat\'u\v{s} used to relate the entropy vectors of pure states to polymatroids~\cite{Matus2012},
\begin{equation}\label{eq:lift}
  h_S(X)=S(X)+\sum_{i\in X}S(i)\qquad(X\subseteq N).
\end{equation}
For a pure state, $h_S$ is a polymatroid (submodularity is strong subadditivity, and monotonicity is the Araki--Lieb inequality~\cite{ArakiLieb1970}), it is self-dual, and $\lambda_{h_S}=2S$, where $\lambda_h(X)=h(X)+h(N\setminus X)-h(N)$ is the connectivity function (Supplementary Lemma~\ref{lem:A-selfdual}).

\begin{proposition}[Dictionary]\label{thm:reduction}
Let $d\ge 2$ and $n\ge 1$. For every pure-symmetric vector $S$ on $N$ the following are equivalent:
\begin{enumerate}
\item[(i)] $S\in\StabZ{d}{n}$;
\item[(ii)] $h_S\in\LcZ{d}{n+1}$;
\item[(iii)] $S$ lies in the closed convex cone generated by the entropy vectors of CSS states of $\Z_d$ qudits.
\end{enumerate}
Moreover, $\StabZ{d}{n}=\lambda\bigl(\LcZ{d}{n+1}\bigr)$, and $S\mapsto h_S$ is a linear bijection from $\StabZ{d}{n}$ onto the self-dual slice $\{h\in\LcZ{d}{n+1}:h^*=h\}$ of $\LcZ{d}{n+1}$, with inverse $h\mapsto\tfrac12\lambda_h$.
\end{proposition}

For qubits, for example, $S$ is a limit of positive combinations of qubit stabilizer entropy vectors if and only if $h_S$ is a limit of positive combinations of rank functions of $n+1$ subspaces over $\GF(2)$, one of which belongs to the purifier. Two features of such ranks thus become features of entanglement: the purifier is one more variable, which is the source of the five-party gap (Sec.~\ref{sec:fourfive}), and the arithmetic of $d$ enters through the inequalities valid for subgroup arrangements of exponent dividing $d$, such as the characteristic-dependent linear rank inequalities~\cite{BKL2011,DFZ2015rate,DFZ2015char}; this is the source of the dependence on the local dimension. The proof combines three facts (Methods); nothing depends on $d$ being prime.

\begin{corollary}[CSS states suffice; doubling]\label{cor:doubling}
Let $S$ be the entropy vector of a stabilizer state of $m$ qudits of dimension $d$. Then $2S$ is the entropy vector of a CSS state of $2m$ qudits of dimension $d$.
\end{corollary}

For prime $d$, Corollary~\ref{cor:doubling} follows from two entropy formulas of Osthus, who asked whether CSS states realize the entropy vectors of all stabilizer states~\cite[Ch.~6]{Osthus2024}; for every $d$, it answers this question for cones.

\begin{corollary}[Inclusion under divisibility]\label{cor:divisibility}
If $d$ divides $d'$, then $\StabZ{d}{n}\subseteq\StabZ{d'}{n}$ for every $n$. In particular, $\StabZ{p}{n}\subseteq\StabZ{d}{n}$ for every prime $p$ dividing $d$.
\end{corollary}

\begin{proof}
A group of exponent dividing $d$ has exponent dividing $d'$, so $\LcZ{d}{n+1}\subseteq\LcZ{d'}{n+1}$; apply Proposition~\ref{thm:reduction} (entropies are measured in bits for all $d$).
\end{proof}

Up to a factor $2$, the dictionary thus reduces the search for stabilizer states with prescribed entanglement entropies to a search for classical codes over $\Z_d$. Separations persist when parties are added: if $S\in\StabZ{d}{m}\setminus\StabZ{d'}{m}$ and $n>m$, give $n-m$ new parties no qudits; merging them with the purifier maps $\StabZ{d'}{n}$ linearly into $\StabZ{d'}{m}$ and the new vector back to $S$.
\subsection{Steane's code and its qudit counterpart cannot imitate each other}
\label{sec:six}
\label{sec:steane}

The code states of Steane's code and of the color code show that the dependence on the arithmetic of $d$ extends from recovery patterns to the entanglement of stabilizer states. Label seven systems by the points of the Fano plane, as in Sec.~\ref{sec:qss}. The $[7,3]$ simplex code over $\GF(2)$ has as generator the $3\times7$ matrix whose columns are the seven nonzero vectors of $\GF(2)^3$. Its CSS state, with one qubit per point of the Fano plane, is the logical $\ket{0_L}$ of Steane's $[[7,1,3]]$ code~\cite{Steane1996prl,Steane1996}, which we call the \emph{Steane state}. It is a state of seven systems (six parties and a purifier), not the eight-system cube state of Sec.~\ref{sec:qss}, which also contains the reference. Its entropy vector, in bits, is $\SF=\lambda_{\rF}$ and its lift is
\begin{equation}\label{eq:hF}
  h_{\SF}=\rF+\rF^*=:\hF ,
\end{equation}
where $\rF$ is the rank function of the Fano matroid (Supplementary Lemma~\ref{lem:A-css}). In the Steane state, $I(i;j)=0$ for any two systems, and $I(i;jk)=0$ unless $\{i,j,k\}$ is a line, in which case $I(i;jk)=S(i)=1$ bit: the correlations sit on the lines of the Fano plane. Any of the seven systems may serve as the purifier, and other eigenstates of logical Pauli operators, such as $\ket{+_L}$, have the same entropy vector.

Over $\GF(p)$ with $p$ odd, the same $3\times7$ matrix generates a code whose columns form the non-Fano configuration, in which the points $W$, $X$ and $Y$ are not collinear, because the determinant of the vectors $110$, $101$ and $011$ is $-2$. Its CSS state, with one qudit on each point, is the \emph{non-Fano state}; its lift is $\hNF=\rNF+\rNF^*$, where $\rNF$ is the rank function of the non-Fano matroid. Let $\SNF$ be its entropy vector in units of $\log p$, which does not depend on $p$. Both $\SF$ and $\SNF$ take the value $1$ on every system, and they differ only on the cut $WXY\mid ABCZ$, where $\SNF$ is larger by $1$. The non-Fano state is a logical basis state of the qudit counterpart of Steane's code: for every $d\ge2$, the logical basis state of the color code of $\Z_d$ qudits of Sec.~\ref{sec:qss} is the CSS state of the code generated by the same matrix over $\Z_d$ (Supplementary Note~\ref{app:qss-color}).

\begin{theorem}[Six parties]\label{thm:steane}
For every odd $d\ge3$ and every $n\ge6$,
\[
  \StabZ{2}{n}\not\subseteq\StabZ{d}{n}.
\]
A witness is the logical $\ket{0_L}$ of Steane's code, with one qubit for each of six parties and one qubit as the purifier.
\end{theorem}

\begin{theorem}[Six parties, the reverse direction]\label{thm:reverse}
For $d\ge2$, let $S^{(d)}$ be the entropy vector, in bits, of the logical basis state of the color code of $\Z_d$ qudits, that is, of the CSS state, with one qudit on each of the seven systems, of the code generated over $\Z_d$ by the rows of the $3\times7$ matrix whose columns are the seven nonzero vectors of $\{0,1\}^3$, labeled as above. Then
\[
  \ftwelve\bigl(S^{(d)}\bigr)=-\log_2\frac{d}{\gcd(d,2)} ,
\]
where $\ftwelve$ is a linear form that is nonnegative on $\StabZ{2}{6}$, with $Z$ as the purifier (Proposition~\ref{prop:f12}, Methods). Hence $\StabZ{d}{n}\not\subseteq\StabZ{2}{n}$ for every $d\ge3$ and every $n\ge6$, and $\StabZ{2}{n}\subsetneq\StabZ{d}{n}$ for every even $d\ge4$ and every $n\ge6$. With Theorem~\ref{thm:steane}, the qubit cone and the cone of every odd local dimension are incomparable from six parties on.
\end{theorem}

For an odd prime $p$, $S^{(p)}$ is $\log_2p$ times $\SNF$.

The difference rests on $1+1=0$: the points $110$, $101$ and $011$ lie on a line exactly when $110+101=011$, that is, in characteristic~$2$. Up to four parties nothing depends on the local dimension (Proposition~\ref{thm:four}, Methods), so the smallest number of parties at which the cone depends on it is five or six.

None of the published characteristic-dependent inequalities that we checked gives these separations when applied directly. By Proposition~\ref{thm:reduction}, for prime $d$ a dependence at six parties is a characteristic dependence of the seven-variable linear rank cone on its self-dual slice. The Fano and non-Fano configurations are the classical examples of characteristic dependence~\cite{Oxley2011}, and seven-variable linear rank inequalities that distinguish them are known~\cite{BKL2011,DFZ2015rate,PenaSarria2019,PenaMacias2023}. But all nine that we checked, the seven published ones and the two of Lemma~\ref{lem:JI}, are blind on the self-dual slice: they hold for every pure quantum state, for every bijection between their variables and the six parties and the purifier (Proposition~\ref{prop:blind}, Methods, and Supplementary Note~\ref{app:blind}). We overcome this with two transfer arguments, which apply a characteristic-dependent inequality not to the subgroups of the parties but to subgroups derived from them. For Theorem~\ref{thm:steane}, these are the pieces of the parties' subgroups generated within the lines of the Fano plane, and the result is an explicit integer form $\fnine$ with $67$ terms that is nonnegative on the lift of $\StabZ{d}{6}$ for every odd $d$ and takes the value $-3$ at $\hF$. For Theorem~\ref{thm:reverse}, they are the restrictions of the stabilizer group to the points, which are orthogonal where two lines meet, because the group is isotropic. Read through the connectivity function, $\ftwelve$ also yields a seven-variable linear rank inequality for characteristic $2$ that does not follow from the Shannon inequalities and the eight inequalities of Proposition~\ref{prop:blind} (Supplementary Note~\ref{app:reverse-g12}).

Larger code states inherit the separations, with one more party, when seven of their qudits carry these configurations and the other qudits determine the codeword; among them is a logical basis state of the $15$-qubit code that Bravyi and Kitaev used to distill magic states~\cite{BravyiKitaev2005} (Supplementary Note~\ref{app:cubes-sandwich}).
\subsection{Divisibility, not size, decides}
\label{sec:dimension}
\label{sec:divisibility}

Theorems~\ref{thm:steane} and~\ref{thm:reverse} separate qubits from every odd local dimension at six parties. With more parties, characteristic-dependent inequalities can be used directly, because the purifier can then be contracted or can hold a copy of each codeword, which removes the self-duality that makes them blind at six parties, and the stabilizer cones are ordered exactly by divisibility (Fig.~\ref{fig:landscape}b).

\begin{theorem}[Divisibility]\label{thm:divisibility}
Let $d,d'\ge2$. If $d\mid d'$, then $\StabZ{d}{n}\subseteq\StabZ{d'}{n}$ for every $n$. If $d\nmid d'$, then $\StabZ{d}{n}\not\subseteq\StabZ{d'}{n}$ for every $n\ge2\max(d,d')+2$. Consequently, for $n\ge2\max(d,d')+2$,
\[
  \StabZ{d}{n}\subseteq\StabZ{d'}{n}\iff d\mid d',
\]
and the stabilizer cones of the local dimensions $2\le d\le(n-2)/2$ are pairwise distinct.
\end{theorem}

So $\Z_6$ qudits reproduce all entropy patterns of qubits and of qutrits alike, whereas, from six parties on, qutrits and $\Z_4$ qudits are incomparable (Theorem~\ref{thm:twopowers}, Methods). Even at a fixed size, the algebraic structure of the alphabet matters: up to phases, the Pauli group of a Galois qudit of dimension $p^k$, defined with the field trace~\cite{Ketkar2006}, is that of $k$ qudits of dimension $p$, so Galois qudits of dimension $4$ reproduce exactly the qubit patterns (Supplementary Proposition~\ref{prop:galois}), whereas $\Z_4$ qudits reproduce strictly more from six parties on (Theorem~\ref{thm:reverse}). Composite local dimensions need inequalities for finite abelian groups rather than vector spaces, which Lemma~\ref{lem:JI} provides; the witnesses are self-dual \emph{generalized cube states}, CSS states of the codes of affine functions on the points of a configuration, and \emph{doubled codes}, CSS states in which the purifier holds a second copy of each codeword (Methods). The bound $2\max(d,d')+2$ is not optimal in general: two primes $p<q$ give incomparable cones from $2p+2$ parties on if $p$ is odd (Theorem~\ref{thm:primes}, Methods) and from six on if $p=2$, and $\StabZ{d}{n}\subseteq\StabZ{2^a}{n}$ holds if and only if $d$ divides $2^a$ for every $n\ge6$ (Theorem~\ref{thm:twopowers}).
\subsection{Up to five parties}
\label{sec:fourfive}
\label{sec:four}
\label{sec:five}

The dictionary has a second consequence: the purifier becomes one more variable of a rank function, so rank inequalities in $n+1$ variables constrain the entropies of $n$ parties; those that follow from common information, such as~\eqref{eq:star} below, do so in every local dimension. The quantum linear rank cone $\QLR_n$ of Bao, Cheng, Hern\'andez-Cuenca and Su~\cite{BCHS2020} is obtained from the linear rank inequalities in $n$ variables by replacing monotonicity with weak monotonicity, $S(AB)+S(BC)\ge S(A)+S(C)$, and completing the result under the symmetry group $S_{n+1}$ of purification~\cite[Sec.~2.3]{BCHS2020}. It contains the stabilizer cone of every local dimension for $n\le5$ (Methods). For $n\le4$ the two cones agree~\cite{LMRW2013,BCHS2020}, and the four-party stabilizer cone is the same for every local dimension (Proposition~\ref{thm:four}, Methods). For $n=5$ equality was left open~\cite{BCHS2020,BCHS2020gap} (Introduction). Consider five parties $A,\dots,E$ with purifier $F$, and the inequality
\begin{equation}\label{eq:star}
\begin{aligned}
  &S(BCD)+S(ABCD)+S(ABE)\\
  &\quad+2S(BCE)+2S(ADE)+S(CDE)\\
  &\ge 2S(BC)+S(ABD)+S(ACD)+S(AE)\\
  &\quad+S(DE)+S(BCDE)+S(ABCDE).
\end{aligned}
\tag{$\bigstar$}
\end{equation}

\begin{theorem}[Five-party gap]\label{thm:five}
Inequality~\eqref{eq:star} holds on $\Stab_5$, the stabilizer cone of all local dimensions, and an explicit extreme ray $r$ of $\QLR_5$ (Supplementary Table~\ref{tab:witness}) violates it. Hence
\[
  \Stab_5\subsetneq\QLR_5 ,
\]
and in particular $\StabZ{d}{5}\subsetneq\QLR_5$ for every $d$. Moreover, \eqref{eq:star} defines a facet of $\StabZ{d}{5}$ for every $d$, and of $\Stab_5$.
\end{theorem}

The witness has small integer values: $r$ is $2$ on $A$, $B$, $E$ and $F$ and $1$ on $C$ and $D$; $3$ on every pair except $CD$, where it is $2$; $3$ on the triples $ABE$, $ADE$, $BCD$, $BCE$ and $CDE$ and on their complements, and $4$ on the other triples; the other values follow from purity. On $r$ the left side of~\eqref{eq:star} is $24$ and the right side $25$. This answers the question of Bao \emph{et al.}: the stabilizer cone is not cut out by the five-variable linear rank inequalities that define $\QLR_5$. Inequality~\eqref{eq:star} comes from a six-variable linear rank inequality in which the purifier is one of the variables: it follows from Shannon inequalities and one common information, $(U_A+U_D)\cap(U_B+U_E)$, of the subgroups of the six systems (Methods). Six-variable inequalities of Dougherty, Freiling and Zeger~\cite{DFZ2009}, applied in the same way, also detect the gap (Supplementary Note~\ref{app:five-dfz}). The gap thus needs no new kind of inequality, only the purifier as a variable. Since $r$ satisfies strong subadditivity and weak monotonicity~\cite{LiebRuskai1973}, inequality~\eqref{eq:star} does not follow from them; whether every quantum state satisfies it, and whether $r$ lies in the quantum entropy cone, we do not know.

\label{sec:scale}Theorem~\ref{thm:five} does not depend on the following computation. An enumeration of the extreme rays of $\QLR_5$, stopped at a preset ceiling on the number of orbits, found $1{,}917{,}706$ orbits under $S_6$, the \emph{catalog}, which may be incomplete. Exact certificates exclude all but $1{,}796$ of them from $\Stab_5$; $641$ lie in the stabilizer cone of every local dimension, the $59$ known before for qubits~\cite[Sec.~4.2]{BCHS2020gap} and $582$ that we realized by graph states, and $1{,}155$ are undecided (Supplementary Note~\ref{app:five-catalogue}). The inequalities used at five parties hold over every field and for every abelian group, so they cannot detect a dependence on the local dimension; whether the five-party cone depends on it is open.

\section{Discussion}
\label{sec:discussion}

Carrying a qubit code over to qudits can preserve its parameters, as for Steane's code and the seven-qudit color code, without preserving its structure: for the recovery patterns of these two codes, the arithmetic of the local dimension decides whether they survive, for every code that encodes one qudit with one qudit of the same dimension per site and not only for stabilizer codes. The two patterns differ in whether $110$, $101$ and $011$ form a line, which needs $1+1=0$, and through the entanglement entropies this line constrains every such code, stabilizer code or not. So, for qudits of odd dimension, as on qutrit hardware, the seven-qudit color code is not merely a convenient counterpart of Steane's code: its pattern is the only one of the two available with one qudit per site; for dimension $6$, neither is. With several qudits per site, by contrast, the obstruction becomes a cost of stabilizer codes: those of $\Z_d$ qudits with Steane's recovery pattern need at least $28/3$ qudits per encoded qudit instead of $7$ whenever $d$ has an odd prime factor, whereas qubit codes embedded into $\Z_d$ qudits approach $7$ when many qudits are encoded. The bounds concern recovery patterns, not distances.

Within the stabilizer formalism, which underlies quantum error correction and measurement-based computation, the local dimension is more than a measure of how much information a qudit carries: it determines which entropy patterns are available, even in the limit, and with enough parties through divisibility rather than size. So $\Z_6$ qudits have the patterns of qubits and of qutrits, and from six parties on $\Z_4$ qudits have strictly more patterns than Galois qudits of the same dimension, a difference between the two standard constructions of qudit stabilizer codes, over $\Z_d$ and over finite fields~\cite{Ketkar2006}. Through the dictionary, the purifier becomes one more variable, which separates the stabilizer cone from $\QLR_5$, and the arithmetic of the alphabet, which purity hides from the characteristic-dependent inequalities that we checked, is brought back by the transfer arguments. The separations are exact but thin: measured by the largest difference in a single entropy, the Steane state lies between $1/952$ and $1/12$ bit from the cones of odd local dimension (Supplementary Note~\ref{app:margins}). Apart from Theorem~\ref{thm:ideal}, they concern stabilizer states: general quantum states reproduce the entropies of every witness in every local dimension. Still, from six parties on, as inner approximations of the quantum entropy cone, the stabilizer cones of $\Z_4$ qudits, and of all local dimensions together, are strictly larger than the qubit cone. Hypergraph and holographic models cannot detect the dependence: the hypergraph entropy cone, which contains the holographic one~\cite{Bao2015cone,BCHS2020}, lies in the stabilizer cone of every prime local dimension~\cite{WalterWitteveen2021}, hence of every local dimension (Corollary~\ref{cor:divisibility}).

The results also bear on a conjecture of Linden \emph{et al.} that the inequalities obtained from common information, added to the Shannon inequalities, describe the stabilizer cone~\cite{LMRW2013}. Applied to five of the six systems at a time, as in $\QLR_5$, they do not suffice at five parties (Theorem~\ref{thm:five}); applied to all systems, purifier included, they hold in every local dimension, so for every $d$, with enough parties, they do not describe the cone of $\Z_d$ qudits, and from six parties on they describe neither that of qubits nor that of an odd local dimension (Theorems~\ref{thm:steane}--\ref{thm:divisibility}). For all local dimensions together the question remains open (Supplementary Note~\ref{app:discussion}). The inequalities $\fnine$ and $\ftwelve$ and those of Lemmas~\ref{lem:JI} and~\ref{lem:fano-odd} can certify that a given pattern is impossible in a given alphabet, or that a recovery pattern needs larger shares. More generally, the dictionary and the transfer arguments carry tools from classical network coding, secret sharing and matroid theory to quantum codes and entanglement. Conversely, Lemma~\ref{lem:fano-odd} also gives a classical result: with a bound of Jafari and Khazaei for $\GF(2)$-linear schemes~\cite[Prop.~4.5]{JafariKhazaei2021}, every abelian classical scheme for the twelve-player union of the ports of the Fano and non-Fano matroids whose group has an elementary abelian Sylow $2$-subgroup has information ratio greater than $7/6$, which bears on a question of Jafari and Khazaei (Methods).

Several questions remain open. For stabilizer codes, we do not know the optimal information ratios for the two recovery patterns in the dimensions in which ideal schemes do not exist: for Steane's pattern with qutrits, for instance, and for that of the color code with qubits, they lie between $4/3$ and $7/5$ (Supplementary Note~\ref{app:qss-upper}). Nor do we know which other recovery patterns are rigid in the sense of Sec.~\ref{sec:qss}, for instance the ports of other matroids that admit ideal quantum schemes~\cite{SarvepalliRaussendorf2010}, or for which access structures entanglement enlarges the set of dimensions in which ideal schemes exist, as for the threshold pattern of the $36$ officers~\cite{Rather2022}, and for which it does not, as for the ports of $\mathrm{AG}(3,2)$ and of $R_8$ (Theorem~\ref{thm:ideal}). Does the five-party cone $\StabZ{d}{5}$ depend on $d$? By Proposition~\ref{prop:sandwich}, for two primes a dependence would give a six-variable linear rank inequality that holds in one characteristic but not in the other; the six-variable linear rank inequalities are not known completely~\cite{DFZ2015char}, and whether the characteristic matters on six variables is open~\cite{Khazaei2026}. Our data show no dependence at five parties, but this is evidence, not proof: the catalog may be incomplete, $1{,}155$ of its orbits are undecided, and the stabilizer cones may have extreme rays that are not extreme rays of $\QLR_5$. For many pairs of local dimensions, including any two odd primes, comparisons with fewer parties than our witnesses use remain open. Supplementary Note~\ref{app:discussion} discusses composite local dimensions, which relate to a question of Khazaei on abelian and mixed-linear entropy cones~\cite{Khazaei2026}, and the conjecture of Linden \emph{et al.} in more detail. Finally, for prime local dimensions the six-party witnesses differ only in whether $2=0$, and the inequalities that separate them assume that $2$ acts injectively on the group or annihilates it. Gross's discrete Wigner function, which on pure states of odd dimension is nonnegative exactly for stabilizer states, is likewise defined through the inverse of $2$~\cite{Gross2006}; we do not know whether the two dichotomies between qubits and odd local dimensions, in entanglement and in phase space, are related.
\section{Methods}
\label{sec:methods}

The Supplementary Information contains the full proofs; Supplementary Table~\ref{tab:summary} lists the results with what they build on and their proof status, and Supplementary Note~\ref{app:repo} locates every result in the repository. Sections~\ref{sec:methods-proofs}--\ref{sec:sandwich} give the tools behind the results on stabilizer cones and state Theorems~\ref{thm:twopowers} and~\ref{thm:primes}; Sec.~\ref{sec:methods-qss} gives those behind the results on codes and secret sharing, with an outline of the proof of Theorem~\ref{thm:ideal}; Sec.~\ref{sec:methods-five} treats four and five parties; and Secs.~\ref{sec:methods-comp}--\ref{sec:methods-prior} describe the computations, the data and code, the use of AI tools and the relation to earlier work.

\subsection{Rank functions and the dictionary}
\label{sec:methods-proofs}

A \emph{polymatroid} on a finite set $\mathcal{E}$ is a monotone and submodular function $h$ on the subsets of $\mathcal{E}$ with $h(\emptyset)=0$; polymatroids form the Shannon cone~\cite{Yeung1997}, cut out by the \emph{elemental} inequalities $h(i|\mathcal{E}\setminus i)\ge0$ and $I(i;j|K)\ge0$, where $h(X|Y)=h(XY)-h(Y)$, $I(X;Y|Z)=h(XZ)+h(YZ)-h(XYZ)-h(Z)$, and juxtaposition denotes union. Subgroup rank functions are, up to the duality of finite abelian groups, the abelian case of the group-characterizable entropy functions of Chan and Yeung~\cite{ChanYeung2002}. $\Lc_m$, the closed convex hull of the cones $\LcZ{p}{m}$ of $\GF(p)$-linear rank functions over all primes $p$, is the cone of linear rank functions over all fields, whose valid inequalities are the \emph{linear rank inequalities}~\cite{Hammer2000,DFZ2009}. The \emph{contraction} of $h$ by $K\subseteq\mathcal{E}$ is $T\mapsto h(T\cup K)-h(K)$ on the subsets of $\mathcal{E}\setminus K$; for subgroups $U_i$ of a finite abelian group $V$ it is the rank function of the images of the other subgroups in $V/U_K$, where $U_K=\sum_{i\in K}U_i$, a group whose order and exponent divide those of $V$. Two subgroups $U_X$ and $U_Y$ have a \emph{common information}, $Z=U_X\cap U_Y$, with $h(XZ)=h(X)$, $h(YZ)=h(Y)$ and $h(Z)=I(X;Y)$; Shannon inequalities applied to such extensions give valid inequalities that are not Shannon inequalities, the first being the Ingleton inequality~\cite{Ingleton1971,Hammer2000}. For prime $d$, $\Z_d=\GF(d)$.

Proposition~\ref{thm:reduction} combines three facts (Supplementary Note~\ref{app:reduction} and Supplementary Fig.~\ref{fig:dictionary}). (a) For every stabilizer state of $\Z_d$ qudits, $h_S$ is the rank function of an arrangement of $n+1$ subgroups of a finite abelian group of exponent dividing $d$. (b) For a code $C\le\Z_d^m$ whose coordinates are distributed among the systems, let $r_C(X)=\log\abs{\pr_X C}$, where $\pr_X$ keeps the coordinates held by $X$; then $\ket{C}$ has entropy vector $\lambda_{r_C}$, and every subgroup rank function of an abelian group of exponent dividing $d$ equals $r_C$ for some such $C$. (c) If $S$ is pure-symmetric, then $h_S$ is self-dual and $\lambda_{h_S}=2S$; conversely, if $h$ is self-dual, then $S=\tfrac12\lambda_h$ is pure-symmetric and $h_S=h$; this is due to Mat\'u\v{s}~\cite[Thm.~1 and Remark~1]{Matus2012}, and Jowett, Mo and Whittle found the same correspondence for connectivity functions~\cite{JowettMoWhittle2016}. By (a), (i) implies (ii); by (b), $\lambda(\LcZ{d}{n+1})$ lies in the closed cone generated by CSS entropy vectors; and if (ii) holds, then $S=\lambda_{h_S/2}$ by (c). Fact~(a) is closely related to results of Gross and Walter, of Linden, Mat\'u\v{s}, Ruskai and Winter, of Majenz and of Osthus~\cite{GrossWalter2013,LMRW2013,Majenz2014,Osthus2024}; to our knowledge, its converse, that every self-dual vector of the rank cone is attained in the closure by CSS states, has not been stated beyond identically self-dual matroids, which give secret sharing schemes~\cite{SarvepalliRaussendorf2010}, and for prime $d$ it also follows from (c) and the entropy formula of Ref.~\cite{Osthus2024}. Corollary~\ref{cor:doubling} reads the Lagrangian subgroup of a stabilizer state as a classical code in which each system holds both coordinates of each of its qudits (Ref.~\cite{LTA2024} and references therein; Supplementary Note~\ref{app:reduction-css}). If $d$ is divisible by the square of a prime, stabilizer states include states outside the Clifford orbit of $\ket{0\cdots0}$~\cite{FengLuo2024}; our separations also hold for the orbit, since their witnesses can be taken in it (Supplementary Note~\ref{app:reduction-stab}).

Let $\Stab_n$ and $\Stab^{\mathrm{pr}}_n$ be the closed convex hulls of $\bigcup_{d\ge2}\StabZ{d}{n}$ and of $\bigcup_p\StabZ{p}{n}$ over primes $p$. For pure-symmetric $S$,
\begin{equation}\label{eq:union}
\begin{aligned}
  \Stab_n&=\bigl\{S:h_S\in\Lab_{n+1}\bigr\},\\
  \Stab^{\mathrm{pr}}_n&=\bigl\{S:h_S\in\Lc_{n+1}\bigr\},
\end{aligned}
\end{equation}
where $\Lab_m$ is the closed convex cone generated by all subgroup rank functions of finite abelian groups (Supplementary Note~\ref{app:reduction-unions}). A linear form $f$ on vectors indexed by the subsets of $N$ is \emph{balanced} if $f(\ind{i})=0$ for every $i\in N$, where $\ind{i}(X)=1$ if $i\in X$ and $0$ otherwise.

\begin{corollary}[Balanced inequalities]\label{cor:balanced}
A pure-symmetric vector $S$ lies in $\StabZ{d}{n}$ if and only if $f(S)\ge 0$ for every balanced linear form $f$ on $N$ that is nonnegative on $\LcZ{d}{n+1}$. The same holds for $\Stab^{\mathrm{pr}}_n$ with the balanced $(n+1)$-variable linear rank inequalities.
\end{corollary}

The proof adapts Chan's balancing argument~\cite{Chan2003} (Supplementary Note~\ref{app:reduction-balanced}); the direction ``only if'' is in essence due to Gross and Walter~\cite{GrossWalter2013}. For $n\ge6$, beyond the range of Ref.~\cite{BCHS2020}, we define $\QLR_n$ by the balanced parts $f_{\mathrm{bal}}(h)=f(h)-\sum_if(\ind{i})\,h(i\mid\text{rest})$ of all $n$-variable linear rank inequalities $f$ on every $n$ of the $n+1$ systems, together with weak monotonicity. For five parties, $\QLR_5$ is defined by $31$ inequalities up to symmetry~\cite[footnote~16]{BCHS2020}, from the five-variable list of Dougherty, Freiling and Zeger~\cite{DFZ2009}, and by Theorem~12 of Ref.~\cite{LMRW2013} stabilizer entropy vectors satisfy the inequalities of Kinser~\cite{Kinser2011} and of Dougherty \emph{et al.} By Corollary~\ref{cor:balanced}, $\StabZ{p}{n}\subseteq\QLR_n$ for every prime $p$; for $n\le5$ this holds for every $d$, since the five-variable linear rank inequalities follow from Shannon inequalities and common information~\cite[Secs.~2 and~3]{DFZ2009}, and for $n\ge6$ it holds for squarefree $d$ and is open otherwise (Supplementary Note~\ref{app:discussion}).

\subsection{Inequalities for the configuration \texorpdfstring{$J-I$}{J-I}}
\label{sec:lemma1}

Let $t\ge2$ and $M=t+1$. In $\Rg^M$, where $\Rg$ is a field or the ring $\Z_N$ for some modulus $N\ge2$, consider the configuration
\begin{equation}\label{eq:JI}
  A_i=\langle e_i\rangle,\quad B_i=\langle c-e_i\rangle,\quad C=\langle c\rangle\qquad(i=1,\dots,M),
\end{equation}
where $c=(1,\dots,1)$, so that the vectors $c-e_i$ are the columns of $J-I$. Since $\det(J-I)=(-1)^t t$, over a field $B_1+\dots+B_M$ has dimension $M$ if the characteristic does not divide $t$ and $M-1$ if it does; for $t=2$ the configuration is the Fano configuration in characteristic $2$ and the non-Fano configuration otherwise.

\begin{lemma}[Inequalities for $J-I$]\label{lem:JI}
Let $t\ge2$, $M=t+1$, and let $A_1,\dots,A_M$, $B_1,\dots,B_M$ and $C$ be subgroups of a finite abelian group $V$ with rank function $h$. Write $A=\sum_kA_k$, $A_{\ne i}=\sum_{k\ne i}A_k$, and define the nonnegative error terms
\begin{align*}
  \alpha&=\textstyle\sum_kh(A_k)-h(A), & \gamma&=h(C|A),\\
  \delta_i&=I(C;A_{\ne i}), & \eps_i&=h(C|A_iB_i),\\
  \zeta_i&=h(B_i|A_{\ne i}), & \eta_i&=h(B_i|A_iC).
\end{align*}
\begin{itemize}
\item[\LB] If multiplication by $t$ is injective on $V$, then
\[
  M\,h(C)\le h(B_1\cdots B_M)+M\Bigl[\gamma+M\alpha+\textstyle\sum_i(\delta_i+\eps_i+\zeta_i)\Bigr].
\]
\item[\LA] If $tV=0$, then
\[
  h(B_1\cdots B_M)\le(M-1)\,h(C)+M\alpha+M\gamma+\textstyle\sum_i(\zeta_i+\eta_i).
\]
\end{itemize}
\end{lemma}

The proof (Supplementary Note~\ref{app:lemma1}) removes the overlaps of the $A_i$ by a quotient and compares the $B_i$ with a twisted diagonal copy of $C$; the hypothesis on $t$ enters only through $M-1=t$. On the configuration~\eqref{eq:JI} over a field every error term vanishes, so \LB{} fails by one unit when the characteristic divides $t$ and \LA{} when it does not. For vector spaces, \LB{} follows from Theorem~9 of Pe\~na-Macias and Sarria~\cite{PenaSarria2021} by adding Shannon terms, and their Theorem~7 plays the role of \LA{} (Supplementary Note~\ref{app:lemma1-ps}; see also Refs.~\cite{PenaSarria2019,PenaMacias2022} and~\cite[Cor.~2.7]{PenaMacias2023}). Deriving $C$ as $(A_1+B_1)\cap\dots\cap(A_M+B_M)$ gives versions \LB$^*$ and \LA$^*$ in $2M$ variables (Supplementary Note~\ref{app:lemma2}).

\subsection{Blindness}
\label{sec:blind}

\begin{proposition}[Blindness]\label{prop:blind}
Let $f$ be any of the eight seven-variable characteristic-dependent inequalities of Supplementary Tables~\ref{tab:eight} and~\ref{tab:eight-coeffs}, which are \LA{} and \LB{} with $t=2$ and six inequalities from the literature~\cite{DFZ2015rate,BKL2011,PenaSarria2019} (their validity is not used). Then $f+f\circ\Dual$ is a nonnegative rational combination of Shannon inequalities on seven variables. Consequently $f(h)\ge0$ for every self-dual polymatroid $h$ on seven elements, in every relabeling of the variables, and for six parties and every bijection between the seven variables and the six parties and the purifier, the pull-back of $f$ holds for every pure quantum state.
\end{proposition}

The certificates are exact (Supplementary Note~\ref{app:blind}, which shows in the same way that the ninth inequality that we checked, the first inequality of Corollary~2.7 of Ref.~\cite{PenaMacias2023} with $n=3$, is blind as well); Supplementary Note~\ref{app:misprints} records two printed inequalities that are false as stated. A dependence at six parties requires a seven-variable inequality $g$ with $g\circ\Dual=g$ that is valid for one local dimension and violated by the lift of a stabilizer state of another; Proposition~\ref{prop:blind} does not cover substitution instances or inequalities applied to families derived from the parties' subgroups.

\subsection{Transfer through a derived family}
\label{sec:transfer}

Let $U_P$ ($P\in\{A,B,C,W,X,Y,Z\}$) be subgroups of a finite abelian group with rank function $h$. For a line $\ell=\{i,j,k\}$ put
\begin{equation}\label{eq:derived}
  F_i^\ell=U_i\cap(U_j+U_k),\qquad \hat F_i=\sum_{\ell\ni i}F_i^\ell\ \le U_i ,
\end{equation}
and let $\hat h$ be the rank function of the derived family $(\hat F_P)$ (Supplementary Fig.~\ref{fig:fano}). Each piece $F_i^\ell$ is the common information of $U_i$ and $U_j+U_k$, and the pieces of a line generate each other: if $v=x+y\in U_i$ with $x\in U_j$ and $y\in U_k$, then $x=v-y\in U_j\cap(U_i+U_k)$, and likewise for $y$. With submodularity, this bounds the ranks of the derived family from above on points, on lines and on $ABCZ$, and from below on every set, by linear forms in $h$ (Supplementary Lemma~\ref{lem:transfer}). Any inequality valid for subgroup arrangements of a given kind holds for $\hat h$, and if its terms with positive coefficients are ranks of points, lines or $ABCZ$, bounding each term, from above or from below according to the sign of its coefficient, gives a new inequality for $h$, as in the common-information method~\cite{Hammer2000,DFZ2009}. Applied to \LB{} with $t=2$ and $(A_1,A_2,A_3,B_1,B_2,B_3,C)=(A,B,C,Y,X,W,Z)$, that is, to an inequality $\LB_2(h)\ge0$ that holds in groups of odd order and takes the value $-1$ on the Fano configuration (Supplementary Eq.~\eqref{eq:LB2}), this gives
\begin{equation}\label{eq:f9}
  0\le\LB_2(\hat h)\le\tfrac13\fnine(h)
\end{equation}
for an explicit integer linear form $\fnine$ with $67$ terms (Supplementary Eq.~\eqref{eq:F-f9}).

\begin{proof}[Proof of Theorem~\ref{thm:steane}]
For a stabilizer state of $\Z_d$ qudits with $d$ odd, $h_S$ is a subgroup rank function in a group of odd order (fact~(a)), so $\fnine(h_S)\ge0$ by Eq.~\eqref{eq:f9}, and by linearity and closedness $\fnine\ge0$ on the lift of $\StabZ{d}{6}$. At $\hF$, each bound used in Eq.~\eqref{eq:f9} equals the corresponding rank of the Fano configuration, so $\fnine(\hF)/3=\LB_2(\rF)=-1$ (Supplementary Note~\ref{app:transfer}). For $n>6$, add parties that hold nothing.
\end{proof}

On the slice of self-dual vectors invariant under the collineations of the Fano plane, with $v_1=h(i)$ on points and, for a pair $\{i,j\}$, a line $\ell\ni i$ and a noncollinear triple $T\ni i$,
\begin{equation}\label{eq:invariants}
  \tau=I(i;j),\quad \mu=I(i;\ell\setminus i),\quad \nu=I(i;T\setminus i),\quad \beta=v_1-2\mu,
\end{equation}
one has
\begin{equation}\label{eq:invariant-f9}
  \tfrac13\fnine=210\,\beta+30\,\tau+21\,\nu-\mu
\end{equation}
(Supplementary Note~\ref{app:transfer}). On lifts, $\beta=2S(i\mid\ell\setminus i)$; so a stabilizer state of odd local dimension whose entropy vector is invariant under the collineations and has $\mu>0$ also has $\beta$, $\tau$ or $\nu$ positive, unlike the Steane state, where $\mu=1$ and $\tau=\nu=\beta=0$. The form $\fnine$ is not a nonnegative combination of Shannon inequalities and relabelings of the eight inequalities of Proposition~\ref{prop:blind}. Since the smallest face of the Shannon cone that contains $\hF$ is spanned by $\rF$ and $\rF^*$, Theorem~\ref{thm:steane} also shows that for every odd prime $q$ no nonzero combination $a\rF+b\rF^*$ with $a,b\ge0$ is a limit of positive combinations of $\GF(q)$-linear rank functions (Supplementary Proposition~\ref{prop:faces}).

\subsection{Transfer through the stabilizer group}
\label{sec:reverse}

Take $Z$ as the purifier. The non-Fano configuration has the six lines other than $WXY$; the \emph{corners} $A$, $B$, $C$, $Z$ lie on three of them and the \emph{diagonal points} $W$, $X$, $Y$ on two. For a pure-symmetric vector $S$ on the seven systems, a set $T$, a line $\ell$ and a point $P\in\ell$, put $k(T)=\sum_{i\in T}S(i)-S(T)$ and $k_P(\ell)=k(\ell)-k(\ell\setminus P)$, and, for a corner $P$ and a line $\ell\ni P$, $e_P^\ell=2S(P)-k_P(\ell')-k_P(\ell'')$, where $\ell'$ and $\ell''$ are the other two lines through $P$. Let
\begin{equation}\label{eq:UZ}
\begin{aligned}
  \mathcal{U}_Z&=4S(Z)-k_Z(AYZ)-k_Z(BXZ)-k_Z(CWZ),\\
  \mathcal{U}_W&=2S(W)-k_W(ABW)+e_A^{ABW}+e_B^{ABW},\\
  \mathcal{U}_X&=2S(X)-k_X(ACX)+e_A^{ACX}+e_C^{ACX},\\
  \mathcal{U}_Y&=2S(Y)-k_Y(BCY)+e_B^{BCY}+e_C^{BCY},
\end{aligned}
\end{equation}
and
\begin{equation}\label{eq:f12}
\begin{aligned}
  \ftwelve(S)={}&5\,\mathcal{U}_Z+4\,\mathcal{U}_W+\mathcal{U}_X+\mathcal{U}_Y\\
  &+k(WXY)+6\,k(ABC)+\textstyle\sum_{\{P,P'\}}k(PP')\\
  &-\textstyle\sum_\ell k(\ell)-3\,k(ABW)-3\,k(CWZ),
\end{aligned}
\end{equation}
where $\{P,P'\}$ runs over the six pairs of corners and $\ell$ over the six lines; in the entropies, $\ftwelve$ has $31$ terms (Supplementary Eq.~\eqref{eq:R-f12}). At $\SNF$, $k$ is $1$ on the six lines and $0$ on pairs and on the other triples, which gives $\ftwelve(\SNF)=5+4+1+1-6-3-3=-1$; at $\SF$, also $k(WXY)=1$, and $\ftwelve(\SF)=0$.

\begin{proposition}[A qubit inequality]\label{prop:f12}
For every $S\in\StabZ{2}{6}$, with $Z$ as the purifier, $\ftwelve(S)\ge0$.
\end{proposition}

For a qubit stabilizer state in normal form, with $S(i)$ qubits on each system $i$ (Supplementary Lemma~\ref{lem:A-normal}), $k(T)$ is the dimension of the space of stabilizer elements supported in $T$, and $k_P(\ell)$ that of the restriction to $P$ of those supported in $\ell$. Where two lines meet in a point, these restrictions are orthogonal, because the stabilizer group is isotropic; dimension counts in the symplectic spaces of the systems bound the restrictions to $P=Z,W,X,Y$ by the $\mathcal{U}_P$, and \LA{} with $t=2$, applied to their images in a quotient, gives $\ftwelve(S)\ge0$ (Supplementary Note~\ref{app:reverse}). The value of $\ftwelve$ at $S^{(d)}$ in Theorem~\ref{thm:reverse} follows from the Smith invariants of the sets of columns of the $3\times7$ matrix, which are all $1$ except for those of $\{W,X,Y\}$, which are $1$, $1$ and $2$ (Supplementary Note~\ref{app:reverse-sevenpoint}). For $d\ge3$ this value is negative, while $S^{(d)}\in\StabZ{d}{6}$; with Proposition~\ref{prop:f12} and Corollary~\ref{cor:divisibility}, this proves Theorem~\ref{thm:reverse}. A version of \LA{} for groups of exponent $2^a$ gives a variant of Proposition~\ref{prop:f12} for every power of $2$ (Supplementary Note~\ref{app:reverse-powers}).

\begin{theorem}[Six parties, powers of two]\label{thm:twopowers}
Let $a\ge1$ and $f^{(a)}=\ftwelve+(1-1/a)\,\mathcal{U}_Z$. Then $f^{(a)}(S)\ge0$ for every $S\in\StabZ{2^a}{6}$, with $Z$ as the purifier. Consequently, for all $n\ge6$, $d\ge2$ and $a\ge1$,
\[
  \StabZ{d}{n}\subseteq\StabZ{2^a}{n}\iff d\mid2^a .
\]
In particular, for every odd $d\ge3$ the cones of $\Z_d$ and of $\Z_{2^a}$ qudits are incomparable from six parties on, and $\StabZ{2^j}{n}\subsetneq\StabZ{2^a}{n}$ for $1\le j<a$ and $n\ge6$.
\end{theorem}

The witnesses are $\SNF$ and the vectors $S^{(2^j)}$ of Theorem~\ref{thm:reverse} with $j>a$. As $a\to\infty$, the vectors $S^{(2^a)}/a$ of $\StabZ{2^a}{6}$ tend to $\SNF$, which lies in none of these cones: the union of the increasing cones $\StabZ{2^a}{6}$ is not closed.

\subsection{Many parties}
\label{sec:sandwich}
\label{sec:cubes}

\begin{proposition}[Sandwich]\label{prop:sandwich}
Let $d,d'\ge2$ and $n\ge1$. Then
\[
  \LcZ{d}{n}\not\subseteq\LcZ{d'}{n}\ \Longrightarrow\ \StabZ{d}{n}\not\subseteq\StabZ{d'}{n}\ \Longrightarrow\ \LcZ{d}{n+1}\not\subseteq\LcZ{d'}{n+1}.
\]
\end{proposition}

The second implication is Proposition~\ref{thm:reduction}. The first uses doubled codes: if a subgroup rank function $r$ on the $n$ parties is realized by a code $C$, the CSS state of $\{(c,c):c\in C\}$, in which the purifier holds the second copy, has $S(X)=r(X)$ for every set $X$ of parties (Supplementary Note~\ref{app:cubes-sandwich}).

\begin{definition}[Generalized cube states]\label{def:cube}
Let $t\ge2$, $M=t+1$, and let $\Rg$ be $\GF(p)$ or $\Z_N$. Consider the $2t+4$ points $e_1,\dots,e_M$, $c-e_1,\dots,c-e_M$, $c$ and $0$ of $\Rg^M$, and the code $C_t(\Rg)$ of all affine functions $x\mapsto a_0+a\cdot x$ evaluated at these points. Put one qudit of dimension $\abs{\Rg}$ on each point; the point $e_i$ belongs to party $A_i$, the point $c-e_i$ to party $B_i$, the point $c$ to party $C$, and the origin to the purifier $O$. The \emph{generalized cube state} is the CSS state $\ket{C_t(\Rg)}$ on $n=2t+3$ parties.
\end{definition}

For $t=2$ the points are the eight vertices of the cube $\{0,1\}^3$ (Supplementary Fig.~\ref{fig:cube}), $(A_1,A_2,A_3,B_1,B_2,B_3,C,O)$ is $(A,B,C,Y,X,W,Z,R)$ of Sec.~\ref{sec:qss}, and over $\GF(2)$ the code is the extended Hamming code $[8,4,4]$.

\begin{lemma}[Self-duality]\label{lem:cube-selfdual}
If $t-1$ is a unit in $\Rg$, the rank function $r$ of $C_t(\Rg)$ satisfies $r^*=r$. Hence the generalized cube state has $S(X)=2r(X)-\abs{X}\log\abs{\Rg}$ and $h_S=2r$, and contracting $h_S$ by the purifier gives twice the rank function of the configuration~\eqref{eq:JI}.
\end{lemma}

With these witnesses, the inequalities of Lemma~\ref{lem:JI} and their versions with a derived $C$ separate the cones of any two primes.

\begin{theorem}[Any two primes]\label{thm:primes}
Let $p<q$ be primes. For every $n\ge2p+3$, and for every $n\ge2p+2$ if $p$ is odd,
\[
  \StabZ{p}{n}\not\subseteq\StabZ{q}{n}\qquad\text{and}\qquad\StabZ{q}{n}\not\subseteq\StabZ{p}{n}.
\]
In particular, the cones of $\Z_3$ and of every prime $q\ge5$ are incomparable from eight parties on; for $p=2$, Theorems~\ref{thm:steane} and~\ref{thm:reverse} give six parties.
\end{theorem}

For $n=2p+3$, contract $h_S$ by the purifier and evaluate \LB{} or \LA{} with $t=p$: by Proposition~\ref{thm:reduction} and Lemma~\ref{lem:JI}, the pulled-back \LB{} is nonnegative on $\StabZ{q}{n}$ and the pulled-back \LA{} on $\StabZ{p}{n}$, while they take the value $-2$ on the generalized cube states over $\GF(p)$ and over $\GF(q)$, respectively, which are self-dual since $0<p-1<p<q$. For $n=2p+2$ with $p$ odd, \LB$^*$ and \LA$^*$ and Proposition~\ref{prop:sandwich} give both non-inclusions, with doubled codes as witnesses. The generalized cube states over $\GF(p)$ and over $\GF(q)$ with $t=p$, in units of $\log p$ and $\log q$, differ on exactly one bipartition (Supplementary Proposition~\ref{prop:onecut}); for $p=2$ it separates the vertices of the cube $\{0,1\}^3$ of odd weight from those of even weight. Yet no limit of positive combinations of entropy vectors of qubit stabilizer states lies on the ray of the entropy vector of the qutrit cube state.

Theorem~\ref{thm:divisibility} uses the witnesses of Theorem~\ref{thm:primes} over $\GF(p)$, the cube states over $\GF(2)$ and $\Z_4$, and doubled codes of the configuration~\eqref{eq:JI} over $\Z_{p^{j+1}}$, where $p^j$ is the largest power of a prime $p$ that divides $d'$ and $p^{j+1}$ divides $d$ (Supplementary Note~\ref{app:cubes-divisibility}). Its bound $2\max(d,d')+2$ is not optimal in general: besides Theorems~\ref{thm:primes} and~\ref{thm:twopowers}, for instance, $\StabZ{3}{n}\subsetneq\StabZ{9}{n}$ for $n\ge8$. Galois qudits of dimension $p^k$, with the Pauli group defined through the field trace~\cite{Ketkar2006}, give exactly the cone $\StabZ{p}{n}$ (Supplementary Proposition~\ref{prop:galois}).

\subsection{Secret sharing}
\label{sec:methods-qss}

For every state, stabilizer or not, the lift turns the access conditions into linear ones: for every set $T$ of players,
\begin{equation}\label{eq:qss-lift}
  h_S(R)+h_S(T)-h_S(RT)=I(R;T),
\end{equation}
because the terms $\sum_{P\in T}S(P)$ cancel. Through the lift, stabilizer schemes of prime local dimension and classical linear schemes correspond (Supplementary Lemmas~\ref{lem:qss-linear} and~\ref{lem:qss-css}); the converse direction is a version of a construction of Smith~\cite{Smith2000}. The following inequality holds in every group of odd order and assumes nothing about access.

\begin{lemma}[Fano lines in groups of odd order]\label{lem:fano-odd}
Let $T_0,T_1,\dots,T_6$ be subgroups of a finite abelian group of odd order, with rank function $h$, and write $I(0;X)=h(0)+h(X)-h(0X)$. Then
\begin{equation}\label{eq:fano-odd}
\begin{aligned}
&h(1)+h(2)+h(3)\\
&\quad\ge I(0;123)+I(0;14)+I(0;25)+I(0;36)\\
&\qquad-I(0;4)-I(0;5)-I(0;6)\\
&\qquad-I(0;234)-I(0;135)-I(0;126)-I(0;456).
\end{aligned}
\end{equation}
The same holds for subspaces of a vector space over a field of characteristic other than $2$, with $h$ the dimension. For the Fano configuration over $\GF(2)$, labeled so that $014$, $025$, $036$, $234$, $135$, $126$ and $456$ are its lines, the left side is $3$ and the right side is $4$.
\end{lemma}

The proof (Supplementary Note~\ref{app:qss-contract}) follows the argument of Ref.~\cite[Prop.~4.4]{JafariKhazaei2021} for linear schemes, with quotients in place of complements, which need not exist in groups; the order of the group enters only through the injectivity of multiplication by $2$. For Proposition~\ref{prop:qss}~(i), let $\ell=\{x_1,x_2,x_3\}$ be a line, $q$ a player not on it, and $x_{i+3}$ the third point of the line through $q$ and $x_i$; apply Lemma~\ref{lem:fano-odd} to the contraction by $q$ of the lift, with $T_0$ from $R$ and $T_i$ from $x_i$. By Eq.~\eqref{eq:qss-lift} the terms $I(0;X)$ are $2S(R)$ for $X=123$, $14$, $25$ and $36$ and $0$ for the other seven sets, while the left side is at most $2\sum_iS(x_i)$, so $\sum_{P\in\ell}S(P)\ge4\,S(R)$. For $\Z_d$ qudits with $d$ not a power of $2$, the factor of the state on its $\Z_{p^a}$ qudits, for an odd prime power $p^a$ that divides $d$ exactly, is a stabilizer scheme of odd local dimension by itself, which gives the bound in numbers of qudits (Supplementary Note~\ref{app:qss-contract}). Part~(ii) follows from a bound of Jafari and Khazaei for $\GF(2)$-linear schemes~\cite[Prop.~4.5]{JafariKhazaei2021}, applied to the contraction of a qubit scheme in normal form (Supplementary Note~\ref{app:reduction-normal}) by one of $A$, $B$, $C$ and $Z$. Part~(iii) uses $(\lambda,\omega)$-decompositions in the sense of van Dijk, Jackson and Martin~\cite{vDJM1998} (here $\lambda$ is not the connectivity function) into seven ideal schemes for access structures in which one line and its complement exchange their status: relabeled cube codes over $\Z_d$ for $\Gamma_{\mathrm F}$, and the binary cube code with six relabeled codes over $\GF(4)$ for $\Gamma_{\mathrm{NF}}$. For part~(iv), the lift of an ideal scheme is a multiple of the rank function of $\mathrm{AG}(3,2)$ or of $R_8$ (Supplementary Lemma~\ref{lem:K-forcing}, by certified linear programs, or, without a computer, as for Theorem~\ref{thm:ideal}); a normal form of the Lagrangian subgroup of the state (Supplementary Note~\ref{app:reduction-stab}), reached along the lines of the Fano plane, then shows that $2x=0$ on the phase space of a player for $\Gamma_{\mathrm F}$, and that this phase space has odd order for $\Gamma_{\mathrm{NF}}$ (Supplementary Lemma~\ref{lem:K-two}). 

\emph{Outline of the proof of Theorem~\ref{thm:ideal}} (Supplementary Note~\ref{app:qss-general}, which uses the same step numbers; Fig.~\ref{fig:ideal}). Let $\mathcal M$ be $\mathrm{AG}(3,2)$ for $\Gamma_{\mathrm F}$ and $R_8$ for $\Gamma_{\mathrm{NF}}$, the matroids whose ports at $R$ are these access structures, and consider a scheme whose secret and seven shares all have dimension $D$. \emph{Reduction and purity.} Every share has entropy $\log D$, and by Eq.~\eqref{eq:qss-lift} the reference is an ideal element, in the sense of Mat\'u\v{s}, of the lift of the entropy vector, which is a polymatroid for every quantum state. By his Theorem~2~\cite{Matus2012}, the lift is a multiple of the rank function of a matroid whose circuits through $R$ are those of $\mathcal M$; both matroids are connected, so by Lehman's theorem~\cite{Lehman1964} this matroid is $\mathcal M$. Hence the joint state of $R$ and the shares is pure, and $S(X)=(2r_{\mathcal M}(X)-\abs X)\log D$ for every set $X$ of the eight systems; this step needs no computer, and the certified forcing lemma (Supplementary Lemma~\ref{lem:K-forcing}) gives the same. \emph{Step~1.} Identify the eight systems with the points of $\GF(2)^3$ as in Sec.~\ref{sec:qss}, and call a direction $\delta$ good if any two of the four pairs $\{v,v+\delta\}$ form a plane of $\mathcal M$: every direction for $\mathrm{AG}(3,2)$, and for $R_8$ the directions of the edges and of the body diagonals of the cube. Along a good direction, the four pairs have the entropies of a GHZ state of local dimension $D^2$, and the structure theorem of Hayden \emph{et al.} for equality in strong subadditivity~\cite{HJPW2004} shows that they are in such a state, in suitable bases (Supplementary Lemma~\ref{lem:K-ghz}). Since any three of the eight systems are in a maximally mixed state, every basis vector of a pair $\{v,v+\delta\}$ is maximally entangled, of the form $(A_j\otimes1)\ket{\Phi}$ for a unitary $A_j$ on $v$ alone, where $\ket{\Phi}$ is a fixed maximally entangled state of the pair; so the operators below are tensor products of unitaries on single systems. \emph{Step~2.} Every plane of $\mathcal M$ splits into pairs along at least two good directions, and for each such splitting the support of its reduced state is spanned by the products of corresponding basis vectors of its two pairs; comparing two of these spans shows that the unitaries $A_jA_1^{-1}$ commute up to scalars. By Supplementary Lemma~\ref{lem:K-nice} they form a nice error basis~\cite{Knill1996} indexed by an abelian group (this commutation is what rules out nonabelian index groups), and tensor products of such unitaries on the four systems of the plane give $D^2$ unitaries that fix the state. \emph{Step~3.} At each system these unitaries generate, up to scalars, one group that does not depend on the direction, and the groups of all systems are isomorphic. \emph{Step~4.} The unitaries from all planes fix the state, and their factors at each system lie in the group of Step~3, whose elements commute up to scalars; so the commutator of any two of them is a scalar that fixes the state, hence $1$. They generate an abelian group without scalars other than $1$, with at least $D^8$ elements, from the three planes formed by one pair of a good direction with each of its other pairs and from a plane that contains no pair of that direction, and at most $D^8$, because the projector onto their common fixed space, which contains the state, has trace $D^8$ divided by their number; so the state is a stabilizer state. \emph{Step~5.} A symplectic basis for the commutator form of the group at each system (Supplementary Lemma~\ref{lem:K-symp}) gives a unitary on that system after which the stabilizers are, up to phases, tensor products of Weyl operators of qudits whose dimensions multiply to $D$, on a phase space $V$ of order $D^2$. \emph{Step~6.} The phase-space labels of the stabilizers form a subgroup of $V^8$ whose projection to every set $X$ of systems has order $\abs V^{r_{\mathcal M}(X)}$. By Supplementary Lemma~\ref{lem:K-two}, the labels of the stabilizers that act trivially on $R$ take, on the players and after an automorphism of the phase space of each player, the normal form $(u,v,w,u+v,u+w,v+w,u+v+w)$ with $u,v,w\in V$; at $W$, $X$ and $Y$ they take $\abs V^2$ values if $\{R,W,X,Y\}$ is a plane of $\mathcal M$, as for $\mathrm{AG}(3,2)$, which forces $2x=0$ for every $x\in V$, and $\abs V^3$ values if it is not, as for $R_8$, which forces $\abs V$ to be odd. So the systems consist of qubits for $\Gamma_{\mathrm F}$ and of qudits of odd dimensions for $\Gamma_{\mathrm{NF}}$. \emph{Conversely,} copies of the qubit cube state and, for odd $D$, the cube state over $\Z_D$ are such schemes. Steps~1--6 use only the pure state and its entropies, so every pure state of eight systems of dimension $D$ with the entropies found in the first part is a stabilizer state up to local unitaries. For $D=2$ every automorphism of $V=\GF(2)^2$ is symplectic, so a code with Steane's recovery pattern and $D=2$ is Steane's code up to local unitaries (Supplementary Note~\ref{app:qss-general}).

Lemma~\ref{lem:fano-odd} also extends the bound $4/3$ for the port of the Fano matroid to abelian schemes, in the sense of Ref.~\cite{JafariKhazaei2021}, whose group has odd order; with Prop.~4.5 of that work it shows that, for the twelve-player union of the ports of the Fano and non-Fano matroids, every abelian scheme whose Sylow $2$-subgroup is elementary abelian has information ratio greater than $7/6$, the value that the mixed-linear schemes of Jafari and Khazaei approach, which bears on a question that they left open (Supplementary Note~\ref{app:qss-contract}).

\subsection{Four and five parties}
\label{sec:methods-five}

Let $C_4$ be the cone of pure-symmetric vectors $S$ on $N=\{0,1,2,3,4\}$ whose lift $h_S$ satisfies the Shannon inequalities and every instance
\begin{equation}\label{eq:ingleton}
  I(A;B|K)\le I(A;B|C K)+I(A;B|D K)+I(C;D|K)
\end{equation}
of the Ingleton inequality~\cite{Ingleton1971}, with $A,B,C,D$ nonempty and $A,B,C,D,K$ pairwise disjoint subsets of $N$.

\begin{proposition}[Four parties]\label{thm:four}
For every $d\ge2$, $\StabZ{d}{4}=C_4$. In particular, the stabilizer cone of at most four parties does not depend on the local dimension.
\end{proposition}

The inclusion $\StabZ{d}{4}\subseteq C_4$ holds because subgroup arrangements have common informations, which imply the Ingleton inequality~\cite{Hammer2000}. Conversely, the cone $C_4$ has $46$ extreme rays in seven orbits (Supplementary Table~\ref{tab:four}; exact enumeration with \texttt{lrs}~\cite{lrs}), and each is realized, at scale one or two, by one integer weighted graph state over every prime field at once: the cut blocks of its weight matrix, the submatrices between the two sides of a cut, have all nonzero Smith invariants equal to $1$; Corollary~\ref{cor:divisibility} gives every $d$ (Supplementary Note~\ref{app:four}). Proposition~\ref{thm:four} also follows from the literature, in the closure and up to scaling: Linden, Mat\'u\v{s}, Ruskai and Winter characterized the four-party cone of all local dimensions jointly~\cite[Sec.~5]{LMRW2013}, and the extreme rays of $C_4$ are hypergraph entropy vectors~\cite[Sec.~3.3.1]{BCHS2020}, which lie in the stabilizer cone of every prime local dimension~\cite{WalterWitteveen2021}, hence, by Corollary~\ref{cor:divisibility}, of every local dimension. For Theorem~\ref{thm:five}, the common information $Z=(U_A+U_D)\cap(U_B+U_E)$ of a subgroup arrangement $U_A,\dots,U_F$ with rank function $h$ satisfies
\begin{equation}\label{eq:ci}
\begin{gathered}
  g(ADZ)=h(AD),\qquad g(BEZ)=h(BE),\\
  g(Z)=I_h(AD;BE),
\end{gathered}
\end{equation}
where $g$ is the rank function of the seven subgroups. A nonnegative integer combination of $26$ Shannon inequalities on these seven elements, in which $Z$ enters only through the three terms of Eq.~\eqref{eq:ci}, becomes the balanced six-variable linear rank inequality
\begin{equation}\label{eq:star-h}
\begin{aligned}
  &h(ABCD)+h(ABE)+h(BCE)+h(ADE)\\
  &\quad+h(ABF)+h(BCF)+h(ADF)+h(AEF)\\
  &\ge h(BC)+h(ABD)+h(AE)+h(ABCDE)\\
  &\quad+h(AF)+h(ABCF)+h(BEF)+h(ADEF),
\end{aligned}
\end{equation}
and eliminating $F$ by purity at $h=h_S$ gives~\eqref{eq:star}. The lift of $r$ satisfies every Ingleton instance on the six systems. The vector $r$ satisfies all $10{,}860$ facet inequalities of $\QLR_5$, the $105$ that are tight at $r$ have rank $30$, and the left side of~\eqref{eq:star} minus the right side is $-1$ at $r$; graph states with weights $0$ and $1$ give, over every prime field, entropy vectors that span $\R^{31}$ and a space of dimension $30$ on the hyperplane of~\eqref{eq:star} (Supplementary Note~\ref{app:five}).

\subsection{Computations}
\label{sec:methods-comp}

The finite computations behind the results are of five kinds (Supplementary Notes~\ref{app:four}, \ref{app:five}, \ref{app:blind}, \ref{app:transfer} and~\ref{app:qss}). First, polyhedral computations: the extreme rays of the four-party cone were enumerated exactly with \texttt{lrs}~\cite{lrs}, and those of $\QLR_5$ by a symmetry-aware adjacency decomposition~\cite{Bremner2009} in exact arithmetic, from its irredundant description by $10{,}860$ facets. We built this description from the Ingleton inequalities and the $24$ five-variable linear rank inequalities of Ref.~\cite{DFZ2009}, read from the table of the PSITIP software~\cite{Li2023psitip}, which lacks the Ingleton form of Eq.~(39) of Ref.~\cite{DFZ2009}, so we added it; built in the same way, the five-variable linear rank cone has the $7{,}943$ extreme rays in $162$ orbits reported in Ref.~\cite{DFZ2009}. The enumeration started from the $59$ orbits known before~\cite{BCHS2020gap} and one more, and stopped when the catalog exceeded a preset ceiling of $1.9$ million orbits (Supplementary Note~\ref{app:five-catalogue}). Second, certificates: for the orbits of the catalog, in a cutting-plane loop, linear programs ask for an extension of the lift of an orbit representative to a polymatroid on seven elements in which the added element is a common information of two sets of systems, as for~\eqref{eq:star}; infeasibility gives a Farkas certificate, which we converted to exact rational form and verified in integer arithmetic, with conditional Ingleton instances added for some orbits. The six-variable inequalities of Ref.~\cite{DFZ2009}, Eqs.~(44)--(63), were pulled back to linear forms on pure-symmetric vectors for all $720$ assignments of their variables to the six systems (Supplementary Note~\ref{app:five-dfz}). The certificates of Proposition~\ref{prop:blind} write $f+f\circ\Dual$ as nonnegative rational combinations of elemental Shannon inequalities, and rational multipliers for elemental Shannon inequalities and access equations certify each of the $2\times2\times255$ bounds which show that the lift of an ideal scheme, restricted to the reference and the players, is forced to be a multiple of a matroid rank function (Supplementary Lemma~\ref{lem:K-forcing}); a checker rebuilds every inequality and equation that a certificate uses from its description and verifies the identity in rational arithmetic. The form $\fnine$ was expanded by computer and checked by an independent re-derivation. Third, realizations: graph states over $\GF(2)$, $\GF(3)$ and $\GF(4)$, and integer weight matrices, were found by randomized searches, and a realization for every prime is certified either by the Smith forms of all cut blocks or by realizations over $\Q$ combined with realizations at their finitely many exceptional primes. Fourth, searches for schemes: the $\GF(4)$ sub-schemes of Proposition~\ref{prop:qss}~(iii) were found by a linear program over a large sample of ideal schemes over $\GF(2^k)$, $k\le4$, and the qubit scheme for $\Gamma_{\mathrm{NF}}$ with a secret of one qubit and shares of two qubits (Supplementary Note~\ref{app:qss-upper}) by a randomized search; both are verified on all $128$ sets of players. Fifth, an exhaustive enumeration: the stabilizer codes that encode one qubit into seven with distance at least $3$, through the eight-qubit graph states of the code and the reference, up to local Clifford unitaries and with the reference as a fixed vertex (the $1{,}044$ graphs on the seven sites up to relabeling, each with the $127$ nonempty sets of sites adjacent to the reference); this list is complete because every stabilizer state is equivalent under local Clifford unitaries to a graph state~\cite{VandenNest2004}; here the enumeration is the proof, and a verification script repeats it, generating the graphs itself (Supplementary Note~\ref{app:qss-color}). Statements about faces and facets rest on rank computations in exact arithmetic. Randomized searches and floating-point solutions only propose candidates; every result rests on an exact verification or on the exhaustive enumeration, and the one floating-point observation that we report (Supplementary Note~\ref{app:five-oneci}) is marked as such.

\subsection{Data, code and verification}
\label{app:repro}

All data, certificates and code are in the repository of Ref.~\cite{CaoSong2026data}, whose release \reprelease{} is archived at Zenodo; the code is under the MIT license and the data are under the CC BY 4.0 license. The code is written in Python, with NumPy, SciPy and SymPy, and in C; Normaliz, which the repository redistributes, serves only for cross-checks, and the linear programs of Supplementary Note~\ref{app:qss-lp} were solved in floating point with HiGHS and certified in exact arithmetic. The command \texttt{python3 selftest.py -{}-full} runs $59$ gate checks, which re-derive the stored certificates in exact rational arithmetic with a checker that is separate from the code that produced them, and must end with \texttt{ALL PASS}; it takes about $20$ minutes on a computer with two cores. The numbers quoted in this paper were recomputed by separate verification scripts, written for this paper (also with AI assistance; see the next section), that read the stored data but not the repository code, except for the transcriptions of the published inequalities, which were checked against the sources term by term. Theorem~\ref{thm:five} can also be checked with the Python standard library alone and without data files: the script \path{checks/check_thm5.py} rebuilds the $10{,}860$ facets of $\QLR_5$ from their $31$ classes, checks that $r$ satisfies them, that the facets tight at $r$ have rank $30$ and that the left side of~\eqref{eq:star} minus the right side is $-1$ at $r$, and verifies the certificate of~\eqref{eq:star}; \path{checks/check_rev_sec4.py}, which needs no data either, checks the facet statement. Supplementary Note~\ref{app:repo} gives the requirements, describes the gates and the scripts, and relates the numbering of the working notes in the repository, which contain the original statements and proofs, to ours.

\subsection{Use of AI tools}
\label{app:repro-ai}

This work made substantial use of generative AI: the Claude models Opus~5, Fable~5, Fable~5.1 and Opus~5.5 by Anthropic, between August and October 2026. Every computational step used in a proof rests on an exact certificate, an exact recomputation or an exhaustive enumeration, checked by verification programs written separately from the code that produced the results and, for most results, also by the gates described above (Supplementary Table~\ref{tab:repro}); the authors checked all proofs line by line. The models were used to design the computations, write the code in the repository and the verification scripts, run part of the computations, draft and check proofs, check the literature and the references, draft the text of this paper and write the TikZ code for the figures; no figure was produced by an image-generating model. The authors chose the research questions and the directions to pursue, decided which claims to make, and reviewed the output at each stage. The authors also asked ChatGPT (OpenAI) for comments on an earlier draft, and Microsoft Copilot (running OpenAI's GPT-6.1 Sol) for comments on the presentation of this version; no text, proof, result or computation in the paper was produced by them. No other AI tools were used. No AI tool is an author, and the authors take full responsibility for the content of this paper.

\subsection{Relation to earlier work}
\label{sec:methods-prior}

Supplementary Table~\ref{tab:summary} marks, result by result, which ingredients are new. For prime local dimensions every part of the dictionary follows from known results~\cite{GrossWalter2013,LMRW2013,Osthus2024}, although its converse direction does not seem to have been stated (Sec.~\ref{sec:methods-proofs}); for composite $d$ the converse direction appears to be new (Supplementary Note~\ref{app:reduction-related}). The nine seven-variable characteristic-dependent inequalities that we checked, seven from the literature, one of them in corrected form (Supplementary Note~\ref{app:misprints}), and the two of Lemma~\ref{lem:JI}, are blind on the self-dual slice when applied directly to the parties and the purifier (Proposition~\ref{prop:blind} and Supplementary Note~\ref{app:blind}). New are the transfer arguments, which yield inequalities that are not blind; inequalities for the configuration $J-I$ that are valid for finite abelian groups (Lemma~\ref{lem:JI}), for which inequalities over fields were known~\cite{PenaSarria2021,PenaMacias2023}; a rank inequality for groups of odd order that extends an argument of Jafari and Khazaei for linear schemes~\cite{JafariKhazaei2021} (Lemma~\ref{lem:fano-odd}); and the rigidity argument for Theorem~\ref{thm:ideal}. Classically, $\mathrm{AG}(3,2)$ is representable only in characteristic $2$ and $R_8$ only in other characteristics~\cite{Oxley2011}. The classical counterparts of the statements of Theorem~\ref{thm:ideal} on the dimension are due to Mat\'u\v{s}, who showed that the Fano and non-Fano matroids, $\mathrm{AG}(3,2)$ and $R_8$ contracted at any point, have representations by partitions only in degrees that are, respectively, powers of $2$ and odd~\cite{Matus1999}; applied to the partitions defined by the phase-space labels of a stabilizer scheme (Supplementary Note~\ref{app:qss-ideal}), this gives Proposition~\ref{prop:qss}~(iv) for $\Gamma_{\mathrm{NF}}$ and a power of $2$ for $\Gamma_{\mathrm F}$, as Theorem~\ref{thm:ideal} does, and his proof, which shows that the group defining a representation of the Fano matroid is a power of $\Z_2$, gives $d=2$ once that group is identified with the phase space of a player (Supplementary Note~\ref{app:qss-ideal}), as Supplementary Lemma~\ref{lem:K-two} does directly. Earlier examples of a dependence of quantum secret sharing on the local dimension concerned threshold schemes, which need a large enough alphabet when each share is a single qudit~\cite{CleveGottesmanLo1999,Keet2010}, and lower bounds above $1$ on share sizes were derived from entropy inequalities by Sarvepalli~\cite{Sarvepalli2011}. Other parts rest on known ingredients. For prime local dimensions, the bound of Proposition~\ref{prop:qss}~(i) follows from bounds of Jafari and Khazaei and of Pe\~na-Macias~\cite{JafariKhazaei2021,PenaMacias2023} for linear schemes, and that of~(ii) from a bound of Jafari and Khazaei; Proposition~\ref{thm:four} follows from the literature for prime local dimensions, in the closure and up to scaling, and with Corollary~\ref{cor:divisibility} for every $d$; and Theorem~\ref{thm:primes}, for at least $2p+3$ parties, from inequalities of Pe\~na-Macias and Sarria~\cite{PenaSarria2021}. For Theorem~\ref{thm:five}, what is new is the use of the purifier as a variable, the facet and the catalog: six-variable inequalities of Dougherty, Freiling and Zeger~\cite{DFZ2009} also detect the gap once the purifier is used in this way (Supplementary Note~\ref{app:five-dfz}). The codes of Proposition~\ref{prop:qss} use constructions of Smith~\cite{Smith2000} and of van Dijk, Jackson and Martin~\cite{vDJM1998}.
\section*{Data availability}
The data and certificates that support the findings of this article are openly available in the archived release \reprelease{} of the repository at Zenodo, \url{https://doi.org/10.5281/zenodo.23841282} (Ref.~\cite{CaoSong2026data}); the development repository is \url{https://github.com/RuifengCao/quantum-linear-rank-cone}. The Methods describe how to verify every result.

\section*{Code availability}
The code that produced the results, the verification scripts and the gate checks of Sec.~\ref{app:repro} are openly available in the same release, under the licenses stated in Sec.~\ref{app:repro}.

\begin{acknowledgments}
This study received no funding.
\end{acknowledgments}

\section*{Author contributions}
X.S. and R.C. contributed equally to this work. X.S.: Data curation, Formal analysis, Investigation, Resources, Validation, Writing -- review \& editing. R.C.: Conceptualization, Data curation, Formal analysis, Investigation, Methodology, Project administration, Software, Validation, Visualization, Writing -- original draft, Writing -- review \& editing.

\section*{Competing interests}
All authors declare no financial or non-financial competing interests.

%

\clearpage

\onecolumngrid
\setcounter{page}{1}
\renewcommand{\thepage}{S\arabic{page}}
\begin{center}
{\large\bfseries Supplementary Information}\\[4pt]
{\large \papertitle}\\[4pt]
Xidan Song and Ruifeng Cao\\[6pt]
{\small Equations, tables and figures of this Supplementary Information carry the prefix S, and its lemmas and propositions (Supplementary Lemmas and Propositions) are numbered by note.\\
All other numbers of equations, results, figures and tables refer to the main text; Example~\ref{ex:q2} is stated in Supplementary Note~\ref{app:five-catalogue}. The references cited here are listed at the end of this Supplementary Information.}
\end{center}
\label{si:start}
\twocolumngrid
\setcounter{section}{0}
\setcounter{equation}{0}
\setcounter{figure}{0}
\setcounter{table}{0}
\renewcommand{\thesection}{\arabic{section}}
\renewcommand{\thesubsection}{\thesection.\arabic{subsection}}
\renewcommand{\theequation}{S\arabic{equation}}
\renewcommand{\thefigure}{S\arabic{figure}}
\renewcommand{\thetable}{S\arabic{table}}
\renewcommand{\theapplemma}{\thesection.\arabic{applemma}}
\makeatletter
\def\p@subsection{}
\def\p@subsubsection{}
\def\@seccntformat#1{\def\@tempa{#1}\def\@tempb{section}%
  \ifx\@tempa\@tempb SUPPLEMENTARY NOTE \csname the#1\endcsname.\quad\else\csname the#1\endcsname.\quad\fi}
\makeatother

\section{Proof of the dictionary (Proposition~\ref{thm:reduction})}
\label{app:reduction}

Supplementary Fig.~\ref{fig:dictionary} summarizes the dictionary. Throughout this note, $d\ge 2$ is an integer, $\omega=e^{2\pi i/d}$, and logarithms are to base $2$. All groups are finite and abelian and are written additively.

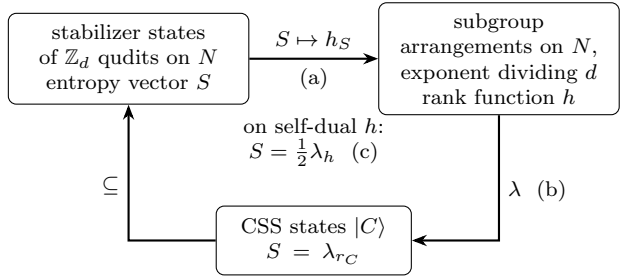
\begin{figure}[t]
\centering
\begin{tikzpicture}[font=\footnotesize,
  box/.style={draw, rounded corners=3pt, align=center, inner sep=4pt, text width=2.9cm},
  arr/.style={-{Stealth[length=2mm]}, thick}]
\node[box] (st) at (0,0) {stabilizer states\\ of $\Z_d$ qudits on $N$\\ entropy vector $S$};
\node[box] (rk) at (4.9,0) {subgroup\\ arrangements on $N$,\\ exponent dividing $d$\\ rank function $h$};
\node[box, text width=2.3cm] (css) at (2.45,-2.4) {CSS states $\ket{C}$\\ $S=\lambda_{r_C}$};
\draw[arr] (st.east) -- node[above] {$S\mapsto h_S$} node[below] {(a)} (rk.west);
\draw[arr] (rk.south) |- node[pos=0.3, right] {$\lambda$\ \ (b)} (css.east);
\draw[arr] (css.west) -| node[pos=0.7, left] {$\subseteq$} (st.south);
\node[align=center] at (2.45,-1.1) {on self-dual $h$:\\ $S=\tfrac12\lambda_h$\ \ (c)};
\end{tikzpicture}
\caption{The dictionary of Proposition~\ref{thm:reduction} (facts (a)--(c) of the Methods). The lift $S\mapsto h_S$ sends stabilizer entropy vectors to self-dual subgroup rank functions on the $n+1$ systems (facts~(a) and~(c)); the connectivity map $\lambda$ sends every rank function to a CSS entropy vector (fact~(b)); and on self-dual vectors the two maps are inverse to each other up to the factor $2$ (fact~(c)).}
\label{fig:dictionary}
\end{figure}

\subsection{Relation to earlier work}
\label{app:reduction-related}

Fact~(a) is closely related to known results: Gross and Walter obtained stabilizer entropies from classical phase-space models and concluded that stabilizer states satisfy every balanced classical information inequality~\cite{si:GrossWalter2013}; Linden, Mat\'u\v{s}, Ruskai and Winter showed that they satisfy linear rank inequalities associated with subgroups of the error group~\cite{si:LMRW2013}; and Majenz showed, under a technical condition on the local dimension, that they respect further inequalities valid for random variables from linear codes~\cite{si:Majenz2014}. For prime local dimension, Osthus expressed the entropies of CSS states, and twice those of general stabilizer states, through ranks of submatrices and asked whether CSS states realize the entropy vectors of all stabilizer states~\cite[Sec.~3.7 and Chs.~4 and~6]{si:Osthus2024}; Corollary~\ref{cor:doubling} of the main text, which for prime $d$ follows from his two expressions, answers this at the level of cones. In quantum secret sharing, ideal schemes are described by self-dual matroids~\cite{si:Matus2012}, and every identically self-dual matroid representable over a finite field induces a pure-state scheme~\cite{si:SarvepalliRaussendorf2010}, the CSS state of a code that represents the matroid, whose lift is twice the rank function. To our knowledge, the converse of fact~(a), that every self-dual vector of the rank cone is attained, in the closure, by CSS states, has not been stated beyond this matroid case. Two features deserve emphasis: the symplectic structure plays no role at the level of cones ($\lambda_r$ is a CSS entropy vector for every subgroup arrangement, with no isotropy condition), although isotropy helps inside proofs (Proposition~\ref{prop:f12}); and nothing in Proposition~\ref{thm:reduction} depends on $d$ being prime, since subgroup arrangements satisfy every inequality that follows from Shannon inequalities and common information.

\emph{Purity is essential.} Proposition~\ref{thm:reduction} concerns pure-symmetric vectors. For $n=1$ and $N=\{0,1\}$, the vector with $S(0)=S(1)=S(01)=1$ is not pure-symmetric, although its lift $h_S=(2,2,3)$ is the rank function of two planes in $\GF(2)^3$ that meet in a line.

\subsection{Stabilizer states of \texorpdfstring{$\Z_d$}{Zd} qudits and their entropies}
\label{app:reduction-stab}

On $\mathbb{C}^d$ let $X\ket{j}=\ket{j+1}$ and $Z\ket{j}=\omega^j\ket{j}$, with $j\in\Z_d$. For $m$ qudits and $v=(x,z)\in\Z_d^{2m}$ put $W(v)=Z^{z_1}X^{x_1}\otimes\dots\otimes Z^{z_m}X^{x_m}$, the Weyl (generalized Pauli) operators. Then $W(v)W(v')=\omega^{-x\cdot z'}W(v+v')$, and $W(v)W(v')=\omega^{\langle v,v'\rangle}W(v')W(v)$ with the symplectic form $\langle v,v'\rangle=x'\cdot z-x\cdot z'$, which is nondegenerate on $\Z_d^{2m}$. Let $\mathcal{P}_m$ be the group generated by the $W(v)$ and the scalars $\xi$ with $\xi^{2d}=1$ (the generalized Pauli group; the extra phases accommodate the usual conventions for even $d$). Every $W(v)$ with $v\ne 0$ is traceless.

A pure \emph{stabilizer state} is a unit vector $\ket{\psi}$ that is, up to a phase, the unique common $+1$ eigenvector of an abelian subgroup $G\le\mathcal{P}_m$ containing no scalar other than $1$~\cite{si:Gottesman1999qudit,si:Hostens2005}. Projecting $G$ to phase space, $\xi W(v)\mapsto v$, is injective on $G$; its image $L$ is isotropic, $\langle L,L\rangle=0$. The projector onto the common $+1$ eigenspace is $\abs{G}^{-1}\sum_{g\in G}g$, of trace $d^m/\abs{G}$. Uniqueness of $\ket{\psi}$ forces $\abs{G}=d^m$, so $\abs{L}=d^m$. Since $\abs{L^\perp}=d^{2m}/\abs{L}$ for a nondegenerate form, $L=L^\perp$: $L$ is a Lagrangian subgroup, and $\ket{\psi}\bra{\psi}=d^{-m}\sum_{g\in G}g$. If $d$ is divisible by the square of a prime, $L$ need not be isomorphic to $\Z_d^m$: for $d=4$ and $m=1$, the operators $X^2$ and $Z^2$ generate such a group $G$, with common $+1$ eigenvector $(\ket{0}+\ket{2})/\sqrt2$ and $L\cong\Z_2^2$. For $d=p^k$, the CSS states of codes that are not free, as in Lemma~\ref{lem:A-realize} below, are of this kind. For squarefree $d$, by contrast, $L\cong\Z_d^m$ always: by the Chinese remainder theorem $L$ is the sum of its $p$-parts, which are Lagrangian subspaces of the symplectic $\GF(p)$-spaces $(d/p)\Z_d^{2m}\cong\GF(p)^{2m}$.

The qudits are distributed among the systems in $N$; system $i$ holds $m_i$ of them, and $m_X=\sum_{i\in X}m_i$. Write $V_X\le\Z_d^{2m}$ for the vectors supported on the qudits of $X$, $\pr_X$ for the projection onto these coordinates, and $L_X=L\cap V_X$.

\begin{applemma}[Entropy formula]\label{lem:A-entropy}
For every $X\subseteq N$,
\begin{equation}\label{eq:A-entropy}
  S(X)=m_X\log d-\log\abs{L_X}=\log\abs{\pr_X L}-m_X\log d .
\end{equation}
\end{applemma}

\begin{proof}
Taking the partial trace of $d^{-m}\sum_{g\in G}g$ over the complement of $X$ annihilates every $g$ that acts nontrivially there, because nonzero Weyl operators are traceless. Hence $\rho_X=d^{-m_X}\sum_{g\in G_X}g|_X$, where $G_X$ is the subgroup of elements supported on $X$ and $g|_X$ is the restriction. The map $g\mapsto g|_X$ is a faithful unitary representation of $G_X$ in which only the identity has nonzero trace, so $\Pi_X=\abs{G_X}^{-1}\sum_{g\in G_X}g|_X$ is a projector of rank $d^{m_X}/\abs{G_X}$. Thus $\rho_X$ is maximally mixed on its support and $S(X)=m_X\log d-\log\abs{G_X}$, with $\abs{G_X}=\abs{L_X}$. For the second expression, note that $\pr_X$ restricted to $L$ has kernel $L_{N\setminus X}$, so $\log\abs{\pr_XL}=m\log d-\log\abs{L_{N\setminus X}}$, and use $S(X)=S(N\setminus X)$.
\end{proof}

For qubits this is equivalent to the formula of Fattal \emph{et al.}~\cite{si:Fattal2004}; for general $d$ see also Refs.~\cite{si:GrossWalter2013,si:LMRW2013}. For a graph state of $\Z_d$ qudits with symmetric weight matrix $W$ (zero diagonal), the Lagrangian subgroup is the row space of $[\,I\mid W\,]$, and Eq.~\eqref{eq:A-entropy} gives $S(X)=\log\abs{W[Q_X,Q_{N\setminus X}]\,\Z_d^{m_{N\setminus X}}}$, where $Q_X$ is the set of qudits of $X$. For prime $d$ this is $\rank_{\GF(d)}W[Q_X,Q_{N\setminus X}]$ in units of $\log d$; for $d=2$ it is the cut-rank formula of Hein, Eisert and Briegel~\cite{si:Hein2004}.

\emph{Two definitions.} If $d$ is divisible by the square of a prime, the stabilizer states defined above include states outside the Clifford orbit of $\ket{0\cdots0}$~\cite{si:FengLuo2024}. The separations between stabilizer cones proved in the main text hold for both definitions. Each explicit witness is the CSS state (Supplementary Note~\ref{app:reduction-css}) of a code over $\Z_{d_0}$, for some $d_0\mid d$, with a generator matrix $[\,I\mid A\,]$ in systematic form, up to the order of the coordinates. Over $\Z_d$, the matrix $[\,I\mid(d/d_0)A\,]$ generates a systematic code with the same entropy vector: each projection of the new code has $(d/d_0)^{k}$ times as many elements as that of the old one, where $k$ is the number of systematic coordinates it keeps; the logarithm of this factor is a sum of contributions of the single systems, which cancels in the entropy formula $S=\lambda_{r_C}$ of fact~(b) of the Methods (Sec.~\ref{sec:methods-proofs}). The CSS state of a systematic code is prepared from $\ket{0\cdots0}$ by Fourier and SUM gates, so it lies in the Clifford orbit. The first implication of Proposition~\ref{prop:sandwich}, the one used for separations, holds for both definitions as well, although its witnesses are built from arbitrary codes (Supplementary Note~\ref{app:cubes-sandwich}).

\subsection{Self-duality and the lift}
\label{app:reduction-lift}

\begin{applemma}[Fact (c); Mat\'u\v{s}~\cite{si:Matus2012}]\label{lem:A-selfdual}
If $S$ is pure-symmetric on $N$, then $h_S^*=h_S$ and $\lambda_{h_S}=2S$. Conversely, if $h$ is a self-dual set function with $h(\emptyset)=0$, then $S=\tfrac12\lambda_h$ is pure-symmetric and $h_S=h$.
\end{applemma}

\begin{proof}
For pure-symmetric $S$ we have $S(N)=0$ and $S(N\setminus X)=S(X)$, so
\begin{align*}
  h_S^*(X)&=\sum_{i\in X}2S(i)+S(N\setminus X)+\sum_{i\notin X}S(i)-\sum_{i\in N}S(i)\\
  &=h_S(X),\\
  \lambda_{h_S}(X)&=S(X)+S(N\setminus X)-S(N)\\
  &\quad+\sum_{i\in X}S(i)+\sum_{i\notin X}S(i)-\sum_{i\in N}S(i)=2S(X).
\end{align*}
Conversely, self-duality gives $h(N\setminus X)=h(X)+h(N)-\sum_{i\in X}h(i)$, hence $\lambda_h(X)=2h(X)-\sum_{i\in X}h(i)$ and $\lambda_h(i)=h(i)$. So $S=\tfrac12\lambda_h$ satisfies $S(X)=S(N\setminus X)$, $S(\emptyset)=0$, and $h_S(X)=h(X)-\tfrac12\sum_{i\in X}h(i)+\tfrac12\sum_{i\in X}h(i)=h(X)$.
\end{proof}

We use the following duality for finite abelian groups. For a group $A$ let $\widehat A=\mathrm{Hom}(A,\mathbb{Q}/\Z)$ be its character group; $\abs{\widehat A}=\abs{A}$, and $\widehat A$ has the same exponent as $A$. A homomorphism $\varphi\colon A\to B$ has a dual $\widehat\varphi\colon\widehat B\to\widehat A$, $\chi\mapsto\chi\circ\varphi$. The image of $\widehat\varphi$ is isomorphic to $\widehat{\varphi(A)}$, because characters of $\varphi(A)$ extend to $B$; hence $\abs{\widehat\varphi(\widehat B)}=\abs{\varphi(A)}$.

\begin{applemma}[Projections versus sums]\label{lem:A-proj}
Let $\Lambda\le\bigoplus_{i\in N}Q_i$ be a subgroup of a direct sum of finite abelian groups, let $\varphi_i\colon\Lambda\to Q_i$ be the coordinate maps, and put $U_i=\widehat{\varphi_i}(\widehat{Q_i})\le\widehat\Lambda$. Then $\abs{\pr_X\Lambda}=\abs{\sum_{i\in X}U_i}$ for every $X\subseteq N$.
\end{applemma}

\begin{proof}
The dual of $\pr_X\colon\Lambda\to\bigoplus_{i\in X}Q_i$ maps $\bigoplus_{i\in X}\widehat{Q_i}$ onto $\sum_{i\in X}U_i$; apply the preceding remark.
\end{proof}

\begin{applemma}[Fact (a): the lift]\label{lem:A-lift}
Let $S$ be the entropy vector of a stabilizer state of $\Z_d$ qudits. Then there are subgroups $U_i$ ($i\in N$) of a finite abelian group of exponent dividing $d$ such that $h_S(X)=\log\abs{\sum_{i\in X}U_i}$ for all $X\subseteq N$. For $d=p$ prime, the $U_i$ are subspaces of a $\GF(p)$-vector space.
\end{applemma}

\begin{proof}
By Eq.~\eqref{eq:A-entropy}, $S(i)=m_i\log d-\log\abs{L_i}$ and $S(X)=\log\abs{\pr_XL}-m_X\log d$, so
\[
  h_S(X)=\log\abs{\pr_XL}-\sum_{i\in X}\log\abs{L_i}.
\]
Since $L_i=L\cap V_i$ consists of elements of $L$ supported on system $i$, $L_i\le\pr_iL$. Put $Q_i=\pr_iL/L_i$, and let $\Lambda$ be the image of $L$ under $L\to\bigoplus_iQ_i$, $v\mapsto(\pr_iv+L_i)_i$. The map $\pr_XL\to\bigoplus_{i\in X}Q_i$ has image $\pr_X\Lambda$ and kernel $\bigoplus_{i\in X}L_i$: an element of $\pr_XL$ whose components lie in the $L_i$ is a sum of elements of the $L_i$, and conversely each $L_i\le L$ lies in $\pr_XL$. Hence $\log\abs{\pr_X\Lambda}=h_S(X)$, and Lemma~\ref{lem:A-proj} gives subgroups $U_i\le\widehat\Lambda$ with $h_S(X)=\log\abs{\sum_{i\in X}U_i}$. The exponent of $\widehat\Lambda$ equals that of $\Lambda$, a subquotient of $\Z_d^{2m}$, and so divides $d$. For prime $d$ all groups involved are $\GF(d)$-vector spaces.
\end{proof}

Lemma~\ref{lem:A-lift} should be compared with the classical model of Gross and Walter~\cite[Thm.~2]{si:GrossWalter2013}, whose phase-space entropies are $\log\abs{\pr_XL}=S(X)+m_X\log d$ (this is the rank function $r_L$ used in Corollary~\ref{cor:doubling}), and with the abelian quotients in the entropy formula of Ref.~\cite[Prop.~9]{si:LMRW2013}. The lift differs from $r_L$ by the modular function $X\mapsto\sum_{i\in X}(S(i)-m_i\log d)$; this particular choice is what makes $h_S$ self-dual (Lemma~\ref{lem:A-selfdual}). In the normal form of Lemma~\ref{lem:A-normal} below, $h_S=r_L$.

\subsection{CSS states}
\label{app:reduction-css}

Let $C\le\Z_d^m$ be a subgroup and let $C^\perp=\{b\in\Z_d^m:b\cdot c=0\ \forall c\in C\}$. Then $\abs{C}\,\abs{C^\perp}=d^m$, and the state $\ket{C}=\abs{C}^{-1/2}\sum_{c\in C}\ket{c}$ is the common $+1$ eigenvector of the commuting operators $X^c=\bigotimes_jX^{c_j}$ ($c\in C$) and $Z^b=\bigotimes_jZ^{b_j}$ ($b\in C^\perp$). These generate an abelian group of order $\abs{C}\abs{C^\perp}=d^m$ without nontrivial scalars, so $\ket{C}$ is a stabilizer state, the CSS state of $C$~\cite{si:CalderbankShor1996,si:Steane1996}.

\begin{applemma}[Fact (b): CSS entropies]\label{lem:A-css}
Distribute the $m$ coordinates of $C$ among the systems in $N$ and put $r_C(X)=\log\abs{\pr_XC}$. Then the entropy vector of $\ket{C}$ is $\lambda_{r_C}$.
\end{applemma}

\begin{proof}
Fix $X$ and write $c=(c_X,c_{\bar X})$ with $\bar X=N\setminus X$. Let $K_X=\{c_X:c\in C,\ c_{\bar X}=0\}$, and define $K_{\bar X}$ in the same way. Then
\[
  \rho_X=\abs{C}^{-1}\sum_{w\in\pr_{\bar X}C}\ket{\psi_w}\bra{\psi_w},\qquad \ket{\psi_w}=\sum_{c\in C:\,c_{\bar X}=w}\ket{c_X}.
\]
The vectors $c_X$ with $c_{\bar X}=w$ form a coset of $K_X$ in $\pr_XC$. The map $w\mapsto$ (this coset) is a surjective homomorphism $\pr_{\bar X}C\to\pr_XC/K_X$ with kernel $K_{\bar X}$. So every coset occurs $\abs{K_{\bar X}}$ times. Different cosets give orthogonal vectors of the same norm. Hence $\rho_X$ is proportional to a projector of rank $\abs{\pr_XC}/\abs{K_X}$, and
\begin{align*}
  S(X)&=\log\abs{\pr_XC}-\log\abs{K_X}\\
  &=\log\abs{\pr_XC}+\log\abs{\pr_{\bar X}C}-\log\abs{C},
\end{align*}
using $\abs{K_X}=\abs{C}/\abs{\pr_{\bar X}C}$. This is $\lambda_{r_C}(X)$.
\end{proof}
Lemma~\ref{lem:A-css} is also the special case $L=C\times C^\perp$ of Eq.~\eqref{eq:A-entropy}; we included the short direct proof.

\begin{applemma}[Every subgroup rank function is a code]\label{lem:A-realize}
Let $U_i$ ($i\in N$) be subgroups of a finite abelian group $V$ of exponent dividing $d$. Then there is a subgroup $C\le\Z_d^m$, with coordinates distributed among the systems, such that $r_C(X)=\log\abs{\sum_{i\in X}U_i}$ for every $X$.
\end{applemma}

\begin{proof}
Write $V\cong\bigoplus_j\Z_{d_j}$ with $d_j\mid d$ and embed $\Z_{d_j}$ into $\Z_d$ as $(d/d_j)\Z_d$, so that $V\le\Z_d^k$. Choose generators of each $U_i$, and let $G=[G_i]_{i\in N}$ be the $k\times m$ matrix whose column block $G_i$ consists of the generators of $U_i$; system $i$ holds the coordinates of block $i$. Let $C=\{y^{\mathsf T}G:y\in\Z_d^k\}$. Then $\pr_XC$ is the row space of $G_X=[G_i]_{i\in X}$ and $\sum_{i\in X}U_i$ is its column space. By the Smith normal form, $G_X=PDQ$ with $P,Q$ invertible over $\Z$ and $D$ diagonal, and both spaces have the same order as the corresponding space of $D$.
\end{proof}

Lemmas~\ref{lem:A-css} and~\ref{lem:A-realize} together prove fact~(b), and Lemmas~\ref{lem:A-lift} and~\ref{lem:A-selfdual} facts~(a) and~(c). This completes the proof of Proposition~\ref{thm:reduction}.

\begin{proof}[Proof of Corollary~\ref{cor:doubling}]
Let $L\le\Z_d^{2m}$ be the Lagrangian subgroup of the state (Supplementary Note~\ref{app:reduction-stab}). Read $L$ as a classical code of length $2m$, in which each system holds the two coordinates ($x$ and $z$) of each of its qudits. By the entropy formula~\eqref{eq:A-entropy}, $r_L(X)=\log\abs{\pr_X L}=S(X)+m_X\log d$, where $m_X$ is the number of qudits held by $X$. Since $\abs{L}=d^m$,
\begin{align*}
  \lambda_{r_L}(X)&=S(X)+S(N\setminus X)\\
  &\quad+(m_X+m_{N\setminus X})\log d-\log\abs{L}=2S(X).
\end{align*}
By Lemma~\ref{lem:A-css}, the CSS state $\ket{L}$ of $2m$ qudits has entropy vector $\lambda_{r_L}=2S$.
\end{proof}

\subsection{Duality}
\label{app:reduction-dual}

For a set function $h$ on $N$ with $h(\emptyset)=0$ the dual is $h^*(X)=\sum_{i\in X}h(i)+h(N\setminus X)-h(N)$; this is the polymatroid duality used by Kaced~\cite{si:Kaced2018}. We write $\Dual h=h^*$. A direct computation gives $\Dual^2h=h-\sum_{i\in N}h(i\mid N\setminus i)\,\ind{i}$. So $\Dual^2$ is the identity on set functions with $h(i\mid N\setminus i)=0$ for every $i$, in particular on self-dual functions and on lifts $h_S$. Moreover, $\Dual$ commutes with relabelings of $N$.

\begin{applemma}[Duality preserves rank cones]\label{lem:A-dual}
If $h$ is the rank function of subgroups $U_i$ of a group $V$ of exponent dividing $d$, then so is $h^*$. Consequently $\Dual$ maps $\LcZ{d}{m}$ into itself.
\end{applemma}

\begin{proof}
Let $\sigma\colon\bigoplus_iU_i\to V$ be the sum map and $K=\ker\sigma$. For $X\subseteq N$, $\log\abs{\sum_{i\in X}U_i}=\sum_{i\in X}h(i)-\log\abs{K\cap\bigoplus_{i\in X}U_i}$. Hence $h^*(X)=\log\abs{K}-\log\abs{K\cap\bigoplus_{i\notin X}U_i}=\log\abs{\pr_XK}$. By Lemma~\ref{lem:A-proj}, applied to $K\le\bigoplus_iU_i$, this is a subgroup rank function in $\widehat K$, whose exponent divides $d$. Linearity and continuity of $\Dual$ give the second statement.
\end{proof}

\subsection{Unions over local dimensions}
\label{app:reduction-unions}

Let $\Lab_m$ be the closed convex hull of $\bigcup_{d\ge2}\LcZ{d}{m}$, and $\Lc_m$ that of $\bigcup_{p}\LcZ{p}{m}$ over primes. The latter is the closed convex cone generated by the rank functions of subspace arrangements over all fields. Indeed, an arrangement over an arbitrary field has the same rank function as an arrangement over some finite field: choose spanning vectors with entries in a finitely generated subring $R$ of the field, and reduce modulo a maximal ideal of $R$ that avoids the product of finitely many nonzero minors that certify the ranks. Such an ideal exists because $R$, a finitely generated $\Z$-algebra, is a Jacobson ring whose residue fields at maximal ideals are finite. Minors that vanish over the field vanish after reduction, so all ranks are preserved. Finally, an arrangement over $\GF(p^k)$ is an arrangement over $\GF(p)$ with all ranks multiplied by $k$ (restriction of scalars). So $\Lc_m$ is the cone cut out by all $m$-variable linear rank inequalities.

To prove Eq.~\eqref{eq:union} for $\Stab_n$, let $h_S\in\Lab_{n+1}$. Since $\lambda$ is linear and continuous and maps $\LcZ{d}{n+1}$ into $\StabZ{d}{n}$, it maps $\Lab_{n+1}$ into $\Stab_n$, so $S=\lambda_{h_S/2}\in\Stab_n$. Conversely, $h_S\in\LcZ{d}{n+1}\subseteq\Lab_{n+1}$ for $S\in\StabZ{d}{n}$, and the lift is linear and continuous. The same argument over primes proves the second identity in Eq.~\eqref{eq:union}.

\subsection{Balanced inequalities}
\label{app:reduction-balanced}

\begin{proof}[Proof of Corollary~\ref{cor:balanced}]
For a set function $h$ on $N$ and $i\in N$, write $h(i\mid N\setminus i)=h(N)-h(N\setminus i)$. If $h$ is the rank function of subgroups $U_j$, replace $U_i$ by $U_i'=U_i\cap\sum_{j\ne i}U_j$. For every $A\le\sum_{j\ne i}U_j$ we have $U_i'\cap A=U_i\cap A$, hence $\abs{U_i+A}/\abs{U_i'+A}=\abs{U_i}/\abs{U_i'}=2^{h(i\mid N\setminus i)}$. So the new rank function is $h-h(i\mid N\setminus i)\,\ind{i}$. It is again a subgroup rank function, and the quantities $h(j\mid N\setminus j)$, $j\ne i$, do not change. Doing this for every $i$ shows that, if $f\ge0$ on $\LcZ{d}{n+1}$, then so is the balanced form
\[
  f_{\mathrm{bal}}(h)=f(h)-\sum_{i\in N}f(\ind{i})\,h(i\mid N\setminus i).
\]
For pure-symmetric $S$ we have $h_S(i\mid N\setminus i)=S(i)-S(N\setminus i)=0$, so $f_{\mathrm{bal}}(h_S)=f(h_S)$. A balanced form vanishes on every $\ind{i}$, and $h_S=S+\sum_iS(i)\ind{i}$, so $f_{\mathrm{bal}}(h_S)=f_{\mathrm{bal}}(S)$.

If $S\in\StabZ{d}{n}$ and $f$ is balanced and nonnegative on $\LcZ{d}{n+1}$, then $f(S)=f(h_S)\ge0$ by Proposition~\ref{thm:reduction}. If $S\notin\StabZ{d}{n}$, then $h_S\notin\LcZ{d}{n+1}$, and since the cone is closed and convex there is a linear $f\ge0$ on it with $f(h_S)<0$. Then $f_{\mathrm{bal}}$ is balanced, nonnegative on the cone, and $f_{\mathrm{bal}}(S)=f(h_S)<0$. The statement for $\Stab^{\mathrm{pr}}_n$ follows in the same way from Eq.~\eqref{eq:union}.
\end{proof}

\subsection{Normal form for prime local dimension}
\label{app:reduction-normal}

\begin{applemma}[Normal form]\label{lem:A-normal}
Let $p$ be prime. If $S$ is the entropy vector of a stabilizer state of $\Z_p$ qudits, then $S$ is also the entropy vector of a stabilizer state in which every system $i$ holds exactly $S(i)/\log p$ qudits.
\end{applemma}

\begin{proof}
By Eq.~\eqref{eq:A-entropy}, $L_i=L\cap V_i$ is an isotropic subspace of dimension $k=m_i-S(i)/\log p$ in the symplectic space $V_i$ of system $i$. A Clifford unitary $U_i$ on the qudits of $i$ maps it onto the span of the first $k$ $z$-coordinates, and a Pauli correction fixes the signs. Afterwards the state is stabilized by $Z_1,\dots,Z_k$, so it equals $\ket{0}^{\otimes k}\otimes\ket{\psi'}$. Every element of the new stabilizer group commutes with $Z_1,\dots,Z_k$, so it acts on these $k$ qudits as a product of powers of $Z$; multiplying by such powers, which lie in the group, shows that the restrictions to the remaining qudits form an abelian group of order $p^{m-k}$ that contains no scalar other than $1$ and stabilizes $\ket{\psi'}$. So $\ket{\psi'}$ is a stabilizer state. Since $U_i$ acts within system $i$, it changes no entropy of a union of systems, and discarding the $k$ qudits in a product state changes no entropy either. Repeat for every system.
\end{proof}

In the normal form, and with entropies in units of $\log p$, $L$ has dimension $m=\sum_iS(i)$, and Eq.~\eqref{eq:A-entropy} gives $\dim\pr_XL=S(X)+\sum_{i\in X}S(i)=h_S(X)$: the lift is simply the rank function of the coordinate projections of $L$. The normal form also reduces the search for a stabilizer state with a prescribed entropy vector $kS$ to a finite search: system $i$ holds exactly $kS(i)$ qudits, and by local-Clifford equivalence one may restrict to graph states.
\section{Four parties: extreme rays and unimodular realizations}
\label{app:four}

This note completes the proof of Proposition~\ref{thm:four}.

\subsection{The cone \texorpdfstring{$C_4$}{C4}}

Index the five systems as $N=\{0,1,2,3,4\}$ with purifier $0$. We use the 15 coordinates $S(X)$, $\emptyset\ne X\subseteq\{1,2,3,4\}$; values on sets containing $0$ follow from purity. The cone $C_4$ of Proposition~\ref{thm:four} (Methods, Sec.~\ref{sec:methods-five}) is cut out by the pull-backs, through $S\mapsto h_S$, of the 85 elemental Shannon inequalities on five elements and of the 480 Ingleton instances~\eqref{eq:ingleton} with pairwise disjoint $A,B,C,D,K$ (counted with repetitions). These yield 100 distinct nonzero rows; the five rows $h(i\mid N\setminus i)\ge0$ vanish on lifts. Up to relabeling, the Ingleton instances on five elements are the four forms (36)--(39) of Ref.~\cite{si:DFZ2009}: the plain inequality on four elements, the two instances in which one argument is a pair, and the instance conditioned on the fifth element. Exact vertex enumeration with \texttt{lrs}~\cite{si:lrs} gives 46 extreme rays, which form seven orbits under the group $S_5$ of relabelings of $N$ (Table~\ref{tab:four}).

\begin{table}[t]
\caption{Orbit representatives of the 46 extreme rays of $C_4=\StabZ{d}{4}$. The representative lists $S(X)$ for $X=1,2,12,3,13,23,123,4,14,24,124,34,134,234,1234$ (binary order); $S(0)=S(1234)$. The last two columns give the multiplicity $\kappa$ of the unimodular realization and its number of qudits.}
\label{tab:four}
\begin{ruledtabular}
\begin{tabular}{clccc}
orbit & representative & size & $\kappa$ & qudits\\
\hline
1 & \texttt{000000011111111} & 10 & 1 & 2\\
2 & \texttt{011111111111111} & 5 & 1 & 4\\
3 & \texttt{011112211222211} & 5 & 2 & 8\\
4 & \texttt{111111111111111} & 1 & 1 & 5\\
5 & \texttt{112122222222222} & 10 & 1 & 7\\
6 & \texttt{112122312232332} & 5 & 1 & 6\\
7 & \texttt{112233223322332} & 10 & 1 & 8\\
\end{tabular}
\end{ruledtabular}
\end{table}

Orbit~1 is a Bell pair between one party and the purifier, orbit~2 a Greenberger--Horne--Zeilinger (GHZ) state of four systems, and orbit~4 the GHZ state of all five systems. Orbit~3 has entropy $1$ on every system of a four-system subset and $2$ on every pair of it. It is not the entropy vector of any four-qubit state, since there is no absolutely maximally entangled state of four qubits~\cite{si:Higuchi2000}, but it is the entropy vector, in units of $\log 3$, of a state of four qutrits~\cite[Sec.~5]{si:LMRW2013}, and twice it is the entropy vector of an eight-qubit graph state with two qubits per system (the realization below).

\subsection{Unimodular realizations}

For each extreme ray $r$ we store an integer symmetric matrix $W$ with zero diagonal on $\kappa\sum_{i\in N}r(i)$ vertices, of which $\kappa\,r(i)$ belong to system $i$. For every $X\subseteq\{1,2,3,4\}$ the block $W[Q_X,Q_{N\setminus X}]$ has all nonzero Smith invariants equal to $1$, and its rank is $\kappa\,r(X)$. A matrix whose nonzero Smith invariants all equal $1$ has the same rank over $\Q$ and over every finite field. By the cut-rank formula of Supplementary Note~\ref{app:reduction-stab}, $W\bmod p$ defines a weighted graph state of $\Z_p$ qudits with entropy vector $\kappa r$ in units of $\log p$, for every prime $p$. So every extreme ray of $C_4$ lies in $\StabZ{p}{4}$, and by Corollary~\ref{cor:divisibility} in $\StabZ{d}{4}$ for every $d$. Together with $\StabZ{d}{4}\subseteq C_4$ this proves Proposition~\ref{thm:four} for $n=4$.

The matrices come from qubit graph states found by search. For 41 rays the stored matrix has entries in $\{0,1\}$; for five rays (one in orbit~5 and four in orbit~7, unrelated to the multiplicity $\kappa$) some edges carry the sign $-1$, found by a search over sign patterns. We verified the Smith invariants and, as a cross-check, the ranks over $\GF(p)$ for all primes $p\le19$.

Linden, Mat\'u\v{s}, Ruskai and Winter used qutrit states for two of the orbits~\cite[Sec.~5]{si:LMRW2013}. As Osthus noted, their state (R6) has, as printed, $S(ab)=1$ instead of $2$, and he gave a qubit state for this ray~\cite[Sec.~3.8]{si:Osthus2024}; the matrices above realize it for every prime. In particular, qubits realize every extreme ray of $C_4$, at twice the scale for orbit~3. The matrices realize each ray without closure, at scale at most two, for every prime at once; and since $C_4$ is cut out by the Shannon and Ingleton inequalities, Proposition~\ref{thm:four} also shows that, on the self-dual slice, the 24 new five-variable inequalities of Dougherty, Freiling and Zeger~\cite{si:DFZ2009} follow from the Shannon and Ingleton inequalities.

\subsection{Fewer parties}

For $n<4$, adding $4-n$ parties that hold no qudits embeds $\StabZ{d}{n}$ into $\StabZ{d}{4}$, with image in the face $C_4\cap\{S(i)=0\text{ for }n<i\le4\}$; the inequalities $S(i)\ge0$ are valid on $C_4$. On this face $h_S(Xi)=h_S(X)$ for each such $i$, because $0\le h_S(Xi)-h_S(X)\le h_S(i)=2S(i)=0$; hence $S(Xi)=S(X)$, and the face is a copy of a cone of $n$-party vectors. It is generated by the extreme rays of $C_4$ that lie on it. In the realization of such a ray the empty parties hold no vertices, so $W\bmod p$ is a graph state of the remaining systems, for every prime $p$. Hence $\StabZ{d}{n}$ is the same cone for every $d$, namely this face of $C_4$.
\section{Five parties: certificates}
\label{app:five}

Throughout this note the six systems are $A,\dots,E$ and the purifier $F$, and $Z$ denotes an additional seventh element. For a set function $g$ on subsets of $\{A,\dots,F,Z\}$ we write $g(X|K)=g(XK)-g(K)$ and $I(X;Y|K)=g(XK)+g(YK)-g(XYK)-g(K)$. Every term below is a Shannon inequality or an Ingleton instance, hence nonnegative for the rank function $g$ of any subgroup arrangement of seven elements.

\subsection{The inequality~\texorpdfstring{\eqref{eq:star}}{(*)}}
\label{app:five-star}

\begin{table}[t]
\caption{The witness $r$ of Theorem~\ref{thm:five}, an extreme ray of $\QLR_5$ outside $\Stab_5$. The table lists all singletons, pairs and triples; each triple is shown with its complement in parentheses, which has the same value. The values on the remaining sets follow from purity. On this vector, the left-hand side of~\eqref{eq:star} is $24$ and the right-hand side is $25$.}
\label{tab:witness}
\begin{ruledtabular}
\begin{tabular}{lcccccc}
singletons & $A$ & $B$ & $C$ & $D$ & $E$ & $F$\\
 & 2 & 2 & 1 & 1 & 2 & 2\\
\hline
pairs & $AB$ & $AC$ & $AD$ & $AE$ & $AF$ & $BC$\\
 & 3 & 3 & 3 & 3 & 3 & 3\\
 & $BD$ & $BE$ & $BF$ & $CD$ & $CE$ & $CF$\\
 & 3 & 3 & 3 & 2 & 3 & 3\\
 & $DE$ & $DF$ & $EF$ & & & \\
 & 3 & 3 & 3 & & & \\
\hline
triples & $ABC$ & $ABD$ & $ABE$ & $ACD$ & $ACE$ & \\
 & ($DEF$) & ($CEF$) & ($CDF$) & ($BEF$) & ($BDF$) & \\
 & 4 & 4 & 3 & 4 & 4 & \\
 & $ADE$ & $BCD$ & $BCE$ & $BDE$ & $CDE$ & \\
 & ($BCF$) & ($AEF$) & ($ADF$) & ($ACF$) & ($ABF$) & \\
 & 3 & 3 & 3 & 4 & 3 & \\
\end{tabular}
\end{ruledtabular}
\end{table}

Let $g$ be the rank function of $U_A,\dots,U_F,Z$ with $Z=U_{AD}\cap U_{BE}$. The certificate is
\begin{equation}\label{eq:B-cert}
\begin{aligned}
0\le{}&5\,g(Z|ABCDEF)+2\,I(F;Z|ABCDE)\\
&+I(A;B|Z)+I(A;C|BFZ)+I(A;E|BCDZ)\\
&+I(A;Z|BCDEF)+I(A;Z|BE)+I(B;D|AEZ)\\
&+I(B;D|AZ)+I(B;Z|ACDEF)+I(B;Z|AF)\\
&+I(C;Z|BDEF)+I(C;Z|ADEF)\\
&+I(C;Z|ABDE)+I(C;Z|BE)+I(C;Z|ABD)\\
&+I(D;E|AFZ)+I(D;E|BCZ)+I(D;Z|BEF)\\
&+I(D;Z|ABCEF)\\
&+I(E;F|BZ)+I(E;Z|ABCF)+I(E;Z|AD)\\
&+I(F;Z|AE)+I(F;Z|AD)+I(F;Z|BC).
\end{aligned}
\end{equation}
When the right-hand side is expanded into values of $g$, every term that contains $Z$ cancels except $3g(ADZ)+3g(BEZ)-g(Z)$. The substitutions~\eqref{eq:ci} replace these by $2h(AD)+2h(BE)+h(ABDE)$, and the right-hand side becomes the left-hand side of~\eqref{eq:star-h} minus its right-hand side. Substituting $h=h_S$ and using purity gives~\eqref{eq:star} with coefficient one. We checked the identity in exact integer arithmetic.

At the witness $r$ of Table~\ref{tab:witness}, all 10{,}860 facet inequalities of $\QLR_5$ are nonnegative. Exactly 105 of them vanish, and these 105 rows have rank 30 over $\Q$. Since $\QLR_5$ is pointed and 31-dimensional, $r$ spans an extreme ray. The inequality~\eqref{eq:star} is valid on $\StabZ{2}{5}$, and 140 of the 760 distinct entropy vectors of six-qubit graph states with one qubit per system attain equality. Their rank is 30, while all 760 vectors have rank 31, so $\StabZ{2}{5}$ is full-dimensional and~\eqref{eq:star} defines a facet of it. The 760 vectors were recomputed from all $2^{15}$ graphs on six vertices.

For an odd prime $p$, the graph state of six qudits whose weights are the entries of a $0/1$ adjacency matrix $W$ has $S(X)=\rank_{\GF(p)}W[X,N\setminus X]$ in units of $\log p$ (Supplementary Note~\ref{app:reduction-stab}). Each cut block has at most three rows or at most three columns, and every minor of order at most $3$ of a $0/1$ matrix lies in $\{0,\pm1,\pm2\}$, which are units or zero modulo $p$; so the rank over $\GF(p)$ equals the rank over $\Q$, and the entropy vectors do not depend on the odd prime $p$. The $2^{15}$ graphs give $835$ distinct vectors, which include the $760$ qubit vectors and have rank $31$; exactly $140$ of them attain equality in~\eqref{eq:star}, the same $140$ as for qubits, and their rank is $30$. So~\eqref{eq:star} defines a facet of $\StabZ{p}{5}$ for every odd prime $p$. For every $d$, choose a prime $p\mid d$: then $\StabZ{p}{5}\subseteq\StabZ{d}{5}\subseteq\Stab_5$ (Corollary~\ref{cor:divisibility}), \eqref{eq:star} is valid on $\Stab_5$, and the $30$ linearly independent vectors on its hyperplane lie in $\StabZ{p}{5}$; so~\eqref{eq:star} defines a facet of $\StabZ{d}{5}$ and of $\Stab_5$. The script \path{check_rev_sec4.py} repeats this computation.

\subsection{The catalog}
\label{app:five-catalogue}

This subsection gives the details behind Sec.~\ref{sec:scale}. Before this work, 59 orbits of extreme rays of $\QLR_5$ under $S_6$ were known to lie in the stabilizer cone~\cite[Sec.~4.2]{si:BCHS2020gap}: 58 of them, all except their vector~15, through hypergraph realizations and the theorem of Walter and Witteveen that the entropy vectors of hypergraphs lie in the stabilizer cone of every prime local dimension~\cite{si:WalterWitteveen2021}, and vector~15 through an explicit graph state. We tested every orbit of the catalog against the inequalities that follow, as~\eqref{eq:star} does, from Shannon inequalities and one common information of two sets of systems. Infeasible linear programs were converted into exact certificates, each of which yields an inequality valid on $\StabZ{d}{5}$ for every $d$. Orbits not excluded in this way were then tested with conditional Ingleton instances on the seven-element extension, which are also valid for subgroup arrangements. The result is:
\begin{itemize}
\item 1{,}915{,}910 orbits are excluded from $\Stab_5$ by exact certificates: 1{,}914{,}541 by one common information and Shannon inequalities, and 1{,}369 with conditional Ingleton instances in addition;
\item 641 orbits lie in the stabilizer cone of every local dimension, so they are extreme rays of $\StabZ{d}{5}$ for every $d$ and of $\Stab_5$: the 59 known ones and 582 realized by graph states in this work;
\item 1{,}155 orbits are undecided.
\end{itemize}
Because the enumeration was not completed, these proportions need not represent $\QLR_5$ as a whole. A separate publication will report the catalog in detail. The rest of this subsection describes the enumeration, the tests, the realizations and the facets.

The extreme rays of $\QLR_5$ were enumerated by adjacency decomposition~\cite{si:Bremner2009}, starting from the 59 known orbits~\cite{si:BCHS2020gap} and the vector $q_2$ of Example~\ref{ex:q2}, and working in exact arithmetic on the irredundant description of $\QLR_5$ by 10{,}860 facets in 31 classes. This description was built from a machine-readable transcription of the Ingleton forms and the 24 five-variable inequalities of Ref.~\cite{si:DFZ2009}, parsed from the table of these inequalities in the PSITIP software~\cite{si:Li2023psitip}, which lacks the Ingleton form of Eq.~(39) of Ref.~\cite{si:DFZ2009}; we added it. Thirteen of the 24 were also compared with a text extraction of Ref.~\cite{si:DFZ2009}. As cross-checks, the five-variable linear rank cone built from the same transcription has 7{,}943 extreme rays in 162 orbits, the numbers reported in Ref.~\cite{si:DFZ2009}, and $\QLR_5$ has 31 facet classes, the number given in Ref.~\cite{si:BCHS2020}. Every orbit found was certified to be an extreme ray by an exact rank computation. The enumeration ran in several sessions, each resuming the catalog of the previous one; the last, which expanded orbits with small coordinates first, stopped when the catalog exceeded a preset ceiling of $1.9$ million orbits, at 1{,}917{,}706 orbits. So the catalog depends on the order of the search, and it is possibly a proper subset of the set of all orbits.

The orbits were tested in a cutting-plane loop. An orbit that violates an $S_6$ image of an inequality already found is excluded by it. For each remaining orbit representative $r$ and pairs of distinct nonempty sets $X,Y\subseteq\{A,\dots,F\}$ (first disjoint pairs; for the orbits not excluded, all pairs), a linear program asks for a polymatroid $g$ on seven elements that extends $h_r$ and satisfies the common-information conditions $g(XZ)=h_r(X)$, $g(YZ)=h_r(Y)$ and $g(Z)=I(X;Y)$. If it is infeasible for some pair, a Farkas certificate gives a nonnegative combination of Shannon inequalities as in~\eqref{eq:B-cert}, which we converted to exact rational form and verified in integer arithmetic; it yields an inequality valid on $\StabZ{d}{5}$ for every $d$ and violated by $r$, whose $S_6$ images were then evaluated on the whole catalog. The $4{,}804$ certificates obtained in this way (139 of them in an earlier run of the same test on 300 orbits with entries at most $5$) fall into 747 classes under $S_6$, and every one of the $1{,}914{,}541$ orbits excluded in this stage violates an $S_6$ image of one of them. In the second stage, the same programs were run with all conditional Ingleton instances on the seven elements added.

Realizations were first searched among graph states for small multiples of the orbit representatives: over $\GF(2)$ and $\GF(3)$ in the normal form of Lemma~\ref{lem:A-normal}, and over $\GF(4)$, whose graph states were then reduced to qubit states by restriction of scalars (Supplementary Note~\ref{app:cubes-galois}). Every realization was verified by an independent rank computation. Each of the 641 orbits is now certified to lie in $\StabZ{d}{5}$ for every $d$, in one of two ways; $r$ denotes the orbit representative and $\kappa\ge1$ a multiplicity.
\begin{itemize}
\item \emph{Unimodular realizations}, as in Supplementary Note~\ref{app:four}: an integer symmetric matrix $W$, zero inside each system, in which system $i$ owns $\kappa\,r(i)$ vertices and every cut block has rank $\kappa\,r(X)$ over $\Q$ and all nonzero Smith invariants equal to $1$. Then $W\bmod p$ is a graph state of $\Z_p$ qudits with entropy vector $\kappa r$, in units of $\log p$, for every prime $p$, and Corollary~\ref{cor:divisibility} gives every $d$. There are such matrices for 566 orbits: $\kappa=1$ for 515 of them, $\kappa=2$ for 50 and $\kappa=3$ for one, with at most 40 vertices.
\item \emph{Covering certificates} for the other 75 orbits. If the cut blocks of an integer matrix $W$ as above have the right ranks over $\Q$, then $W\bmod p$ realizes $\kappa r$ over $\GF(p)$ for every prime $p$ that divides no nonzero Smith invariant of a cut block; the others, finitely many, are the exceptional primes of $W$. For each of these orbits we give one or more such matrices, 125 in all, such that the primes exceptional for all of them are among $2$ and $3$, and realizations over $\GF(2)$ or $\GF(3)$, 74 in all, for these primes. So the orbit lies in $\StabZ{p}{5}$ for every prime $p$, and by Corollary~\ref{cor:divisibility} in $\StabZ{d}{5}$ for every $d$.
\end{itemize}
In particular, each of the 641 orbits has an explicit realization by qubit graph states and one by qutrit graph states; this includes the 12 of the 58 hypergraph-realizable orbits of Ref.~\cite{si:BCHS2020gap} for which no explicit stabilizer state was known to us, so that their membership in the stabilizer cone rested only on hypergraph realizations and the theorem of Walter and Witteveen~\cite{si:WalterWitteveen2021}; their qubit realizations are unimodular ones at $\kappa=2$. The matrices were found by a randomized search, which proves nothing. Every stored matrix is verified by exact Smith forms and, independently, by rank computations modulo the prime factors of one nonzero minor of maximal order of each cut block, since every prime that divides a Smith invariant of the block divides all such minors (Supplementary Note~\ref{app:repo-verify}).

A class is a facet of $\StabZ{d}{5}$ for every $d$ if 30 linearly independent vectors that lie in $\StabZ{d}{5}$ for every $d$ lie on its hyperplane. The 760 entropy vectors of six-qubit graph states with one qubit per system are also entropy vectors of graph states over every odd prime field (Supplementary Note~\ref{app:five-star}), so they lie in $\StabZ{d}{5}$ for every $d$, as do the images under $S_6$ of the 641 orbit representatives. With these vectors, 119 of the 747 classes pass the test, certified by exact rank computations modulo a large prime, which bound the rank over $\Q$ from below (Supplementary Note~\ref{app:repo-verify}); none of the 747 classes is violated by any of these vectors.

The vector $q_2$ of Example~\ref{ex:q2}, one of the starting points of the enumeration, shows that the distinction between cones and sets of entropy vectors matters even for small examples.

\begin{example}[A ray that needs twice the scale]\label{ex:q2}
Consider five parties $A,\dots,E$ and the purifier $F$, and let $q_2$ be the vector with entropy $1$ on every single system and $2$ on every pair (in units of $\log p$ for $\Z_p$ qudits), and with
\[
  q_2(X)=\begin{cases}2,& X\in\{ABC,ABD,ACD,ABE,BCE\},\\ 3,& X\in\{BCD,ACE,ADE,BDE,CDE\},\end{cases}
\]
on the triples that do not contain $F$ (the triples containing $F$ follow by purity). The vector $q_2$ is an extreme ray of $\QLR_5$. Its double $2q_2$ is the entropy vector of a 12-qubit graph state with two qubits per system, and also of a 12-qutrit weighted graph state. So $q_2\in\StabZ{2}{5}\cap\StabZ{3}{5}$. But $q_2$ itself is not the entropy vector of any qubit or qutrit stabilizer state. Since every system has entropy $1$, such a state could be brought to a form in which each system holds a single qudit (Lemma~\ref{lem:A-normal}). Every stabilizer state of six qubits or qutrits is local-Clifford equivalent to a (weighted) graph state~\cite{si:VandenNest2004,si:Bahramgiri2006}, and local Clifford unitaries do not change entropies. An exhaustive search over all $2^{15}$ six-vertex graphs and all $3^{15}$ weighted six-vertex graphs over $\GF(3)$ finds no state with entropy vector $q_2$. The two realizations of $2q_2$ and the search code are part of the data accompanying this paper (Methods).
\end{example}

\subsection{One common information is not enough}
\label{app:five-oneci}

Some excluded orbits pass every single-common-information test: for every pair of sets $X,Y\subseteq\{A,\dots,F\}$, the lift $h_r$ extends to a polymatroid on seven elements in which the seventh element is a common information of $X$ and $Y$ in the sense of Eq.~\eqref{eq:ci}. For 18 orbits we verified this with exact rational extensions for all pairs; floating-point computations, which we did not certify, indicate that it holds for all 1{,}369 orbits excluded with conditional Ingleton instances. These orbits are nevertheless outside $\Stab_5$: conditional Ingleton instances on the extension exclude them. So $\Stab_5$ is strictly smaller than the intersection of $\QLR_5$ with all inequalities obtainable from one common information.

One of the 18 orbits, with entries at most $5$, has the representative
\begin{align*}
  r'=(&2,2,3,2,4,4,4,2,4,4,5,3,5,4,4,2,\\
  &4,4,4,3,4,4,4,3,4,4,4,3,3,3,2)
\end{align*}
in binary order of the subsets of $\{A,\dots,E\}$ ($A$ is the least significant bit, so the entries are $r'(A),r'(B),r'(AB),r'(C),\dots,r'(ABCDE)$). It is an extreme ray of $\QLR_5$ (the tight facets have rank 30). For each of the $\binom{63}{2}=1{,}953$ pairs $X,Y$ there is a rational polymatroid extension in which $Z$ is a common information of $X$ and $Y$: for 1{,}338 pairs it is stored, and for the others it is explicit ($Z$ a copy of the smaller set when $X\subseteq Y$ or $Y\subseteq X$, and $Z=0$ when $I(X;Y)=0$). All extensions were verified in integer arithmetic against the Shannon inequalities on seven elements. Nevertheless, $r'$ violates
\begin{equation}\label{eq:B-oneci}
\begin{aligned}
&S(A)+S(AB)+3S(ABC)+S(AD)+2S(CD)\\
&\quad+2S(BCD)+2S(ACE)+2S(BCE)\\
&\quad+S(ADE)+S(BDE)+S(CDE)\\
&\ge S(B)+S(AC)+3S(BC)+S(D)+2S(ABD)\\
&\quad+3S(ACD)+S(AE)+2S(ABCE)\\
&\quad+S(DE)+S(BCDE)
\end{aligned}
\end{equation}
by one unit. This inequality holds on $\Stab_5$. With $Z=U_A\cap U_B$, the combination
\begin{equation*}
\begin{aligned}
0\le{}&6\,g(Z|ABCDEF)+2I(A;Z|BCDEF)\\
&+I(B;Z|ADEF)+I(B;Z|ACDE)\\
&+I(C;Z|ABDEF)+2I(C;Z|BDEF)\\
&+I(C;Z|AB)+I(C;Z|B)+I(D;Z|BF)\\
&+I(D;Z|ABCEF)\\
&+I(D;Z|BE)+I(D;Z|BC)+I(E;Z|ACD)\\
&+I(E;Z|ABD)+I(E;Z|AD)+I(E;Z|BC)\\
&+I(E;Z|AC)+2I(F;Z|ABCDE)\\
&+I(F;Z|ABCE)+I(F;Z|B)+I(F;Z|A)\\
&+\Ingl(B,E;A,D|CZ)\\
&+\Ingl(E,F;A,B|DZ)+\Ingl(B,Z;A,E)\\
&+\Ingl(A,C;D,Z)+\Ingl(A,D;C,Z)\\
&+\Ingl(A,C;F,Z|BDE)\\
&+\Ingl(A,D;E,Z|F)+\Ingl(D,Z;C,F|E)
\end{aligned}
\end{equation*}
of 21 Shannon inequalities and 8 Ingleton instances becomes~\eqref{eq:B-oneci} after the substitutions $g(AZ)=h(A)$, $g(BZ)=h(B)$, $g(Z)=I(A;B)$ and $h=h_S$. Here $\Ingl(a,b;c,d|K)$ denotes the right-hand side minus the left-hand side of the Ingleton instance $I(a;b|K)\le I(a;b|cK)+I(a;b|dK)+I(c;d|K)$ of Eq.~\eqref{eq:ingleton}; it is nonnegative for subgroup arrangements because contraction by $K$ preserves them.

\subsection{Published six-variable inequalities}
\label{app:five-dfz}

We transcribed the twenty six-variable inequalities printed in Sec.~6 of Ref.~\cite{si:DFZ2009}, Eqs.~(44)--(63), term by term from the printed equations, and pulled each of them back, for all $720$ assignments of its variables to the systems $A,\dots,F$, to a linear form on pure-symmetric vectors. All twenty are balanced. As further checks of the transcription, which detect only errors that make an inequality invalid, none of them is violated by 6{,}000 random arrangements of six subspaces over $\GF(2)$, $\GF(3)$ and $\GF(5)$, by the entropy vectors of six-qubit graph states, or by any of the 641 orbits known to lie in the stabilizer cone. The pulled-back forms exclude 149{,}922 orbits of the catalog, 45{,}292 of them by Eqs.~(44)--(48) alone; all of these orbits are also excluded by our certificates from one common information. The vector
\begin{align*}
  r''=(&3,3,4,3,4,5,5,2,4,4,5,4,5,5,4,3,\\
  &5,5,5,4,5,5,4,5,5,5,5,5,4,5,3),
\end{align*}
in the same order as $r'$ above, is an extreme ray of $\QLR_5$: it satisfies all facet inequalities, and the 57 that vanish at $r''$ have rank 30. With the variables $(A,B,C,D,E,F)$ of Eq.~(45) of Ref.~\cite{si:DFZ2009} assigned to the systems $(D,C,E,F,B,A)$, that inequality reads
\begin{align*}
  I(C;D)\le{}&I(D;E)+I(C;F|E)+I(B;D|F)\\
  &+I(A;D|B)+I(C;D|A),
\end{align*}
and at $r''$ its left-hand side is $1$ while every term on the right vanishes. The witness $r$ of Table~\ref{tab:witness} satisfies all twenty pulled-back inequalities. We could not retrieve the full list of the sharp six-variable inequalities from the address given in Ref.~\cite{si:DFZ2009}.
\section{Proof of Lemma~\ref{lem:JI}, and inequalities with a derived \texorpdfstring{$C$}{C}}
\label{app:lemma1}

Let $V$ be a finite abelian group and write $h(X)=\log\abs{X}$ for subgroups, with sums of subgroups as joins. For subgroups $X$, $Y$, $Y'$, $W_1,\dots,W_k$ and $W\le W'$ of $V$ we use three elementary facts about indices:
\begin{itemize}
\item[(F1)] $[X:X\cap W_1\cap\dots\cap W_k]\le\prod_j[X:X\cap W_j]$;
\item[(F2)] $[X:X\cap W]\le[X:X\cap W']\,[W':W]$, because $(X\cap W')/(X\cap W)$ embeds into $W'/W$;
\item[(F3)] $\log[X:X\cap(Y+Y')]=h(X|YY')$, because $X/(X\cap(Y+Y'))\cong(X+Y+Y')/(Y+Y')$.
\end{itemize}
The same proof works for subspaces of a vector space over any field, with dimensions in place of logarithms of orders.

\subsection{Proof of Lemma~\ref{lem:JI}}\label{app:lemma1-proof}

\emph{Step 0: separating the $A_i$.} Let $O_i=A_i\cap A_{\ne i}$. Then $\log\abs{O_i}=I(A_i;A_{\ne i})\le\alpha$, by subadditivity. Put $O=\sum_iO_i$; then $\log\abs{O}\le M\alpha$, and $O\le A_{\ne i}$ for every $i$. Let $q\colon V\to V/O$ and $\bar A_i=q(A_i)$. The sum $\bar A=\sum_i\bar A_i$ is direct: if $\sum_ia_i=\sum_io_i$ with $a_i\in A_i$ and $o_i\in O_i$, then $a_i-o_i\in A_i\cap A_{\ne i}=O_i$ for each $i$, so $q(a_i)=0$. Let $\pi_i\colon\bar A\to\bar A_i$ be the projections. Then $q(A_{\ne i})=\ker\pi_i$, and for $y\in\bar A$,
\begin{equation}\label{eq:C-twist}
  \sum_{i=1}^M\bigl(y-\pi_i(y)\bigr)=My-y=t\,y .
\end{equation}

\emph{Proof of \LB.} Assume that multiplication by $t$ is injective on $V$; then it is injective on $V/O$. Let $B_i'=B_i\cap A_{\ne i}$ and $C_1=C\cap\bigcap_i(A_i+B_i')$; note $C_1\le A$.
\begin{enumerate}
\item By (F3), $[B_i:B_i']=2^{\zeta_i}$, hence $[A_i+B_i:A_i+B_i']\le2^{\zeta_i}$. By (F2) and (F3), $\log[C:C\cap(A_i+B_i')]\le\eps_i+\zeta_i$, and by (F1), $\log\abs{C_1}\ge h(C)-\sum_i(\eps_i+\zeta_i)$.
\item Let $x\in C_1$ and write $x=a+b$ with $a\in A_i$ and $b\in B_i'\le A_{\ne i}$. Then $\pi_i(qx)=q(a)$, so $qx-\pi_i(qx)=q(b)\in q(B_i)$.
\item Let $Y=q(C_1)$ and $K_i=Y\cap\ker\pi_i$. The preimage of $K_i$ under $q$ is $(C_1+O)\cap A_{\ne i}=(C_1\cap A_{\ne i})+O$, by the modular law and $O\le A_{\ne i}$. Hence $\log\abs{K_i}\le\log\abs{C\cap A_{\ne i}}=\delta_i$.
\item Define $\Phi\colon Y^M\to\bar A$ by $\Phi(y_1,\dots,y_M)=\sum_i(y_i-\pi_iy_i)$. By step~2, the image of $\Phi$ lies in $q(B_1+\dots+B_M)$. If $\Phi(y)=0$, put $\sigma=\sum_iy_i$. Applying $\pi_k$ gives $\pi_k(\sigma-y_k)=0$, so $\kappa_k:=\sigma-y_k\in K_k$, and $\sum_k\kappa_k=t\,\sigma$. The map $y\mapsto(\kappa_k)_k$ is injective on $\ker\Phi$, because $\kappa$ determines $\sigma$ (multiplication by $t$ is injective) and then $y_k=\sigma-\kappa_k$. Hence $\abs{\ker\Phi}\le\prod_k\abs{K_k}$.
\item Combining, $h(B_1\cdots B_M)\ge\log\abs{q(\sum_iB_i)}\ge M\log\abs{Y}-\sum_k\delta_k$, and $\log\abs{Y}\ge\log\abs{C_1}-\log\abs{O}\ge h(C)-\sum_i(\eps_i+\zeta_i)-M\alpha$. Thus
\[
  M\,h(C)\le h(B_1\cdots B_M)+M\textstyle\sum_i(\eps_i+\zeta_i)+M^2\alpha+\sum_i\delta_i ,
\]
which implies \LB{} since $\gamma$ and the $\delta_i$ are nonnegative.
\end{enumerate}

\emph{Proof of \LA.} Assume $tV=0$; then also $t(V/O)=0$. Let $C_A=C\cap A$, so that $\log[C:C_A]=\gamma$, and $B_i''=B_i\cap A_{\ne i}\cap(A_i+C_A)$.
\begin{enumerate}
\item By (F1), $[B_i:B_i'']\le[B_i:B_i\cap A_{\ne i}]\,[B_i:B_i\cap(A_i+C_A)]$. The first factor is $2^{\zeta_i}$. By (F2) the second is at most $[B_i:B_i\cap(A_i+C)]\,[A_i+C:A_i+C_A]\le2^{\eta_i}2^{\gamma}$. Hence $\log[B_i:B_i'']\le\zeta_i+\eta_i+\gamma$.
\item Let $b\in B_i''$ and write $b=a+x$ with $a\in A_i$ and $x\in C_A$. Since $b\in A_{\ne i}$, $\pi_i(qb)=0$, so $q(a)=-\pi_i(qx)$ and $q(b)=qx-\pi_i(qx)$.
\item Define $\Phi\colon q(C_A)^M\to\bar A$ as before. By step~2 its image contains $q(B_1''+\dots+B_M'')$. By Eq.~\eqref{eq:C-twist} and $t(V/O)=0$, its kernel contains the diagonal $\{(y,\dots,y)\}$. Hence $\log\abs{q(\sum_iB_i'')}\le(M-1)\log\abs{q(C_A)}\le(M-1)\,h(C)$.
\item Finally, $h(B_1\cdots B_M)\le\log\abs{\sum_iB_i''}+\sum_i\log[B_i:B_i'']$ and $\log\abs{\sum_iB_i''}\le\log\abs{q(\sum_iB_i'')}+\log\abs{O}$. Together with $\log\abs{O}\le M\alpha$ and step~1 this gives \LA.
\end{enumerate}

\emph{The configuration~\eqref{eq:JI}.} All error terms vanish: the $A_i$ are independent, $c\in A$, $\langle c\rangle\cap A_{\ne i}=0$ (the $i$th coordinate of a multiple $\lambda c$ is $\lambda$), $c=e_i+(c-e_i)\in A_i+B_i$, $c-e_i\in A_{\ne i}$, and $c-e_i\in A_i+C$. Moreover $h(C)=\log\abs{\Rg}$, and $h(B_1\cdots B_M)$ is $\log$ of the order of the column space of $J-I$ over $\Rg$. The adjugate of $J-I$ is $\pm(J-tI)$, whose off-diagonal entries are $\pm1$, so the Smith form of $J-I$ over $\Z$ is $\mathrm{diag}(1,\dots,1,t)$. Over $\Z_N$ ($N\ge2$ a modulus, as in Sec.~\ref{sec:lemma1}) the column space therefore has order $N^{M-1}\,N/\gcd(t,N)$. Over a field of characteristic $\ell$ its dimension is $M-1$ if $\ell\mid t$ and $M$ otherwise. Hence \LB{} fails by $\log\gcd(t,N)$ (by one unit over a field with $\ell\mid t$), and \LA{} fails by $\log(N/\gcd(t,N))$ (by one unit over a field with $\ell\nmid t$).

The bound in step~5 of the proof of \LB{} is stronger than stated, since it needs neither $\gamma$ nor the factor $M$ in front of $\sum_i\delta_i$; we state the weaker form because it is the one used to derive the inequality $\fnine$ of Sec.~\ref{sec:transfer}. The stated forms of \LA{} and \LB{} were also tested numerically: on all $\{0,1\}$-point configurations for $t=2$ over $\GF(2)$ and $\GF(3)$ (self-test gate G31b), on random point configurations over rings $\Z_N$ (gate G32b), and on random perturbations of the configuration~\eqref{eq:JI} over $\GF(2)$, $\GF(3)$, $\GF(5)$, $\Z_4$ and $\Z_9$ for $t\le4$ (script \path{check_sec5.py}; Supplementary Table~\ref{tab:repro}). The form whose hypothesis holds was never violated, and the other form was.

\subsection{Deriving \texorpdfstring{$C$}{C}}
\label{app:lemma2}

The subgroup $C$ of the configuration~\eqref{eq:JI} is $(A_1+B_1)\cap\dots\cap(A_M+B_M)$. Deriving it in this way, and bounding the terms of Lemma~\ref{lem:JI} that involve $C$ by terms of the $A_i$ and $B_i$, gives the following inequalities in $2M$ variables, with the lower bound $\hcap(C)$ for $h(C)$.

\begin{applemma}[Deriving $C$]\label{lem:JIstar}
Let $t\ge2$, $M=t+1$, and let $A_1,\dots,A_M$, $B_1,\dots,B_M$ be subgroups of a finite abelian group $V$ with rank function $h$. Write $P_m=A_m+B_m$ and juxtapose indices for sums, $P_{1k}=P_1+P_k$ and $P_{12k}=P_1+P_2+P_k$; let $A$, $A_{\ne i}$, $\alpha$ and $\zeta_i$ be as in Lemma~\ref{lem:JI}, choose an index $j(i)\ne i$ for every $i$, and put
\begin{align*}
  \hcap(C)={}&h(P_1)+h(P_2)-h(P_{12})\\
  &+\textstyle\sum_{k=3}^{M}\bigl(h(P_k)-h(P_{1k})-h(P_{2k})+h(P_{12k})\bigr),\\
  \gamma^*={}&h(A_1\cdots A_MB_1\cdots B_M)-h(A),\\
  \delta^*_i={}&I(A_{j(i)}B_{j(i)};B_i)+I(A_i;A_{\ne i}\mid B_i),\\
  \theta_i={}&I(A_i;A_{j(i)}B_{j(i)}).
\end{align*}
\begin{itemize}
\item[\LB$^*$] If multiplication by $t$ is injective on $V$, then
\[
  M\,\hcap(C)\le h(B_1\cdots B_M)+M\Bigl[\gamma^*+M\alpha+\textstyle\sum_i(\delta^*_i+\zeta_i)\Bigr].
\]
\item[\LA$^*$] If $tV=0$, then
\begin{align*}
  h(B_1\cdots B_M)\le{}&\textstyle\sum_i\bigl(h(P_i)-h(A_i)\bigr)-\hcap(C)+\sum_i\theta_i\\
  &+M\alpha+M\gamma^*+\textstyle\sum_i\zeta_i .
\end{align*}
\end{itemize}
\end{applemma}

On the subgroups $A_i=\langle e_i\rangle$ and $B_i=\langle c-e_i\rangle$ of~\eqref{eq:JI} with $M\ge4$, $\hcap(C)=1$ and $h(P_i)-h(A_i)=1$ in units of $\log\abs{\Rg}$, and every error term vanishes; so over a field \LB$^*$ fails by one unit when the characteristic divides $t$ and \LA$^*$ when it does not, and over $\Z_N$ they fail by $\log\gcd(t,N)$ and by $\log(N/\gcd(t,N))$. For $t=2$, $\hcap(C)=0$ on the configuration and neither inequality is violated, so for qubits we use the transfers of Sec.~\ref{sec:six} (Methods, Secs.~\ref{sec:transfer} and~\ref{sec:reverse}). For $t=3$, \LB$^*$ and \LA$^*$ play the roles of Theorems~3.1 and~4.1 of Ref.~\cite{si:DFZ2015char}, the first of which is incorrect as printed.

\emph{Proof of Lemma~\ref{lem:JIstar}.}

Let $C=P_1\cap\dots\cap P_M$ and apply Lemma~\ref{lem:JI} to $A_1,\dots,A_M$, $B_1,\dots,B_M$ and $C$. We bound the terms that involve $C$, using $h(X)+h(Y)=h(X+Y)+h(X\cap Y)$ for subgroups $X$ and $Y$; in particular $I(X;Y)=h(X\cap Y)$.
\begin{enumerate}
\item $h(C)\ge \hcap(C)$. Put $C^{(2)}=P_1\cap P_2$, so that $h(C^{(2)})=h(P_1)+h(P_2)-h(P_{12})$, and $C^{(k)}=C^{(k-1)}\cap P_k$ for $k=3,\dots,M$; then $C=C^{(M)}$. Since $C^{(k-1)}\le P_1\cap P_2$, the sum $C^{(k-1)}+P_k$ lies in $(P_1+P_k)\cap(P_2+P_k)$, whose rank is $h(P_{1k})+h(P_{2k})-h(P_{12k})$ because $(P_1+P_k)+(P_2+P_k)=P_{12k}$. Hence $h(C^{(k)})=h(C^{(k-1)})+h(P_k)-h(C^{(k-1)}+P_k)\ge h(C^{(k-1)})+h(P_k)-h(P_{1k})-h(P_{2k})+h(P_{12k})$.
\item $\gamma=h(C\mid A)\le\gamma^*$, because $C\le\sum_m(A_m+B_m)$.
\item $\eps_i=h(C\mid A_iB_i)=0$, because $C\le P_i$.
\item $\delta_i=h(C\cap A_{\ne i})\le\delta^*_i$. Put $E_i=P_i\cap A_{\ne i}$ and $E'_i=B_i\cap A_{\ne i}\le E_i$. Since $P_i+A_{\ne i}=A+B_i$, we have $h(E_i)-h(E'_i)=h(A_i\mid B_i)-h(A_i\mid A_{\ne i}B_i)=I(A_i;A_{\ne i}\mid B_i)$. With $j=j(i)$, $C\cap A_{\ne i}\le P_j\cap E_i$, and $(P_j\cap E_i)/(P_j\cap E'_i)$ embeds into $E_i/E'_i$. Hence $\delta_i\le h(P_j\cap E'_i)+h(E_i)-h(E'_i)\le h(P_j\cap B_i)+I(A_i;A_{\ne i}\mid B_i)=\delta^*_i$.
\item $\eta_i=h(B_i\mid A_iC)=h(P_i)-h(A_i)-h(C)+h(A_i\cap C)$, because $A_i+B_i+C=P_i$; and $h(A_i\cap C)\le h(A_i\cap P_{j(i)})=\theta_i$.
\end{enumerate}
\LB{} reads $M\,h(C)\le h(B_1\cdots B_M)+M\bigl[\gamma+M\alpha+\sum_i(\delta_i+\eps_i+\zeta_i)\bigr]$; bounding $h(C)$ from below by step~1 and $\gamma$, $\delta_i$ and $\eps_i$ from above by steps~2--4 gives \LB$^*$. In \LA{}, $h(B_1\cdots B_M)\le(M-1)\,h(C)+M\alpha+M\gamma+\sum_i(\zeta_i+\eta_i)$, step~5 turns the right-hand side into $\sum_i\bigl(h(P_i)-h(A_i)\bigr)-h(C)+\sum_ih(A_i\cap C)+M\alpha+M\gamma+\sum_i\zeta_i$; bounding $-h(C)$ by $-\hcap(C)$, $h(A_i\cap C)$ by $\theta_i$ and $\gamma$ by $\gamma^*$ gives \LA$^*$. Every step uses only $\abs{X+Y}\,\abs{X\cap Y}=\abs{X}\,\abs{Y}$, inclusions of subgroups and embeddings of quotients, so Lemma~\ref{lem:JIstar} holds for subgroups of finite abelian groups under the hypotheses of Lemma~\ref{lem:JI}.

\emph{The configuration~\eqref{eq:JI} without $C$.} Let $M\ge4$ and $\Rg=\GF(\ell)$ or $\Z_N$, and measure ranks in units of $\log\abs{\Rg}$. Then $P_m=\langle e_m,c\rangle$. For $k\ge3$ and one $l\notin\{1,2,k\}$, the $4\times4$ minor of $(e_1,e_2,e_k,c)$ in the coordinates $1$, $2$, $k$ and $l$ equals $1$; so $h(P_m)=2$, $h(P_{12})=h(P_{1k})=h(P_{2k})=3$, $h(P_{12k})=4$ and $\hcap(C)=1$. The error terms vanish: $\alpha=0$; $\gamma^*=0$, since every $B_i\le A=\Rg^M$; $\zeta_i=0$, since $c-e_i\in A_{\ne i}$; $I(A_i;A_{\ne i}\mid B_i)=0$, since $B_i\le A_{\ne i}$; and, for $j=j(i)$, $\theta_i=I(A_i;A_jB_j)=0$ and $I(A_jB_j;B_i)=0$, because comparing the $i$th coordinate and a coordinate outside $\{i,j\}$ shows that $\langle e_j,c\rangle$ meets $\langle e_i\rangle$ and $\langle c-e_i\rangle$ only in $0$. So every $\delta^*_i$ vanishes. Finally $h(P_i)-h(A_i)=1$. So \LB$^*$ reads $M\le h(B_1\cdots B_M)$ and \LA$^*$ reads $h(B_1\cdots B_M)\le M-1$, and the value of $h(B_1\cdots B_M)$ computed above gives the failures stated above. For $M=3$, by contrast, $P_{12}$, $P_{13}$, $P_{23}$ and $P_{123}$ are all of $\Rg^3$, so $\hcap(C)=2+2-3+(2-3-3+3)=0$ and neither inequality is violated.

\subsection{Relation to the inequalities of Pe\~na-Macias and Sarria}
\label{app:lemma1-ps}

For vector spaces, \LB{} follows from Theorem~9 of Ref.~\cite{si:PenaSarria2021}, with their $n$ equal to $t$, by adding Shannon terms: as linear forms, \LB{} equals $M$ times that inequality plus $M\sum_{i=2}^{t}I(A_1\cdots A_{i-1};A_i)$. Their Theorem~7 plays the role of \LA{}: on the configuration~\eqref{eq:JI} it reduces to $h(B_1\cdots B_M)\le t\,h(C)$. For $t=2$, \LB{} also follows from the inequality (b) of Ref.~\cite{si:PenaSarria2019} and from the corrected inequality (6.17) of Ref.~\cite{si:BKL2011} by adding Shannon terms (Supplementary Note~\ref{app:blind}). These theorems are proved for vector spaces; Lemma~\ref{lem:JI} proves both inequalities for finite abelian groups.
\section{The eight seven-variable inequalities}
\label{app:blind}

\begin{table}[t]
\caption{The eight seven-variable characteristic-dependent linear rank inequalities of Proposition~\ref{prop:blind}. Variables are labeled by the points $A,B,C,W,X,Y,Z=100,010,001,110,101,011,111$ of $\GF(2)^3$. ``Fano'' is the configuration of these seven vectors over $\GF(2)$, ``non-Fano'' the same vectors over a field of odd characteristic; each inequality holds in the characteristic shown and takes the value $-1$ on the configuration in the last column, with (b) multiplied by $M=3$. Equation and theorem numbers are those of the cited arXiv versions; Table~\ref{tab:eight-coeffs} lists all eight forms. The corrected form of (6.17) is that of Supplementary Note~\ref{app:misprints}; its validity rests on the proof of Ref.~\cite{si:BKL2011} with the missing terms restored.}
\label{tab:eight}
\begin{ruledtabular}
\begin{tabular}{lccc}
inequality & source & char. & fails on\\
\hline
(65) & \cite{si:DFZ2015rate} & odd & Fano\\
(91) & \cite{si:DFZ2015rate} & $2$ & non-Fano\\
(6.8) & \cite{si:BKL2011} & $2$ & non-Fano\\
(6.17), corrected & \cite{si:BKL2011} & odd & Fano\\
(a), $t=2$ & \cite{si:PenaSarria2019} & $2$ & non-Fano\\
(b), $t=2$ & \cite{si:PenaSarria2019} & odd & Fano\\
\LA, $t=2$ & Lemma~\ref{lem:JI} & $2$ & non-Fano\\
\LB, $t=2$ & Lemma~\ref{lem:JI} & odd & Fano\\
\end{tabular}
\end{ruledtabular}
\end{table}

As in Table~\ref{tab:eight}, the seven variables are the points of $\GF(2)^3$: $A$, $B$ and $C$ are $100$, $010$ and $001$, then $W$, $X$ and $Y$ are $110$, $101$ and $011$, and $Z$ is $111$. The inequalities (65) and (91) of Dougherty, Freiling and Zeger are taken from Ref.~\cite{si:DFZ2015rate} (arXiv version, Theorems 8.6 and 8.8), and the inequalities (6.8) and (6.17) of Blasiak, Kleinberg and Lubetzky from Ref.~\cite{si:BKL2011} (arXiv version, Theorems 6.4 and 6.6), with (6.17) in the corrected form of Supplementary Note~\ref{app:misprints}. The inequalities (a) and (b) are taken from Example~6 of Ref.~\cite{si:PenaSarria2019} with $t=2$ and $M=3$, \ie, seven variables, under the identification $(A_1,A_2,A_3,B_1,B_2,B_3,C)=(A,B,C,Y,X,W,Z)$; we multiply (b) by $M=3$ to clear its denominator. All were transcribed term by term from the cited sources. The inequalities \LA{} and \LB{} are those of Lemma~\ref{lem:JI} with $t=2$ and the same identification; \LB{} is written out in Eq.~\eqref{eq:LB2}. Table~\ref{tab:eight-coeffs} lists all eight forms. As a check of the transcriptions, each takes the value $-1$ on the configuration listed in Table~\ref{tab:eight} and is nonnegative on the other one. In this labeling, (6.8) equals (a)${}+I(B;C)$, the corrected (6.17) equals \LB${}-6I(A;B)$, and (b) equals \LB${}-3I(A;B)-3I(B;C)$; nevertheless no two of the eight forms are related by a relabeling of the variables. Since $\Dual$ maps the Shannon cone on seven elements into itself, the first and the last of these identities make Proposition~\ref{prop:blind} for (6.8) and \LB{} a consequence of the statements for (a) and (b). Modulo conditional mutual informations of either sign, the eight forms fall into five classes, distinguished by their values on the modular vectors $\ind{i}$: those of (65), (91), \LA, (a) and (b); the class of (a) contains (6.8), and that of (b) contains \LB{} and the corrected (6.17). This does not transfer blindness to the corrected (6.17), which differs from \LB{} by a negative Shannon term, so Proposition~\ref{prop:blind} is certified for each of the eight forms separately. Theorems~7 and~9 of Ref.~\cite{si:PenaSarria2021}, with their parameter $n=2$ and the same identification, add nothing new here: Theorem~7 (characteristic $2$) coincides with (6.8), and three times Theorem~9 (odd characteristic) equals (b)${}+3I(B;C)$, which is \LB${}-3I(A;B)$. Both are therefore blind as well.

\begin{table*}[t]
\caption{The eight inequalities of Table~\ref{tab:eight}, each written as $0\le\sum_Tc_T\,h(T)$; the table lists the coefficients $c_T$ (empty entries are zero), with the variables labeled as in Table~\ref{tab:eight}. Column (b) is $M=3$ times the inequality as printed, and column (6.17) is the corrected form of Supplementary Note~\ref{app:misprints}.}
\label{tab:eight-coeffs}
\begin{ruledtabular}
\footnotesize
\begin{tabular}{lrrrrrrrr}
set & (65) & (91) & (6.8) & (6.17) & (a) & (b) & \LA & \LB\\
\hline
$A$ & $3$ & $5$ & $2$ & $3$ & $2$ & $6$ & $3$ & $9$\\
$B$ & $4$ & $4$ & $2$ & $3$ & $1$ & $3$ & $3$ & $9$\\
$C$ & $3$ & $5$ & $3$ & $9$ & $2$ & $6$ & $3$ & $9$\\
$W$ & $1$ & $1$ &  &  &  &  &  & \\
$X$ & $1$ & $1$ &  &  &  &  &  & \\
$Y$ & $1$ & $1$ &  &  &  &  &  & \\
$Z$ & $1$ & $3$ & $11$ & $6$ & $11$ & $6$ & $2$ & $6$\\
$AB$ & $-5$ & $-2$ & $3$ & $6$ & $3$ & $3$ & $-1$ & \\
$AC$ &  & $-4$ & $2$ &  & $2$ &  & $-1$ & \\
$AX$ & $-2$ &  &  &  &  &  &  & \\
$AY$ &  &  &  & $-3$ &  & $-3$ &  & $-3$\\
$AZ$ &  &  & $-1$ &  & $-1$ &  & $-1$ & \\
$BC$ & $-5$ & $-3$ & $2$ &  & $3$ & $3$ & $-1$ & \\
$BX$ &  &  &  & $-3$ &  & $-3$ &  & $-3$\\
$BZ$ &  &  & $-1$ &  & $-1$ &  & $-1$ & \\
$CW$ & $-3$ &  &  & $-3$ &  & $-3$ &  & $-3$\\
$CZ$ &  &  & $-1$ &  & $-1$ &  & $-1$ & \\
$WY$ & $-3$ &  &  &  &  &  &  & \\
$WZ$ &  & $-1$ &  &  &  &  &  & \\
$XZ$ & $-1$ & $-3$ &  &  &  &  &  & \\
$YZ$ & $-2$ & $-2$ &  &  &  &  &  & \\
$ABC$ & $-5$ & $-13$ & $-4$ & $-12$ & $-4$ & $-12$ & $-6$ & $-12$\\
$ABW$ & $5$ & $2$ & $1$ & $3$ & $1$ & $3$ & $1$ & $3$\\
$ABZ$ &  &  & $-3$ & $-3$ & $-3$ & $-3$ &  & $-3$\\
$ACX$ & $2$ & $4$ & $1$ & $3$ & $1$ & $3$ & $1$ & $3$\\
$ACZ$ &  &  & $-3$ & $-3$ & $-3$ & $-3$ &  & $-3$\\
$AYZ$ & $2$ & $2$ & $1$ & $3$ & $1$ & $3$ & $1$ & $3$\\
$BCY$ & $5$ & $3$ & $1$ & $3$ & $1$ & $3$ & $1$ & $3$\\
$BCZ$ &  &  & $-3$ & $-3$ & $-3$ & $-3$ &  & $-3$\\
$BXZ$ & $1$ & $3$ & $1$ & $3$ & $1$ & $3$ & $1$ & $3$\\
$CWZ$ & $3$ & $1$ & $1$ & $3$ & $1$ & $3$ & $1$ & $3$\\
$WXY$ & $3$ & $-1$ & $-1$ & $1$ & $-1$ & $1$ & $-1$ & $1$\\
$ABCZ$ &  & $6$ & $1$ & $3$ & $1$ & $3$ & $3$ & $3$\\
$CWXY$ &  & $1$ &  &  &  &  &  & \\
\end{tabular}
\end{ruledtabular}
\end{table*}

\emph{Proof of Proposition~\ref{prop:blind}.} If $h$ is a self-dual polymatroid, then $f(h)=\frac12\bigl(f(h)+f(h^*)\bigr)=\frac12(f+f\circ\Dual)(h)\ge0$, by the certificates below. The Shannon cone is invariant under relabelings, and $\Dual$ commutes with them, so the same holds in every relabeling. For a pure state on six parties and a purifier, $h_S$ is a polymatroid: it is submodular by strong subadditivity, monotone because $S(Xi)+S(i)\ge S(X)$, and self-dual by purity (Lemma~\ref{lem:A-selfdual}).\qed

For each inequality $f\ge0$ we computed $g=f+f\circ\Dual$, where $(f\circ\Dual)(h)=f(h^*)$, and solved a linear program for nonnegative multipliers $c_k$ with $g=\sum_kc_k\theta_k$, where the $\theta_k$ are the elemental Shannon inequalities on seven variables, $h(i\mid N\setminus i)\ge0$ and $I(i;j\mid K)\ge0$. On the support of the solution we recomputed the multipliers exactly as rationals and verified the identity in exact arithmetic. The certificates involve at most 114 elemental inequalities each; the script \path{check_sec6.py} recomputes and verifies them (Supplementary Table~\ref{tab:repro}).

The certificates show more than the application to quantum states needs: each $f$ is nonnegative on every self-dual polymatroid, not only on lifts of quantum states. On the self-dual vectors of interest the inequalities are far from tight. Over all $5{,}040$ relabelings, \LB{} takes values at least $17$ and (65) at least $19$ on $\hF$, while \LA{} takes values at least $6$ and (91) at least $16$ on $\hNF$.

\emph{A further inequality.} Corollary~2.7 of Pe\~na-Macias~\cite{si:PenaMacias2023} gives two characteristic-dependent rank inequalities for the configuration~\eqref{eq:JI}, one valid over finite fields whose characteristic does not divide $n-1$ and one over finite fields whose characteristic divides it. For $n=3$, with the subspaces $A_1,A_2,A_3$, $B_1,B_2,B_3$ and $C$ of Eq.~\eqref{eq:JI}, the first reads, as we transcribed it from the journal PDF,
\begin{align*}
h(C)\le{}&\tfrac14\bigl(h(A_1)+h(A_2)+h(A_3)\bigr)+h(C\mid A_1A_2A_3)\\
&+\textstyle\sum_i\bigl[h(C\mid A_iB_i)+I(C;A_{\ne i})+I(C;B_i)\\
&\qquad+I(C;A_{\ne i}B_i)\bigr]+I(C;B_1B_2B_3).
\end{align*}
It takes the value $-1/4$, as the right side minus the left side, on the configuration over $\GF(2)$, and $3/4$ over $\GF(3)$ and $\GF(5)$. It is blind as well: four times $f+f\circ\Dual$, for $f$ the right side minus the left side, is a nonnegative combination of $88$ elemental Shannon inequalities with rational coefficients (\path{check_rev_blind_pm.py}). We have not tested his other inequalities.

\emph{Symmetrization.} A dependence on the local dimension at six parties requires a seven-variable inequality $g$ with $g\circ\Dual=g$ that is valid for one local dimension and violated by the lift of a stabilizer state of another (Sec.~\ref{sec:blind}). Indeed, if $f\ge0$ on $\LcZ{d'}{7}$ and $f(h_S)<0$ for a stabilizer state of local dimension $d$, we may assume that $f$ is balanced (Supplementary Note~\ref{app:reduction-balanced}), since $h_S(i\mid N\setminus i)=0$; then $g=f+f\circ\Dual$ is nonnegative on $\LcZ{d'}{7}$, because $\Dual$ maps this cone into itself (Lemma~\ref{lem:A-dual}), it satisfies $g\circ\Dual=g$, because $\Dual^2h=h-\sum_ih(i\mid N\setminus i)\,\ind{i}$ for $h(\emptyset)=0$ and $f$ is balanced, and $g(h_S)=2f(h_S)<0$.
\section{The transfer and the six-party inequality}
\label{app:transfer}

\begin{figure}[t]
\centering
\begin{tikzpicture}[scale=1.6,
  pt/.style={circle, fill=black, inner sep=1.6pt},
  lab/.style={font=\footnotesize}]
\coordinate (A) at (90:1.6);
\coordinate (B) at (210:1.6);
\coordinate (C) at (330:1.6);
\coordinate (W) at ($(A)!0.5!(B)$);
\coordinate (X) at ($(A)!0.5!(C)$);
\coordinate (Y) at ($(B)!0.5!(C)$);
\coordinate (Z) at (0,0);
\draw[thick, gray] (A)--(B)--(C)--cycle;
\draw[thick, gray] (A)--(Y); \draw[thick, gray] (B)--(X); \draw[thick, gray] (C)--(W);
\draw[thick, gray] (Z) circle (0.8);
\draw[very thick, blue!70!black] (A)--(B);
\draw[very thick, blue!70!black] (A)--(C);
\draw[very thick, blue!70!black] (A)--(Y);
\foreach \p in {A,B,C,W,X,Y,Z} \node[pt] at (\p) {};
\node[lab, above] at (A) {$A=100$};
\node[lab, below left] at (B) {$B=010$};
\node[lab, below right] at (C) {$C=001$};
\node[lab, left] at (W) {$W=110$};
\node[lab, right] at (X) {$X=101$};
\node[lab, below] at (Y) {$Y=011$};
\node[lab, right, fill=white, inner sep=1pt] at (Z) {$Z=111$};
\end{tikzpicture}
\caption{The Fano plane on the seven systems of Theorem~\ref{thm:steane}. Each point $i$ lies on three lines; for each line $\ell=\{i,j,k\}$ the derived family of the Methods (Sec.~\ref{sec:transfer}) keeps the piece $F_i^\ell=U_i\cap(U_j+U_k)$ of $U_i$ that is generated within the line, and $\hat F_i$ is the sum of the three pieces (shown for $i=A$, whose lines $ABW$, $ACX$ and $AYZ$ are highlighted).}
\label{fig:fano}
\end{figure}
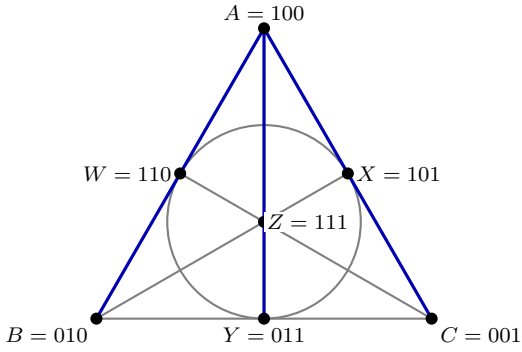

\subsection{The transfer lemma}
\label{app:transfer-lemma}

We keep the labels of Sec.~\ref{sec:qss} and the derived family $(\hat F_P)$ of Eq.~\eqref{eq:derived}, with rank function $\hat h$. Write $m_i^\ell=I(i;\ell\setminus i)$, which equals $\log\abs{F_i^\ell}$. For the three lines $\ell_1,\ell_2,\ell_3$ through $i$, let $Q_{ab}=(\ell_a\cup \ell_b)\setminus i$, and put
\begin{equation}\label{eq:ui}
  u_i=\tfrac13\Bigl[2\sum_{a<b}I(i;Q_{ab})-\sum_am_i^{\ell_a}\Bigr],\qquad
  l_i=\tfrac13\sum_am_i^{\ell_a}.
\end{equation}

\begin{applemma}[Transfer]\label{lem:transfer}
The rank function $\hat h$ of the derived family satisfies the following bounds, all linear in $h$:
\begin{enumerate}
\item[(i)] $l_i\le\hat h(i)\le u_i$ for every point $i$;
\item[(ii)] $\hat h(Q)\ge h(Q)-\sum_{i\in Q}\bigl(h(i)-l_i\bigr)$ for every set $Q$ of points;
\item[(iii)] $\hat h(\ell)\le u_i+u_j+u_k-m_i^\ell$ for every line $\ell=\{i,j,k\}$ and every $i\in \ell$;
\item[(iv)] $\hat h(ABCZ)\le u_A+u_B+u_C+u_Y+u_Z-m_Y^{BCY}-m_Z^{AYZ}$.
\end{enumerate}
\end{applemma}

The inequality \LB{} of Lemma~\ref{lem:JI} with $t=2$ and $(A_1,A_2,A_3,B_1,B_2,B_3,C)=(A,B,C,Y,X,W,Z)$ reads
\begin{equation}\label{eq:LB2}
\begin{aligned}
0\le{}&9h(A)+9h(B)+9h(C)-12h(ABC)+6h(Z)\\
&+3h(ABCZ)-3h(ABZ)-3h(ACZ)-3h(BCZ)\\
&+3h(AYZ)+3h(BXZ)+3h(CWZ)-3h(AY)\\
&-3h(BX)-3h(CW)+3h(BCY)+3h(ACX)\\
&+3h(ABW)+h(WXY),
\end{aligned}
\end{equation}
it holds for subgroup arrangements in groups of odd order and takes the value $-1$ on the Fano configuration; we write $\LB_2(h)$ for its right-hand side. Bounding the positive terms of Eq.~\eqref{eq:LB2} at $\hat h$ by (i), (iii) (with $i$ the first point of $\ell$ in the order $A,B,C,W,X,Y,Z$) and (iv) of Lemma~\ref{lem:transfer}, and the negative terms by (ii), gives Eq.~\eqref{eq:f9}. At $\hF$, $m_i^\ell=1$ and $u_i=l_i=1$, the bounds (iii) equal $2$, the bound (iv) equals $3$, and the lower bounds (ii) equal the ranks of the Fano configuration; so each bound used in Eq.~\eqref{eq:f9} equals the corresponding rank of the Fano configuration, which gives $\fnine(\hF)/3=\LB_2(\rF)=-1$ in the proof of Theorem~\ref{thm:steane}.

\subsection{Proof of Lemma~\ref{lem:transfer}}

Let $U_A,\dots,U_Z$ be subgroups of a finite abelian group with rank function $h$, and define $F_i^\ell$, $\hat F_i$ and $\hat h$ as in Eq.~\eqref{eq:derived}. With dimensions in place of logarithms of orders, all statements hold for subspaces of a vector space over any field. For a set $Q$ of points we write $U_Q=\sum_{j\in Q}U_j$.

\begin{enumerate}
\item[(a)] \emph{Size of a piece.} $\log\abs{F_i^\ell}=h(i)+h(jk)-h(ijk)=m_i^\ell$.
\item[(b)] \emph{Pieces generate each other.} $F_i^\ell\le F_j^\ell+F_k^\ell$. If $v=x+y\in U_i$ with $x\in U_j$ and $y\in U_k$, then $x=v-y\in U_j\cap(U_i+U_k)=F_j^\ell$ and $y=v-x\in F_k^\ell$.
\item[(c)] \emph{Upper bound for points.} Let $\ell_1,\ell_2,\ell_3$ be the lines through $i$. For $a\ne b$, $F_i^{\ell_a}+F_i^{\ell_b}\le U_i\cap U_{Q_{ab}}$, a subgroup of order $2^{I(i;Q_{ab})}$. For subgroups $P,P',P''$ we have $\abs{P+P'+P''}\,\abs{P}\le\abs{P+P'}\,\abs{P+P''}$, because $P\le(P+P')\cap(P+P'')$. With $P=F_i^{\ell_a}$, $P'=F_i^{\ell_b}$ and $P''=F_i^{\ell_c}$ this gives
\[
  \hat h(i)\le\log\abs{F_i^{\ell_a}+F_i^{\ell_b}}+\log\abs{F_i^{\ell_a}+F_i^{\ell_c}}-\log\abs{F_i^{\ell_a}}
\]
for $\{a,b,c\}=\{1,2,3\}$, hence $\hat h(i)\le I(i;Q_{ab})+I(i;Q_{ac})-m_i^{\ell_a}$. Averaging over $a$ gives $\hat h(i)\le u_i$.
\item[(d)] \emph{Lower bounds.} $\hat F_i\ge F_i^\ell$ for each of the three lines, so $\hat h(i)\ge\max_{\ell}m_i^\ell\ge l_i$. For a set $Q$ of points, the natural map $\bigoplus_{i\in Q}U_i/\hat F_i\to U_Q/\sum_{i\in Q}\hat F_i$ is onto, so $h(Q)\le\hat h(Q)+\sum_{i\in Q}\bigl(h(i)-\hat h(i)\bigr)$, and (ii) follows from $\hat h(i)\ge l_i$.
\item[(e)] \emph{Lines.} For $\ell=\{i,j,k\}$, (b) gives $F_i^\ell\le\hat F_i\cap(\hat F_j+\hat F_k)$, so $I_{\hat h}(i;jk)\ge m_i^\ell$ and, by subadditivity and (c), $\hat h(\ell)=\hat h(i)+\hat h(jk)-I_{\hat h}(i;jk)\le u_i+u_j+u_k-m_i^\ell$.
\item[(f)] \emph{The set $ABCZ$.} By monotonicity $\hat h(ABCZ)\le\hat h(ABCYZ)$. By (b), $F_Y^{BCY}\le\hat F_Y\cap(\hat F_B+\hat F_C)$, so $\hat h(ABCY)\le\hat h(ABC)+\hat h(Y)-m_Y^{BCY}$. Again by (b), $F_Z^{AYZ}\le\hat F_Z\cap(\hat F_A+\hat F_Y)$, so $\hat h(ABCYZ)\le\hat h(ABCY)+\hat h(Z)-m_Z^{AYZ}$. Together with $\hat h(ABC)\le\hat h(A)+\hat h(B)+\hat h(C)$ and (c) this gives (iv).
\end{enumerate}

\subsection{The inequality \texorpdfstring{$\fnine$}{f9}}

Substituting the bounds of Lemma~\ref{lem:transfer} into Eq.~\eqref{eq:LB2} as described in Supplementary Note~\ref{app:transfer-lemma} and multiplying by $3$ gives
\begin{widetext}
\begin{equation}\label{eq:F-f9}
\begin{aligned}
\fnine(h)={}&36\,h(A)+45\,h(B)+54\,h(C)+18\,h(W)+21\,h(X)+21\,h(Y)+45\,h(Z)-10\,h(AB)\\
&-10\,h(AC)-42\,h(AW)-42\,h(AX)-45\,h(AY)-13\,h(AZ)-22\,h(BC)-51\,h(BW)\\
&-36\,h(BX)-42\,h(BY)-10\,h(BZ)-36\,h(CW)-51\,h(CX)-51\,h(CY)-10\,h(CZ)\\
&-13\,h(WX)-10\,h(WY)-51\,h(WZ)-13\,h(XY)-51\,h(XZ)-51\,h(YZ)-36\,h(ABC)\\
&+103\,h(ABW)-9\,h(ABZ)+103\,h(ACX)-9\,h(ACZ)+100\,h(AYZ)+115\,h(BCY)\\
&-9\,h(BCZ)+88\,h(BXZ)+88\,h(CWZ)+36\,h(WXY)+48\,h(ABCZ)+92\,h(ABXY)\\
&+92\,h(ACWY)+104\,h(AWXZ)+98\,h(BCWX)+98\,h(BWYZ)+98\,h(CXYZ)-42\,h(ABCWX)\\
&-42\,h(ABCWY)-14\,h(ABCWZ)-42\,h(ABCXY)-14\,h(ABCXZ)-20\,h(ABCYZ)\\
&-14\,h(ABWXY)-42\,h(ABWXZ)-42\,h(ABWYZ)-36\,h(ABXYZ)-14\,h(ACWXY)\\
&-42\,h(ACWXZ)-36\,h(ACWYZ)-42\,h(ACXYZ)-20\,h(AWXYZ)-20\,h(BCWXY)\\
&-36\,h(BCWXZ)-42\,h(BCWYZ)-42\,h(BCXYZ)-14\,h(BWXYZ)-14\,h(CWXYZ).
\end{aligned}
\end{equation}
\end{widetext}
The expansion was done by computer and checked by an independent re-derivation. Its values are $\fnine(\hF)=\fnine(\rF)=-3$, $\fnine(\rF^*)=0$, $\fnine(\hNF)=81$ and $\fnine(\rNF)=33$.

The form $\fnine$ is not a nonnegative combination of Shannon inequalities and relabelings of the eight inequalities of Table~\ref{tab:eight}, or of Theorems~7 and~9 of Ref.~\cite{si:PenaSarria2021} with $n=2$: the form $g=\fnine+\fnine\circ\Dual$ has $g(\hF)=-6$, whereas for such a combination $g$ would be a nonnegative combination of Shannon inequalities, by Proposition~\ref{prop:blind} (Supplementary Note~\ref{app:blind} for the theorems of Ref.~\cite{si:PenaSarria2021}) and because $\Dual$ maps the Shannon cone into itself~\cite[Prop.~3]{si:Kaced2018}.

\subsection{The invariant slice}

A vector on the subsets of the Fano plane that is invariant under its collineation group takes one value on each orbit of subsets: points ($v_1$), pairs ($v_2$), lines ($v_3$), noncollinear triples ($v_3'$), complements of lines, the other four-element sets, and sets of five, six and seven elements. Self-duality determines the last five values: $v_7=\frac72v_1$, $v_6=v_7$, $v_5=v_2-2v_1+v_7$, and the two kinds of four-element sets take $v_3-3v_1+v_7$ and $v_3'-3v_1+v_7$. Substituting into Eq.~\eqref{eq:F-f9} gives
\[
  \tfrac13\fnine=-130\,v_1-430\,v_2+421\,v_3-21\,v_3' ,
\]
which equals $210\,\beta+30\,\tau+21\,\nu-\mu$ with $\tau=2v_1-v_2$, $\mu=v_1+v_2-v_3$, $\nu=v_1+v_2-v_3'$ and $\beta=v_1-2\mu$; this is Eq.~\eqref{eq:invariant-f9}. For $\hF$, $(v_1,v_2,v_3,v_3')=(2,4,5,6)$, so $\tau=\beta=\nu=0$ and $\mu=1$.

\subsection{Faces}

Among the $679$ elemental Shannon inequalities on seven elements, $385$ vanish at $\hF$ and $490$ at each of $\rF$ and $\rF^*$; their ranks are $125$ and $126$. For the non-Fano vectors the numbers are $367$ (rank $125$) at $\hNF$ and $487$ (rank $126$) at each of $\rNF$ and $\rNF^*$. All ranks were computed exactly over $\Q$.

\begin{appproposition}[Faces]\label{prop:faces}
The smallest face of the Shannon cone on seven elements that contains $\hF$ is $\cone(\rF,\rF^*)$, a two-dimensional cone, and $\rF$ and $\rF^*$ are extreme rays of the Shannon cone. The same holds for the lift $\hNF=\rNF+\rNF^*$ of the non-Fano state. Consequently, if $\hF$ is a positive combination of polymatroids, for instance of rank functions of subspace arrangements, each of them is of the form $a\rF+b\rF^*$ with $a,b\ge0$. Moreover, for every prime $q$,
\[
  \hF\in\LcZ{q}{7}\iff\LcZ{q}{7}\cap\cone(\rF,\rF^*)\ne\{0\},
\]
and likewise for the non-Fano pair.
\end{appproposition}

\emph{Proof of Proposition~\ref{prop:faces}.} Every inequality tight at $\hF=\rF+\rF^*$ is tight at both summands, so the face of $\hF$ contains the extreme rays $\rF$ and $\rF^*$. They are linearly independent and the face has dimension $127-125=2$, so the face is the cone they span. If $\hF=\sum_kc_kg_k$ with $c_k>0$ and polymatroids $g_k$, every $g_k$ lies in this face. If $\hF\in\LcZ{q}{7}$, then $\hF$ itself is a nonzero element of the intersection. Conversely, if $a\rF+b\rF^*\in\LcZ{q}{7}$ with $a,b\ge0$ not both zero, then applying $\Dual$, which maps $\LcZ{q}{7}$ into itself (Lemma~\ref{lem:A-dual}) and swaps $\rF$ and $\rF^*$ (the Fano matroid has no coloops, so $\Dual^2\rF=\rF$), and adding gives $(a+b)\hF\in\LcZ{q}{7}$. The same argument, with the ranks above, applies to the non-Fano pair.
\section{The reverse direction at six parties}
\label{app:reverse}

\subsection{Proof of Proposition~\ref{prop:f12}}
\label{app:reverse-proof}

\emph{Setting.} Let $S$ be the entropy vector, in bits, of a qubit stabilizer state of the seven systems $N=\{A,B,C,W,X,Y,Z\}$, labeled as in Sec.~\ref{sec:qss}, with the corners, diagonal points and forms of Sec.~\ref{sec:reverse}. By Lemma~\ref{lem:A-normal} we may assume that every system $P$ holds exactly $S(P)$ qubits. Let $L$ be the Lagrangian subspace of the state in $V=\bigoplus_{P\in N}V_P$, $V_P=\GF(2)^{2S(P)}$, with the symplectic form $\omega=\sum_P\omega_P$, where $\omega_P$ is the form $\langle\cdot,\cdot\rangle$ of Supplementary Note~\ref{app:reduction-stab} on the qubits of $P$ (in this note, $\omega$ does not denote a root of unity). Let $\pr_P$ be the restriction to the coordinates of $P$, and $L_T$ the subspace of elements of $L$ supported in $T\subseteq N$. By Eq.~\eqref{eq:A-entropy},
\[
  \dim L_T=\textstyle\sum_{P\in T}S(P)-S(T)=k(T),
\]
in particular $L_P=0$ for every system $P$. For a line $\ell$ and $P\in\ell$, the kernel of $\pr_P$ on $L_\ell$ is $L_{\ell\setminus P}$, so $\dim\pr_P(L_\ell)=k_P(\ell)$. If two lines $\ell\ne\ell'$ meet in $P$, then for $s\in L_\ell$ and $s'\in L_{\ell'}$ the supports meet only in $P$, and isotropy gives $\omega_P(\pr_Ps,\pr_Ps')=\omega(s,s')=0$; so $\pr_P(L_\ell)$ and $\pr_P(L_{\ell'})$ are orthogonal. For a subspace $U$ of a symplectic space let $\rk(U)$ be the rank of the form restricted to $U$; then $\dim(U\cap U^\perp)=\dim U-\rk(U)$.

\begin{applemma}[Corners]\label{lem:R-corners}
Let $V_1,V_2,V_3$ be pairwise orthogonal subspaces of a symplectic space of dimension $2s$ over a field, of dimensions $d_1,d_2,d_3$. Then $\rk(V_i)\le2s-d_j-d_k$ for $\{i,j,k\}=\{1,2,3\}$, and $\dim(V_1+V_2+V_3)\le4s-d_1-d_2-d_3$.
\end{applemma}

\begin{proof}
Let $U=V_j+V_k$. Since $U\le V_i^\perp$, we have $V_i\cap U\le V_i\cap V_i^\perp$, so $\dim(V_i+U)\ge d_i+\dim U-(d_i-\rk V_i)=\dim U+\rk V_i$. The intersection $V_j\cap V_k$ is orthogonal to $V_i$, $V_j$ and $V_k$, so it lies in $(V_i+U)^\perp$. Hence $d_j+d_k=\dim U+\dim(V_j\cap V_k)\le\dim U+2s-\dim(V_i+U)\le2s-\rk V_i$. For $V_\Sigma=V_1+V_2+V_3$: since the $V_i$ are pairwise orthogonal, the pull-back of the form on $V_\Sigma$ along the surjection $V_1\oplus V_2\oplus V_3\to V_\Sigma$ is the direct sum of the forms on the $V_i$, and a pull-back along a surjection has the same rank; so $\rk V_\Sigma=\sum_i\rk V_i$. Since $V_\Sigma\cap V_\Sigma^\perp\le V_\Sigma^\perp$, we get $\dim V_\Sigma-\rk V_\Sigma\le2s-\dim V_\Sigma$, hence $\dim V_\Sigma\le s+\frac12\sum_i\rk V_i\le s+\frac12\sum_i(2s-d_j-d_k)=4s-\sum_id_i$.
\end{proof}

\begin{applemma}[A line transports the rank]\label{lem:R-line}
For a line $\ell=\{P,P',Q\}$, $\rk\pr_Q(L_\ell)\le\rk\pr_P(L_\ell)+\rk\pr_{P'}(L_\ell)$.
\end{applemma}

\begin{proof}
For $s,s'\in L_\ell$, isotropy gives $\omega_P(\pr_Ps,\pr_Ps')+\omega_{P'}(\pr_{P'}s,\pr_{P'}s')+\omega_Q(\pr_Qs,\pr_Qs')=0$. So the pull-back of the form on $\pr_Q(L_\ell)$ along the surjection $\pr_Q\colon L_\ell\to\pr_Q(L_\ell)$ is minus the sum of the pull-backs of the forms on $\pr_P(L_\ell)$ and on $\pr_{P'}(L_\ell)$. A pull-back along a surjection has the same rank, and the rank of a sum of two forms is at most the sum of their ranks.
\end{proof}

\begin{applemma}[Points]\label{lem:R-points}
Let $\Lambda=\sum_\ell L_\ell$, the sum over the six lines, and let $D_P=\pr_P(\Lambda)$, of dimension $\Delta_P$. Then $\Delta_Z\le \mathcal{U}_Z$, $\Delta_W\le \mathcal{U}_W$, $\Delta_X\le \mathcal{U}_X$ and $\Delta_Y\le \mathcal{U}_Y$.
\end{applemma}

\begin{proof}
$D_P=\sum_{\ell\ni P}\pr_P(L_\ell)$. At $Z$, the restrictions of $L_{AYZ}$, $L_{BXZ}$ and $L_{CWZ}$ are pairwise orthogonal, of dimensions $k_Z(\cdot)$, and Lemma~\ref{lem:R-corners} with $s=S(Z)$ gives $\Delta_Z\le \mathcal{U}_Z$. At $W$, $D_W=U+U'$ with $U=\pr_W(L_{ABW})$ and $U'=\pr_W(L_{CWZ})\le U^\perp$, since the two lines meet only in $W$. So $\Delta_W\le\dim(U+U^\perp)=2S(W)-\dim(U\cap U^\perp)=2S(W)-k_W(ABW)+\rk U$. By Lemma~\ref{lem:R-line}, $\rk U\le\rk\pr_A(L_{ABW})+\rk\pr_B(L_{ABW})$, and Lemma~\ref{lem:R-corners} at the corner $A$ (lines $ABW$, $ACX$, $AYZ$) and at the corner $B$ (lines $ABW$, $BCY$, $BXZ$) bounds these ranks by $e_A^{ABW}$ and $e_B^{ABW}$; so $\Delta_W\le \mathcal{U}_W$. For $X$ use the line $ACX$, with the corners $A$ and $C$, and for $Y$ the line $BCY$, with the corners $B$ and $C$.
\end{proof}

\begin{applemma}[Supports in $\Lambda$]\label{lem:R-supports}
Let $z(T)$ be the dimension of the subspace of elements of $\Lambda$ supported in $T\subseteq N$. Then (i) $z(T)\le k(T)$, so $z(P)=0$ for every system $P$; (ii) $z(\ell)=k(\ell)$ for every line $\ell$; and (iii) $z(ABCZ)\ge k(ABW)+k(CWZ)-\Delta_W$.
\end{applemma}

\begin{proof}
(i) $\Lambda\le L$. (ii) $L_\ell\le \Lambda$. (iii) The map $L_{ABW}\oplus L_{CWZ}\to D_W$, $(s,s')\mapsto\pr_Ws+\pr_Ws'$, is onto, since $ABW$ and $CWZ$ are the only lines through $W$; so its kernel $K$ has dimension $k(ABW)+k(CWZ)-\Delta_W$. On $K$, the map $(s,s')\mapsto s+s'$ takes values in the elements of $\Lambda$ supported in $ABCZ$, and it is injective: $s+s'=0$ means $s=-s'\in L_{ABW}\cap L_{CWZ}=L_W=0$.
\end{proof}

\begin{applemma}[The dual arrangement]\label{lem:R-dual}
Let $D=\bigoplus_{P\in N}D_P$, which contains $\Lambda$, and let $R_P$ be the image of $D_P$ in $D/\Lambda$. Then $\dim\sum_{P\in T}R_P=\sum_{P\in T}\Delta_P-z(T)$ for every $T\subseteq N$.
\end{applemma}

\begin{proof}
Every element of $\Lambda$ has its $P$-component in $D_P$, so $\Lambda\le D$. With $D_T=\bigoplus_{P\in T}D_P$, $\sum_{P\in T}R_P=(D_T+\Lambda)/\Lambda\cong D_T/(D_T\cap \Lambda)$, and $D_T\cap \Lambda$ is the subspace of elements of $\Lambda$ supported in $T$.
\end{proof}

\emph{Assembly.} Let $r(T)=\dim\sum_{P\in T}R_P=\sum_{P\in T}\Delta_P-z(T)$, and apply \LA{} of Lemma~\ref{lem:JI}, with $t=2$, to the subspaces $R_P$ over $\GF(2)$, with $(A_1,A_2,A_3,B_1,B_2,B_3,C)=(R_A,R_B,R_C,R_Y,R_X,R_W,R_Z)$ as in Eq.~\eqref{eq:LB2}. By Lemma~\ref{lem:R-supports}, $r(P)=\Delta_P$ and $r(\ell)=\sum_{P\in\ell}\Delta_P-k(\ell)$. So $h(B_1B_2B_3)=r(WXY)=\Delta_W+\Delta_X+\Delta_Y-z(WXY)$ and $h(C)=\Delta_Z$; $\alpha=z(ABC)$ and $\gamma=r(ABCZ)-r(ABC)=\Delta_Z-z(ABCZ)+z(ABC)$; $\zeta_1=r(BCY)-r(BC)=\Delta_Y-k(BCY)+z(BC)$, and likewise $\zeta_2=\Delta_X-k(ACX)+z(AC)$ and $\zeta_3=\Delta_W-k(ABW)+z(AB)$; $\eta_1=r(AYZ)-r(AZ)=\Delta_Y-k(AYZ)+z(AZ)$, and likewise $\eta_2=\Delta_X-k(BXZ)+z(BZ)$ and $\eta_3=\Delta_W-k(CWZ)+z(CZ)$. The terms $\Delta_A$, $\Delta_B$ and $\Delta_C$ cancel, and \LA{} becomes
\begin{equation}\label{eq:R-assembly}
\begin{aligned}
  0\le{}&5\,\Delta_Z+\Delta_W+\Delta_X+\Delta_Y+z(WXY)+6\,z(ABC)\\
  &+\textstyle\sum_{\{P,P'\}}z(PP')-3\,z(ABCZ)-\sum_\ell k(\ell),
\end{aligned}
\end{equation}
with $\{P,P'\}$ running over the six pairs of corners. Bound $z(WXY)$, $z(ABC)$ and $z(PP')$ from above by $k$ (Lemma~\ref{lem:R-supports}(i)), and $-3\,z(ABCZ)$ by $-3\,k(ABW)-3\,k(CWZ)+3\,\Delta_W$ (Lemma~\ref{lem:R-supports}(iii)). Then $\Delta_Z$, $\Delta_W$, $\Delta_X$ and $\Delta_Y$ have the coefficients $5$, $4$, $1$ and $1$, and Lemma~\ref{lem:R-points} bounds them from above. The result is $0\le\ftwelve(S)$, with $\ftwelve$ as in Eq.~\eqref{eq:f12}, for the entropy vector $S$ of every qubit stabilizer state of the seven systems, and by linearity on the closed cone $\StabZ{2}{6}$. This proves Proposition~\ref{prop:f12}.

\subsection{A self-dual inequality for characteristic \texorpdfstring{$2$}{2}}
\label{app:reverse-g12}

For a set function $h$ on the seven systems put $g_{12}(h)=\ftwelve(\lambda_h)$. If $h$ is the rank function of seven subspaces over a field of characteristic $2$, which may be taken to be $\GF(2)$ (Supplementary Note~\ref{app:reduction-unions} and restriction of scalars), then $\lambda_h$ is the entropy vector of a qubit CSS state (fact~(b) of the Methods, Sec.~\ref{sec:methods-proofs}), so $g_{12}(h)\ge0$ by Proposition~\ref{prop:f12}. Since $\lambda_{h^*}=\lambda_h$, the form $g_{12}$ satisfies $g_{12}\circ\Dual=g_{12}$, and $g_{12}(\rNF)=g_{12}(\rNF^*)=-1$, while $g_{12}(\rF)=g_{12}(\rF^*)=0$. So $g_{12}$ is a seven-variable linear rank inequality for characteristic $2$ of the kind that Sec.~\ref{sec:blind} shows to be necessary for a dependence at six parties. It is not a nonnegative combination of Shannon inequalities and of relabelings of the eight inequalities of Table~\ref{tab:eight}, since all of these are nonnegative at the self-dual polymatroid $\hNF$ (Proposition~\ref{prop:blind}), whereas $g_{12}(\hNF)=-2$. As for the Fano matroid, Proposition~\ref{prop:faces} turns Theorem~\ref{thm:reverse} into a strengthening of the classical fact that the non-Fano matroid is not representable in characteristic $2$~\cite{si:Oxley2011}.

\subsection{The expansion of \texorpdfstring{$\ftwelve$}{f12}}

By purity, a vector $S$ on the seven systems is determined by its values on the nonempty subsets of the six parties $A$, $B$, $C$, $W$, $X$, $Y$. In these coordinates,
\begin{widetext}
\begin{equation}\label{eq:R-f12}
\begin{aligned}
\ftwelve(S)={}&3\,S(A)+3\,S(B)+3\,S(C)-3\,S(W)-5\,S(AB)-2\,S(AC)-S(AW)-S(AX)\\
&-5\,S(AY)-2\,S(BC)-S(BW)-5\,S(BX)-S(BY)-5\,S(CW)-4\,S(CX)-4\,S(CY)\\
&-6\,S(ABC)+10\,S(ABW)+7\,S(ACX)+7\,S(BCY)-S(WXY)\\
&+11\,S(ABXY)+11\,S(ACWY)+11\,S(BCWX)-5\,S(ABCWX)-5\,S(ABCWY)\\
&-2\,S(ABCXY)-S(ABWXY)-S(ACWXY)-S(BCWXY)+2\,S(ABCWXY).
\end{aligned}
\end{equation}
\end{widetext}
Here $S(ABXY)=S(CWZ)$, $S(ACWY)=S(BXZ)$, $S(BCWX)=S(AYZ)$ and $S(ABCWXY)=S(Z)$. For instance, the coefficient $2$ of $S(Z)$ is $20-15=5$ from $5\,\mathcal{U}_Z$, plus $3$ from the pairs $AZ$, $BZ$ and $CZ$, minus $3$ from the lines through $Z$ and minus $3$ from $-3\,k(CWZ)$. The expansion was done by computer and checked by an independent re-derivation from Eq.~\eqref{eq:f12}.

\subsection{Powers of two: proof of Theorem~\ref{thm:twopowers}}
\label{app:reverse-powers}

The proof of Proposition~\ref{prop:f12} carries over to $\Z_{2^a}$ qudits, with logarithms of orders in place of dimensions, except for its first and its last step. The first, Lemma~\ref{lem:A-normal}, holds for prime local dimensions only; the compression below replaces it. The last, \LA{}, needs $2V=0$; we first replace that step.

\begin{applemma}[\LA{} for exponent $t^a$]\label{lem:R-LAa}
Let $t\ge2$, $M=t+1$ and $a\ge1$, and let $V$ be a finite abelian group with $t^aV=0$. For subgroups $A_1,\dots,A_M$, $B_1,\dots,B_M$ and $C$ of $V$, with the error terms of Lemma~\ref{lem:JI},
\[
  h(B_1\cdots B_M)\le\Bigl(M-\frac1a\Bigr)h(C)+M\alpha+M\gamma+\textstyle\sum_i(\zeta_i+\eta_i).
\]
\end{applemma}

\begin{proof}
For $a=1$ this is \LA. In general we follow the proof of \LA{} in Supplementary Note~\ref{app:lemma1}; only step~3 changes. Let $Y=q(C_A)$ and let $\Phi\colon Y^M\to\bar A$ be as there. By Eq.~\eqref{eq:C-twist}, $\Phi(y,\dots,y)=t\,y$, so the kernel of $\Phi$ contains the diagonal of the subgroup $Y[t]=\{y\in Y:ty=0\}$, and $\log\abs{q(\sum_iB_i'')}\le M\log\abs{Y}-\log\abs{Y[t]}$. Write $Y$ as a direct sum of cyclic groups of prime-power orders $p^e$. Since $t^aY=0$, each of them has $e\le a\,v_p(t)$, and its elements killed by $t$ form a subgroup of order $p^{\min(e,v_p(t))}\ge p^{e/a}$. Hence $\abs{Y[t]}\ge\abs{Y}^{1/a}$, and $\log\abs{q(\sum_iB_i'')}\le(M-\frac1a)\log\abs{Y}\le(M-\frac1a)\,h(C)$. Step~4 concludes as before.
\end{proof}

On the configuration~\eqref{eq:JI} over $\Z_{p^a}$ with $t=p$, every error term vanishes, $h(C)=a\log p$ and $h(B_1\cdots B_M)=\bigl((M-1)a+a-1\bigr)\log p$ (Supplementary Note~\ref{app:lemma1}), so equality holds.

\emph{Symplectic abelian groups.} A \emph{symplectic abelian group} is a finite abelian group $V$ with a biadditive form $\omega\colon V\times V\to\Q/\Z$ that is alternating, $\omega(v,v)=0$, and nondegenerate: $v\mapsto\omega(v,\cdot)$ is an isomorphism of $V$ onto its character group $\widehat V$. For a subgroup $U$ let $U^\perp=\{v:\omega(v,U)=0\}$. For an alternating biadditive form $\beta$ on a finite abelian group $W$ let $\mathrm{rad}\,\beta=\{w:\beta(w,W)=0\}$ and $\rk\beta=\log\abs{W/\mathrm{rad}\,\beta}$; for a subgroup $U$ of a symplectic group, $\rk(U)=\log\abs{U}-\log\abs{U\cap U^\perp}$ is the rank of $\omega$ restricted to $U$. We use three elementary facts.
\begin{itemize}
\item[(H1)] $\abs{U^\perp}=\abs{V}/\abs{U}$ and $U^{\perp\perp}=U$. Indeed, the restriction $\widehat V\to\widehat U$ is onto, because characters of $U$ extend to $V$ (Supplementary Note~\ref{app:reduction-lift}), and $v\mapsto\omega(v,\cdot)$ maps $U^\perp$ onto its kernel; and $U\le U^{\perp\perp}$, where both have the same order.
\item[(H2)] $\abs{U+U'}\,\abs{U\cap U'}=\abs{U}\,\abs{U'}$.
\item[(H3)] A pull-back along a surjective homomorphism has the same rank; $\rk(\beta_1+\beta_2)\le\rk\beta_1+\rk\beta_2$, because $\mathrm{rad}\,\beta_1\cap\mathrm{rad}\,\beta_2\le\mathrm{rad}(\beta_1+\beta_2)$ and $W/(\mathrm{rad}\,\beta_1\cap\mathrm{rad}\,\beta_2)$ embeds into $W/\mathrm{rad}\,\beta_1\times W/\mathrm{rad}\,\beta_2$; the rank of an orthogonal direct sum is the sum of the ranks; and $\rk(-\beta)=\rk\beta$.
\end{itemize}

\emph{Setting.} A system $P$ with $m_P$ qudits of dimension $2^a$ carries $V_P=\Z_{2^a}^{2m_P}$ with $\omega_P(u,v)=\langle u,v\rangle/2^a$, where $\langle\cdot,\cdot\rangle$ is the symplectic form of Supplementary Note~\ref{app:reduction-stab}, and the stabilizer group maps onto a Lagrangian subgroup $L=L^\perp$ of $V=\bigoplus_PV_P$. By Lemma~\ref{lem:A-entropy}, $S(T)=\log\abs{\pr_TL}-\frac12\log\abs{V_T}$ for every set $T$ of systems, where $V_T=\bigoplus_{P\in T}V_P$.

\emph{Compression.} The subgroup $L_P$ of elements of $L$ supported in $P$ is isotropic, and $\pr_PL=L_P^\perp$ inside $V_P$: an element $x\in V_P$ is orthogonal to $\pr_PL$ if and only if it is orthogonal to $L$, that is, if and only if $x\in L\cap V_P=L_P$; now use (H1). Replace $V_P$ by $V'_P=L_P^\perp/L_P$ with the induced form, which is nondegenerate by (H1), and $L$ by its image $L'$ in $V'=\bigoplus_PV'_P$. The kernel of $L\to L'$ is $\bigoplus_PL_P$: an element of $L$ whose components lie in the $L_P$ is the sum of these components, each of which lies in $L$. So $\abs{L'}^2=\abs{L}^2/\prod_P\abs{L_P}^2=\abs{V'}$, and the isotropic subgroup $L'$ is Lagrangian by (H1). For every $T$, the kernel of $\pr_TL\to\pr_TL'$ is $\bigoplus_{P\in T}L_P$, by the same argument, so $\log\abs{\pr_TL}-\frac12\log\abs{V_T}$ does not change. We may therefore assume $L_P=0$ for every $P$. Then $S(P)=\frac12\log\abs{V_P}$, and since $L_T$ is the kernel of $\pr_{N\setminus T}$ on $L$ and $\abs{L}^2=\abs{V}$, purity gives $\log\abs{L_T}=\frac12\log\abs{V_T}-S(T)=k(T)$. (The groups $V'_P$ need not be of the form $\Z_{2^a}^{2m}$; nothing below needs this.)

\emph{The proof.} Lemmas~\ref{lem:R-corners}--\ref{lem:R-dual} hold in this setting, with $\log\abs{\cdot}$ in place of $\dim$, with symplectic abelian groups $V$ of order $2^{2s}$ in Lemma~\ref{lem:R-corners}, and with the same proofs, using (H1)--(H3) for the facts of linear algebra used there. In Lemma~\ref{lem:R-corners}, (H2) gives $\abs{V_i+U}\ge\abs{V_i}\,\abs{U}/\abs{V_i\cap V_i^\perp}$ and the dimension count for $V_j$ and $V_k$, (H1) gives the order of $(V_i+U)^\perp$, and the bound for $V_\Sigma$ uses (H3) and (H1). Lemma~\ref{lem:R-line} is (H3). In Lemma~\ref{lem:R-points}, $\abs{U+U^\perp}=\abs{U}\,\abs{U^\perp}/\abs{U\cap U^\perp}$ by (H2) and $\abs{U^\perp}=\abs{V_W}/\abs{U}$ by (H1). Lemmas~\ref{lem:R-supports} and~\ref{lem:R-dual} hold verbatim, and $D/\Lambda$ has exponent dividing $2^a$. Apply Lemma~\ref{lem:R-LAa} with $M=3$ and $t=2$ to the family $(R_P)$ in $D/\Lambda$, in place of \LA. Its right-hand side contains $(3-\frac1a)h(C)=2\,\Delta_Z+(1-\frac1a)\,\Delta_Z$ instead of $2\,\Delta_Z$, so Eq.~\eqref{eq:R-assembly} holds with $(6-\frac1a)\,\Delta_Z$ in place of $5\,\Delta_Z$, and the same bounds give $0\le\ftwelve(S)+(1-\frac1a)\,\mathcal{U}_Z(S)=f^{(a)}(S)$ for the entropy vector $S$ of every stabilizer state of $\Z_{2^a}$ qudits, and by linearity on the closed cone $\StabZ{2^a}{6}$. This proves the inequality of Theorem~\ref{thm:twopowers}. For the rest of the theorem, the form $\mathcal{U}_Z$ equals $S(Z)+\sum_{\ell\ni Z}\bigl(S(\ell)-S(\ell\setminus Z)\bigr)$ and takes the value $1$ at $\SNF$ and at $\SF$, so $f^{(a)}(\SNF)=-1/a$ and, by Supplementary Note~\ref{app:reverse-sevenpoint}, $f^{(a)}(S^{(2^j)})=1-j+(1-1/a)\,j=1-j/a$. If $d\mid2^a$, the inclusion is Corollary~\ref{cor:divisibility}. Otherwise either $d$ has an odd prime factor $p$, and then $\SNF\in\StabZ{p}{6}\subseteq\StabZ{d}{6}$, or $d=2^j$ with $j>a$, and then $S^{(d)}\in\StabZ{d}{6}$; in both cases $f^{(a)}$ is negative there. For odd $d$, $\StabZ{2^a}{n}\not\subseteq\StabZ{d}{n}$ by Theorem~\ref{thm:steane}, since $\StabZ{2}{n}\subseteq\StabZ{2^a}{n}$. For $n>6$, add parties that hold nothing.

\subsection{Seven-point states: proof of Theorem~\ref{thm:reverse}}
\label{app:reverse-sevenpoint}

Let $G$ be the $3\times7$ matrix and $G_T$ its columns in a set $T$ of systems. If $s_1,\dots,s_r$ are the nonzero Smith invariants of $G_T$ over $\Z$, then for every modulus $N\ge2$ (not the set of systems) the projection of the code over $\Z_N$ to $T$ has order $\prod_iN/\gcd(s_i,N)$, by the Smith normal form, and $r$ is the rank of $G_T$ over $\Q$. All $2\times2$ minors of $G$ lie in $\{0,\pm1\}$, and the only $3\times3$ minor of absolute value larger than $1$ is $\det(W,X,Y)=\pm2$. A set of columns of rank $3$ other than $\{W,X,Y\}$ contains a basis of $\Q^3$ other than $\{W,X,Y\}$, whose determinant is $\pm1$. So every set of columns has Smith invariants equal to $1$, except $\{W,X,Y\}$, whose Smith invariants are $1$, $1$ and $2$; and the rank over $\Q$ equals the rank over $\GF(p)$ for every odd $p$, which is the rank function $\rNF$ of the non-Fano matroid. By the Chinese remainder theorem, the order of a projection of the code over $\Z_d$ is the product of the orders for the prime powers $p^a$ that divide $d$ exactly, so $S^{(d)}$ is the sum of the vectors $S^{(p^a)}$ (Lemma~\ref{lem:A-css}). For odd $p$, every projection has order $p^{a\rNF(T)}$, so $S^{(p^a)}=a\log_2p\cdot\SNF$. For $p=2$ the same holds except for $T=\{W,X,Y\}$, where the order is $2^{3a-1}$; by Lemma~\ref{lem:A-css} this lowers the entropy by one bit on $WXY$ and on its complement $ABCZ$, so $S^{(2^a)}=\SF+(a-1)\SNF$, and $\ftwelve(S^{(2^a)})=1-a$. Summing over the prime powers gives the value of Theorem~\ref{thm:reverse}. The non-inclusions follow from Proposition~\ref{prop:f12}, since this value is negative for every $d\ge3$; the inclusions follow from Corollary~\ref{cor:divisibility}; and for $n>6$ we add parties that hold nothing.

\subsection{Distance of the witnesses from the cones}
\label{app:margins}

Let $g\cdot S\ge0$ be a linear inequality that holds on a cone, written in the $2^n-1$ entropies of the nonempty sets of parties (the purifier is eliminated by purity), and let a vector $S_w$ satisfy $g\cdot S_w=-v<0$. For every $S$ in the cone, $v\le g\cdot(S-S_w)\le\norm{g}_1\max_X\abs{S(X)-S_w(X)}$. So every vector of the cone differs from $S_w$ by at least $v/\norm{g}_1$ in some entropy, and its Euclidean distance from $S_w$ is at least $v/\norm{g}_2$. For $\fnine$ and the Steane state ($v=3$ bits; $\fnine$ is a form in the lift, and in the $63$ entropies it has $\norm{g}_1=2{,}856$ and $\norm{g}_2\approx590.2$, for any choice of the purifier), this gives $1/952\approx1.05\times10^{-3}$ bit and a Euclidean distance of at least $5.0\times10^{-3}$ bit from $\StabZ{d}{6}$, $d$ odd, and from the closed convex hull of the union of these cones, on which $\fnine$ holds as well. For $\ftwelve$ and the non-Fano vector $\SNF$, normalized to entropy $1$ on every system, with $Z$ as the purifier ($v=1$, $\norm{g}_1=129$, $\norm{g}_2\approx28.97$), it gives $1/129$ bit and $0.034$ bit from $\StabZ{2}{6}$. These are lower bounds from a single inequality each. For the Steane state an explicit positive combination gives an upper bound: let $\SNF^{(\ell)}$ be the non-Fano vector relabeled by a collineation of the Fano plane that maps $WXY$ to the line $\ell$, so that $\SNF^{(\ell)}$ exceeds $\SF$ by $1$ on $\ell$ and on its complement and agrees with it elsewhere. The vector $\frac{5}{36}\sum_\ell\SNF^{(\ell)}$, a positive combination of entropy vectors of stabilizer states of every odd prime local dimension, differs from $\SF$ by $-\SF(X)/36$ on the sets $X$ that are not lines or complements of lines and by $1/12$ on the others; since $\SF$ takes values between $1$ and $3$, it lies within $1/12$ bit of $\SF$ in every entropy. So the distance of the Steane state from $\StabZ{d}{6}$, $d$ odd, measured by the largest difference in a single entropy, lies between $1/952$ and $1/12$ bit (Corollary~\ref{cor:divisibility} gives every odd $d$). We have not computed the distances to the cones more precisely.
\section{Doubled codes, generalized cube states and Galois qudits}
\label{app:cubes}

\subsection{Doubled codes and Proposition~\ref{prop:sandwich}}
\label{app:cubes-sandwich}

\begin{applemma}[Doubled codes]\label{lem:D-doubled}
Let $C\le\Z_d^m$ be a code whose coordinates are distributed among the parties $1,\dots,n$, with rank function $r=r_C$ on $\{1,\dots,n\}$. Let $C^{(2)}=\{(c,c):c\in C\}\le\Z_d^{2m}$, where the parties hold their coordinates of the first copy and the purifier holds the whole second copy. Then the CSS state of $C^{(2)}$ has $S(X)=r(X)$ for every set $X$ of parties.
\end{applemma}

\begin{proof}
For a set $X$ of parties, the complement $N\setminus X$ contains the purifier, whose coordinates determine the codeword. So $r_{C^{(2)}}(N\setminus X)=\log\abs{C}=r_{C^{(2)}}(N)$, and Lemma~\ref{lem:A-css} gives $S(X)=r_{C^{(2)}}(X)=r(X)$.
\end{proof}

The code $C^{(2)}$ need not have a generator matrix in systematic form; if $d$ is divisible by the square of a prime, its CSS state then need not lie in the Clifford orbit of $\ket{0\cdots0}$ (Supplementary Note~\ref{app:reduction-stab}). If the $k$ rows of a matrix $G$ generate $C$, the code generated by $[\,G\mid I_k\,]$, in which the parties hold the coordinates of the first block and the purifier holds the last $k$ coordinates, is systematic, and the same argument gives $S(X)=r(X)$ for every set $X$ of parties: the purifier's coordinates determine the codeword. So the first implication of Proposition~\ref{prop:sandwich}, proved below, also holds for the cones generated by the Clifford orbit of $\ket{0\cdots0}$.

\begin{proof}[Proof of Proposition~\ref{prop:sandwich}]
The second implication is the contrapositive of: $\LcZ{d}{n+1}\subseteq\LcZ{d'}{n+1}$ implies $\StabZ{d}{n}\subseteq\StabZ{d'}{n}$, which follows from Proposition~\ref{thm:reduction}. For the first, suppose $\LcZ{d}{n}\not\subseteq\LcZ{d'}{n}$. Since $\LcZ{d}{n}$ is generated by subgroup rank functions, one of them, $r$, lies outside $\LcZ{d'}{n}$, and there is a linear form $f\ge0$ on $\LcZ{d'}{n}$ with $f(r)<0$. Its balanced version $f_{\mathrm{bal}}(h)=f(h)-\sum_if(\ind{i})\,h(i\mid\text{rest})$ is again nonnegative on $\LcZ{d'}{n}$ (Supplementary Note~\ref{app:reduction-balanced}). Each $\ind{i}$ is a positive multiple of a subgroup rank function of exponent dividing $d'$, so $f(\ind{i})\ge0$ and $f_{\mathrm{bal}}(r)\le f(r)<0$. For a stabilizer state of $\Z_{d'}$ qudits, the restriction of $h_S$ to subsets of the parties lies in $\LcZ{d'}{n}$, so $f_{\mathrm{bal}}(S)=f_{\mathrm{bal}}(h_S)\ge0$; this is a valid linear inequality on $\StabZ{d'}{n}$. Realize $r$ by a code over $\Z_d$ (Lemma~\ref{lem:A-realize}). The doubled CSS state of Lemma~\ref{lem:D-doubled} lies in $\StabZ{d}{n}$ and has $f_{\mathrm{bal}}(S)=f_{\mathrm{bal}}(r)<0$.
\end{proof}

\emph{The witnesses of Theorem~\ref{thm:primes} for $n=2p+2$.} Let $M=p+1$ and let $\mathcal C$ be the row space of the $M\times2M$ matrix $[\,I\mid J-I\,]$ over $\GF(p)$ or over $\GF(q)$, with one coordinate of the first copy of $\mathcal C$ for each of the $2p+2$ parties and the second copy for the purifier. Every element of the configuration lies in the sum of the others, because $e_i=(c-e_j)-\sum_{k\ne i,j}e_k$ for any $j\ne i$ and $c-e_i=\sum_{j\ne i}e_j$; so its rank function $h$ has $h(i\mid\text{rest})=0$ for every element $i$ (there are no coloops), and the balanced forms of \LB$^*$ and \LA$^*$ agree with the forms themselves on the configuration. By Lemma~\ref{lem:D-doubled}, they take on the respective doubled codes their values on the configuration, $-1$ in units of $\log p$ and of $\log q$ (Supplementary Note~\ref{app:lemma2}).

\emph{Code states that contain the two configurations.} If the other qudits of a code state determine the codeword, the patterns of Theorems~\ref{thm:steane} and~\ref{thm:reverse} on seven of its qudits can be detected by a simpler argument, with the other qudits as the purifier. So the logical zero state of the CSS code $[[2^m-1,2^m-1-2m,3]]$ built from the Hamming code~\cite{si:CalderbankShor1996,si:Steane1996}, $m\ge4$, lies outside $\StabZ{d}{7}$ for every odd $d$ when the seven parties are qubits that form a Fano subplane; for $m=4$ it is the CSS state of the binary simplex code of length $15$, a logical basis state of the $15$-qubit code with which Bravyi and Kitaev distill magic states~\cite[Sec.~VI]{si:BravyiKitaev2005}. In the Steane and non-Fano states themselves the seven qudits are the whole state, the lift is self-dual, and Proposition~\ref{prop:blind} applies, which is why we needed the transfer arguments to obtain inequalities that detect them. In detail: if the qudits outside seven coordinates of a code $C$ over $\Z_d$ determine the codeword, then in the CSS state of $C$ these seven qudits, taken as the parties with all other qudits as the purifier, have $S(X)=r_C(X)$, by the argument of Lemma~\ref{lem:D-doubled}. For the binary simplex code of dimension $m\ge4$, whose CSS state is that logical zero state, this holds for the seven qubits of every Fano subplane, because the other $2^m-8$ points span $\GF(2)^m$; these seven qubits then have the entropies of the Fano matroid. The form \LB{} with $t=2$, balanced over the seven parties as in the proof of Proposition~\ref{prop:sandwich}, is nonnegative on $\StabZ{d}{7}$ for every odd $d$ and equals $-1$ on this state, since the Fano matroid has no coloops; so the state lies outside $\StabZ{d}{7}$ for every odd $d$. Likewise, for an odd prime $p$, the seven qudits $\{0,1\}^3\setminus\{0\}$ of the CSS state of the $p$-ary simplex code of dimension $3$, whose other $p^2+p-6$ points span $\GF(p)^3$, have the entropies of the non-Fano matroid, and \LA{} with $t=2$ shows that this state lies outside the qubit cone.

\subsection{Self-duality of the generalized cube codes}
\label{app:cubes-selfdual}

Let $\Rg$ be $\GF(p)$ or $\Z_N$ for a modulus $N\ge2$ (not the set of systems), and let $C\le\Rg^{\mathcal{E}}$ be a code whose coordinate set $\mathcal{E}$ is partitioned among the systems. Then
\begin{equation}\label{eq:D-dualcode}
  r_{C^\perp}(X)=\abs{X}\log\abs{\Rg}+r_C(\mathcal{E}\setminus X)-r_C(\mathcal{E})
\end{equation}
for every set $X$ of coordinates. Indeed, $\pr_X$ restricted to $C^\perp$ has kernel $C^\perp\cap\Rg^{\mathcal{E}\setminus X}$, which is the annihilator of $\pr_{\mathcal{E}\setminus X}C$ inside $\Rg^{\mathcal{E}\setminus X}$, of order $\abs{\Rg}^{\abs{\mathcal{E}\setminus X}}/\abs{\pr_{\mathcal{E}\setminus X}C}$; and $\abs{C^\perp}=\abs{\Rg}^{\abs{\mathcal{E}}}/\abs{C}$.

\begin{proof}[Proof of Lemma~\ref{lem:cube-selfdual}]
Let $G$ be the $(M+1)\times(2M+2)$ generator matrix with rows $1,x_1,\dots,x_M$ evaluated at the points, and $w$ the weights $-1$ on each $e_i$, $+1$ on each $c-e_i$, $-(t-1)$ on $c$ and $t-1$ on $0$. Over $\Z$:
\begin{align*}
  \textstyle\sum_xw_x&=-M+M-(t-1)+(t-1)=0,\\
  \textstyle\sum_xw_xx_j&=-1+(M-1)-(t-1)=0,\\
  \textstyle\sum_xw_xx_jx_k&=(M-2)-(t-1)=0\quad(j\ne k),
\end{align*}
and $\sum_xw_xx_j^2=\sum_xw_xx_j$. Hence $G\,\mathrm{diag}(w)\,G^{\mathsf T}=0$, and $C\,\mathrm{diag}(w)\le C^\perp$, where $C=C_t(\Rg)$. The columns at $0,e_1,\dots,e_M$ form a unimodular matrix, so $C$ is free of rank $M+1$ and $\abs{C}=\abs{\Rg}^{M+1}=\abs{\Rg}^{(2M+2)/2}$. If $t-1$ is a unit, all weights are units, so $\abs{C\,\mathrm{diag}(w)}=\abs{C}=\abs{C^\perp}$ and $C^\perp=C\,\mathrm{diag}(w)$. Scaling coordinates by units does not change the rank function, so $r_{C^\perp}=r_C=:r$. By Eq.~\eqref{eq:D-dualcode}, $r(X)=\abs{X}\log\abs{\Rg}+r(\mathcal{E}\setminus X)-r(\mathcal{E})$ for every $X$, which is $r^*=r$ since $r(\{x\})=\log\abs{\Rg}$ for every point (the constant functions lie in $C$). Then $S=\lambda_r=2r-\abs{X}\log\abs{\Rg}$ and $h_S=S+\sum_iS(i)\ind{i}=2r$. The purifier's column is the first unit vector; the quotient by it maps the column $(1,x)$ to $x$, so the contraction of $r$ by the purifier is the rank function of the configuration~\eqref{eq:JI}.
\end{proof}

\subsection{The one-cut proposition}
\label{app:cubes-onecut}

\begin{appproposition}[One cut]\label{prop:onecut}
Let $t=p$ be prime and $q>p$ prime, and measure the entropies of the generalized cube states over $\GF(p)$ and over $\GF(q)$ in units of $\log p$ and $\log q$, respectively. The two entropy vectors differ on exactly one bipartition, $\{A_1,\dots,A_M,C\}\mid\{B_1,\dots,B_M,O\}$, where the entropy is $p$ over $\GF(p)$ and $p+2$ over $\GF(q)$.
\end{appproposition}

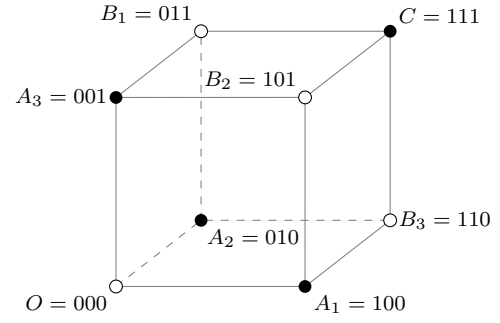
\begin{figure}[t]
\centering
\begin{tikzpicture}[scale=1.25,
  odd/.style={circle, fill=black, inner sep=1.7pt},
  even/.style={circle, draw=black, fill=white, inner sep=1.7pt},
  lab/.style={font=\footnotesize}]
\coordinate (O) at (0,0);
\coordinate (A) at (2,0);
\coordinate (B) at (0.9,0.7);
\coordinate (C) at (0,2);
\coordinate (W) at ($(A)+(B)$);
\coordinate (X) at ($(A)+(C)$);
\coordinate (Y) at ($(B)+(C)$);
\coordinate (Z) at ($(A)+(B)+(C)$);
\draw[gray] (O)--(A)--(X)--(C)--cycle;
\draw[gray] (A)--(W)--(Z)--(X);
\draw[gray] (C)--(Y)--(Z);
\draw[gray, dashed] (O)--(B)--(W);
\draw[gray, dashed] (B)--(Y);
\node[even] at (O) {}; \node[lab, below left] at (O) {$O=000$};
\node[odd] at (A) {}; \node[lab, below right] at (A) {$A_1=100$};
\node[odd] at (B) {}; \node[lab, below right, xshift=-1pt] at (B) {$A_2=010$};
\node[odd] at (C) {}; \node[lab, left] at (C) {$A_3=001$};
\node[even] at (W) {}; \node[lab, right] at (W) {$B_3=110$};
\node[even] at (X) {}; \node[lab, above left] at (X) {$B_2=101$};
\node[even] at (Y) {}; \node[lab, above left] at (Y) {$B_1=011$};
\node[odd] at (Z) {}; \node[lab, above right] at (Z) {$C=111$};
\end{tikzpicture}
\caption{The cube state ($t=2$): one qudit on each vertex of $\{0,1\}^3$, and the code of affine functions on the vertices. The origin $O$ is the purifier. The vertices of odd weight (filled) and of even weight (open) form the only cut on which the entropy, in units of $\log p$, depends on the prime local dimension $p$: it is $2$ for $p=2$ and $4$ for every odd prime $p$ (Proposition~\ref{prop:onecut} below).}
\label{fig:cube}
\end{figure}

Proposition~\ref{prop:onecut} follows from the cases $\ell=t$ and $\ell>t$ of the following statement. Let $t\ge2$, $M=t+1$, let $\Kf$ be a field of characteristic $\ell$ with $\ell\nmid t-1$ ($\ell=0$ allowed), and let $r_{\Kf}$ be the rank function of the vectors $(1,x)$, $x\in\{e_j,c-e_j,c,0\}$, of $\Kf^{M+1}$. Call a set $X$ of points a \emph{transversal} if it contains exactly one of $e_j$ and $c-e_j$ for every $j$, and put $a=\#\{j:c-e_j\in X\}$.

\begin{applemma}[Ranks over different fields]\label{lem:D-onecut}
$r_{\Kf}(X)=r_{\Q}(X)-1$ if $X$ is a transversal and either $c\in X$, $0\notin X$, $\ell\mid t-a$ and $a\ne t$, or $0\in X$, $c\notin X$, $\ell\mid a-1$ and $a\ne1$. For every other $X$, $r_{\Kf}(X)=r_{\Q}(X)$.
\end{applemma}

\begin{proof}
A dependency among the vectors of $X$ is a vector $\lambda\in\Kf^X$ with $\sum_x\lambda_x=0$ and $\sum_x\lambda_xx=0$. Write $u_j=\lambda_{c-e_j}$ (zero if $c-e_j\notin X$), $\beta=\sum_ju_j$, and $\kappa=-\beta-\lambda_c$ (with $\lambda_c=0$ if $c\notin X$). The $j$th coordinate of $\sum_x\lambda_xx=0$ reads $\lambda_{e_j}+\beta-u_j+\lambda_c=0$, \ie, $\lambda_{e_j}=u_j+\kappa$, and then $\sum_x\lambda_x=0$ gives $\lambda_0=(t-1)\beta+t\lambda_c$. So dependencies correspond to the choices of the $u_j$ and $\lambda_c$ subject to
\begin{itemize}
\item[(i)] $u_j+\kappa=0$ for every $j$ with $e_j\notin X$;
\item[(ii)] $(t-1)\beta+t\lambda_c=0$ if $0\notin X$.
\end{itemize}
The number of unknowns does not depend on $\Kf$, so the nullity $\abs{X}-r_{\Kf}(X)$ depends on $\Kf$ only through the rank of this system. Let $D=\{j:e_j,c-e_j\in X\}$, $D'=\{j:\text{only }c-e_j\in X\}$, $a=\abs{D'}$, and $s=\sum_{j\in D}u_j$.

If some $j$ has neither $e_j$ nor $c-e_j$ in $X$, then (i) gives $\kappa=0$ and $u_j=0$ on $D'$, hence $\beta=s=-\lambda_c$, and (ii) becomes $\lambda_c=0$. So the system is equivalent to $u_j=0$ ($j\in D'$), $s+\lambda_c=0$ and, if $0\notin X$, $\lambda_c=0$; these equations are triangular with pivots $\pm1$ in disjoint sets of unknowns, so the rank of the system is the same over every field.

Otherwise every $j$ lies in $D$, in $D'$, or has only $e_j$ in $X$. Then (i) says $u_j=-\kappa$ on $D'$, so $\beta=s-a\kappa$; we treat $\kappa$ as an unknown subject to its definition, which becomes $\mathrm{E}_1\colon(a-1)\kappa-s-\lambda_c=0$, and (ii) becomes $\mathrm{E}_2\colon-a(t-1)\kappa+(t-1)s+t\lambda_c=0$.
\begin{itemize}
\item $c,0\in X$: only $\mathrm{E}_1$; rank $1$ over every field.
\item $c\in X$, $0\notin X$: if $D\ne\emptyset$, then $\mathrm{E}_1$ and $\mathrm{E}_2$ are independent over every field, since proportionality would force the coefficients of $s$ and $\lambda_c$ to give $-\mu=t-1$ and $-\mu=t$. If $D=\emptyset$, $X$ is a transversal, the unknowns are $\kappa$ and $\lambda_c$, and the determinant is $(a-1)t-a(t-1)=a-t$: the rank drops exactly when $\ell\mid t-a$, and over $\Q$ exactly when $a=t$.
\item $0\in X$, $c\notin X$: $\lambda_c=0$ and $\mathrm{E}_1$ reads $(a-1)\kappa=s$. If $D\ne\emptyset$ the rank is $1$. If $D=\emptyset$, $X$ is a transversal and the rank drops exactly when $\ell\mid a-1$, and over $\Q$ exactly when $a=1$.
\item $c,0\notin X$: $\lambda_c=0$. If $D\ne\emptyset$, then $\mathrm{E}_2+(t-1)\mathrm{E}_1=-(t-1)\kappa$, so the rank is $2$ because $\ell\nmid t-1$. If $D=\emptyset$, the only unknown is $\kappa$, with coefficients $a-1$ and $-a(t-1)$, which never vanish together.
\end{itemize}
So $r_{\Kf}$ differs from $r_\Q$ exactly in the two cases stated, where the nullity is larger by one.
\end{proof}

When $\ell\nmid t-1$ the cube code is identically self-dual over $\Kf$ (Lemma~\ref{lem:cube-selfdual}; for $\ell=0$ the same weights work), so $S_{\Kf}=2r_{\Kf}-\abs{X}$ in units of $\log\abs{\Kf}$, and $S_{\Kf}$ differs from the formal vector $S_\Q=2r_\Q-\abs{X}$ exactly on the sets of Lemma~\ref{lem:D-onecut}, by $-2$. These sets come in complementary pairs. If $\ell>t$ there are none, so $S_{\Kf}=S_\Q$. If $\ell=t$, the only solutions are $a=0$ in the first case and $a=M$ in the second, \ie, the set $\{A_1,\dots,A_M,C\}$ and its complement. This proves Proposition~\ref{prop:onecut}, with $S_\Q=t+2$ on that cut.

\subsection{Galois qudits}
\label{app:cubes-galois}

\begin{appproposition}[Galois qudits]\label{prop:galois}
Up to phases, the Pauli group of a Galois qudit of dimension $p^k$, defined with the field trace~\cite{si:Ketkar2006}, is that of $k$ qudits of dimension $p$, and CSS states over $\GF(p)$ remain CSS states over $\GF(p^k)$ with the same ranks. Consequently, the entropy vectors of stabilizer states of Galois qudits of dimension $p^k$ generate exactly the cone $\StabZ{p}{n}$.
\end{appproposition}

\begin{proof}
The generalized Pauli operators of a Galois qudit of dimension $q=p^k$ are $X(a)Z(b)$, $a,b\in\GF(q)$, with $X(a)\ket{x}=\ket{x+a}$ and $Z(b)\ket{x}=\omega^{\mathrm{tr}(bx)}\ket{x}$, $\omega=e^{2\pi i/p}$~\cite{si:Ketkar2006}. Expanding $a$ in a basis of $\GF(q)$ over $\GF(p)$ and $b$ in the trace-dual basis identifies these operators, up to phases, with the Weyl operators of $k$ qudits of dimension $p$, and the trace-symplectic form with the standard one. So every stabilizer state of Galois qudits (with additive or with $\GF(q)$-linear stabilizer group) is a stabilizer state of $\Z_p$ qudits, and its entropy vector lies in $\StabZ{p}{n}$. Conversely, if $C$ is a code over $\GF(p)$, the CSS state of $C\otimes\GF(q)$ on Galois qudits has entropy vector $k$ times that of the CSS state of $C$, because extending scalars does not change ranks. By Corollary~\ref{cor:doubling} these vectors generate $\StabZ{p}{n}$.
\end{proof}

\subsection{The \texorpdfstring{$\Z_4$}{Z4} cube state}
\label{app:cubes-z4}

The cube state over $\Z_4$ violates the pulled-back \LA{} and \LB{} with $t=2$ (each takes the value $-2$ bits), so it lies outside $\StabZ{p}{7}$ for every prime $p$. Nevertheless it lies in $\Stab^{\mathrm{pr}}_7$, the closed convex hull of the union of the prime cones: all minors of its integer generator matrix $G$ lie in $\{0,\pm1,\pm2\}$ (its rows are $(1,\dots,1)$ and the three coordinate functions on $\{0,1\}^3$; a minor of order at most $3$ is the determinant of a $0/1$ matrix of order at most $3$, whose absolute value is at most $2$, and a minor of order $4$ is $\pm6$ times the volume of a tetrahedron with vertices in $\{0,1\}^3$, which is $0$, $1/6$ or $1/3$), so $\log_2\abs{\pr_XC}=\rank_{\GF(2)}G_X+\rank_{\GF(3)}G_X$ for every set $X$ of columns, and its entropy vector in bits is the sum of the qubit cube vector (in bits) and the qutrit cube vector (in units of $\log 3$). Likewise, the $\Z_4$ seven-point state of Theorem~\ref{thm:reverse} has entropy vector $\SF+\SNF$ in bits, so it lies in $\Stab^{\mathrm{pr}}_6$ and, by that theorem, outside the qubit cone.

\subsection{Proof of Theorem~\ref{thm:divisibility}}
\label{app:cubes-divisibility}

The first statement is Corollary~\ref{cor:divisibility}. If $d\nmid d'$, choose a prime $p$ with $v_p(d)>v_p(d')=:j$, where $v_p(d)$ is the exponent of $p$ in $d$. We give a witness in $\StabZ{d}{n_0}$ outside $\StabZ{d'}{n_0}$ for some $n_0\le2\max(d,d')+2$; for larger $n$, add parties that hold nothing (Sec.~\ref{sec:reduction}).
\begin{itemize}
\item $j=0$, $p$ odd. The doubled-code witness of Theorem~\ref{thm:primes} over $\GF(p)$, with $n_0=2p+2$ parties, lies in $\StabZ{p}{n_0}\subseteq\StabZ{d}{n_0}$ and violates the balanced \LB$^*$ with $t=p$. This inequality is valid on $\StabZ{d'}{n_0}$, because multiplication by $p$ is injective on every group of exponent dividing $d'$ (Lemma~\ref{lem:JIstar} and Proposition~\ref{thm:reduction}, as in the proof of Proposition~\ref{prop:sandwich}).
\item $j=0$, $p=2$. Then $d$ is even and $d'$ is odd, so $d'\ge3$. The generalized cube state over $\GF(2)$ with $t=2$, with $n_0=7\le2d'+2$ parties, lies in $\StabZ{2}{7}\subseteq\StabZ{d}{7}$ and violates the pulled-back \LB{} with $t=2$, which is valid on $\StabZ{d'}{7}$ because $2$ is a unit modulo $d'$ (proof of Theorem~\ref{thm:primes}).
\item $j\ge1$, $d'\ge3$. Let $t=d'$ and $N=p^{j+1}$, which divides $d$. The doubled code of the configuration~\eqref{eq:JI} over $\Z_N$ gives a state with $n_0=2d'+2$ parties in $\StabZ{N}{n_0}\subseteq\StabZ{d}{n_0}$. The configuration has no coloops over $\Z_N$ either, and $M=d'+1\ge4$ and $\gcd(t,N)=p^j$, so the balanced \LA$^*$ with $t=d'$ takes the value $-\log p$ on this state. The inequality is valid on $\StabZ{d'}{n_0}$, because $d'V=0$ for every group $V$ of exponent dividing $d'$. For example, $d=9$ and $d'=3$ give $n_0=8$, and $d=12$ and $d'=6$ give $n_0=14$ (Fig.~\ref{fig:landscape}b).
\item $j\ge1$, $d'=2$. Then $p=2$ and $4\mid d$. The generalized cube state over $\Z_4$ with $t=2$, with $n_0=7$ parties, lies in $\StabZ{4}{7}\subseteq\StabZ{d}{7}$, and it is self-dual by Lemma~\ref{lem:cube-selfdual}, since $t-1=1$ is a unit. The subgroup $B_1+B_2+B_3$ of $\Z_4^3$ has index $2$ (the Smith form of $J-I$ is $\mathrm{diag}(1,\dots,1,t)$), so the pulled-back \LA{} with $t=2$, which is valid on $\StabZ{2}{7}$, takes the value $-2\log2$ on this state.
\end{itemize}
\section{Two incorrect printed inequalities}
\label{app:misprints}

We record two printed inequalities that are false as stated, which we came across while checking published characteristic-dependent inequalities. The failure of the first one was already noted by Pe\~na-Macias and Sarria~\cite{si:PenaSarria2021}; it is relevant to Table~\ref{tab:eight}. The second one is not used anywhere in this paper.

\emph{Blasiak, Kleinberg and Lubetzky, Eq.~(6.17).} In the labeling of Table~\ref{tab:eight}, Eq.~(6.17) of Ref.~\cite{si:BKL2011} (Theorem~6.6 of arXiv version~1, the only arXiv version; claimed for every characteristic other than $2$) reads
\[
\begin{aligned}
0\le{}&3h(A)+3h(B)+9h(C)+6h(Z)+6h(AB)\\
&-12h(ABC)+3h(ABW)+3h(ACX)+3h(BCY)\\
&-3h(ABZ)-3h(ACZ)-3h(BCZ)\\
&+3h(ABCZ)+h(WXY).
\end{aligned}
\]
As printed, it fails over every field: for $A,B,C=\langle e_1\rangle,\langle e_2\rangle,\langle e_3\rangle$, $Z=\langle e_1+e_2+e_3\rangle$ and $W=X=Y=0$, the right-hand side equals $-3$. Pe\~na-Macias and Sarria gave the same example, attributed the error to an upper bound in the proof, and proved a different inequality for this configuration (their Theorem~9 with $n=2$)~\cite{si:PenaSarria2021}. The proof of Theorem~6.6 bounds the terms $\sum_{i+j=111}\delta[111\mid i,j]$ (in the notation of Ref.~\cite{si:BKL2011}) in the text, but they are missing from Eq.~(6.16). Restoring them adds
\begin{align*}
3\bigl[&h(AYZ)-h(AY)+h(BXZ)\\
&-h(BX)+h(CWZ)-h(CW)\bigr]
\end{align*}
to the right-hand side. The corrected form takes the value $6$ on the example above and $-1$ on the Fano configuration, and it equals \LB{} with $t=2$ (Eq.~\eqref{eq:LB2}) minus $6\,I(A;B)$. Its validity rests on the proof of Ref.~\cite{si:BKL2011} with these terms restored; Proposition~\ref{prop:blind} does not depend on it. We found no problem with Eq.~(6.8) of the same reference, which coincides with Theorem~7 of Ref.~\cite{si:PenaSarria2021} for their parameter $n=2$ (Supplementary Note~\ref{app:blind}). The conference version of Ref.~\cite{si:BKL2011} asserts the existence of both inequalities (its Theorem~6.2) but does not print them.

\emph{Dougherty, Freiling and Zeger, Theorem~3.1 of Ref.~\cite{si:DFZ2015char}.} Theorem~3.1, identical in the published version and in arXiv version~1, states an eight-variable inequality on $A,B,C,D,W,X,Y,Z$ for every characteristic other than $3$:
\[
\begin{aligned}
H(A)\le{}&8H(Z)+29H(Y)+3H(X)+8H(W)\\
&-6H(D)-17H(C)-8H(B)-17H(A)\\
&+55H(Z|A,B,C)+35H(Y|W,X,Z)\\
&+50H(X|A,C,D)+49H(W|B,C,D)\\
&+18H(A|B,D,Y)+7H(B|D,X,Z)\\
&+H(B|A,W,X)+7H(C|D,Y,Z)\\
&+7H(C|B,X,Y)+3H(C|A,W,Y)\\
&+6H(D|A,W,Z)\\
&+49\bigl(H(A)+H(B)+H(C)+H(D)\\
&\qquad-H(A,B,C,D)\bigr).
\end{aligned}
\]
As printed, it is false over every field: let $C=X$ be the same one-dimensional subspace and all other subspaces zero. The left-hand side is $0$, and the right-hand side is $3-17+7+3=-4$. A coefficient error in the final expansion of the proof seems likely. The same paper states an eight-variable inequality that holds in characteristic $3$ and fails in every other characteristic~\cite[Thms.~4.1 and~4.2]{si:DFZ2015char}; through the doubled codes of Proposition~\ref{prop:sandwich} it would give $\StabZ{q}{8}\not\subseteq\StabZ{3}{8}$ for every prime $q\ne3$. We have not proved it ourselves and do not use it; Lemma~\ref{lem:JIstar} gives eight-variable inequalities of this kind for both characteristics, with a proof (Supplementary Note~\ref{app:lemma2}).
\section{Secret sharing: access structures, contractions and bounds}
\label{app:qss}

This note proves Lemma~\ref{lem:fano-odd} and gives the details behind Proposition~\ref{prop:qss} and Theorem~\ref{thm:ideal}. In a classical secret sharing scheme the secret and the shares are random variables; a set of players is authorized if its shares determine the secret and unauthorized if they are independent of it, and the scheme is linear over a field $\Kf$ if the secret and the shares are linear functions of a uniformly random vector over $\Kf$; the secret may consist of several field elements~\cite{si:JafariKhazaei2021}. Every scheme is obtained from a pure-state one by discarding a share, and every pure-state scheme requires that a set of players is authorized exactly when its complement is not~\cite{si:Gottesman2000}, as for $\Gamma_{\mathrm F}$ and $\Gamma_{\mathrm{NF}}$. Through the lift, stabilizer schemes of prime local dimension and classical linear schemes correspond; the converse direction is a version of a construction of Smith~\cite{si:Smith2000}.

\begin{applemma}[Stabilizer schemes are linear schemes]\label{lem:qss-linear}
Let $p$ be prime. Every stabilizer scheme of local dimension $p$, pure or mixed, can be brought to the normal form of Lemma~\ref{lem:A-normal} without changing its access structure or increasing its information ratio, which is then $\max_PS(P)/S(R)$. In normal form, $h_S$ is the entropy function of a $\GF(p)$-linear classical secret sharing scheme for the same access structure, in which the secret has entropy $2S(R)$ and the share of $P$ has entropy $2S(P)$.
\end{applemma}

\begin{applemma}[Linear schemes give stabilizer schemes]\label{lem:qss-css}
Let $\Gamma$ be an access structure in which a set of players is authorized if and only if its complement among the players is not, and let $C\le\Z_d^m$ be a code whose coordinates are distributed among the players and $R$, which holds $m_R$ of them. If a uniformly random codeword is a classical scheme for $\Gamma$ whose secret, its part at $R$, is uniform on $\Z_d^{m_R}$, then the CSS state $\ket{C}$, with $R$ as the reference system, is a pure stabilizer scheme for $\Gamma$ in which the secret consists of $m_R$ qudits and each player holds one qudit for each of its coordinates.
\end{applemma}

Over $\GF(p)$ we may drop every coordinate that is a linear function of the other coordinates of the same system, so for such an access structure and a prime $p$ the optimal information ratio of stabilizer schemes of local dimension $p$ is that of $\GF(p)$-linear schemes.

\emph{Proof of Lemma~\ref{lem:qss-linear}.} The normal form applies local Clifford unitaries and removes qudits that are in a fixed pure state, so it changes no entropy and the access structure stays the same, and the information ratio does not increase; afterwards every system $P$ holds $S(P)/\log p$ qudits, so the information ratio is $\max_PS(P)/S(R)$. In normal form, $h_S(T)$ is $\log p$ times the dimension of the projection of the Lagrangian subgroup $L$ to the coordinates of the qudits held by a set $T$ of systems (Supplementary Note~\ref{app:reduction-normal}). So a uniformly random element of $L$, with its coordinates distributed among the systems, is a $\GF(p)$-linear scheme whose entropy function is $h_S$, and by Eq.~\eqref{eq:qss-lift} the quantum and the classical conditions for authorized and for unauthorized sets coincide. Finally, $h_S(P)=2S(P)$ and $h_S(R)=2S(R)$.\qed

\emph{Proof of Lemma~\ref{lem:qss-css}.} Let $r(X)=\log\abs{\pr_XC}$, let $\mathcal Q$ be the set of players and $N=\mathcal Q\cup\{R\}$. By fact~(b) of the Methods (Supplementary Lemma~\ref{lem:A-css}), $S(X)=r(X)+r(N\setminus X)-r(N)$. The set $\mathcal Q$ is authorized, so $r(N)=r(\mathcal Q)$ and $S(R)=r(R)=m_R\log d$. For a set $T$ of players and $U=\mathcal Q\setminus T$, this gives $I(R;T)=\bigl[r(R)+r(T)-r(RT)\bigr]+\bigl[r(RU)-r(U)\bigr]$. If $T$ is authorized, then $U$ is not, and both brackets equal $r(R)$; if $T$ is unauthorized, then $U$ is authorized, and both vanish.\qed

In the rest of this note, the players are the seven nonzero vertices of the cube, labeled as in Table~\ref{tab:eight} ($A=100$, $B=010$, $C=001$, $W=110$, $X=101$, $Y=011$, $Z=111$); the origin is the reference system $R$; and, for mixed-state schemes, one more system $E$ purifies the joint state of $R$ and the shares.

\subsection{The access structures}
\label{app:qss-access}

By Eq.~\eqref{eq:qss-lift}, the access structure of a scheme can be read off from the lift. For the cube states of Definition~\ref{def:cube} with $t=2$, $h_S=2r$, where $r$ is the rank function of the code of affine functions on the cube (Lemma~\ref{lem:cube-selfdual}), so a set $T$ of players is authorized if and only if the origin column of the generator matrix, $(1,0,0,0)^{\mathsf T}$, lies in the span of the columns of $T$, and unauthorized otherwise. Every set of players is either authorized or unauthorized, and $T$ is authorized if and only if its complement among the players is unauthorized; $64$ sets of players, counting the empty set, are unauthorized and $64$ are authorized.
\begin{itemize}
\item Over $\GF(2)$, the minimal authorized sets of $\Gamma_{\mathrm F}$ are the seven lines $ABW$, $ACX$, $AYZ$, $BCY$, $BXZ$, $CWZ$ and $WXY$ of the Fano plane. Indeed, $(1,0,0,0)^{\mathsf T}$ is a sum of columns $(1,x)^{\mathsf T}$ with $x\in Q$ if and only if $\abs{Q}$ is odd and the points of $Q$ sum to $0$; no single point and no five points have sum $0$, the seven points do, and three points have sum $0$ exactly when they form a line. The matroid of the code is the binary affine cube $\mathrm{AG}(3,2)$, and $\Gamma_{\mathrm F}$ is its port at $R$: it consists of the sets $T$ of players such that $T\cup\{R\}$ contains a circuit through $R$.
\item Over $\GF(p)$ with $p$ odd, the minimal authorized sets of $\Gamma_{\mathrm{NF}}$ are the six lines other than $WXY$ and the sets $ABCZ$, $AWXY$, $BWXY$, $CWXY$ and $WXYZ$. The access structures differ only in $WXY$, which is authorized for $\Gamma_{\mathrm F}$ only, and $ABCZ$, which is authorized for $\Gamma_{\mathrm{NF}}$ only. All minors of the integer generator matrix lie in $\{0,\pm1,\pm2\}$ (Supplementary Note~\ref{app:cubes-z4}), so its rank over $\GF(p)$ on every set of columns is the same for every odd $p$, and $\Gamma_{\mathrm{NF}}$ is the same for every odd $p$. The matroid is the real affine cube $R_8$, obtained from $\mathrm{AG}(3,2)$ by relaxing the two complementary circuit-hyperplanes $\{R,W,X,Y\}$ and $\{A,B,C,Z\}$, which are affine planes over $\GF(2)$ but not over $\Q$ (relaxing a circuit-hyperplane turns it into a basis); $\Gamma_{\mathrm{NF}}$ is its port at $R$.
\end{itemize}

The same construction gives such a pair of access structures for every pair of primes $p<q$: with $t=p$, the generalized cube codes of Definition~\ref{def:cube} over $\GF(p)$ and over $\GF(q)$ are ideal schemes for their ports at the purifier $O$, and by Proposition~\ref{prop:onecut} these ports differ only in the status of $\{B_1,\dots,B_M\}$ and of its complement, as $\Gamma_{\mathrm F}$ and $\Gamma_{\mathrm{NF}}$ differ in that of $WXY$ and $ABCZ$. Contracting such a port at the player $C$ gives, after exchanging the labels $A_i$ and $B_i$, the port at $c$ of the configuration~\eqref{eq:JI} itself, as the columns $(1,v)$ of the cube code show; for these ports Pe\~na-Macias proved lower bounds on the information ratio of linear schemes~\cite[Cor.~2.9]{si:PenaMacias2023}, and since contracting does not increase the optimal information ratio, they hold for the ports at $O$ as well: over finite fields whose characteristic is not $p$, linear schemes for the port of the code over $\GF(p)$ have information ratio at least $(p+2)/(p+1)$, and over finite fields of characteristic $p$, linear schemes for the port of the code over $\GF(q)$ have information ratio at least $(p+3)/(p+2)$. For $p=2$ these are $4/3$ and $5/4$; for the latter, the bound of Jafari and Khazaei behind Proposition~\ref{prop:qss}~(ii) gives $4/3$~\cite[Prop.~4.5]{si:JafariKhazaei2021}.

\subsection{Contractions, a rank inequality for groups of odd order, and the bounds \texorpdfstring{$4/3$}{4/3}}
\label{app:qss-contract}

For a player $q$, the contraction of an access structure $\Gamma$ is $\Gamma/q=\{T:T\cup\{q\}\in\Gamma\}$, on the other players. Let $h$ be the entropy function of a linear scheme for $\Gamma$ in which $\{q\}$ is unauthorized. Conditioning on the share of $q$, that is, restricting to the codewords whose coordinates held by $q$ vanish, gives a linear scheme over the same field with entropy function $T\mapsto h(Tq)-h(q)$. It is a scheme for $\Gamma/q$: its conditions for a set $T$ are those of the original scheme for $T\cup\{q\}$, because $h(Rq)=h(R)+h(q)$. Its secret has the same entropy $h(R)$, and the share of $P$ has entropy $h(Pq)-h(q)\le h(P)$.

Jafari and Khazaei consider the access structures $\mathcal F=14+25+36+123+156+246+345$ on the players $1,\dots,6$ (minimal qualified sets), the port of the Fano matroid, and $\mathcal N=\mathcal F+456$, a port of the non-Fano matroid~\cite{si:JafariKhazaei2021}. In every linear scheme for $\mathcal F$ over a field of odd characteristic, the shares of the players of every minimal authorized (qualified) set of size $3$ have entropies that sum to at least $4$ times that of the secret~\cite[Prop.~4.4~(II)]{si:JafariKhazaei2021}; in every linear scheme for $\mathcal N$ over a field of characteristic $2$, the same holds for the set $456$~\cite[Prop.~4.5~(III)]{si:JafariKhazaei2021}. A computation over all bijections shows:
\begin{itemize}
\item for every player $q$, $\Gamma_{\mathrm F}/q$ is isomorphic to $\mathcal F$, and its minimal authorized sets of size $3$ are the four lines of the Fano plane that miss $q$;
\item for $q\in\{A,B,C,Z\}$, $\Gamma_{\mathrm{NF}}/q$ is isomorphic to $\mathcal N$, with the set $456$ corresponding to the other three of $A$, $B$, $C$ and $Z$; for $q\in\{W,X,Y\}$ it is isomorphic to neither.
\end{itemize}

\emph{Proof of Lemma~\ref{lem:fano-odd}.} Write $T_X=\sum_{i\in X}T_i$, so that $I(0;X)=\log\abs{T_0\cap T_X}$. We prove the stronger inequality
\begin{equation}\label{eq:K-fano-strong}
\begin{aligned}
  I(0;123)\le{}&h(1|04)+h(2|05)+h(3|06)\\
  &+I(0;234)+I(0;135)\\
  &+I(0;126)+I(0;456);
\end{aligned}
\end{equation}
since $h(i|0,i{+}3)=h(i)-I(0;i,i{+}3)+I(0;i{+}3)-I(i;i{+}3)$, Eq.~\eqref{eq:fano-odd} is Eq.~\eqref{eq:K-fano-strong} plus $I(1;4)+I(2;5)+I(3;6)\ge0$.

Let $D$ be the group of triples $a=(a_1,a_2,a_3)\in T_1\times T_2\times T_3$ with $\sigma(a)=a_1+a_2+a_3\in T_0$. Then $\sigma(D)=T_0\cap T_{123}$, so $\abs{D}=\abs{T_0\cap T_{123}}\,\abs{\ker\sigma}$. For $i=1,2,3$ put $Y_i=T_i\cap(T_0+T_{i+3})$, and let $D_0=\{a\in D:a_i\in Y_i\text{ for }i=1,2,3\}$, the kernel of the homomorphism $D\to\bigoplus_iT_i/Y_i$, $a\mapsto(a_i+Y_i)_i$; hence $\abs{D}/\abs{D_0}\le\prod_i[T_i:Y_i]$, and by the second isomorphism theorem $T_i/Y_i\cong(T_0+T_i+T_{i+3})/(T_0+T_{i+3})$, so $\log[T_i:Y_i]=h(i|0,i{+}3)$.

Let $a\in D_0$ and $s=\sigma(a)$. Write $a_i=s_i+b_{i+3}$ with $s_i\in T_0$ and $b_{i+3}\in T_{i+3}$. Then $k_1=s-s_1=a_2+a_3+b_4$ lies in $T_0\cap T_{234}$, and likewise $k_2=s-s_2=a_1+a_3+b_5\in T_0\cap T_{135}$ and $k_3=s-s_3=a_1+a_2+b_6\in T_0\cap T_{126}$. Adding, $b_4+b_5+b_6=k_1+k_2+k_3-2s$, which lies in $T_0$ and in $T_{456}$. So $2s$ lies in the subgroup $K=(T_0\cap T_{234})+(T_0\cap T_{135})+(T_0\cap T_{126})+(T_0\cap T_{456})$ of $T_0$. Multiplication by $2$ is an automorphism of $T_0$, whose order is odd, so the elements $s\in T_0$ with $2s\in K$ form a subgroup of order $\abs{K}$. Hence $\abs{\sigma(D_0)}\le\abs{K}$, and $\abs{D_0}\le\abs{\ker\sigma}\,\abs{K}$. Combining,
\[
  \abs{T_0\cap T_{123}}=\frac{\abs{D}}{\abs{\ker\sigma}}\le\frac{\abs{K}\,\abs{D}}{\abs{D_0}}\le\abs{K}\prod_{i=1}^3[T_i:Y_i],
\]
and $\abs{K}$ is at most the product of the orders of its four summands. Taking logarithms gives Eq.~\eqref{eq:K-fano-strong}. For subspaces over a field of characteristic other than $2$ the same argument works with dimensions, since multiplication by $2$ is invertible. For the Fano configuration over $\GF(2)$ with lines $014$, $025$, $036$, $234$, $135$, $126$ and $456$, every point has rank $1$; the point $0$ lies in the span of $123$, which is the whole plane, and on the lines $014$, $025$ and $036$, but it is none of the points $4$, $5$, $6$ and lies on none of the lines $234$, $135$, $126$ and $456$. So the left side of Eq.~\eqref{eq:fano-odd} is $3$ and the right side is $1+3-0-0=4$; in Eq.~\eqref{eq:K-fano-strong} the left side is $1$ and the right side $0$.\qed

\emph{Proof of the bound in Proposition~\ref{prop:qss}~(i).} Let a stabilizer scheme for $\Gamma_{\mathrm F}$ of odd local dimension be given, pure or mixed. For Galois qudits of dimension $p^k$, read every Galois qudit as $k$ qudits of dimension $p$ (Supplementary Proposition~\ref{prop:galois}); the state, its entropies and the dimensions of the secret and the shares stay the same. So we may assume $\Z_d$ qudits with $d$ odd. By Lemma~\ref{lem:A-lift}, the lift $h=h_S$ of the joint state is the rank function of subgroups $U_P$, one for each system $P$, of a finite abelian group of exponent dividing $d$, hence of odd order; we write $U_X$ for the sum of the $U_P$ with $P\in X$ and use only the subgroups of $R$ and of the players. Let $\ell=\{x_1,x_2,x_3\}$ be a line, $q$ a player not on $\ell$, and $x_{i+3}$ the third point of the line through $q$ and $x_i$; then $x_4$, $x_5$, $x_6$ are the three points off $\ell\cup\{q\}$. The contraction of $h$ by $q$ is the rank function of the images $T_0$ of $U_R$ and $T_i$ of $U_{x_i}$ in the quotient by $U_q$, a group of odd order, and Lemma~\ref{lem:fano-odd} applies to it. For a set $X\subseteq\{1,\dots,6\}$ let $X'$ be the corresponding set of players. The term $I(0;X)$ of the contraction is $h(Rq)+h(X'q)-h(RX'q)-h(q)=I(R;X'q)-I(R;q)$, by Eq.~\eqref{eq:qss-lift}. Here $q$ is unauthorized, the sets $\ell$ and $\{x_i,x_{i+3},q\}$ are lines, the sets $\{x_{i+3},q\}$ contain no line, and $\{x_2,x_3,x_4,q\}$, $\{x_1,x_3,x_5,q\}$, $\{x_1,x_2,x_6,q\}$ and $\{x_4,x_5,x_6,q\}$ are the complements of the lines $\{x_1,x_5,x_6\}$, $\{x_2,x_4,x_6\}$, $\{x_3,x_4,x_5\}$ and $\ell$, hence unauthorized (we checked these incidences for all $28$ pairs $(\ell,q)$, \path{check_rev_qss.py}, part (m)). So $I(0;X)=2S(R)$ for $X=123$, $14$, $25$, $36$ and $I(0;X)=0$ for the other seven sets in Eq.~\eqref{eq:fano-odd}, whose right side is therefore $8S(R)$. Its left side is $\sum_i\bigl(h(x_iq)-h(q)\bigr)\le\sum_ih(x_i)=2\sum_iS(x_i)$. Hence $\sum_{P\in\ell}S(P)\ge4\,S(R)$. Some player $P$ of $\ell$ has $S(P)\ge\frac43S(R)$; since $S(R)$ is the logarithm of the dimension of the secret and $S(P)$ is at most the logarithm of the dimension of the share, the information ratio is at least $4/3$. Every player lies on three of the seven lines, so summing over the lines gives $3\sum_PS(P)\ge28\,S(R)$.\qed

\emph{Decomposition into prime powers.} For every $d$, the Lagrangian subgroup $L\le\Z_d^{2k}$ of a stabilizer state of $k$ qudits is the direct sum of its Sylow components $L_p\le\Z_{p^a}^{2k}$, one for each prime power $p^a$ that divides $d$ exactly, and $\abs{\pr_TL}=\prod_p\abs{\pr_TL_p}$ for every set $T$ of systems; since $L_p$ is isotropic and $\prod_p\abs{L_p}=d^k$, each $L_p$ is Lagrangian. By Eq.~\eqref{eq:A-entropy}, $S=\sum_pS_p$, where $S_p$ is the entropy vector of the stabilizer state of $\Z_{p^a}$ qudits with Lagrangian subgroup $L_p$, in which every system holds as many qudits as in the given state; by the Chinese remainder theorem, which identifies a $\Z_d$ qudit with one $\Z_{p^a}$ qudit for each $p$, the given state is the tensor product of these states. Conversely, the symplectic form of $\Z_d^{2k}$ is the sum of those of the Sylow components (identified with $\Z_{p^a}^{2k}$ up to a unit multiple of the form), so for Lagrangian subgroups $L_p\le\Z_{p^a}^{2k}$, one for each $p$, the sum $\bigoplus_pL_p$ is Lagrangian, and the tensor product of stabilizer states of $\Z_{p^a}$ qudits with the same numbers of qudits per system is a stabilizer state of $\Z_d$ qudits. The same holds, with the same proof, for any factorization $d=d_1d_2$ into coprime factors. For every set $T$ of players, $0\le I_p(R;T)\le2S_p(R)$ and $I(R;T)=\sum_pI_p(R;T)$; so each factor of a scheme meets the access conditions by itself (trivially if $S_p(R)=0$), and a tensor product of schemes for the same access structure is a scheme for it. If the secret consists of $m$ qudits, that is, $S(R)=m\log d=\sum_pma\log p$, then $S_p(R)=ma\log p$ in every factor, since each $S_p(R)$ is an integer multiple of $\log p$, while a player that holds $n_P$ qudits has $S_p(P)\le n_Pa\log p$. Without this assumption the factors can carry different shares of the secret, and for even $d$ the bounds in numbers of qudits below fail (Supplementary Note~\ref{app:qss-upper}).

\emph{Local dimensions that are not powers of $2$.} Let $p$ be an odd prime that divides $d$, exactly as $p^a$, and consider a stabilizer scheme for $\Gamma_{\mathrm F}$ of $\Z_d$ qudits, pure or mixed, in which the secret consists of $m$ qudits and player $P$ holds $n_P$ qudits. Its factor of $\Z_{p^a}$ qudits is a stabilizer scheme for $\Gamma_{\mathrm F}$ of odd local dimension, so the bound just proved gives $\sum_{P\in\ell}n_Pa\log p\ge\sum_{P\in\ell}S_p(P)\ge4\,S_p(R)=4ma\log p$, that is, $\sum_{P\in\ell}n_P\ge4m$, for every line $\ell$. For odd local dimensions this also follows from the bound for entropies, since $S(P)$ is at most $n_P$ times, and $S(R)$ equals $m$ times, the logarithm of the dimension of a qudit. Hence some player of every line holds at least $\frac43m$ qudits, and $3\sum_Pn_P\ge28m$.

\emph{Approximate schemes.} For a stabilizer state and a set $T$ of players, $h(R)=2S(R)$ and, by Eq.~\eqref{eq:qss-lift}, $I(R;T)=h(R)+h(T)-h(RT)=\log\abs{U_R}+\log\abs{U_T}-\log\abs{U_R+U_T}=\log\abs{U_R\cap U_T}$; so $2S(R)-I(R;T)=\log[U_R:U_R\cap U_T]$. Both are logarithms of orders of finite groups of exponent dividing the local dimension $d$ (read Galois qudits of dimension $p^k$ as qudits of dimension $p$), hence $0$ or at least $\log_2p\ge1$ bits, where $p$ is the smallest prime factor of $d$. Now let the state be a stabilizer state of odd local dimension in which, for some $\varepsilon\ge0$, every set $T\in\Gamma_{\mathrm F}$ has $I(R;T)\ge2S(R)-\varepsilon$ and every other nonempty set of players has $I(R;T)\le\varepsilon$, in bits. In the proof above, each term of the right side of Eq.~\eqref{eq:fano-odd} for the contraction is $I(R;X'q)-I(R;q)$, with a positive sign for the four sets $X=123$, $14$, $25$, $36$ and a negative sign for the seven others; so the right side equals the sum of $I(R;X'q)$ over the first four, minus the sum over the seven others, plus $3\,I(R;q)$. The first four terms are at least $2S(R)-\varepsilon$, the seven others at most $\varepsilon$, and $I(R;q)\ge0$; so $2\sum_{P\in\ell}S(P)\ge8S(R)-4\varepsilon-7\varepsilon$, that is, $\sum_{P\in\ell}S(P)\ge4\,S(R)-\frac{11}{2}\varepsilon$. Since every nonzero deficit or leakage is at least $\log_23$ bits for odd $d$, smaller errors force an exact scheme, and the bound matters only for $\varepsilon\ge\log_23$. For $\Z_d$ qudits with $d$ not a power of $2$, let $p^a$ be the power of an odd prime that divides $d$ exactly: since $I(R;T)=\sum_{p'}I_{p'}(R;T)$ and $2S(R)-I(R;T)=\sum_{p'}\bigl(2S_{p'}(R)-I_{p'}(R;T)\bigr)$ with nonnegative terms, the factor of $\Z_{p^a}$ qudits meets the same conditions with the same $\varepsilon$, and with a secret of $m$ qudits and shares of $n_P$ qudits it gives $\sum_{P\in\ell}n_P\ge4m-\frac{11\varepsilon}{2a\log_2p}$.

\emph{Abelian schemes.} An abelian scheme in the sense of Ref.~\cite{si:JafariKhazaei2021} is group-characterizable with an abelian main group: it is given by subgroups $G_0,\dots,G_n$ of a finite abelian group $G$, and the entropy of the variables indexed by $A$ is $\log[G:\bigcap_{i\in A}G_i]$. With the annihilators $U_i\le\widehat G$ of the $G_i$ in the dual group, $\bigl(\bigcap_{i\in A}G_i\bigr)^\perp=\sum_{i\in A}U_i$ has order $[G:\bigcap_{i\in A}G_i]$, so the entropy function is the rank function of the arrangement $(U_i)$ in $\widehat G\cong G$ (cf.\ Ref.~\cite[Prop.~2.3]{si:JafariKhazaei2021}). For the port $\mathcal F$ of the Fano matroid and a group of odd order, Lemma~\ref{lem:fano-odd} with the labels of $\mathcal F$, where $123$, $14$, $25$ and $36$ are authorized and $4$, $5$, $6$, $234$, $135$, $126$ and $456$ are not, gives Proposition~4.4~(II) of Ref.~\cite{si:JafariKhazaei2021} for such schemes; by the symmetry of $\mathcal F$ the same holds for its other three triples, and so the information ratio is at least $4/3$.

\emph{The twelve-player access structure $\mathcal F+\mathcal N$.} Jafari and Khazaei place $\mathcal F$ and $\mathcal N$ on the disjoint sets of players $\{1,\dots,6\}$ and $\{7,\dots,12\}$ and call their union $\mathcal F+\mathcal N$: a set of players is authorized if its part in the first set is authorized for $\mathcal F$ or its part in the second set is authorized for $\mathcal N$~\cite[Sec.~6]{si:JafariKhazaei2021}. Consider an abelian scheme for $\mathcal F+\mathcal N$ whose group has an elementary abelian Sylow $2$-subgroup; mixed-linear schemes, direct sums of linear schemes over fields~\cite[Sec.~6.2]{si:JafariKhazaei2021}, are of this kind. Splitting the group into its Sylow $2$-subgroup and its subgroup of odd order splits every entropy into two terms, each the entropy of an abelian scheme for $\mathcal F+\mathcal N$ (the access conditions are sums of nonnegative terms), with secret entropies $\sigma_2+\sigma_o=H(S)$, where $\sigma_2$ is the logarithm of a power of $2$ and $\sigma_o$ that of an odd number. Restricted to the players of $\mathcal F$, the odd part is an abelian scheme of odd order for $\mathcal F$, so by Lemma~\ref{lem:fano-odd} the shares of each minimal triple of $\mathcal F$ have entropies of at least $4\sigma_o$ in total in it, and at least $3\sigma_2$ in the $2$-part, since in every scheme the share of a player that lies in a minimal authorized set has at least the entropy of the secret. Restricted to the players of $\mathcal N$, the $2$-part is a $\GF(2)$-linear scheme for $\mathcal N$, so by Prop.~4.5~(III) of Ref.~\cite{si:JafariKhazaei2021} the triple that corresponds to $456$ has share entropies of at least $4\sigma_2$ in total in it, and at least $3\sigma_o$ in the odd part. The two bounds sum to $7H(S)$, so one of these two triples has shares with entropies of at least $\frac72H(S)$ in total, and the information ratio is at least $7/6$; equality would need $4\sigma_o+3\sigma_2=4\sigma_2+3\sigma_o$, that is, $\sigma_2=\sigma_o$, which is impossible since $H(S)>0$. So the information ratio of every such scheme exceeds $7/6$. Jafari and Khazaei gave mixed-linear schemes with information ratio $7/6+\varepsilon$ for every $\varepsilon>0$, found by computer that $7/6$ is the optimum for mixed-linear schemes, and left open whether abelian schemes do better~\cite[Sec.~6.2]{si:JafariKhazaei2021}; for abelian schemes whose Sylow $2$-subgroup is elementary abelian they do not.

\emph{Proof of the bound in Proposition~\ref{prop:qss}~(ii).} Let a stabilizer scheme for $\Gamma_{\mathrm{NF}}$ of qubits be given; Galois qudits of characteristic $2$ are read as qubits. By Lemma~\ref{lem:qss-linear} we may assume that it is in normal form; its lift $h=h_S$ is then the entropy function of a $\GF(2)$-linear scheme for $\Gamma_{\mathrm{NF}}$, with $h(P)=2S(P)$ and $h(R)=2S(R)$. For $q\in\{A,B,C,Z\}$ the contracted scheme is a linear scheme for $\Gamma_{\mathrm{NF}}/q\cong\mathcal N$, in which the other three of $A$, $B$, $C$, $Z$ form the set $456$, so Proposition~4.5~(III) of Ref.~\cite{si:JafariKhazaei2021} gives $\sum_P\bigl(h(Pq)-h(q)\bigr)\ge4\,h(R)$ over these three players, hence $S(P)+S(P')+S(P'')\ge4\,S(R)$, and the information ratio is at least $4/3$. For $\Z_d$ qudits with $d=2r$, $r$ odd, the factor of qubits of a scheme for $\Gamma_{\mathrm{NF}}$ in which the secret consists of $m$ qudits and player $P$ holds $n_P$ qudits is a qubit scheme for $\Gamma_{\mathrm{NF}}$ with $S_2(R)=m$ bits and $S_2(P)\le n_P$ bits, so $n_P+n_{P'}+n_{P''}\ge4m$ for any three $P$, $P'$, $P''$ of $A$, $B$, $C$, $Z$.\qed

\subsection{Schemes with information ratios \texorpdfstring{$7/5$}{7/5} and \texorpdfstring{$2$}{2}}
\label{app:qss-upper}

\emph{The ranges of the optimal ratios.} By Proposition~\ref{prop:qss}, for $\Gamma_{\mathrm F}$ the optimal information ratio of stabilizer schemes lies between $4/3$ and $7/5$ in every odd local dimension and, for secrets of whole qudits, for $d=2r$ with $r>1$ odd; for $\Gamma_{\mathrm{NF}}$ it lies between $4/3$ and $7/5$ for qubits, for Galois qudits of characteristic $2$ and, for secrets of whole qudits, for $d=2r$ with $r$ odd. For $\Gamma_{\mathrm F}$ with $4\mid d$ and $d$ not a power of $2$ we know only the lower bound $4/3$, for secrets of whole qudits; for $\Gamma_{\mathrm F}$ with $d=2^a$, $a\ge2$, and $\Gamma_{\mathrm{NF}}$ with $4\mid d$ only ideal schemes are excluded. For even $d$ the restriction to secrets of whole qudits matters: otherwise the prime-power factors of a scheme can be balanced, as in the mixed-linear schemes of Ref.~\cite[Sec.~6.2]{si:JafariKhazaei2021}; for $d=6$, seven qubit cube states times the scheme of $\Z_3$ qudits of (iii), with two qutrits in a fixed state added to the reference, give a scheme for $\Gamma_{\mathrm F}$ with information ratio $7\log6/(7\log2+5\log3)\approx1.21$. With a secret of five qudits, the schemes of (iii) for $\Gamma_{\mathrm F}$ are optimal at their largest share, since by (i) some share has at least $20/3$, hence $7$, qudits; by (ii), so are those for $\Gamma_{\mathrm{NF}}$, with a secret of ten qudits.

\emph{The construction.} The schemes of Proposition~\ref{prop:qss}~(iii) come from the decomposition construction of van Dijk, Jackson and Martin, applied to $(\lambda,\omega)$-decompositions in their sense (Ref.~\cite{si:vDJM1998}, Sec.~5) with $\lambda=6$ and $\omega=1$. Let $\mathcal Q$ be the set of players. For a set $\mathcal S$ of lines of the Fano plane, let $\Gamma^{\mathcal S}$ be the access structure in which a set of players is authorized if and only if it contains a line not in $\mathcal S$, or a line in $\mathcal S$ and one more player, or the complement of a line in $\mathcal S$: each line in $\mathcal S$ and its complement exchange their status. So $\Gamma_{\mathrm F}=\Gamma^{\emptyset}$ and $\Gamma_{\mathrm{NF}}=\Gamma^{\{WXY\}}$. Every $\Gamma^{\mathcal S}$ is self-dual: a set of players is authorized if and only if its complement in $\mathcal Q$ is not. Its minimal authorized sets are the lines not in $\mathcal S$, the complements of the lines in $\mathcal S$, and the sets $\ell\cup\{x\}$ with $\ell\in\mathcal S$ and $x\notin\ell$, since the $35$ sets of four players are the seven complements of lines, which contain no line, and the $28$ sets $\ell\cup\{x\}$, and every set of five players contains two lines; its maximal unauthorized sets are their complements.

Let $\Rg$ be $\Z_d$ with $d$ odd, a finite field of odd characteristic, or $\GF(4)$, and suppose that for every line $\ell$ the access structure $\Gamma^{\mathcal S\triangle\{\ell\}}$ has an ideal $\Rg$-linear scheme. In the scheme of $\ell$, every line other than $\ell$, and the complement of every such line, has the status that it has for $\Gamma^{\mathcal S}$, and the other minimal authorized sets of $\Gamma^{\mathcal S}$, which consist of a line in $\mathcal S$ and one more player, are authorized in all seven schemes, while their complements, three players not on a line, are authorized in none. So every minimal authorized set of $\Gamma^{\mathcal S}$ is authorized in at least six of the seven schemes ($\lambda=6$), and every maximal unauthorized set in at most one ($\omega=1$). Number the lines $\ell_1,\dots,\ell_7$, let $e_1,\dots,e_7$ be the standard basis of $\Rg^7$, and use the seven schemes, with independent randomness, to share the seven coordinates of
\[
  s=\sum_{i=1}^{5}u_i\,(e_{i+1}-e_7)+\xi\,b,
\]
where $u\in\Rg^5$ is the secret, $\xi\in\Rg$ is uniformly random, and $b=(1,1,1,1,1,-1,-4)$ in odd characteristic and $b=(1,1,1,1,1,\alpha,\alpha^2)$ for $\GF(4)$, with $\alpha^2=\alpha+1$. The coordinates of $s$ sum to zero, and $b$, whose first coordinate is $1$, completes $e_2-e_7,\dots,e_6-e_7$ to a basis of the vectors with this property; so any six coordinates determine $s$, and hence $u$. Every coordinate of $b$ is a unit, so each single coordinate of $s$ is uniform and independent of $u$. Hence every authorized set of $\Gamma^{\mathcal S}$ recovers $u$, and every unauthorized set learns at most one coordinate of $s$ and nothing about $u$. This checks the access conditions directly. With this ramp scheme for their domain structure, the construction is an instance of Theorem~4 and Corollary~5 of Ref.~\cite{si:vDJM1998}, which need no algebraic structure; their Corollary~6 assumes sub-secrets with $q$ values for a prime power $q$ larger than the number of sub-schemes, an assumption that fails here for $\Z_3$, $\Z_5$, $\Z_7$ and $\GF(4)$, for instance. Each player holds seven coordinates and the secret has five, so the information ratio is $7/(\lambda-\omega)=7/5$. The combined code has $5+7\cdot7$ coordinates over $\Rg$; if $\Rg$ is a field of order $p^k$, restriction of scalars to $\GF(p)$ multiplies every number of coordinates by $k$, and Lemma~\ref{lem:qss-css} then gives a stabilizer scheme with the same ratio. The seven sub-secrets $s_\ell$ allow no larger secret: if a function $u'$ of $s$ is determined by any six coordinates of $s$ and is independent of each single coordinate, then for two lines $\ell\ne\ell'$ we have $H(u')=H(u'\mid s_{\ell'})\le H\bigl((s_{\ell''})_{\ell''\ne\ell,\ell'}\bigr)\le5\log\abs{\Rg}$. It remains to give the seven ideal sub-schemes for $\mathcal S=\emptyset$ and for $\mathcal S=\{WXY\}$.

\emph{The access structure $\Gamma_{\mathrm F}$ for odd $d$.} Over $\Z_d$ with $d$ odd, all minors of the integer generator matrix of the cube code lie in $\{0,\pm1,\pm2\}$ (Supplementary Note~\ref{app:cubes-z4}), so the nonzero Smith invariants of every set of columns are $1$ or $2$, and the orders of the projections are $d$ raised to the ranks over $\Q$; hence the scheme of the cube code over $\Z_d$ is an ideal scheme for $\Gamma_{\mathrm{NF}}$, as over $\GF(p)$. For each line $\ell$ of the Fano plane let $\sigma_\ell$ be a collineation with $\sigma_\ell(WXY)=\ell$. The scheme of $\ell$ draws a uniformly random affine function $f_\ell(x)=s_\ell+a_\ell\cdot x$ on $\{0,1\}^3$, with $a_\ell\in\Z_d^3$, and gives every player $P$ the value $f_\ell\bigl(\sigma_\ell^{-1}(P)\bigr)$, where points are read as vectors in $\{0,1\}^3$. Its access structure is the image of $\Gamma_{\mathrm{NF}}$ under $\sigma_\ell$, that is, $\Gamma^{\{\ell\}}$: $\Gamma_{\mathrm F}$ with $\ell$ removed and $\mathcal Q\setminus\ell$ added. These are the seven sub-schemes for $\mathcal S=\emptyset$, and their sub-secrets $s_\ell$ are the coordinates of $s$ above, with $\Rg=\Z_d$. We checked the access conditions of the combined code on all $128$ sets of players, with the orders of the projections computed exactly, and the entropies of its CSS state, for $d=3,5,7,9,11,13,15,21,25,27,45$ (\path{check_rev_qss.py}, part (h)); without $\xi$, the complement of a line learns a coordinate of $u$. For Galois qudits of dimension $p^k$ with $p$ odd, the same construction over $\GF(p^k)$ gives a $\GF(p^k)$-linear scheme, which by restriction of scalars is $\GF(p)$-linear (Supplementary Proposition~\ref{prop:galois}), and Lemma~\ref{lem:qss-css} applies.

\emph{The access structure $\Gamma_{\mathrm{NF}}$ for qubits.} Let $\alpha\in\GF(4)$ with $\alpha^2=\alpha+1$. The binary cube code has the generator matrix below with $\alpha$ replaced by $0$, whose columns for $R$, $A$, $B$ and $C$ are the unit vectors; the matrix itself, with two entries changed to $\alpha$,
\[
\setlength{\arraycolsep}{4pt}
\begin{array}{cccc|cccc}
R&A&B&C&W&X&Y&Z\\\hline
1&0&0&0&1&1&1&\alpha\\
0&1&0&0&1&1&0&1\\
0&0&1&0&1&\alpha&1&1\\
0&0&0&1&0&1&1&1
\end{array}
\]
generates a code over $\GF(4)$ whose matroid is $\mathrm{AG}(3,2)$ with the four circuit-hyperplanes $\{R,W,X,Y\}$, $\{A,B,C,Z\}$, $\{R,A,C,X\}$ and $\{B,W,Y,Z\}$ relaxed. Its port at $R$ is $\Gamma^{\{WXY,ACX\}}$, so the code is an ideal scheme for this access structure. The $24$ collineations of the Fano plane that fix the line $WXY$ act transitively on the other six lines; relabeling the players by one that maps $ACX$ to $\ell$ gives an ideal scheme for $\Gamma^{\{WXY,\ell\}}$. With the binary cube code, read over $\GF(4)$, for $\Gamma^{\emptyset}=\Gamma_{\mathrm F}$, these are the seven sub-schemes for $\mathcal S=\{WXY\}$, and $b=(1,1,1,1,1,\alpha,\alpha^2)$ has nonzero coordinates with sum $1+\alpha+\alpha^2=0$. The combined code over $\GF(4)$ has $5+7\cdot7$ coordinates. Writing every element of $\GF(4)$ as $c_0+c_1\alpha$ with $c_0,c_1\in\GF(2)$ turns it into a binary code with $10$ coordinates for the secret and $14$ for each player, and by Lemma~\ref{lem:qss-css} its CSS state is a qubit scheme with information ratio $7/5$; over $\GF(2^k)$ the binary code gives one for Galois qudits of dimension $2^k$. We checked the matroid (all $256$ ranks), the seven sub-schemes, and the access conditions on all $128$ sets of players over $\GF(4)$ and for the binary code, with the entropies of the CSS states (\path{check_rev_qss.py}, part (i)); with $b=(1,1,1,1,1,1,0)$, or with the cube code as all seven sub-schemes, the code is not a scheme for $\Gamma_{\mathrm{NF}}$. These sub-schemes were found by a linear program over a large sample of ideal schemes over $\GF(2^k)$, $k\le4$.

\emph{A secret of one qubit.} The binary code with generator matrix
\[
\setlength{\arraycolsep}{2.6pt}
\begin{array}{c|cc|cc|cc|cc|cc|cc|cc}
R&\multicolumn{2}{c|}{A}&\multicolumn{2}{c|}{B}&\multicolumn{2}{c|}{C}&\multicolumn{2}{c|}{W}&\multicolumn{2}{c|}{X}&\multicolumn{2}{c|}{Y}&\multicolumn{2}{c}{Z}\\\hline
1&0&0&0&0&0&0&0&1&1&1&1&1&1&1\\
0&1&0&0&0&0&0&0&1&1&1&1&0&0&1\\
0&0&1&0&0&0&0&1&1&0&1&1&0&1&1\\
0&0&0&1&0&0&0&0&0&0&1&0&0&1&1\\
0&0&0&0&1&0&0&0&1&0&0&1&1&1&0\\
0&0&0&0&0&1&0&0&0&0&1&0&1&1&0\\
0&0&0&0&0&0&1&0&0&1&0&1&0&1&1
\end{array}
\]
(one column for $R$ and two for each player) is a scheme for $\Gamma_{\mathrm{NF}}$: we checked all $128$ sets of players (\path{check_rev_qss.py}, part (f)). It was found by a randomized search. By Lemma~\ref{lem:qss-css}, its CSS state is a qubit scheme with a secret of one qubit and shares of two qubits, the least possible by Proposition~\ref{prop:qss}~(ii), since shares consist of whole qudits (likewise a secret of one qudit needs a share of at least two qudits for $\Gamma_{\mathrm F}$ whenever $d$ is not a power of $2$, and for $\Gamma_{\mathrm{NF}}$ whenever $d\equiv2\pmod4$); over $\GF(2^k)$ the same matrix gives one for Galois qudits of dimension $2^k$.

\emph{Local dimensions $d=2r$ with $r$ odd.} Let $r>1$ be odd. For $\Gamma_{\mathrm F}$, take five copies of the qubit cube state, so that the reference and every player hold five qubits, add two qubits in the state $\ket0$ to every player, and take the tensor product with the scheme of $\Z_r$ qudits above, with a secret of five qudits and shares of seven. By the decomposition into coprime factors (Supplementary Note~\ref{app:qss-contract}), the product is a stabilizer scheme of $\Z_d$ qudits for $\Gamma_{\mathrm F}$ with a secret of five qudits and shares of seven. For $\Gamma_{\mathrm{NF}}$, take the qubit scheme above, with a secret of ten qubits and shares of fourteen, and ten copies of the cube state over $\Z_r$, with four qudits in the state $\ket0$ added to every player; the product is a stabilizer scheme of $\Z_d$ qudits for $\Gamma_{\mathrm{NF}}$ with a secret of ten qudits and shares of fourteen. Both have information ratio $7/5$. We checked both codes over $\Z_6$, $\Z_{10}$, $\Z_{18}$ and $\Z_{30}$ on all $128$ sets of players, with the entropies of their CSS states (\path{check_rev_qss.py}, part (n)).

\emph{Decompositions into ideal schemes.} For $\Gamma_{\mathrm F}$, $7/5$ is the best ratio that a $(\lambda,\omega)$-decomposition into ideal linear schemes gives: \emph{let $p$ be an odd prime; in every $(\lambda,\omega)$-decomposition of $\Gamma_{\mathrm F}$ whose sub-schemes are ideal linear schemes over a common field of characteristic $p$, the average size of a share is at least $7/5$ times that of the secret, and so is the information ratio.} Here the decomposition construction of Ref.~\cite{si:vDJM1998} gives player $P$ one sub-share for each sub-scheme in which $P$ takes part, say $m_P$ of them, and a secret of the size of $\lambda-\omega>0$ sub-secrets, so its information ratio is $\max_Pm_P/(\lambda-\omega)$. Let $\Sigma_j$ be the sub-schemes, $\Sigma_j$ ideal for an access structure $\Gamma_j$ on a set $P_j$ of players; removing the players that lie in no minimal authorized set of $\Gamma_j$ changes no access condition and does not increase any $m_P$, so we may assume that every player of $P_j$ lies in one. Let $n_j=\abs{P_j}$, and let $a_j$ and $c_j$ be the numbers of lines and of complements of lines $T$ with $T\cap P_j\in\Gamma_j$. Every line is authorized in at least $\lambda$ sub-schemes and every complement of a line, a maximal unauthorized set, in at most $\omega$, so $\sum_ja_j\ge7\lambda$ and $\sum_jc_j\le7\omega$, while $\sum_Pm_P=\sum_jn_j$. Hence the average of $m_P/(\lambda-\omega)$ over the seven players is at least $\sum_jn_j/\sum_j(a_j-c_j)$, and it suffices to show $5n_j\ge7(a_j-c_j)$ for every $j$. If $n_j=7$ and $a_j-c_j\ge6$, relabel $\Sigma_j$ by the $168$ collineations of the Fano plane, after extending scalars to a field of characteristic $p$ with at least $168$ elements, which keeps $\Sigma_j$ ideal. Since the stabilizer of a line has order $24$, every line is authorized in exactly $24a_j$ and every complement of a line in exactly $24c_j$ of these schemes, and every player takes part in all of them. The construction of Ref.~\cite{si:vDJM1998}, with the ramp scheme that evaluates at $168$ distinct points a polynomial of degree less than $24a_j$ whose $24c_j$ lowest coefficients are random and whose other coefficients are the secret, combines them into a linear scheme for $\Gamma_{\mathrm F}$ with information ratio $168/(24(a_j-c_j))\le7/6$, and Lemma~\ref{lem:qss-css} gives a stabilizer scheme of Galois qudits of characteristic $p$ with the same ratio, contradicting Proposition~\ref{prop:qss}~(i). So $a_j-c_j\le5$ when $n_j=7$. If $n_j\le6$, a line is authorized only if it meets $P_j$, and then every complement of a line that contains its intersection with $P_j$ is authorized too. For $n_j=6$, let $x$ be the missing player: the three lines through $x$ meet $P_j$ in the three pairs of a partition of $P_j$, the other four lines lie in $P_j$, and the complement of the line through $x$ and one pair is the union of the other two pairs. If $\nu$ of the pairs are authorized, at least $0$, $2$, $3$ or $3$ complements of lines are authorized for $\nu=0,1,2,3$, so $a_j-c_j\le4$. A check of all sets $P_j$ with $n_j\le5$ and all sets of authorized lines gives $a_j-c_j\le2,1,1,0,0$ for $n_j=5,4,3,2,1$ (\path{check_rev_qss.py}, part (j)). In every case $5n_j\ge7(a_j-c_j)$. The sub-schemes of (iii) attain the bound, with $n_j=7$, $a_j=6$ and $c_j=1$.

\subsection{Comparison with a linear program}
\label{app:qss-lp}

Lemma~\ref{lem:fano-odd} does not follow from the characteristic-dependent inequalities that we checked. As a seven-variable inequality, it does not follow from the Shannon inequalities, the Ingleton inequalities (all instances with disjoint arguments) and the four odd-characteristic inequalities of Table~\ref{tab:eight}, (65), (6.17) corrected, (b) and \LB{}, in all $5{,}040$ labelings: an integer polymatroid on seven elements with $h(N)=120$ satisfies all of them, and the two sides of Eq.~\eqref{eq:fano-odd} differ by $-37$ on it. For the port $\mathcal F$ of the Fano matroid, these inequalities and the access equations allow the information ratio $23/22$: an integer polymatroid in which the secret has entropy $22$ and every share at most $23$ satisfies all of them (\path{check_rev_qss.py}, part (m)), whereas Lemma~\ref{lem:fano-odd} gives $4/3$. Pe\~na-Macias obtained $4/3$ for linear schemes from his inequalities~\cite[Cors.~2.7 and~2.9]{si:PenaMacias2023}, which are therefore not among these four; the first inequality of his Corollary~2.7 with $n=3$ is blind on the self-dual slice (Supplementary Note~\ref{app:blind}), and we did not add it to these programs. Nor, in the setting of secret sharing, does it follow from the Shannon inequalities and \LB{} with $t=2$ in all labelings, even with self-duality and the access equations. Let $h$ be the lift of the joint state of a stabilizer scheme for $\Gamma_{\mathrm F}$ of odd local dimension $d$, on the eight systems or, with $E$, on nine. Then:
\begin{itemize}
\item $h$ is a polymatroid, and it is self-dual because the joint state is pure (Lemma~\ref{lem:A-selfdual});
\item $h$ is the rank function of a subgroup arrangement in a finite abelian group of exponent dividing $d$, hence of odd order (Lemma~\ref{lem:A-lift}); so is every minor of $h$ obtained by contracting or deleting elements (Methods, Sec.~\ref{sec:methods-proofs});
\item consequently, every seven-element minor satisfies \LB{} with $t=2$, in every labeling of its elements by $A_1,A_2,A_3,B_1,B_2,B_3,C$, because multiplication by $2$ is injective on a group of odd order;
\item the scheme conditions are the equations of Eq.~\eqref{eq:qss-lift} for $\Gamma_{\mathrm F}$, and the information ratio is at least $\max_Ph(P)/h(R)$.
\end{itemize}
Linear programs give lower bounds on information ratios, and Farr\`as \emph{et al.}\ sharpen them for linear schemes by imposing the common information property~\cite{si:FKMP2018}. Such constraints hold for every subgroup arrangement, hence also for the lift of the qubit cube state, an ideal scheme for $\Gamma_{\mathrm F}$, so for $\Gamma_{\mathrm F}$ they cannot raise the bound above $1$; here we add \LB{} on minors instead, which is valid in groups of odd order and fails for the Fano configuration. The linear program that minimizes $\rho$ over $(h,\rho)$ subject to the elemental Shannon inequalities, the self-duality equations, the access equations for all $127$ nonempty sets of players, $h(R)=2$, $h(P)\le2\rho$ for the seven players, and \LB{} with $t=2$ on minors has optimum $43/41$, both without and with $E$. Without $E$ the minors are those obtained by contracting or deleting one of the eight elements, in all $5{,}040$ labelings ($80{,}640$ inequalities); with $E$ they are those obtained by contracting or deleting $E$ and one other element ($161{,}280$ inequalities). The optima were found in floating point with HiGHS; the lower bound $43/41$ is certified in exact arithmetic by rational dual multipliers, with $72$ inequalities and $109$ equations without $E$ and $140$ and $201$ with $E$ (rational multipliers $y\le0$ for the inequalities $a\cdot x\le0$ and $z$ for the equations $e\cdot x=b$ such that $c-\sum ya-\sum ze\ge0$ componentwise, where $c$ is the objective, so that $\rho\ge\sum zb$ for every feasible point, since all variables are nonnegative), and rounding the floating-point solutions gives exactly feasible points with $\rho=43/41$. At such a point every player has $h(P)\le86/41$, so every line has $\sum_{P\in\ell}h(P)\le258/41<8=4h(R)$. So the inequalities that Lemma~\ref{lem:fano-odd} gives on the contractions, as in the proof of Proposition~\ref{prop:qss}~(i), do not follow from the constraints of this program, although they hold for every stabilizer scheme of odd local dimension. On the qubit cube state, \LB{} with $t=2$ takes the value $-2$ on some minor, so the program excludes it; on the cube state over $\GF(3)$ it is nonnegative on all minors in all labelings. As controls, the program with \LA{} in place of \LB{}, which the qubit cube state satisfies on all minors, and the program with Shannon inequalities only have optimum $1$; $\rho\ge1$ follows because every player lies in a minimal authorized set. Adding the inequalities of Ref.~\cite{si:JafariKhazaei2021} on all contractions to the programs, with or without self-duality, gives exactly $4/3$: the added inequalities imply $\rho\ge4/3$, and rounded floating-point solutions are exactly feasible with $\rho=4/3$. So purity does not raise the bound $4/3$ in these programs. The value $43/41$ is also the lower bound that Khazaei recently obtained for homomorphic schemes for the twelve-player access structure formed by disjoint Fano and non-Fano ports, from recursive common information~\cite{si:Khazaei2026rci}; we know of no relation between the two.

The certificates and the code are in the repository~\cite{si:CaoSong2026data}: \path{data/s38_qss_certificates.json} contains, for each of the two programs, the descriptions of the inequalities and equations used and their rational multipliers. The script \path{scripts/s38_secret_sharing.py} solves the programs and checks exact feasible points (\texttt{-{}-solve}, gate G43b), regenerates the certificates (\texttt{-{}-certify}), verifies them and the contraction statements above (\texttt{-{}-verify}, gate G43), and computes the distances of Supplementary Note~\ref{app:margins}; \path{checks/check_rev_qss.py} independently verifies the certificates, the contraction statements and the distances.

\subsection{Ideal schemes}
\label{app:qss-ideal}

This subsection proves Proposition~\ref{prop:qss}~(iv). For Galois qudits, (i), (ii) and the cube codes over $\GF(q)$ give the same dichotomy by characteristic: $\Gamma_{\mathrm F}$ has an ideal scheme exactly in characteristic $2$, and $\Gamma_{\mathrm{NF}}$ exactly in odd characteristic. Let $\Gamma$ be $\Gamma_{\mathrm F}$ or $\Gamma_{\mathrm{NF}}$, and let $\mathcal M$ be the matroid whose port at $R$ is $\Gamma$: $\mathrm{AG}(3,2)$ for $\Gamma_{\mathrm F}$ and $R_8$ for $\Gamma_{\mathrm{NF}}$ (Supplementary Note~\ref{app:qss-access}).

\emph{Proof of~(iv).} Consider an ideal stabilizer scheme of $\Z_d$ qudits for $\Gamma$, pure or mixed. Its shares consist of whole qudits and it is ideal, so its secret has dimension $d^m$ for some $m\ge1$: $S(R)=m\log d$, and every player holds at most $m$ qudits. Every player $P$ lies in a minimal authorized set $T$, so $2S(R)=I(R;T)-I(R;T\setminus P)=I(R;P|T\setminus P)\le2S(P)$; hence every player holds exactly $m$ qudits, all maximally mixed. Let $L$ be the Lagrangian subgroup of the joint state of the reference, the players and, for a mixed scheme, $E$ (Supplementary Note~\ref{app:reduction-stab}), let $L_0$ be its projection to the reference and the players, and let the reference hold $m_R$ qudits. By Eq.~\eqref{eq:A-entropy}, $\log\abs{\pr_XL_0}=S(X)+m_X\log d$ for every set $X$ of these eight systems, where $m_X$ is the number of qudits of $X$, and Lemma~\ref{lem:K-forcing} below, applied to $h_S/(2m\log d)$, gives $h_S=2m\log d\,r_{\mathcal M}$ on them. Since the reference and every player have entropy $m\log d$, this gives $\abs{\pr_TL_0}=\abs V^{r_{\mathcal M}(T)}$ and $\abs{\pr_{RT}L_0}=d^{\,m_R-m}\abs V^{r_{\mathcal M}(RT)}$ for every set $T$ of players, with $V=\Z_d^{2m}$, the phase space of a player; in particular, $\abs{\pr_RL_0}=d^{\,m_R-m}\abs V$. By Lemma~\ref{lem:K-two} below, $2V=0$ for $\Gamma_{\mathrm F}$, that is, $d=2$, and $\abs V=d^{2m}$ is odd for $\Gamma_{\mathrm{NF}}$, that is, $d$ is odd. Hence $\Gamma_{\mathrm F}$ has an ideal scheme only if $d=2$, and $\Gamma_{\mathrm{NF}}$ only if $d$ is odd. Conversely, the qubit cube state is ideal for $\Gamma_{\mathrm F}$, and for odd $d$ the cube state over $\Z_d$ is ideal for $\Gamma_{\mathrm{NF}}$: all minors of its generator matrix lie in $\{0,\pm1,\pm2\}$ (Supplementary Note~\ref{app:cubes-z4}), so the nonzero Smith invariants of every set of its columns are $1$ or $2$, units modulo $d$, and every projection of the code over $\Z_d$ has the order $d^{r}$ given by the rational rank $r$ (Supplementary Note~\ref{app:qss-access}).

\begin{applemma}[Forcing]\label{lem:K-forcing}
Let $g$ be a polymatroid on $\{R\}\cup\mathcal Q$ with $g(R)=g(P)=1$ for every player $P$, $g(RT)=g(T)$ for every authorized set $T$ of $\Gamma$, and $g(RT)=g(T)+1$ for every other nonempty set $T$ of players. Then $g=r_{\mathcal M}$.
\end{applemma}

\begin{proof}
The rank function $r_{\mathcal M}$ satisfies all these conditions. For every nonempty set $X$ of the eight elements, we solved the linear programs that minimize and maximize $g(X)$ over the polymatroids with these properties, and certified both optima, $r_{\mathcal M}(X)$, by rational multipliers: nonnegative ones for elemental Shannon inequalities $a\cdot g\ge0$ and arbitrary ones for the equations $e\cdot g=b$ above, such that the linear form $g\mapsto\pm g(X)$ equals the combination of the $a$ and $e$, as linear forms in the $255$ values of $g$, and the combination of the right-hand sides $b$ equals $\pm r_{\mathcal M}(X)$. For every feasible $g$ this gives $g(X)\ge r_{\mathcal M}(X)$ and $g(X)\le r_{\mathcal M}(X)$. The $2\times2\times255$ certificates are verified in exact arithmetic (Methods).
\end{proof}

Lemma~\ref{lem:K-forcing} also follows without a computer, for every polymatroid $g$ with these properties: $R$ is an ideal element of $g$ in the sense of Mat\'u\v{s}, so $g$ is a matroid rank function by his Theorem~2~\cite{si:Matus2012}, and Lehman's theorem~\cite{si:Lehman1964} identifies the matroid, as in the proof of Theorem~\ref{thm:ideal} (Supplementary Note~\ref{app:qss-general}). Applied in the same way to the entropy function of an ideal classical scheme in which every participant is essential, Mat\'u\v{s}'s Theorem~2 shows that it is a multiple of a matroid rank function, a form of the theorem of Brickell and Davenport that the access structures of ideal classical schemes are matroid ports~\cite{si:BrickellDavenport1991}. The certificates give an independent, computer-checked proof.

\begin{applemma}[Abelian phase spaces]\label{lem:K-two}
Let $\mathcal M$ be $\mathrm{AG}(3,2)$ or $R_8$, on the reference and the seven players labeled as above, let $V$ be a finite abelian group, and let $L$ be a subgroup of $V_R\oplus\bigoplus_{P\in\mathcal Q}V_P$, where $V_R$ is a finite abelian group and every $V_P$ is a copy of $V$, such that $\abs{\pr_{RT}L}=\abs{\pr_RL}\,\abs V^{\,r_{\mathcal M}(RT)-1}$ for every set $T$ of players. If $\mathcal M=\mathrm{AG}(3,2)$, then $2x=0$ for every $x\in V$; if $\mathcal M=R_8$, then $\abs V$ is odd.
\end{applemma}

\begin{proof}
Let $L''$ be the projection to the players of $\{v\in L:\pr_Rv=0\}$. For a set $T$ of players, restriction maps $\pr_{RT}L$ onto $\pr_RL$, with kernel isomorphic to $\pr_TL''$; so
\[
  \abs{\pr_TL''}=\abs V^{\,r_{\mathcal M}(RT)-1}.
\]
All sets $RT$ used below have the same rank in both matroids, except $RWXY$.

(1) Since $r_{\mathcal M}(RABC)=r_{\mathcal M}(\{R\}\cup\mathcal Q)=4$, the projection $\pr_{ABC}$ maps $L''$ bijectively onto $V_A\oplus V_B\oplus V_C$. So $L''$ consists of the vectors with coordinates $u$, $v$, $w$ at $A$, $B$, $C$, and $u\Phi_{AP}+v\Phi_{BP}+w\Phi_{CP}$ at $P\in\{W,X,Y,Z\}$, for all $u,v,w\in V$, where the $\Phi_{IP}$ are endomorphisms of $V$, written on the right.

(2) The planes $RABW$, $RACX$ and $RBCY$ have rank $3$, as do $RAB$, $RAC$ and $RBC$, so on $L''$ the coordinates at $W$, $X$ and $Y$ are functions of those at $AB$, $AC$ and $BC$, respectively: $\Phi_{CW}=\Phi_{BX}=\Phi_{AY}=0$.

(3) For each of the other nine blocks $\Phi_{IP}$, the set $R\cup(\{A,B,C\}\setminus\{I\})\cup\{P\}$ is not a plane and has rank $4$; so the projection to its players is injective on $L''$, and $\Phi_{IP}$ is injective, hence an automorphism, since $V$ is finite.

(4) Applying an automorphism of $V$ to the coordinates of a single player changes no order $\abs{\pr_TL''}$. Applying such automorphisms to $A$ and $B$, then to $X$ and $C$, and then to $Y$, we may assume that the coordinates at $W$, $X$ and $Y$ are $u+v$, $u+w$ and $v\Phi+w$, and that the coordinate at $Z$ is $u\Psi_A+v\Psi_B+w\Psi_C$, with $\Phi$ and the $\Psi_I$ automorphisms.

(5) The planes $RAYZ$, $RBXZ$ and $RCWZ$ have rank $3$, as do $RAY$, $RBX$ and $RCW$, so the coordinate at $Z$ is a function of those at $AY$, at $BX$ and at $CW$. Hence $\Psi_B=\Phi\Psi_C$, $\Psi_A=\Psi_C$ and $\Psi_A=\Psi_B$, so $(\Phi-1)\Psi_C=0$ and $\Phi=1$. After an automorphism of $Z$,
\[
  L''=\{(u,v,w,u+v,u+w,v+w,u+v+w)\},
\]
with $u,v,w\in V$ and the players in the order $A,B,C,W,X,Y,Z$.

(6) The map $(u,v,w)\mapsto(u+v,u+w,v+w)$ has kernel $\{(x,-x,-x):2x=0\}$, which is isomorphic to $V[2]=\{x\in V:2x=0\}$, so $\abs{\pr_{WXY}L''}=\abs V^3/\abs{V[2]}$. But $\abs{\pr_{WXY}L''}=\abs V^{\,r_{\mathcal M}(RWXY)-1}$, where $r_{\mathcal M}(RWXY)$ is $3$ for $\mathrm{AG}(3,2)$ and $4$ for $R_8$. For $\mathrm{AG}(3,2)$ this gives $V[2]=V$. For $R_8$ it gives $V[2]=0$: no element of $V$ has order $2$, so $\abs V$ is odd by Cauchy's theorem.
\end{proof}

For $\mathrm{AG}(3,2)$, $V$ is thus a vector space over $\GF(2)$, and $L''$ is, up to group automorphisms, not necessarily symplectic, of the players' phase spaces, the binary simplex code, that is, the cube code shortened at $R$, on $\dim V$ coordinates per player, as for copies of the qubit cube state. The lemma also follows from the work of Mat\'u\v{s} on representations by partitions~\cite{si:Matus1999}. The cosets of the kernels of the coordinate projections of $L''$, a group of order $\abs V^3$, form a representation by partitions of degree $\abs V$ of the contraction $\mathcal M/R$, the Fano or the non-Fano matroid, so $\abs V$ is a power of $2$ for $\mathrm{AG}(3,2)$ and odd for $R_8$ (Ref.~\cite{si:Matus1999}, proof of Prop.~4.1). For the Fano matroid his proof gives more: the group that defines such a representation, unique up to isomorphism by his Prop.~3.1, is a power of $\Z_2$ (see also his Sec.~5.6). Here this group is $V$: after steps (1)--(4), the kernel of $(u,v,w)\mapsto(u+v,u+w,v\Phi+w)$, which is isomorphic to that of $\Phi+1$, has order $\abs V$ because $\{W,X,Y\}$ is a line, so $\Phi=-1$, and $(u,v,w)\mapsto(u,-v,-w)$ maps the partitions at $A$, $B$, $C$, $W$, $Y$ and $X$ to those of his representation of the cycle matroid of $K_4$ built from $(V,+)$. So $2x=0$ on $V$ also follows from his proof; the proof above obtains it, and the normal form used for the uniqueness for $D=2$ (Supplementary Note~\ref{app:qss-general}), directly from the orders of the projections. Step~6 of the proof of Theorem~\ref{thm:ideal} for $\Gamma_{\mathrm F}$ and the case $d=2^a$, $a\ge2$, of Proposition~\ref{prop:qss}~(iv) need $2x=0$, which excludes groups such as $\Z_4$ and $\Z_8$. For $\Gamma_{\mathrm{NF}}$, Step~6 uses the half for $R_8$. The certificates of Lemma~\ref{lem:K-forcing} are in \path{data/s39_ideal_certificates.json}; the script \path{scripts/s39_ideal_schemes.py} regenerates and verifies them and computes the orders of the projections of the normal form over $\Z_{2^a}$ for $a\le3$, and \path{checks/check_rev_qss.py}, part (k), verifies them independently.

\subsection{The seven-qudit color code}
\label{app:qss-color}

Steane's code is the smallest triangular color code~\cite{si:Bombin2006,si:Landahl2011}: its seven qubits are the vertices of three faces, $AWXZ$, $BWYZ$ and $CXYZ$, which are the supports of the rows $M_1$, $M_2$ and $M_3$ of the $3\times7$ matrix $M$ whose columns are the points $A,\dots,Z$, and each face carries an $X$-type and a $Z$-type stabilizer. These are the faces of the cube $\{0,1\}^3$ that do not contain the origin. Watson \emph{et al.}\ define color codes of qudits of any dimension, also on lattices punctured at a vertex, by splitting the vertices into two classes, such that neighboring vertices lie in different classes, and using $Z^{-1}$ in place of $Z$ on one of them~\cite{si:Watson2015}; here the classes are $\{A,B,C,Z\}$ and $\{W,X,Y\}$. For $d\ge2$, let $M'$ be $M$ with the columns $W$, $X$ and $Y$ negated, and consider the code of $\Z_d$ qudits stabilized by $X^{M_i}=\bigotimes_PX_P^{M_{iP}}$ and $Z^{M'_i}$, $i=1,2,3$. These operators commute, because each face contains two vertices of each class and any two faces share one vertex of each, so that $MM'^{\mathsf T}=0$ over the integers. For $d=2$ the code is Steane's code. For every $d\ge2$:
\begin{enumerate}
\item[(a)] It is a $[[7,1,3]]_d$ code. The columns $A$, $B$ and $C$ of $M$ and of $M'$ form identity matrices, so the six stabilizers generate a group of order $d^6$, and the code space has dimension $d$. An operator $X^uZ^v$ commutes with all stabilizers exactly when $M'u=0$ and $Mv=0$ modulo $d$. The columns of $M$ are distinct nonzero $0/1$ vectors, and those of $M'$ are the same up to sign, so for two of them some coordinate is nonzero in exactly one; hence no nonzero $u$ or $v$ supported on at most two sites satisfies these equations, and no logical operator has weight $1$ or $2$. The operator $Z_AZ_BZ_W^{-1}$ commutes with all stabilizers and is not a stabilizer: a combination of the rows of $M'$ with entries $1,1,0$ at $A,B,C$ has $-2$, not $-1$, at $W$. So the code has distance $3$.
\item[(b)] Its logical basis state $\ket{0_L}$ stabilized by the $X^{M_i}$ and by every $Z^{v}$ with $Mv=0$ modulo $d$ is the CSS state of the code generated by $M$ over $\Z_d$, whose entropy vector is $S^{(d)}$ of Theorem~\ref{thm:reverse}: the Steane state for $d=2$, and for an odd prime $d$ the non-Fano state.
\item[(c)] Take $X^{\otimes7}$, which commutes with the $Z^{M'_i}$ because each row of $M'$ sums to $0$, as the logical $X$ operator, and let a reference system $R$ be maximally entangled with the logical qudit: the state is $d^{-1/2}\sum_{j}\ket{j}_R(X^{\otimes7})^j\ket{0_L}$. It is the CSS state of the code generated by $(1,1,\dots,1)$ and the rows $(0,M_i)$, with the coordinate of $R$ first, whose columns are $(1,v)^{\mathsf T}$, $v\in\{0,1\}^3$, with $R$ at the origin; so it is the cube state of Definition~\ref{def:cube} with $t=2$ over $\Z_d$. It is an ideal scheme for $\Gamma_{\mathrm F}$ if $d=2$ and for $\Gamma_{\mathrm{NF}}$ if $d$ is odd (Supplementary Notes~\ref{app:qss-access} and~\ref{app:qss-ideal}).
\end{enumerate}
Hence, for odd $d$, the seven-qudit color code recovers its logical qudit from exactly the sets of sites in $\Gamma_{\mathrm{NF}}$: unlike Steane's code, it cannot recover it from the three sites of the line $WXY$, and it can recover it from the complementary four sites $ABCZ$. For even $d\ge4$, by contrast, $WXY$ and $ABCZ$ learn part of the logical qudit, and every other set of sites recovers it or learns nothing about it, according as it is authorized in both access structures or in neither. Indeed, every set of columns of the integer generator matrix of the cube code has nonzero Smith invariants $1$ or $2$ (Supplementary Note~\ref{app:cubes-z4}), so for even $d$ the projection of the code over $\Z_d$ to a set $X$ of systems has order $2^{r_2(X)}(d/2)^{r_{\Q}(X)}$ (Supplementary Note~\ref{app:reverse-sevenpoint}), where $r_2$ and $r_{\Q}$ are the ranks over $\GF(2)$ and over $\Q$, the rank functions of $\mathrm{AG}(3,2)$ and of $R_8$. Since $h_S$ is twice the logarithm of these orders (Lemma~\ref{lem:cube-selfdual}), Eq.~\eqref{eq:qss-lift} gives $I(R;T)=2\,[T\in\Gamma_{\mathrm F}]+2\log(d/2)\,[T\in\Gamma_{\mathrm{NF}}]$ bits, where $[\cdot]$ is $1$ if the condition holds and $0$ otherwise, whereas $2S(R)=2+2\log(d/2)$ bits. For $d=6$, for instance, the state of $R$ and the sites is, by the decomposition into prime powers of Supplementary Note~\ref{app:qss-contract}, the tensor product of the cube states over $\Z_2$ and over $\Z_3$, so $WXY$ recovers the qubit part of the logical qudit and learns nothing about its qutrit part, and $ABCZ$ does the reverse. By Theorem~\ref{thm:reverse}, for every $d\ge3$ no qubit stabilizer state has the entropy pattern of the logical basis state $\ket{0_L}$ of the color code of $\Z_d$ qudits, even up to scaling and in the limit, and by Theorem~\ref{thm:ideal} no code that encodes one qudit of even dimension with one such qudit per site, stabilizer code or not, has the recovery pattern of the color code of odd local dimension. By Proposition~\ref{prop:qss}, a stabilizer code of odd local dimension with the recovery pattern of Steane's code needs on average at least $4/3$ times as many qudits per site as it encodes, and a qubit stabilizer code with the recovery pattern of the seven-qudit color code needs a site with at least $4/3$ times as many. The script \path{checks/check_color_code.py} verifies (a)--(c) and the recovery patterns in exact arithmetic for every $d\le30$.

\emph{Erasures.} For a code, the erasure of a set of sites is correctable exactly when that set is unauthorized, that is, since the joint state of the reference and the sites is pure, exactly when the other sites are authorized. So Steane's code fails exactly when the erased sites contain a line of the Fano plane, and the color code of odd local dimension exactly when they form a set in $\Gamma_{\mathrm{NF}}$. For $3,4,5,6,7$ sites, Steane's code has $7$, $28$, $21$, $7$ and $1$ uncorrectable erasures, and the color code $6$, $29$, $21$, $7$ and $1$. If every site is erased independently with probability $\varepsilon$, the probability that the erasure is uncorrectable is therefore smaller for the color code, by $\varepsilon^3(1-\varepsilon)^4-\varepsilon^4(1-\varepsilon)^3=\varepsilon^3(1-\varepsilon)^3(1-2\varepsilon)$, for every $0<\varepsilon<1/2$. Qubit stabilizer codes cannot match this: every qubit stabilizer code that encodes one qubit into seven with distance at least $3$ has at least $7$ uncorrectable erasures of three sites (see below), and then it fails more often than the color code for every $0<\varepsilon<1/2$. Indeed, let such a code have $N_k$ uncorrectable erasures of $k$ sites. By the distance, every set of at most two sites is unauthorized and its complement authorized, so $N_5=21$, $N_6=7$, $N_7=1$ and $N_k=0$ for $k\le2$; and a set of four sites is uncorrectable unless its complement is authorized, while an authorized set of three sites is uncorrectable, so $N_4\ge35-N_3$. Hence its failure probability exceeds that of the color code by at least $(N_3-6)\varepsilon^3(1-\varepsilon)^4+(6-N_3)\varepsilon^4(1-\varepsilon)^3=(N_3-6)\varepsilon^3(1-\varepsilon)^3(1-2\varepsilon)>0$ when $N_3\ge7$. With the reference, such a code is a stabilizer state of eight qubits, which is equivalent under local Clifford unitaries, which do not change entropies, to a graph state~\cite[Thm.~1]{si:VandenNest2004}; an exhaustive search over the graphs on the seven sites, $1{,}044$ up to isomorphism, and the $127$ nonempty sets of sites adjacent to $R$ gives $3{,}414$ pairs (an isomorphism class of graphs on the sites and a neighborhood of $R$) whose codes have distance at least $3$, and each of them has at least $7$ uncorrectable erasures of three sites, as Steane's code has; the numbers range from $7$ to $23$, so codes with the parameters of Steane's code need not have its recovery pattern. The script \path{checks/check_rev_qss.py}, part (r), recomputes these counts; it generates the graphs itself and checks their number with Burnside's lemma. For even $d\ge4$, the color code fails exactly on the erasures of the sets in $\Gamma_{\mathrm F}\cup\Gamma_{\mathrm{NF}}$, $7$, $29$, $21$, $7$ and $1$ of them for $3,\dots,7$ sites, so it fails more often than Steane's code, by $\varepsilon^4(1-\varepsilon)^3$, and than the color code of odd local dimension, by $\varepsilon^3(1-\varepsilon)^4$; \path{checks/check_color_code.py} verifies this recovery pattern, the values $I(R;WXY)=2$ bits and $I(R;ABCZ)=2\log(d/2)$ bits, and these numbers of uncorrectable erasures for every even $d\le30$.

\subsection{Codes that are not stabilizer codes}
\label{app:qss-general}

Proposition~\ref{prop:qss} concerns stabilizer schemes. Schemes that are not stabilizer schemes avoid the bound of its part~(i) for large secrets, but not when the secret and every share have the same dimension: this subsection proves Theorem~\ref{thm:ideal}, for $\Gamma_{\mathrm F}$ and for $\Gamma_{\mathrm{NF}}$.

\emph{Large secrets.} Let $d\ge3$ not be a power of $2$, let $k\ge1$, and put $m=\lceil k\log_2d\rceil$ and $n=\lceil m\log_d2\rceil$. Take $m$ copies of the qubit cube state, a pure scheme for $\Gamma_{\mathrm F}$ whose secret consists of $m$ qubits, as for $m$ copies of Steane's code. Restrict the secret to a subspace of dimension $d^k\le2^m$ of the encoded space, and embed the $m$ qubits of each player isometrically into $n$ qudits of dimension $d$, which is possible because $d^n\ge2^m$. An authorized set recovers every encoded state of the copies of Steane's code, in particular those of the subspace, after undoing the isometries, and the erasure of every other set is correctable for every encoded state, so these sets learn nothing about the secret. So this is a pure scheme for $\Gamma_{\mathrm F}$ whose secret consists of $k$ qudits and in which every player holds $n$ qudits, and $n/k\to1$ as $k\to\infty$. For $d=3$ and $k=4$ or $5$, it uses $7n=35$ or $42$ qutrits. Every stabilizer scheme of qutrits for $\Gamma_{\mathrm F}$ with such a secret uses at least $\lceil28k/3\rceil=38$ or $47$, because the players of each of the seven lines hold together at least $4k$ qutrits (Proposition~\ref{prop:qss}~(i)) and every player lies on three lines.

\emph{Ideal schemes.} In this subsection a \emph{nice basis} on $\mathbb C^D$ is a family of unitaries $U_g$, indexed by the elements $g$ of a finite group $G$ of order $D^2$, such that $\operatorname{Tr}(U_g^\dagger U_h)=D\,\delta_{gh}$ and $U_gU_h$ is a scalar multiple of $U_{gh}$ for all $g,h\in G$; these are the nice error bases of Knill~\cite{si:Knill1996}, and those indexed by abelian groups were classified by Klappenecker and R\"otteler~\cite{si:KlappeneckerRoetteler2002}; we prove the facts that we need. The Weyl operators $W(v)$, $v\in\Z_D^2$, of Supplementary Note~\ref{app:reduction-stab} form one, indexed by $\Z_D^2$. If $G$ is abelian, then $U_gU_hU_g^{-1}U_h^{-1}=\beta(g,h)$ is a scalar for all $g,h$, and $\beta$ is a bicharacter of $G$ with $\beta(g,g)=1$; it is nondegenerate, because an element $U_g$ that commutes with all $U_h$ commutes with every operator, since the $U_h$ span all operators, so it is a scalar and $g$ is the identity. For qudits of dimensions $n_1,\dots,n_N$, with Weyl operators $X$ and $Z$ on each, a \emph{stabilizer state} is, up to a phase, the only vector fixed by an abelian group of scalar multiples of tensor products of Weyl operators that contains no scalar other than $1$. Weyl operators other than scalars are traceless, so the projector $\abs{\mathcal G}^{-1}\sum_{s\in\mathcal G}s$ onto the common fixed space of such a group $\mathcal G$ has trace $\dim\mathcal H/\abs{\mathcal G}$, and the fixed vector is unique exactly when $\abs{\mathcal G}$ is the total dimension $\dim\mathcal H$.

\begin{applemma}[Four parties with the entropies of a GHZ state]\label{lem:K-ghz}
Let $\ket\Psi$ be a pure state of four systems, each of dimension $q$, with $S(i)=S(ij)=\log q$ for all distinct $i,j$. Then there are orthonormal bases $\{\ket{a^i_n}\}_{n=1}^q$ of the four systems such that $\ket\Psi=q^{-1/2}\sum_n\ket{a^1_n}\ket{a^2_n}\ket{a^3_n}\ket{a^4_n}$. They are unique up to phases and a common relabeling: for any two of the systems, say $1$ and $2$, the product vectors in the support of $\rho_{12}$ are exactly the multiples of the vectors $\ket{a^1_n}\ket{a^2_n}$.
\end{applemma}

\begin{proof}
Every $\rho_i$ is maximally mixed, and by purity $S(ijk)=S(l)=\log q$ whenever $\{i,j,k,l\}=\{1,2,3,4\}$. Hence, for distinct $i,j,k$, $I(i;j)=\log q$ and $I(i;j|k)=S(ik)+S(jk)-S(k)-S(ijk)=0$. By the theorem of Hayden, Jozsa, Petz and Winter~\cite{si:HJPW2004}, the latter gives a decomposition $\mathcal H_k=\bigoplus_m\mathcal H_{k^L_m}\otimes\mathcal H_{k^R_m}$ with $\rho_{ijk}=\bigoplus_mr_m\,\sigma^m_{ik^L_m}\otimes\tau^m_{k^R_mj}$ for states $\sigma^m$, $\tau^m$ and probabilities $r_m$. The state $\sum_mr_m\,\sigma^m_i\otimes\tau^m_j\otimes\ket m\bra m_M$ has the marginal $\rho_{ij}$, so
\begin{align*}
\log q=I(i;j)&\le I(i;jM)=I(i;M)\\
&=S(\rho_i)-\textstyle\sum_mr_mS(\sigma^m_i)\\
&=\log q-\textstyle\sum_mr_mS(\sigma^m_i).
\end{align*}
So every $\sigma^m_i$ is pure, $\sigma^m_{ik^L_m}=\ket{\alpha_m}\bra{\alpha_m}\otimes\sigma^m_{k^L_m}$, and $\rho_{ik}=\sum_mr_m\ket{\alpha_m}\bra{\alpha_m}\otimes\gamma_m$, where $\gamma_m$ is supported on the $m$-th summand of $\mathcal H_k$. Since $\rho_k=\sum_mr_m\gamma_m$ is maximally mixed, $r_m\gamma_m=\Pi_m/q$, where $\Pi_m$ is the projector onto that summand. For $(i,j,k)=(1,2,3)$ and for $(3,2,1)$ this gives
\[
\rho_{13}=\frac1q\sum_m\ket{\alpha_m}\bra{\alpha_m}\otimes\Pi_m=\frac1q\sum_nQ_n\otimes\ket{\beta_n}\bra{\beta_n},
\]
where the $\Pi_m$ are mutually orthogonal projectors on $\mathcal H_3$ with sum $1$, and so are the $Q_n$ on $\mathcal H_1$. Both sums are $1/q$ times projectors, so these projectors have the same range. If some $Q_n$ had rank at least $2$, its range would contain a vector $x$ that is not a multiple of any $\ket{\alpha_m}$; then $x\otimes\ket{\beta_n}$ lies in the range of $\sum_m\ket{\alpha_m}\bra{\alpha_m}\otimes\Pi_m$, and applying $1\otimes\Pi_m$ shows $\Pi_m\ket{\beta_n}=0$ for every $m$, which is impossible. So each $Q_n$ is the projector onto a vector $\ket{a^1_n}$; these form an orthonormal basis of $\mathcal H_1$, and since $\rho_3=\frac1q\sum_n\ket{\beta_n}\bra{\beta_n}$ is maximally mixed, the vectors $\ket{a^3_n}=\ket{\beta_n}$ form one of $\mathcal H_3$. In the same way, with $(2,1,4)$ and $(4,1,2)$, $\rho_{24}=\frac1q\sum_k\ket{a^2_ka^4_k}\bra{a^2_ka^4_k}$ with orthonormal bases $\{\ket{a^2_k}\}$ and $\{\ket{a^4_k}\}$. Hence $\ket\Psi=q^{-1/2}\sum_n\ket{a^1_n}\ket{a^3_n}\ket{\phi_n}_{24}$ with orthonormal vectors $\ket{\phi_n}$ in the support of $\rho_{24}$, that is, $\ket{\phi_n}=\sum_kV_{nk}\ket{a^2_ka^4_k}$ for a unitary matrix $V$. Then $\rho_{12}=\frac1q\sum_{n,k}\abs{V_{nk}}^2\ket{a^1_na^2_k}\bra{a^1_na^2_k}$, whose entropy is $\log q+\frac1q\sum_nH(\abs{V_{n1}}^2,\dots,\abs{V_{nq}}^2)$. It equals $\log q$ only if every row of $V$ has a single nonzero entry; relabeling the $\ket{a^2_k}$ and $\ket{a^4_k}$ and absorbing phases gives the first claim. For the second, if $x\otimes y=\sum_nc_n\ket{a^1_n}\ket{a^2_n}$, then $\langle a^2_n|y\rangle\,x=c_n\ket{a^1_n}$ for every $n$, so $\langle a^2_n|y\rangle\ne0$ for exactly one $n$, and $x$ and $y$ are multiples of $\ket{a^1_n}$ and $\ket{a^2_n}$. Applied to the pairs $\{1,2\}$, $\{1,3\}$ and $\{1,4\}$, this shows that the bases of any other such decomposition are those of this one up to phases and a relabeling common to the four systems.
\end{proof}

\begin{applemma}[Commuting unitary bases]\label{lem:K-nice}
Let $U_1,\dots,U_{D^2}$ be unitaries on $\mathbb C^D$ with $\operatorname{Tr}(U_j^\dagger U_k)=D\,\delta_{jk}$ that commute pairwise up to scalars, and write $U_kU_lU_k^{-1}=\beta(k,l)U_l$.
\begin{enumerate}
\item[(a)] If an operator $Y\ne0$ satisfies $U_kYU_k^{-1}=\lambda_kY$ with scalars $\lambda_k$ for all $k$, then $Y$ is a scalar multiple of the unique $U_l$ with $\beta(k,l)=\lambda_k$ for all $k$.
\item[(b)] The scalar multiples of the $U_j$ form a group, so the $U_j$ form a nice basis indexed by an abelian group of order $D^2$, and $\beta$ is its commutator form.
\end{enumerate}
\end{applemma}

\begin{proof}
The $U_l$ form a basis of the operators on $\mathbb C^D$, and $U_l$ is an eigenvector of every map $Y\mapsto U_kYU_k^{-1}$, with eigenvalue $\beta(k,l)$. The tuples $(\beta(k,l))_k$ are distinct for distinct $l$: if $U_l$ and $U_{l'}$ had the same tuple, $U_l^{-1}U_{l'}$ would commute with every $U_k$, hence with every operator, and would be a scalar, contradicting $\operatorname{Tr}(U_l^\dagger U_{l'})=0$. For (a), write $Y=\sum_ly_lU_l$; then $y_l\ne0$ implies $\beta(k,l)=\lambda_k$ for all $k$, which holds for at most one $l$. For (b), the identity, the products $U_jU_k$ and the inverses $U_j^{-1}$ satisfy the hypothesis of (a), so each is a scalar multiple of some $U_l$; the scalar multiples of the $U_j$ thus form a group, whose quotient by the scalars has order $D^2$ and is abelian.
\end{proof}

\begin{applemma}[Symplectic bases]\label{lem:K-symp}
Let $\beta$ be a nondegenerate bicharacter of a finite abelian group $G$, written additively, with $\beta(g,g)=1$ for all $g\in G$. Then $G$ has elements $g_1,h_1,\dots,g_r,h_r$ such that $G=\bigoplus_{i=1}^r\bigl(\langle g_i\rangle\oplus\langle h_i\rangle\bigr)$, the elements $g_i$ and $h_i$ have the same order $n_i$, a prime power, $\beta(g_i,h_i)$ is a primitive $n_i$-th root of unity, and $\beta(x,y)=1$ whenever $x$ and $y$ lie in different summands $\langle g_i\rangle\oplus\langle h_i\rangle$.
\end{applemma}

\begin{proof}
Expanding $\beta(g+h,g+h)=1$ gives $\beta(h,g)=\beta(g,h)^{-1}$. If $g$ and $h$ have coprime orders, then $\beta(g,h)$ is a root of unity whose order divides both, so $\beta(g,h)=1$; hence the Sylow subgroups of $G$ pair trivially with each other, $\beta$ is nondegenerate on each of them, and it suffices to treat an abelian $p$-group $G\ne0$. Let $g\in G$ have the largest order, $p^k$. The map $x\mapsto\beta(x,\cdot)$ from $G$ to its group of characters is injective, hence bijective, so the character $\beta(g,\cdot)$ has order $p^k$; its image is a cyclic group of order $p^k$, so $\zeta=\beta(g,h)$ is a primitive $p^k$-th root of unity for some $h\in G$. Since $\beta(g,h)^m=\beta(g,mh)$, $p^k$ divides the order of $h$, which is therefore $p^k$. If $ag=bh$ for integers $a$ and $b$, then $1=\beta(g,ag)=\beta(g,bh)=\zeta^b$, so $bh=0$; hence $H=\langle g\rangle+\langle h\rangle$ is the direct sum $\langle g\rangle\oplus\langle h\rangle$, and $\beta(ag+bh,a'g+b'h)=\zeta^{ab'-ba'}$, so $\beta$ is nondegenerate on $H$. Let $H^\perp=\{x\in G:\beta(x,y)=1\text{ for all }y\in H\}$, the kernel of the homomorphism that restricts the characters $\beta(x,\cdot)$ to $H$, so that $\abs{H^\perp}\ge\abs G/\abs H$. Since $H\cap H^\perp=0$, $G=H\oplus H^\perp$, and $\beta$ is nondegenerate on $H^\perp$, because an element of $H^\perp$ that pairs trivially with $H^\perp$ pairs trivially with $G$. Induction on $\abs G$ completes the proof.
\end{proof}

\begin{proof}[Proof of Theorem~\ref{thm:ideal}]
\emph{Reduction.} Let $\Gamma$ be $\Gamma_{\mathrm F}$ or $\Gamma_{\mathrm{NF}}$, and let $\mathcal M$ be $\mathrm{AG}(3,2)$ or $R_8$, respectively, the matroid whose port at $R$ is $\Gamma$ (Supplementary Note~\ref{app:qss-access}). Consider a scheme for $\Gamma$ whose secret and seven shares all have dimension $D$, so that $S(R)=\log D$ and $S(P)\le\log D$ for every player $P$. Every player $P$ lies in a minimal authorized set $T$, so $2S(R)=I(R;T)-I(R;T\setminus P)=I(R;P|T\setminus P)\le2S(P)$, since $\abs{S(P|Y)}\le S(P)$ for every set $Y$, and $S(P)=\log D$. Let $h(X)=S(X)+\sum_{i\in X}S(i)$ be the lift of the entropy vector of the joint state of $R$ and the players, which may be mixed. It is a polymatroid: $h(\emptyset)=0$, $h$ is nondecreasing because $S(Xi)\ge S(X)-S(i)$ (Araki--Lieb~\cite{si:ArakiLieb1970}), and submodular because $S(Xi)+S(Xj)\ge S(X)+S(Xij)$ (strong subadditivity~\cite{si:LiebRuskai1973}). Equation~\eqref{eq:qss-lift} holds for every state, so the access conditions give $h(RT)=h(T)$ for every authorized set $T$ and $h(RT)=h(T)+2\log D$ for every other nonempty set $T$ of players.

\emph{The entropy vector and purity.} Put $g=h/(2\log D)$, so that $g(R)=g(P)=1$ for every player $P$, and $g(RT)-g(T)$ is $0$ for every authorized set $T$ of players and $1=g(R)$ for every other set $T$ of players, the empty set included. Every player $P$ lies in a minimal authorized set $T$, and then $g(R(T\setminus P))-g(T\setminus P)=g(R)$. In the terminology of Mat\'u\v{s}~\cite[Sec.~4]{si:Matus2012}, $R$ is therefore an ideal element of the polymatroid $g$: it is perfect, every player is essential for it, and $g(P)=g(R)$. By his Theorem~2, $g=t\,r$ for the rank function $r$ of a matroid on the eight systems and some $t>0$, and $g(R)=1$ gives $r(R)=1$ and $t=1$. So a set $T$ of players is authorized exactly when $r(RT)=r(T)$, and the circuits of this matroid that contain $R$ are the sets $\{R\}\cup T$ for the minimal authorized sets $T$ of $\Gamma$; since $\Gamma$ is the port of $\mathcal M$ at $R$, these are also the circuits of $\mathcal M$ through $R$. Every system lies in one of them, so both matroids are connected, and a connected matroid is determined by its circuits through one element (Ref.~\cite{si:Lehman1964}, Lemma~(46) on p.~721, whose only hypothesis is that every circuit lies in the union of the circuits through the element); hence $g=r_{\mathcal M}$, the rank function of $\mathcal M$. In particular $g(R\mathcal Q)=4$, and since $g(R\mathcal Q)=S(R\mathcal Q)/(2\log D)+4$, $S(R\mathcal Q)=0$: the joint state of $R$ and the players is a pure state $\ket\Psi$, and $S(X)=(2r_{\mathcal M}(X)-\abs X)\log D$ for every set $X$ of the eight systems. For pure states, Theorem~3 of Mat\'u\v{s}, a quantum counterpart of the theorem of Brickell and Davenport that the access structures of ideal classical schemes are matroid ports~\cite{si:BrickellDavenport1991}, gives $S=t\,r^\vee$, where $r^\vee(X)=r(X)-\frac12\sum_{i\in X}r(i)$, for a tight self-dual matroid $r$ directly, and Lehman's theorem then identifies $r$ with $r_{\mathcal M}$ as above. This part of the proof uses no computer; Lemma~\ref{lem:K-forcing}, applied to $g$, gives the same conclusion by certified linear programs.

Identify the eight systems with the points of $\GF(2)^3$, with $R$ at the origin, as in Sec.~\ref{sec:qss}. The planes of $\mathrm{AG}(3,2)$ are the fourteen sets of four points with zero sum. The planes of $R_8$, the sets of four coplanar vertices of the real cube, are the twelve of them other than the two tetrahedra $\{R,W,X,Y\}$ and $\{A,B,C,Z\}$: the six faces and the six diagonal rectangles. In both matroids every set of at most three points is independent. Hence (A) $\rho_X$ is maximally mixed if $\abs X\le3$, and (B) $S(\Pi)=2\log D$ for every plane $\Pi$ of $\mathcal M$. For a direction $\delta\ne0$ the eight systems form four \emph{$\delta$-pairs} $\{v,v+\delta\}$; we call $\delta$ \emph{good} if any two $\delta$-pairs form a plane of $\mathcal M$. Every direction is good for $\mathrm{AG}(3,2)$. For $R_8$ the good directions are the directions $100$, $010$ and $001$ of the edges of the cube and the direction $111$ of its body diagonals; for each of the directions $110$, $101$ and $011$ of the face diagonals, two of the six unions of two pairs are the two tetrahedra. A plane $\{v,v+\delta,v+\gamma,v+\delta+\gamma\}$ is split into $\delta$-pairs, into $\gamma$-pairs and into $(\delta+\gamma)$-pairs, and the following facts are checked by listing the planes. (G1)~Every plane of $\mathcal M$ is split into pairs along at least two good directions: along all three for $\mathrm{AG}(3,2)$; for $R_8$, along exactly two, a face along the directions of two edges and a diagonal rectangle along that of one edge and that of the body diagonals. (G2)~For any two good directions $\delta\ne\gamma$ and any point $v$, $\{v,v+\delta,v+\gamma,v+\delta+\gamma\}$ is a plane of $\mathcal M$. (G3)~The good directions span $\GF(2)^3$. (G4)~For every good direction $\delta$, some plane of $\mathcal M$ contains no $\delta$-pair: for $R_8$, a face perpendicular to the edges along $\delta$, or any face if $\delta=111$. Fix a basis of every system and put $\ket{\Phi_{vw}}=D^{-1/2}\sum_x\ket x_v\ket x_w$; for every operator $Y$, $(1\otimes Y)\ket{\Phi_{vw}}=(Y^{\mathsf T}\otimes1)\ket{\Phi_{vw}}$.

\emph{Step 1: each good direction gives a GHZ state.} Let $\delta$ be a good direction. By (A) and (B), Lemma~\ref{lem:K-ghz} applies to the four $\delta$-pairs, systems of dimension $q=D^2$: $\ket\Psi=D^{-1}\sum_{j=1}^{D^2}\bigotimes_P\ket{e^P_j}$, the tensor product over the $\delta$-pairs $P$, with orthonormal bases $\{\ket{e^P_j}\}_j$ of the pairs that are unique up to phases and a common relabeling. If $u$ lies in a $\delta$-pair $P'\ne P$, then $\rho_{P\cup\{u\}}=D^{-2}\sum_j\ket{e^P_j}\bra{e^P_j}\otimes\operatorname{Tr}_{P'\setminus u}\ket{e^{P'}_j}\bra{e^{P'}_j}$ is maximally mixed by (A), so every $\ket{e^{P'}_j}$ has maximally mixed marginals: it is maximally entangled. So for a $\delta$-pair $P=\{v,w\}$, written with $v$ first, $\ket{e^P_j}=(A_j\otimes1)\ket{\Phi_{vw}}$ with unitaries $A_j$ on $\mathcal H_v$ such that $\operatorname{Tr}(A_j^\dagger A_k)=D\delta_{jk}$. Let $\mathcal P_\delta(v)$ be the set of scalar multiples of the $A_jA_1^{-1}$, the unitaries that map $\ket{e^P_1}$ to the other $\ket{e^P_j}$ by acting on $v$.

\emph{Step 2: each plane carries a stabilizer group.} Let $\Pi=\{a,b,c,d\}$ (here $d$ is a system, not a local dimension) be a plane of $\mathcal M$ with $\delta$-pairs $\{a,b\}$ and $\{c,d\}$ and $\gamma$-pairs $\{a,c\}$ and $\{b,d\}$, for two good directions $\delta\ne\gamma$ (G1). Let $\ket{e_j}=(A_j\otimes1)\ket{\Phi_{ab}}$ and $\ket{f_j}=(B_j\otimes1)\ket{\Phi_{cd}}$ be the vectors of Step~1 for these $\delta$-pairs, and $\ket{g_k}$ and $\ket{h_k}$ those for the $\gamma$-pairs. By Step~1, $\rho_\Pi=D^{-2}\sum_j\ket{e_jf_j}\bra{e_jf_j}=D^{-2}\sum_k\ket{g_kh_k}\bra{g_kh_k}$, so the support $V$ of $\rho_\Pi$ is spanned by the $\ket{e_j}\ket{f_j}$ and also by the $\ket{g_k}\ket{h_k}$. Let $\Lambda$ be the linear isomorphism from $\mathcal H_a\otimes\mathcal H_b\otimes\mathcal H_c\otimes\mathcal H_d$ to the operators from $\mathcal H_b\otimes\mathcal H_d$ to $\mathcal H_a\otimes\mathcal H_c$ defined on the fixed bases by $\Lambda(\ket x_a\ket y_b\ket z_c\ket w_d)=\ket x_a\ket z_c\bra y_b\bra w_d$ and extended linearly, so that a vector on $b$ and $d$ enters as the bra of the vector with the complex conjugate coordinates. Then $\Lambda(\ket{e_j}\ket{f_j})=D^{-1}A_j\otimes B_j$, where $\mathcal H_b$ and $\mathcal H_d$ are identified with $\mathcal H_a$ and $\mathcal H_c$ through the fixed bases, and $\Lambda(\ket{g_k}\ket{h_k})=\ket{g_k}\bra{\bar h_k}$, where $\ket{\bar h_k}$ has the complex conjugate coordinates. Hence $\Lambda(V)=\{Q\Delta Q'^\dagger:\Delta\ \text{diagonal}\}$ for the unitaries $Q=\sum_k\ket{g_k}\bra k$ and $Q'=\sum_k\ket{\bar h_k}\bra k$, and $A_j\otimes B_j=Q\Delta_jQ'^\dagger$ with $\Delta_j$ diagonal. So the operators $(A_j\otimes B_j)(A_1\otimes B_1)^{-1}=Q\Delta_j\Delta_1^{-1}Q^\dagger$ commute pairwise. With $U_j=A_jA_1^{-1}$ and $U'_j=B_jB_1^{-1}$ this reads $U_jU_k\otimes U'_jU'_k=U_kU_j\otimes U'_kU'_j$, so $U_jU_k=\lambda U_kU_j$ and $U'_jU'_k=\lambda^{-1}U'_kU'_j$ for a scalar $\lambda=\lambda_{jk}$. By Lemma~\ref{lem:K-nice}, the $U_j$ form a nice basis indexed by an abelian group $G_\Pi$, written additively, with commutator form $\beta$, and $\mathcal P_\delta(a)$ is the group of scalar multiples of the $U_j$; the $U'_j$ form one with commutator form $\beta^{-1}$, and with the same multiplication table: by Lemma~\ref{lem:K-nice}~(a), $U'_jU'_k$ is a scalar multiple of $U'_l$ for the $l$ with $U_jU_k\in\mathbb CU_l$, since both have the eigenvalues $\beta(m,l)^{-1}$ under $Y\mapsto U'_mYU'^{-1}_m$. For $k\in G_\Pi$ let
\[
M_k=\bigl(A_1^{-1}U_k^{-1}A_1\bigr)^{\mathsf T},\qquad M'_k=\bigl(B_1^{-1}U'^{-1}_kB_1\bigr)^{\mathsf T},
\]
on $\mathcal H_b$ and $\mathcal H_d$. Then $(U_k\otimes M_k)\ket{e_j}=(U_kU_jU_k^{-1}A_1\otimes1)\ket{\Phi_{ab}}=\beta(k,j)\ket{e_j}$ and $(U'_k\otimes M'_k)\ket{f_j}=\beta(k,j)^{-1}\ket{f_j}$, so
\[
S_k=U_k\otimes M_k\otimes U'_k\otimes M'_k,
\]
on $a,b,c,d$, fixes every $\ket{e_j}\ket{f_j}$, hence every vector of $V$, and $S_k\ket\Psi=\ket\Psi$. If $U_kU_l=cU_{k+l}$ and $U'_kU'_l=c'U'_{k+l}$, then $M_kM_l=c^{-1}M_{k+l}$ and $M'_kM'_l=c'^{-1}M'_{k+l}$, so $S_kS_l=S_{k+l}$, and $\mathcal S_\Pi=\{S_k\}$ is a group of order $D^2$. For $k\ne0$ none of the four factors of $S_k$ is a scalar. The factors lie in $\mathcal P_\delta(a)$, $\mathcal P_\delta(b)$, $\mathcal P_\delta(c)$ and $\mathcal P_\delta(d)$, as defined in Step~1: for $b$, written first, $\ket{e_j}=(A_j^{\mathsf T}\otimes1)\ket{\Phi_{ba}}$ and $A_j^{\mathsf T}(A_1^{\mathsf T})^{-1}=(A_1^{-1}U_jA_1)^{\mathsf T}$. By (G2), every pair $\{v,v+\delta\}$ along a good direction $\delta$ lies in such a plane, together with a second good direction, so every $\mathcal P_\delta(v)$ is a group, the scalar multiples of a nice basis; in particular, it does not depend on the phases of the vectors of Step~1 or on which of them is labeled $1$.

\emph{Step 3: one group for each system.} In the notation of Step~2, the operators $N_k=U_k\otimes U'_k$ on $\mathcal H_a\otimes\mathcal H_c$ commute pairwise and satisfy $\operatorname{Tr}(N_k^\dagger N_l)=D^2\delta_{kl}$, so they span a commutative algebra of dimension $D^2$; their common eigenvectors form an orthonormal basis $\{\ket{\xi_\chi}\}$ of $\mathcal H_a\otimes\mathcal H_c$, with $N_k\ket{\xi_\chi}=\chi(k)\ket{\xi_\chi}$ and distinct eigenvalue tuples $\chi$. In the same way the operators $K_k=M_k\otimes M'_k$ on $\mathcal H_b\otimes\mathcal H_d$, which commute pairwise because $M_kM_l=\beta(k,l)^{-1}M_lM_k$ and $M'_kM'_l=\beta(k,l)M'_lM'_k$, have common eigenspaces of dimension $1$. Since $S_k$ is $N_k\otimes K_k$ up to the order of the factors and fixes $V$, every vector of $V$, written as $\sum_\chi\ket{\xi_\chi}\ket{w_\chi}$, has $K_k\ket{w_\chi}=\chi(k)^{-1}\ket{w_\chi}$ for all $k$. The common eigenspace of the $K_k$ with these eigenvalues has dimension at most $1$; so $\ket{w_\chi}$ is a multiple of a fixed unit vector $\ket{\zeta_\chi}$ if this eigenspace is not zero, and $\ket{w_\chi}=0$ otherwise. Hence $V$ lies in the span of the $\ket{\xi_\chi}\ket{\zeta_\chi}$ over the $\chi$ whose eigenspace is not zero, a space of dimension at most $D^2$; since $V$ has dimension $D^2$, all $D^2$ of these eigenspaces are nonzero and $V=\operatorname{span}\{\ket{\xi_\chi}\ket{\zeta_\chi}\}$. The product vectors in $V$, for the split of $\Pi$ into $\{a,c\}$ and $\{b,d\}$, are then the multiples of the $\ket{\xi_\chi}\ket{\zeta_\chi}$, so by the uniqueness in Lemma~\ref{lem:K-ghz} the $\ket{g_k}$ are the $\ket{\xi_\chi}$, up to phases and order. Write $\ket{\xi_\chi}=(C_\chi\otimes1)\ket{\Phi_{ac}}$. From $N_k\ket{\xi_\chi}=\chi(k)\ket{\xi_\chi}$ we get $U_kC_\chi U'^{\mathsf T}_k=\chi(k)C_\chi$, hence $U_kC_\chi C_{\chi'}^{-1}U_k^{-1}=\chi(k)\chi'(k)^{-1}C_\chi C_{\chi'}^{-1}$, and by Lemma~\ref{lem:K-nice}~(a) every $C_\chi C_{\chi'}^{-1}$ is a scalar multiple of some $U_l$. So $\mathcal P_\gamma(a)\subseteq\mathcal P_\delta(a)$; both consist of the scalar multiples of $D^2$ pairwise orthogonal unitaries, so they are equal. Since any two good directions lie in a plane of $\mathcal M$ through $a$ (G2), the group $\mathcal P(a)=\mathcal P_\delta(a)$ does not depend on the good direction $\delta$. By Step~2, the groups $\mathcal P(a)$ and $\mathcal P(c)$ have, up to scalars, the same multiplication table whenever $a+c$ is a good direction, and by (G3) any two systems are joined by a chain of such steps; so the groups $\mathcal P(v)/\mathbb C^\times$ are all isomorphic to one abelian group $G$ of order $D^2$.

\emph{Step 4: the state is a stabilizer state.} Let $\mathcal S$ be the group generated by the groups $\mathcal S_\Pi$ of Step~2, for all planes of $\mathcal M$ and all choices of good directions, vectors and labels. Its elements are tensor products of elements of the groups $\mathcal P(v)$, and they fix $\ket\Psi$. Any two of them commute up to a scalar, which fixes $\ket\Psi$ and so equals $1$; hence $\mathcal S$ is abelian and contains no scalar other than $1$. It is finite: an element of $\mathcal S$ is determined by its factors up to scalars, since two elements with the same factors up to scalars differ by a scalar in $\mathcal S$, so $\abs{\mathcal S}\le\prod_v\abs{\mathcal P(v)/\mathbb C^\times}=D^{16}$. Fix a good direction $\delta$, its pairs $P_1,\dots,P_4$, and the planes $\Pi_j=P_1\cup P_j$, $j=2,3,4$, and choose the vectors and labels of Steps~1 and~2 once for each pair. On $P_j$ the elements of $\mathcal S_{\Pi_j}$ act as the operators $U'_k\otimes M'_k$ built from the vectors of $P_j$, and on $P_1$ as the operators $U_k\otimes M_k$ built from the vectors and labels of $P_1$ alone, which are therefore the same for $j=2,3,4$; since $(U_k\otimes M_k)(U_l\otimes M_l)=U_{k+l}\otimes M_{k+l}$, these form a group, and each of them is a scalar on both systems of the pair or on neither. If $s_2s_3s_4=1$ with $s_j\in\mathcal S_{\Pi_j}$, restriction to $P_j$ shows $s_j=1$; so $\mathcal S^\delta=\mathcal S_{\Pi_2}\mathcal S_{\Pi_3}\mathcal S_{\Pi_4}$ has $D^6$ elements, each acting on every $\delta$-pair as a scalar on both systems or on neither. Let $\Pi$ be a plane of $\mathcal M$ that does not contain a $\delta$-pair (G4); it meets each $\delta$-pair in one point. An element of $\mathcal S_\Pi$ other than $1$ acts as a non-scalar exactly on the four systems of $\Pi$, so $\mathcal S_\Pi\cap\mathcal S^\delta=\{1\}$ and $\abs{\mathcal S}\ge D^8$. On the other hand, every element of $\mathcal S$ other than $1$ has a factor in some $\mathcal P(v)$ that is not a scalar, which is traceless because it is orthogonal to the scalars in $\mathcal P(v)$; so the projector $\abs{\mathcal S}^{-1}\sum_{s\in\mathcal S}s$ onto the common fixed space of $\mathcal S$ has trace $D^8/\abs{\mathcal S}$, which is at least $1$ because the space contains $\ket\Psi$. Hence $\abs{\mathcal S}=D^8$ and $\ket\Psi\bra\Psi=D^{-8}\sum_{s\in\mathcal S}s$.

\emph{Step 5: Weyl operators.} At each system $v$, the commutator form $\beta_v$ of $\mathcal P(v)$ is a nondegenerate bicharacter of $G$ with $\beta_v(g,g)=1$. Take $g_i$, $h_i$ and $n_i$ as in Lemma~\ref{lem:K-symp} for $\beta_v$; since $G\cong\bigoplus_i\Z_{n_i}^2$, the $n_i$ are determined by $G$ up to order, and $\prod_in_i=D$. Let $u_i,w_i\in\mathcal P(v)$ be unitaries indexed by $g_i$ and $h_i$; since $u_i^{n_i}$ and $w_i^{n_i}$ are scalars, we may rescale them so that $u_i^{n_i}=w_i^{n_i}=1$. Then $u_iw_i=\zeta_iw_iu_i$ with $\zeta_i=\beta_v(g_i,h_i)$, a primitive $n_i$-th root of unity, and $u_i$ and $w_i$ commute with $u_j$ and $w_j$ for $j\ne i$. The $n_i^2$ unitaries $u_i^aw_i^b$, $0\le a,b<n_i$, lie in $\mathcal P(v)$ and are indexed by distinct elements of $G$, so they are pairwise orthogonal, and they span an algebra $\mathcal A_i$ of dimension $n_i^2$ that is closed under multiplication and under adjoints. Its center is trivial: conjugation by $u_i$ and by $w_i$ multiplies $u_i^aw_i^b$ by $\zeta_i^b$ and by $\zeta_i^{-a}$, so a central element has no component on $u_i^aw_i^b$ unless $a=b=0$. So $\mathcal A_i$ is a full matrix algebra $M_{n_i}$, and these algebras commute pairwise. A full matrix subalgebra $\mathcal A\cong M_k$ of $M_D$ splits $\mathbb C^D$ as $\mathbb C^k\otimes\mathbb C^{D/k}$ with $\mathcal A=M_k\otimes1$, and its commutant $1\otimes M_{D/k}$ contains the algebras that commute with it; by induction, and since $\prod_in_i=D$, $\mathbb C^D=\bigotimes_i\mathbb C^{n_i}$, with $u_i$ and $w_i$ acting on the $i$-th factor alone. On that factor, $w_i$ has eigenvalues among the $n_i$-th roots of unity, and $u_i$ maps each eigenspace of $w_i$ onto another, permuting the eigenvalues $\mu\,\zeta_i^{-x}$ cyclically; hence these eigenspaces have dimension $1$, and in the basis $u_i^x\ket{\eta_i}$, $x\in\Z_{n_i}$, for an eigenvector $\ket{\eta_i}$ of $w_i$, the operators $u_i$ and $w_i$ are, up to scalars, the Weyl operators $X$ and $Z^{c_i}$ for a unit $c_i$ of $\Z_{n_i}$. Every element of $\mathcal P(v)$ is a scalar multiple of a product of powers of the $u_i$ and $w_i$, since the $g_i$ and $h_i$ generate $G$. Hence there are unitaries $T_v$ such that $T_v^\dagger\mathcal P(v)T_v$ consists of the scalar multiples of the Weyl operators of qudits of dimensions $n_1,\dots,n_r$, and $(\bigotimes_vT_v^\dagger)\ket\Psi$ is fixed by an abelian group of $D^8$ tensor products of such Weyl operators times scalars, without scalars other than $1$. So, up to local unitaries, $\ket\Psi$ is a stabilizer state of $8r$ qudits, $r$ on each system, of dimensions $n_1,\dots,n_r$.

\emph{Step 6: the dimension.} Identify $G$ with $\bigoplus_i\Z_{n_i}^2$, the phase space of the qudits of one system, and write each element of the group $\mathcal S$ of Step~5, after the local unitaries $T_v$, as a scalar times a tensor product of Weyl operators $W(x_v)$, $x_v\in G$, one on each system $v$. Since $W(x)W(y)$ is a scalar multiple of $W(x+y)$ and $\mathcal S$ contains no scalar other than $1$, the map that sends such an element to $(x_v)_v$ is an injective homomorphism; let $L\le\bigoplus_vG$ be its image, so $\abs L=D^8$. For a set $X$ of systems, let $\mathcal S_X$ be the subgroup of the elements whose Weyl operators outside $X$ are the identity. Weyl operators other than scalars are traceless, so tracing $\ket\Psi\bra\Psi=D^{-8}\sum_{s\in\mathcal S}s$ over the complement of $X$ gives $\rho_X=D^{-\abs X}\sum_{s\in\mathcal S_X}s|_X$, where $s|_X$ is the factor of $s$ on $X$; as in the definition of stabilizer states, this is $\abs{\mathcal S_X}D^{-\abs X}$ times a projector of rank $D^{\abs X}/\abs{\mathcal S_X}$, so $S(X)=\abs X\log D-\log\abs{\mathcal S_X}$. Since $\mathcal S_X$ corresponds to the kernel of the projection of $L$ onto the complement $\bar X$ of $X$, $\abs{\mathcal S_X}=D^8/\abs{\pr_{\bar X}L}$, and by purity
\begin{align*}
  \log\abs{\pr_{\bar X}L}&=S(\bar X)+\abs{\bar X}\log D=h(\bar X)\\
  &=2r_{\mathcal M}(\bar X)\log D=r_{\mathcal M}(\bar X)\log\abs G .
\end{align*}
So Lemma~\ref{lem:K-two} applies to $L$ with $V=V_R=G$, since $r_{\mathcal M}(R)=1$. For $\Gamma_{\mathrm F}$ it gives $2x=0$ for every $x\in G$, so every $n_i$ is $2$: up to the local unitaries $T_v$, $\ket\Psi$ is a stabilizer state of qubits, $r$ on each system, and $D=2^r$. For $\Gamma_{\mathrm{NF}}$ it gives that $\abs G=D^2$ is odd, so $D$ and every $n_i$ are odd: up to the local unitaries $T_v$, $\ket\Psi$ is a stabilizer state of qudits of the odd dimensions $n_1,\dots,n_r$ on each system.

Conversely, if $D$ is odd, the cube state over $\Z_D$ is a pure scheme for $\Gamma_{\mathrm{NF}}$ whose secret and shares have dimension $D$ (proof of Proposition~\ref{prop:qss}~(iv) in Supplementary Note~\ref{app:qss-ideal}); and if $D=2^k$, then $k$ copies of the qubit cube state form a pure scheme for $\Gamma_{\mathrm F}$ whose secret and shares have dimension $D$, since a tensor product of schemes for the same access structure is a scheme for it (Supplementary Note~\ref{app:qss-contract}).
\end{proof}

Steps~1--6 use the scheme only through the pure state $\ket\Psi$ of eight systems of dimension $D$ and its entropies $S(X)=(2r_{\mathcal M}(X)-\abs X)\log D$. They therefore prove an entropic form of the rigidity: every pure state of eight systems of dimension $D$ with these entropies is, up to a unitary on each system, a stabilizer state, of qubits for $\mathcal M=\mathrm{AG}(3,2)$ and of qudits of odd dimensions for $\mathcal M=R_8$, and such a state exists if and only if $D$ is a power of $2$, respectively odd. With $R$ as the reference, such a state is an ideal scheme for $\Gamma$, since these entropies give $I(R;T)\in\{0,2\log D\}$ for every set $T$ of players.

\emph{Codes that correct at least these erasures.} Let a code with one system of dimension $D$ per site encode one system of dimension $D$, and let it correct the erasure of every set of sites that is not in $\Gamma$, for $\Gamma=\Gamma_{\mathrm F}$ or $\Gamma_{\mathrm{NF}}$. With the reference, the code is a pure state, so $I(R;T)+I(R;\bar T)=2S(R)$ for every set $T$ of sites and its complement $\bar T$; the erasure of $T$ is correctable exactly when $I(R;\bar T)=2S(R)$, that is, when $I(R;T)=0$. Both access structures are self-dual: a set $T$ is in $\Gamma$ exactly when $\bar T$ is not, since they are the ports of the identically self-dual matroids $\mathrm{AG}(3,2)$ and $R_8$ (for $\Gamma_{\mathrm F}$, a set of points of the Fano plane contains no line exactly when its complement contains one, because any two lines meet and every set that meets all lines contains one). Hence every $T\notin\Gamma$ has $I(R;T)=0$, and every $T\in\Gamma$, whose complement is not in $\Gamma$, has $I(R;T)=2S(R)$: the code has exactly the recovery pattern $\Gamma$, and Theorem~\ref{thm:ideal} applies to it.

\emph{Approximate schemes in the excluded dimensions.} For a state $\rho$ of $R$ and the seven players, each of dimension $D$, let
\begin{align*}
  \varepsilon_\Gamma(\rho)={}&\abs{S(R)-\log D}+\sum_{T\in\Gamma}\bigl(2S(R)-I(R;T)\bigr)\\
  &+\sum_{T\notin\Gamma}I(R;T),
\end{align*}
the sums running over the sets $T$ of players; every term is nonnegative, since $0\le I(R;T)\le2S(R)$, and $\varepsilon_\Gamma(\rho)=0$ exactly when $\rho$ is a scheme for $\Gamma$ whose secret has dimension $D$. The von Neumann entropy is continuous on the compact set of states on these eight systems, so $\varepsilon_\Gamma$ attains its minimum, and by Theorem~\ref{thm:ideal} this minimum is positive whenever $D$ is not a power of $2$ (for $\Gamma_{\mathrm F}$) or $D$ is even (for $\Gamma_{\mathrm{NF}}$). So in these dimensions every state violates the access conditions by at least a positive amount $\varepsilon_D$ that depends only on $D$ and $\Gamma$; the argument gives no value of $\varepsilon_D$.

For $D=2^k$ the theorem says that every ideal scheme for $\Gamma_{\mathrm F}$, and so every code that encodes $k$ qubits with Steane's recovery pattern and $k$ qubits per site, is, up to local unitaries, a stabilizer state of qubits; $k$ copies of the qubit cube state are an example. In particular, for $D=4$ no ideal scheme is, up to local unitaries, a stabilizer state of one $\Z_4$ qudit per system, in line with Proposition~\ref{prop:qss}~(iv). Likewise, for every ideal scheme for $\Gamma_{\mathrm{NF}}$, and so for every code that encodes one qudit of odd dimension $D$ with the recovery pattern of the seven-qudit color code (Supplementary Note~\ref{app:qss-color}) and one such qudit per site, the joint state with the reference is a stabilizer state up to local unitaries, of qudits of odd prime-power dimensions; and if $D$ is even but not a power of $2$, as for $D=6$, neither $\Gamma_{\mathrm F}$ nor $\Gamma_{\mathrm{NF}}$ has an ideal scheme. Such rigidity is a property of the access structure. The pure ideal schemes of the $((2,3))$ threshold structure, in which any two of three players recover the secret, are exactly the absolutely maximally entangled states of four systems of dimension $D$, with any of the systems as the reference: these entropies give $I(R;P)=0$ and $I(R;PP')=2\log D$ for distinct players $P$ and $P'$; conversely, in an ideal scheme every player has $S(P)=\log D$, as in the proof above, and $I(R;P)=0$ gives $S(RP)=2\log D$, which by purity is also the entropy of the other two players. For $D=3$ these states are unique up to local unitaries~\cite{si:Rather2023}, and the CSS state of a $[4,2,3]$ code over $\GF(3)$, such as the tetracode, is one of them, so there the structure is rigid. For $D=6$ they exist~\cite{si:Rather2022}, but none of them is a stabilizer state up to local unitaries in the sense of the proof above. Such a stabilizer state has, after a unitary on each system, a qubit and a qutrit on each system, and its labels split into their Sylow components, as in Supplementary Note~\ref{app:qss-contract}; both are isotropic, so their orders are at most $2^4$ and $3^4$, hence equal to them, and the state is the tensor product of a stabilizer state of one qubit per system and one of one qutrit per system. Its two-system entropies are then at most $2\log6$, with equality only if the qubit factor is absolutely maximally entangled, and no state of four qubits is~\cite{si:Higuchi2000}. So the code of the $36$ officers~\cite{si:Rather2022}, which has the recovery pattern of this structure, is not a stabilizer code up to local unitaries, and no stabilizer code that encodes one qudit of dimension $6$ with one such qudit per site has that pattern. Steps~1--5 of the proof above do not apply to absolutely maximally entangled states of four systems, since their three-system marginals have entropy $\log D$, not the $3\log D$ of a maximally mixed state of three systems. The script \path{checks/check_rev_qss.py}, part (s), checks that the CSS state of the tetracode is absolutely maximally entangled and that no graph state of four qubits, and none of four $\Z_6$ qudits, is.

\emph{Uniqueness for qubits.} For $D=2$ the theorem can be sharpened: every code that encodes one qubit into seven qubits and corrects every erasure that Steane's code corrects is, up to a unitary on each qubit, Steane's code. Indeed, by the proof above, after a unitary on each system the joint state of $R$ and the sites is a stabilizer state of one qubit per system, whose labels form a subgroup $L$ of $V^8$, $V=\GF(2)^2$, with $\abs{\pr_XL}=\abs V^{r_{\mathcal M}(X)}$ for every set $X$ of systems, $\mathcal M=\mathrm{AG}(3,2)$. By the proof of Lemma~\ref{lem:K-two}, after a group automorphism of $V$ at each site, the labels of the stabilizers that act trivially on $R$ form the group $L''=\{(u,v,w,u+v,u+w,v+w,u+v+w):u,v,w\in V\}$, which is the group of labels of the stabilizers of Steane's code, the $X$-type and $Z$-type operators supported on the codewords of the $[7,3]$ simplex code. Every group automorphism of $\GF(2)^2$ preserves its symplectic form, since $\mathrm{GL}(2,\GF(2))=\mathrm{SL}(2,\GF(2))=\mathrm{Sp}(2,\GF(2))$, so it is induced, up to signs, by conjugation with a single-qubit Clifford unitary. After these Clifford unitaries and a Pauli operator that fixes the signs, the stabilizers that act trivially on $R$ are those of Steane's code, and the code space, the support of the reduced state on the sites, is their common $+1$ eigenspace. For $D=2^k$ with $k\ge2$ the automorphisms of $\GF(2)^{2k}$ need not preserve the symplectic form, and we do not know whether the codes are unique up to local unitaries; nor do we know it for $\Gamma_{\mathrm{NF}}$.

In network coding, networks built on the Fano and non-Fano matroids show that linear network coding depends on the characteristic~\cite{si:DFZ2005insufficiency}. For secret sharing, the classical counterpart of the statements of Theorem~\ref{thm:ideal} on the dimension is due to Mat\'u\v{s}: the Fano matroid has representations by partitions only in degrees that are powers of $2$, and the non-Fano matroid only in odd degrees (Ref.~\cite{si:Matus1999}, Prop.~4.1 and its proof). These matroids are $\mathrm{AG}(3,2)$ and $R_8$ contracted at any point, and the representations by partitions of degree $D$ of a matroid describe the ideal classical schemes for its ports whose secret takes $D$ values (Refs.~\cite{si:BrickellDavenport1991} and~\cite[Sec.~5.3]{si:Matus1999}). Conditioning an ideal classical scheme for $\Gamma_{\mathrm F}$ on a value of the share of a player $q$ gives an ideal scheme for $\Gamma_{\mathrm F}/q$ with the same secret, since $q$ alone learns nothing about the secret, and $\Gamma_{\mathrm F}/q$ is a port of the Fano matroid, as $\Gamma_{\mathrm{NF}}/A$ is one of the non-Fano matroid (Supplementary Note~\ref{app:qss-contract}). With the cube codes over $\GF(2^k)$ and over $\Z_D$ for odd $D$, ideal classical schemes for $\Gamma_{\mathrm F}$ whose secret takes $D$ values therefore exist exactly when $D$ is a power of $2$, and for $\Gamma_{\mathrm{NF}}$ exactly when $D$ is odd. The proof above does not use this result, and Theorem~\ref{thm:ideal} is a quantum counterpart of both of its halves. For $\Gamma_{\mathrm{NF}}$ the proof needs only four of the seven directions of the cube: along the three directions of the face diagonals the four pairs of systems do not form a GHZ state of local dimension $D^2$, since two of the unions of two pairs are the tetrahedra $\{R,W,X,Y\}$ and $\{A,B,C,Z\}$, whose entropy is $4\log D$ rather than $2\log D$.
\section{Further discussion}
\label{app:discussion}

\emph{Composite local dimensions.} For composite local dimensions, by Eq.~\eqref{eq:union}, $\Stab_n=\Stab^{\mathrm{pr}}_n$ if and only if the self-dual slices of $\Lab_{n+1}$ and $\Lc_{n+1}$ coincide; up to the duality of finite abelian groups, $\Lab_m$ and $\Lc_m$ are the abelian and the mixed-linear entropy cones of Khazaei, who notes that the two agree for $m\le5$ and that it is not known whether they can differ~\cite{si:Khazaei2026}. For squarefree $d$ nothing new arises: a group of exponent dividing $d$ is the direct sum of its Sylow subgroups, which are vector spaces over prime fields, and every subgroup splits accordingly, so $\StabZ{d}{n}\subseteq\Stab^{\mathrm{pr}}_n\subseteq\QLR_n$. When $d$ is divisible by the square of a prime, we do not even know whether $\StabZ{d}{n}\subseteq\QLR_n$ for $n\ge6$: through the doubled codes of Sec.~\ref{sec:sandwich}, this inclusion would imply $\LcZ{d}{n}\subseteq\Lc_n$. The $\Z_4$ cube state lies outside the stabilizer cone of every prime local dimension, and the $\Z_4$ seven-point state outside the qubit cone, but both lie in the closed convex hull of the union of the prime cones (Supplementary Note~\ref{app:cubes-z4}). We do not know whether a stabilizer state of such a local dimension can have an entropy vector outside $\Stab^{\mathrm{pr}}_n$; such a state would give a subgroup rank function that violates a linear rank inequality valid over every field. If there is none, then, at the level of cones, composite local dimensions add nothing beyond tensor products of states of prime local dimensions.

\emph{Open cases at six parties.} The two six-party directions differ in reach: $\fnine$ holds for all odd local dimensions at once, so the Steane state lies outside the closed convex hull of the union of their cones, whereas the entropy vector of the non-Fano state, normalized to entropy $1$ on every system, is the limit of normalized entropy vectors of $\Z_{2^a}$ stabilizer states as $a\to\infty$ (Methods, after Theorem~\ref{thm:twopowers}), so that no single inequality separates it from the cones of all powers of two. At six parties, whether $\StabZ{d}{6}\subseteq\StabZ{d'}{6}$ for $d\nmid d'$ remains open unless $d$ is even and $d'$ odd (Theorem~\ref{thm:steane}) or $d'=2$ (Theorem~\ref{thm:reverse}), or more generally $d'$ is a power of $2$ (Theorem~\ref{thm:twopowers}). Open cases include two odd primes, and $\Z_4$ against $\Z_6$: the $\Z_4$ seven-point state, with entropy vector $\SF+\SNF$, lies in $\StabZ{2}{6}+\StabZ{3}{6}\subseteq\StabZ{6}{6}$. For two primes $p<q$, the cones coincide up to four parties and are incomparable from $2p+2$ parties on (Theorems~\ref{thm:steane}, \ref{thm:reverse} and~\ref{thm:primes}); whether $2p+1$ or fewer parties suffice is open, and for $p=2$ this is the five-party question of the Discussion. For that question, the six-variable linear rank inequalities would matter: they are not known completely, and Dougherty reported computations showing that there are more than one billion sharp ones, counting the forms obtained by permuting the variables~\cite{si:Dougherty2014}. We have not examined whether the seven-variable inequalities that the transfer arguments yield give new bounds on the rates of linear network codes~\cite{si:DFZ2015rate}.

\emph{A conjecture on common information.} Linden \emph{et al.} conjectured that ``the closure of the cone generated by the entropy vectors of stabiliser states is identical to that obtained when inequalities obtained from common information [\dots] are added to the classical ones''~\cite[Sec.~5]{si:LMRW2013}. If these inequalities are applied to $n$ of the $n+1$ systems, as in $\QLR_n$, the conjecture fails for five parties (Theorem~\ref{thm:five}): every five-variable linear rank inequality follows from common information~\cite{si:DFZ2009}, but~\eqref{eq:star} comes from the six-variable inequality~\eqref{eq:star-h}, which needs all six systems. If they are applied to the lift on all $n+1$ systems, they hold for subgroup arrangements and hence, by Proposition~\ref{thm:reduction}, for stabilizer states of every local dimension, so they cannot describe the cone of a single local dimension as soon as that of another is not contained in it: for qubits and every odd local dimension from six parties on (Theorems~\ref{thm:steane} and~\ref{thm:reverse}), and for every $d$ with enough parties (Theorem~\ref{thm:divisibility}). For the cone $\Stab_n$ of all local dimensions together, Eq.~\eqref{eq:union} turns the conjecture into the question whether every self-dual polymatroid on $n+1$ elements that satisfies all inequalities obtained from common information lies in $\Lab_{n+1}$. In this form the conjecture holds for up to four parties (Ref.~\cite{si:LMRW2013} and Proposition~\ref{thm:four}); for five parties, the inequalities obtained from a single common information do not suffice (Supplementary Note~\ref{app:five-oneci}), and we do not know whether those obtained from several do. Linden \emph{et al.} also noted that it was not clear whether stabilizer states satisfy the additional linear rank inequalities whose existence had been shown at the time~\cite[Sec.~5]{si:LMRW2013}; for prime local dimensions they satisfy all linear rank inequalities on the $n+1$ systems in balanced form (Corollary~\ref{cor:balanced}, whose ``only if'' direction is in essence due to Gross and Walter~\cite{si:GrossWalter2013}).
\section{Results and the repository}
\label{app:repo}

\subsection{Requirements, gates and verification scripts}
\label{app:repo-verify}

All data, certificates and code are in the repository of Ref.~\cite{si:CaoSong2026data}, whose release \reprelease{} is archived at Zenodo. Our code is under the MIT license and the data are under the CC BY 4.0 license; the repository also redistributes a binary of Normaliz under its own license (GPL-3), which is used only for cross-checks. The code requires Python~3.8 or later with NumPy and SciPy (and SymPy for the verification scripts below), a C compiler for the search programs, and, optionally, \texttt{lrs}~\cite{si:lrs} for the exact enumeration in Proposition~\ref{thm:four}. The command
\begin{center}
\texttt{python3 selftest.py -{}-full}
\end{center}
runs $59$ gate checks and must end with \texttt{ALL PASS}; it takes about $20$ minutes on a computer with two cores. The gates re-derive the stored certificates in exact rational arithmetic with a checker that is separate from the code that produced them, recompute stored realizations from their graph states (some gates only on a sample), and compare the data with a manifest of hashes; a few gates check stored data with the same modules that produced them. Without \texttt{lrs}, the exact enumeration of Proposition~\ref{thm:four} is skipped and only the stored rays are checked. The notes \path{docs/A2_stabilizer_vs_QLR.md} (the A2 note) and \path{docs/R32_reduction_and_local_dimension.md} (the R32 note) contain the original statements and proofs. They are working notes written during this work, mostly before this paper. Their numbering differs from ours (Supplementary Table~\ref{tab:repro}), and so does some notation: they write $E$, $\mathrm{LR}^{(p)}$, $\mathrm{LR}$, $\pi$ and $\Stab_n$ where we write $N$, $\LcZ{p}{m}$, $\Lc_m$, $\lambda$ and $\Stab^{\mathrm{pr}}_n$, and they measure the entropies of $\Z_p$ qudits in units of $\log p$.

The numbers quoted in the Results and in the Supplementary Information were recomputed by separate scripts (last column of Supplementary Table~\ref{tab:repro}). They were written for this paper, also with AI assistance (Sec.~\ref{app:repro-ai}), separately from the code that produced the results: they read the stored data but not the repository code, except for the transcriptions of the published inequalities of Supplementary Table~\ref{tab:eight}, which were checked against the cited sources term by term. Apart from comparisons with explicit state vectors, which are numerical, the scripts compute ranks, group orders and certificates exactly, in integer and rational arithmetic. The facet count of Supplementary Note~\ref{app:five-catalogue} uses ranks modulo the prime $1{,}000{,}003$; since the rank of an integer matrix over $\Q$ is at least its rank modulo any prime, rank $30$ modulo this prime gives $30$ linearly independent vectors on the hyperplane, which certifies the count. For the realizations of Sec.~\ref{sec:scale}, the script computes no Smith forms: for each cut block it computes the rank over $\Q$ and one nonzero minor of maximal order by fraction-free elimination, and the rank modulo each prime factor of this minor; since every prime that divides a Smith invariant of the block divides this minor, this determines the exceptional primes. The scripts are included in the archived release, in the directory \texttt{checks}; the results that are not in the preprint arXiv:2610.18905v1, in Secs.~\ref{sec:qss}, \ref{sec:six}, \ref{sec:dimension} and~\ref{sec:fourfive}, in the Discussion and Supplementary Note~\ref{app:discussion}, and in the Methods, are checked by \path{check_rev_sec4.py}, \path{check_rev_stab5.py}, \path{check_rev_sec5.py}, \path{check_rev_sec6.py}, \path{check_rev_qss.py}, \path{check_color_code.py} and \path{check_rev_blind_pm.py}. The scripts \path{check_thm5.py} and \path{check_rev_sec4.py} need only the Python standard library and no data files; the first checks Theorem~\ref{thm:five} except its facet statement, which the second checks. To check the results from the archive, unpack release \reprelease{} (Ref.~\cite{si:CaoSong2026data}), run \texttt{python3 selftest.py -{}-full} in its top directory, and run each Python script in the last column of Supplementary Table~\ref{tab:repro} with \texttt{python3} in the directory \texttt{checks}, where the scripts find the data (\path{check_sec3.py} compiles \path{q2_gf3.c} with \texttt{gcc}, and \path{check_sec4a.py} needs \texttt{lrs}); each ends with \texttt{ALL OK} if its checks pass, except \path{check_sec4b_facet.py}, which prints the minimum of the left side of~\eqref{eq:star} minus the right side over the vectors that it checks, which must be~$0$, and the rank of the vectors at which this minimum is attained, which must be~$30$.

\subsection{Where each result is}
\label{app:repo-where}

Table~\ref{tab:repro} locates every result of this paper in the repository of Ref.~\cite{si:CaoSong2026data} (release \reprelease{}): the working note that states it, the self-test gates that check it (Methods), and the separately written script that recomputes the numbers quoted here. The working notes number their results differently (Supplementary Note~\ref{app:repo-verify}).

\begin{table*}[t]
\caption{Where each result of this paper is stated in the repository, which self-test gates check it, and which of the scripts written separately for this paper recomputes the numbers quoted here. SN: Supplementary Note.}
\label{tab:repro}
\begin{ruledtabular}
\footnotesize
\begin{tabular}{p{3.9cm}p{6.1cm}p{2.7cm}p{4.0cm}}
\raggedright This paper &\raggedright  Repository note &\raggedright  Gates &\raggedright  Separate check\tabularnewline
\hline
\raggedright Prop.~\ref{thm:reduction}, Cors.~\ref{cor:doubling}, \ref{cor:divisibility} and~\ref{cor:balanced} &\raggedright  R32: Thm.~1 (\S2), \S3.1, Thm.~1(d) and Thm.~5 (\S4.6) &\raggedright  G29 &\raggedright  \path{check_sec3.py}, \path{check_sec3_lift.py}\tabularnewline
\raggedright Example~\ref{ex:q2} &\raggedright  A2: \S5, \S8 &\raggedright  G20, G28 &\raggedright  \path{check_sec3.py}, \path{q2_gf3.c}\tabularnewline
\raggedright Prop.~\ref{thm:four}, Table~\ref{tab:four} &\raggedright  R32: Thm.~7 (\S4.7) &\raggedright  G29, G29b, G33 &\raggedright  \path{check_sec4a.py}\tabularnewline
\raggedright Thm.~\ref{thm:five}, Table~\ref{tab:witness}, SN~\ref{app:five-star}, \ref{app:five-dfz} &\raggedright  A2: \S3, \S4 &\raggedright  G20, G22; none for the facet for every $d$ &\raggedright  \path{check_thm5.py} (no data), \path{check_sec4b.py}, \path{check_sec4b_facet.py}, \path{check_dfz6.py}; facet for every $d$: \path{check_rev_sec4.py}\tabularnewline
\raggedright Sec.~\ref{sec:scale}, SN~\ref{app:five-catalogue}, \ref{app:five-oneci} &\raggedright  A2: \S4, \S6, \S7 &\raggedright  G21, G23, G25, G26 &\raggedright  \path{check_sec4b.py}, \path{check_dfz6.py}, \path{check_realizations.py}\tabularnewline
\raggedright Prop.~\ref{prop:sandwich} &\raggedright  R32: Prop.~3 and its corollary (\S4.3), for primes &\raggedright  G30 &\raggedright  ---\tabularnewline
\raggedright Lemma~\ref{lem:JI} &\raggedright  R32: Lemma~4 (\S4.4), Lemma~4(G) (\S4.6) &\raggedright  G31b, G32b &\raggedright  \path{check_sec5.py}\tabularnewline
\raggedright Lemma~\ref{lem:cube-selfdual}, Thm.~\ref{thm:primes} &\raggedright  R32: Thms.~2 and~2$'$ (\S\S4.1--4.4) &\raggedright  G30, G31 &\raggedright  \path{check_sec5.py}\tabularnewline
\raggedright Suppl.\ Prop.~\ref{prop:onecut} &\raggedright  R32: Prop.~6 (\S4.7) &\raggedright  G33, G33b &\raggedright  \path{check_sec5.py}\tabularnewline
\raggedright Thm.~\ref{thm:divisibility}, Suppl.\ Prop.~\ref{prop:galois}, SN~\ref{app:cubes-z4} &\raggedright  R32: Thm.~5 (\S4.6) &\raggedright  G28, G32, G32b &\raggedright  \path{check_sec5.py}\tabularnewline
\raggedright Prop.~\ref{prop:blind}, Table~\ref{tab:eight} &\raggedright  R32: Prop.~8 (\S4.7) &\raggedright  G33, G34 &\raggedright  \path{check_sec6.py}; the first inequality of Ref.~\cite{si:PenaMacias2023}, Cor.~2.7, $n=3$ (SN~\ref{app:blind}): \path{check_rev_blind_pm.py}\tabularnewline
\raggedright Suppl.\ Lemma~\ref{lem:transfer}, Thm.~\ref{thm:steane}, Suppl.\ Prop.~\ref{prop:faces} &\raggedright  R32: Lemmas~0 and A--E, Thm.~9, faces (\S4.8) &\raggedright  G35, G35b &\raggedright  \path{check_sec6.py}\tabularnewline
\raggedright SN~\ref{app:misprints} &\raggedright  R32: \S4.5; \path{scripts/s28_generalized_cubes.py} with \texttt{-{}-dfz1401} &\raggedright  G34, G34b &\raggedright  \path{check_sec6.py}\tabularnewline
\raggedright Suppl.\ Lemma~\ref{lem:JIstar}; Thms.~\ref{thm:primes} ($n=2p+2$) and~\ref{thm:divisibility} &\raggedright  R32: Lemma~4$^*$, Thm.~17, Cor.~19 (\S4.13) &\raggedright  G42, G42b &\raggedright  \path{check_rev_sec5.py}\tabularnewline%
\raggedright Thm.~\ref{thm:reverse}, Prop.~\ref{prop:f12}, SN~\ref{app:reverse-g12} &\raggedright  R32: Thm.~12, Cor.~13 (\S4.11) &\raggedright  G39, G39b &\raggedright  \path{check_rev_sec6.py}\tabularnewline%
\raggedright Thm.~\ref{thm:twopowers} &\raggedright  R32: Thm.~14, Lemma~4(G$'$) (\S4.12) &\raggedright  G40, G40b &\raggedright  \path{check_rev_sec6.py}\tabularnewline%
\raggedright Sec.~\ref{sec:scale}, SN~\ref{app:five-catalogue} (realizations for every $d$; facets of $\StabZ{d}{5}$) &\raggedright  R32: \S4.9 (realizations); A2: \S4, R48 addendum (facets) &\raggedright  G36, G36b, G37 (realizations); none for the facets &\raggedright  \path{check_rev_stab5.py}\tabularnewline
\raggedright Lemma~\ref{lem:fano-odd}, Suppl.\ Lemmas~\ref{lem:qss-linear} and~\ref{lem:qss-css}, Prop.~\ref{prop:qss}, SN~\ref{app:qss}; distances (SN~\ref{app:margins}) &\raggedright  R32: Lemmas~20, 22, 27 and~31a, Props.~21, 23--25 and~28--31 (\S\S4.15, 4.17, 4.19--4.21), \S4.16 &\raggedright  G43, G43b (the comparison program of SN~\ref{app:qss-lp}); none for Lemmas~\ref{lem:qss-css} and~\ref{lem:fano-odd}, Prop.~\ref{prop:qss}~(i), (iii), (iv), the bound for decompositions and the upper bound $1/12$ &\raggedright  \path{check_rev_qss.py}\tabularnewline
\raggedright SN~\ref{app:cubes-sandwich} (code states of the simplex codes) &\raggedright  R32: Prop.~26 (\S4.18) &\raggedright  none &\raggedright  \path{check_rev_sec6.py}\tabularnewline
\raggedright SN~\ref{app:qss-color} (the seven-qudit color code; erasures) &\raggedright  --- &\raggedright  none &\raggedright  \path{check_color_code.py}; the erasure counts and the search over qubit codes: \path{check_rev_qss.py}, part (r)\tabularnewline
\raggedright Thm.~\ref{thm:ideal}, SN~\ref{app:qss-general} (codes that are not stabilizer codes) &\raggedright  --- &\raggedright  none (no computer-assisted step; the forcing lemma, as for Prop.~\ref{prop:qss}~(iv), is a cross-check) &\raggedright  the forcing lemma: \path{check_rev_qss.py}, part (k); the planes and good directions of $\mathrm{AG}(3,2)$ and $R_8$, (G1)--(G4): part (q); the numbers of qudits for large secrets: part (o); the $((2,3))$ threshold structure: part (s)\tabularnewline
\end{tabular}
\end{ruledtabular}
\end{table*}

Supplementary Table~\ref{tab:summary} lists the main results, what they build on, and their proof status.

\begin{table*}[t]
\caption{The main results, what they build on, and their proof status. ``Exact'' means that the proof rests on rational identities verified in exact arithmetic; ``computer-assisted'' means that part of the proof is a finite computation whose data and code are public (Methods). The access structures $\Gamma_{\mathrm F}$ and $\Gamma_{\mathrm{NF}}$ are those of Sec.~\ref{sec:qss}. SN: Supplementary Note.}
\label{tab:summary}
\begin{ruledtabular}
\footnotesize
\begin{tabular}{lp{1.9cm}p{4.8cm}p{4.9cm}p{3.5cm}}
Result &\raggedright  Parties &\raggedright  Statement &\raggedright  Builds on &\raggedright  Proof\tabularnewline
\hline
Thm.~\ref{thm:ideal} &\raggedright  seven players &\raggedright  codes with one qudit of dimension $D$ per site, and ideal schemes, stabilizer or not, with secret and shares of dimension $D$: for $\Gamma_{\mathrm F}$ iff $D$ is a power of $2$, for $\Gamma_{\mathrm{NF}}$ iff $D$ is odd; up to local unitaries, stabilizer states of qubits ($\Gamma_{\mathrm F}$; Steane's code for $D=2$) or of odd-dimensional qudits ($\Gamma_{\mathrm{NF}}$) &\raggedright  polymatroids with an ideal element~\cite{si:Matus2012}; a connected matroid is determined by its circuits through one element~\cite{si:Lehman1964}; equality in strong subadditivity~\cite{si:HJPW2004}; nice error bases~\cite{si:Knill1996}; symplectic bases; classical counterpart~\cite{si:Matus1999} (not used in the proof; its proof also gives Suppl.\ Lemma~\ref{lem:K-two}, SN~\ref{app:qss-ideal}); rigidity and uniqueness for $D=2$ (new) &\raggedright  proof (SN~\ref{app:qss-general}), without computer\tabularnewline
Prop.~\ref{prop:qss} &\raggedright  seven players &\raggedright  stabilizer secret sharing: ratio $\ge4/3$, also on average, for $\Gamma_{\mathrm F}$ in every odd local dimension and, for secrets of whole qudits, whenever $d$ has an odd prime factor; ratio $\ge4/3$ for $\Gamma_{\mathrm{NF}}$ with qubits, Galois qudits of characteristic $2$ or, for secrets of whole qudits, $d\equiv2\pmod4$; for $\Z_d$ qudits, ideal schemes exactly for $d=2$ ($\Gamma_{\mathrm F}$) and for odd $d$ ($\Gamma_{\mathrm{NF}}$); schemes with ratio $7/5$ for $\Gamma_{\mathrm F}$ ($d$ odd, or $d\equiv2\pmod4$ and $d\ne2$) and for $\Gamma_{\mathrm{NF}}$ ($d\equiv2\pmod4$) &\raggedright  linear schemes~\cite{si:JafariKhazaei2021,si:PenaMacias2023}; Lemma~\ref{lem:fano-odd}, for groups of odd order (new; it extends an argument of Ref.~\cite{si:JafariKhazaei2021}); $(\lambda,\omega)$-decompositions~\cite{si:vDJM1998}; quantum schemes from linear ones~\cite{si:Smith2000}; decomposition into prime powers; ideal schemes from identically self-dual matroids~\cite{si:SarvepalliRaussendorf2010}; representations by partitions~\cite{si:Matus1999}; $\GF(4)$ sub-schemes (new) &\raggedright  proof; the ideal case without a computer, with the forcing lemma by exact certificates as a cross-check; the $\GF(4)$ sub-scheme for $\Gamma_{\mathrm{NF}}$ checked by computer (SN~\ref{app:qss})\tabularnewline
Prop.~\ref{thm:reduction} &\raggedright  all $n$ &\raggedright  lift: $\StabZ{d}{n}\cong$ self-dual slice of $\LcZ{d}{n+1}$, every $d$ &\raggedright  lift of Mat\'u\v{s}~\cite{si:Matus2012}; entropies of stabilizer states~\cite{si:GrossWalter2013,si:LMRW2013}; for prime $d$, Ref.~\cite{si:Osthus2024} &\raggedright  proof (SN~\ref{app:reduction})\tabularnewline
Thm.~\ref{thm:steane} &\raggedright  $n\ge6$ &\raggedright  $\StabZ{2}{n}\not\subseteq\StabZ{d}{n}$ for every odd $d$ &\raggedright  Fano configuration; transfer through a derived family built from common informations (new) &\raggedright  proof; $\fnine$ expanded by computer (SN~\ref{app:transfer})\tabularnewline
Thm.~\ref{thm:reverse} &\raggedright  $n\ge6$ &\raggedright  $\StabZ{d}{n}\not\subseteq\StabZ{2}{n}$ for every $d\ge3$ &\raggedright  non-Fano configuration; transfer through the stabilizer group (new) &\raggedright  proof (SN~\ref{app:reverse})\tabularnewline
Thm.~\ref{thm:divisibility} &\raggedright  $n\ge2m+2$, $m=\max(d,d')$ &\raggedright  $\StabZ{d}{n}\subseteq\StabZ{d'}{n}$ if and only if $d\mid d'$ &\raggedright  inclusion for $d\mid d'$ is immediate; the converse rests on Lemma~\ref{lem:JI} and Suppl.\ Lemma~\ref{lem:JIstar} for groups (new) &\raggedright  proof (SN~\ref{app:lemma1}, \ref{app:cubes})\tabularnewline
Thm.~\ref{thm:five} &\raggedright  $n=5$ &\raggedright  $\Stab_5\subsetneq\QLR_5$; the separating inequality is a facet of $\StabZ{d}{5}$ for every $d$ &\raggedright  question of Ref.~\cite{si:BCHS2020}, answered here; published six-variable inequalities detect the gap as well~\cite{si:DFZ2009} &\raggedright  exact certificate; extremality and facet computer-assisted (SN~\ref{app:five})\tabularnewline
Thm.~\ref{thm:twopowers} &\raggedright  $n\ge6$ &\raggedright  $\StabZ{d}{n}\subseteq\StabZ{2^a}{n}$ if and only if $d\mid2^a$ &\raggedright  the same transfer with an inequality for groups of exponent $2^a$ (new) &\raggedright  proof (SN~\ref{app:reverse-powers})\tabularnewline
Thm.~\ref{thm:primes} &\raggedright  $n\ge2p+3$; $n\ge2p+2$ for odd $p$ &\raggedright  $\StabZ{p}{n}$, $\StabZ{q}{n}$ incomparable for primes $p<q$ &\raggedright  for $n\ge2p+3$, characteristic-dependent inequalities over fields~\cite{si:PenaSarria2021}; $n=2p+2$ via Suppl.\ Lemma~\ref{lem:JIstar} (new) &\raggedright  proof (SN~\ref{app:lemma1}, \ref{app:cubes})\tabularnewline
Prop.~\ref{prop:blind} &\raggedright  $n=6$ &\raggedright  eight characteristic-dependent inequalities hold on every self-dual polymatroid &\raggedright  six of them from Refs.~\cite{si:BKL2011,si:DFZ2015rate,si:PenaSarria2019}; blindness on the self-dual slice (new) &\raggedright  computer-assisted, exact certificates (SN~\ref{app:blind})\tabularnewline
Prop.~\ref{prop:sandwich} &\raggedright  all $n$ &\raggedright  non-inclusions pass from rank cones on $n$ elements to stabilizer cones, and from these to rank cones on $n+1$ elements &\raggedright  doubled codes (a standard purification of a classical code); their use to transfer non-inclusions (new); Prop.~\ref{thm:reduction} &\raggedright  proof (SN~\ref{app:cubes-sandwich})\tabularnewline
Prop.~\ref{thm:four} &\raggedright  $n\le4$ &\raggedright  $\StabZ{d}{n}$ is the same cone for every $d$ &\raggedright  four-party cone of all $d$ together~\cite{si:LMRW2013}; also from Refs.~\cite{si:BCHS2020,si:WalterWitteveen2021} with Cor.~\ref{cor:divisibility} &\raggedright  computer-assisted, exact (SN~\ref{app:four})\tabularnewline
SN~\ref{app:qss-color} &\raggedright  seven sites &\raggedright  erasures: for every erasure probability $0<\varepsilon<1/2$, the color code of odd local dimension fails less often than every qubit stabilizer code that encodes one qubit into seven with distance at least $3$ &\raggedright  graph states up to local Clifford unitaries~\cite{si:VandenNest2004}; we know of no earlier comparison &\raggedright  exhaustive enumeration, repeated by a verification script (SN~\ref{app:qss-color})\tabularnewline
SN~\ref{app:five-catalogue} &\raggedright  $n=5$ &\raggedright  of the $1{,}917{,}706$ orbits of the catalog, $641$ lie in $\StabZ{d}{5}$ for every $d$ and $1{,}155$ are undecided &\raggedright  $59$ orbits known for qubits~\cite{si:BCHS2020gap}; the other $582$ realizations and every $d$ (new) &\raggedright  computer-assisted: exact certificates and realizations (SN~\ref{app:five-catalogue})\tabularnewline
\end{tabular}
\end{ruledtabular}
\end{table*}

\clearpage

\end{document}